\documentclass{article}
\usepackage{amssymb}
\usepackage{titlesec}
\usepackage{graphicx}
\usepackage[hidelinks]{hyperref}
\usepackage{float}
\usepackage{mathtools}
\usepackage{afterpage}
\usepackage{amsmath}
\usepackage[export]{adjustbox}
\usepackage{graphicx}
\usepackage{subcaption}
\usepackage{dsfont}
\usepackage{enumitem}
\usepackage{algorithm}
\usepackage{algpseudocode}

\DeclarePairedDelimiter{\ceil}{\lceil}{\rceil}
\usepackage{color}
\usepackage[english]{babel}
\titleformat{\chapter}[display]
  {\normalfont\LARGE\bfseries}
  {\titleline{}\vspace{5pt}\titleline{}\vspace{1pt}%
  \MakeUppercase{\chaptertitlename} \thechapter}
  {1pc}
  {\titleline{}\vspace{0.5pc}} 
\usepackage{mathtools}
\DeclarePairedDelimiter\abs{\lvert}{\rvert}

\DeclareMathOperator*{\argmin}{arg\,min}
\DeclareMathOperator*{\argmax}{arg\,max}

\DeclareMathOperator{\R}{\mathbb{R}}
\DeclareMathOperator{\N}{\mathbb{N}}
\DeclareMathOperator{\E}{\mathbb{E}}

\newcommand{\footremember}[2]{%
    \footnote{#2}
    \newcounter{#1}
    \setcounter{#1}{\value{footnote}}%
}
\newcommand{\footrecall}[1]{%
    \footnotemark[\value{#1}]%
}

\makeatletter
\renewcommand\section{\@startsection {section}{1}{\z@}%
                               {-3.5ex \@plus -1ex \@minus -.2ex}%
                               {2.3ex \@plus.2ex}%
                               {\normalfont\large\bfseries}}
\renewcommand\subsection{\@startsection{subsection}{2}{\z@}%
                                 {-3.25ex\@plus -1ex \@minus -.2ex}%
                                 {1.5ex \@plus .2ex}%
                                 {\normalfont\bfseries}}

\makeatother

\usepackage[english]{babel} 
\usepackage{amsmath,amsfonts,amsthm} 

\usepackage{bold-extra}

\usepackage{fancyhdr}        
\numberwithin{equation}{section}       
\numberwithin{figure}{section}         
\numberwithin{table}{section}          

\newtheorem{theorem}{Theorem}[section]
\newtheorem{condition}{Condition}[section]
\newtheorem{prop}[theorem]{Proposition}
\newtheorem{corollary}[theorem]{Corollary}
\newtheorem{lemma}[theorem]{Lemma}
\theoremstyle{remark}
\newtheorem{remark}{Remark}[section]
\theoremstyle{definition}

\newtheorem{assumption}{Assumption}
\theoremstyle{definition}

\graphicspath{{img/}}

\title{
\normalfont \normalsize   
\huge Complexity of Markov Chain Monte Carlo for Generalized Linear Models
}
\author{Martin Chak\footremember{alley}{Department of Decision Sciences and BIDSA, Bocconi University}\footnote{martin.chak@unibocconi.it} \and Giacomo Zanella\footrecall{alley} \footnote{giacomo.zanella@unibocconi.it}}
\date{\normalsize\today}

\begin{document}
\maketitle

\begin{abstract}
Markov Chain Monte Carlo (MCMC), Laplace approximation (LA) and variational inference (VI) methods are popular approaches to Bayesian inference, each with trade-offs between computational cost and accuracy. 
However, a theoretical understanding of these differences is missing, particularly when both the sample size~$n$ and the dimension~$d$ are large. 
LA and Gaussian VI are justified by Bernstein-von Mises (BvM) theorems, and recent work has derived the characteristic condition~$n\gg d^2$ for their validity, improving over the condition~$n\gg d^3$. 
In this paper, 
we show for linear, logistic and Poisson regression that for~$n\gtrsim d$, MCMC attains the same complexity scaling in~$n,d$ as first-order optimization algorithms, up to sub-polynomial factors. 
Thus MCMC is competitive with LA and Gaussian VI in complexity, under a scaling between~$n$ and~$d$ more general than BvM regimes. 
Our complexities apply to appropriately scaled priors that are not necessarily Gaussian-tailed, including Student-$t$ and flat priors, with log-posteriors that are not necessarily globally concave or gradient-Lipschitz. 
\end{abstract}

\section{Introduction}
In Bayesian inference, 
a posterior distribution~$\pi$ on~$\mathbb{R}^d$ describes the uncertainty of a~$d$-dimensional parameter in a statistical model given~$n\in\mathbb{N}$ observations of data and prior belief about the parameter. 
Often~$\pi$ is not analytically tractable and approximations of~$\pi$ are used instead; three standard approaches for approximating~$\pi$ are LA~\cite{MR830567}, VI~\cite{MR3671776} and MCMC~\cite{MR2742422}. In the present work, we consider the computational complexity of MCMC to accurately approximate~$\pi$, with emphasis on generalized linear models (GLMs)~\cite{MR3223057} under the frequentist viewpoint for large~$d,n$ and in comparison to LA and Gaussian VI~\cite{MR4804813}.

LA and Gaussian VI aim to approximate~$\pi$ with a normal distribution on~$\mathbb{R}^d$. 
The mean and covariance of the respective ideal normal distributions are determined by the maximizer of the density of~$\pi$ in the former case and by the minimizer with respect to closeness to~$\pi$ in the latter. 
Computationally, these ideal approximations are themselves estimated by an optimization procedure, see~\cite[Section~3.4]{MR4804813} and references within for the computational costs. 
The two approaches are theoretically justified by BvM theorems (see e.g.~\cite[Section~10.2]{MR1652247}), which give conditions under which the rescaled posterior converges to a suitable normal distribution as~$n\rightarrow\infty$. 
The precise errors incurred by these ideal approximations have been analyzed in a series of papers~\cite{MR4923386,katsevich2024l,MR4935629,MR4804813}, in which the characteristic requirement
\begin{equation}\label{ndlv}
n\geq Cd^2\epsilon^{-k}
\end{equation}
was derived, where~$C>0$ is a model-dependent constant,~$\epsilon>0$ is a desired approximation error and~$k\in[2/3,2]$ is an explicit constant depending on the choice between some canonical error criteria. 
Condition~\eqref{ndlv} improves upon the more classical~$n\gg d^3$ scaling, and in~\cite{katsevich2024l}, condition~\eqref{ndlv} is shown to be necessary in the total variation sense for LA. 
These works assume that~$\pi$ admits a density~$\pi(dx)=e^{-U(x)}dx$ w.r.t.\ the Lebesgue measure for~$U:\Theta\rightarrow\mathbb{R}$ with~$\Theta\subseteq\mathbb{R}^d$ and satisfying various extents of smoothness, but they do not require strong convexity of~$U$, with~\cite{MR4923386} explicitly allowing nonconvexity arising from heavy tails in~$\pi$.

On the other hand, MCMC approximates~$\pi$ by simulating trajectories of random vectors on~$\mathbb{R}^d$ which have distributions converging to~$\pi$. Since Dalalyan's paper~\cite{MR3641401}, a number of works have provided nonasymptotic, explicit upper bounds for the closeness of these distributions to~$\pi$. These results bound 
the computational effort required for the distributions to be arbitrarily close to~$\pi$. 
Following~\cite{MR3641401}, a central focus in the literature has been on the case where~$\pi(dx)=e^{-U(x)}dx$ with~$U:\mathbb{R}^d\rightarrow\mathbb{R}$ strongly convex and admitting a globally Lipschitz gradient. The predominant example problem that fits squarely within this setting is Bayesian logistic regression with Gaussian priors, see~\cite[Section~6.2]{MR3641401} or~\cite[Section~5]{durmus2018highdimensional}. Logistic regression is also a primary focus in the works~\cite{MR4923386,katsevich2024l,MR4935629,MR4804813} on LA and Gaussian VI, but no assumption on~$n$ is required for the MCMC bounds. Instead, the bounds in this setting degenerate if the assumed strong convexity constant converges to zero, making them quantitatively vacuous for non-Gaussian priors without strong log-concavity. There are a number of quantitative bounds allowing for non-strong convexity (see e.g.~\cite[Section~7]{pmlr-v83-c},~\cite[Proposition~11]{pmlr-v195-f},~\cite[Theorem~D.7]{MR4763254},~\cite[Theorem~12]{MR3960927},~\cite{MR4577674}) and, further, for heavy-tailed~$\pi$ leading to nonconvex~$U$~(see e.g.~\cite{gb2,MR4723893,NEURIPS2024_77f33e8b}), but these bounds suffer in all cases from a considerably worsened dependence on~$d$ compared to the best known results under strong convexity~\cite{MR4763254,ascolani2024e,gouraud2023hmc,NEURIPS2019_eb86d510}.

\subsection{Summary of contributions}

In this paper, 
we show for GLMs (without strong convexity in the negative log-likelihood) that MCMC can achieve~$\epsilon$-accurate samples under more robust linear conditions between~$n$ and~$d$, 
with a computational effort that scales as in optimization (e.g.~\cite{10.5555/29}) 
up to sub-polynomial factors. 
More specifically, we consider GLMs with~$C^2(\mathbb{R}^d)$ log-posteriors, canonical link functions (linear, logistic and Poisson regression) and appropriately scaled priors (including Gaussian, Cauchy, Student-$t$, flat priors). For linear and logistic regression, under the condition
\begin{equation}\label{ndmc}
n\geq C(d\ln(d) + \ln(\epsilon^{-1})),
\end{equation}
with 
Gaussian design 
and feasible initializations, we show that the Gibbs sampler generates with high probability an~$\epsilon$-accurate sample in total variation with a number of floating point operations (flops) that 
is~$\tilde{O}(nd)$ as~$n,d,\epsilon^{-1}\rightarrow\infty$ subject to~\eqref{ndmc}. 
This~$\tilde{O}(nd)$ complexity in~$n,d,\epsilon$ (where~$\tilde{O}$ hides logarithmic factors and assuming that the complexity of conditional draws is~$\tilde{O}(n)$) is the same as for finding~$\epsilon$-accurate optimizers of~$\pi$ given initializers in 
polynomial proximity. 
Moreover, for logistic regression, we show for unadjusted Hamiltonian Monte Carlo (HMC) as in~\cite{bourabee2022unadjusted} that, to obtain with high probability an~$\epsilon$-accurate sample in scaled 2-Wasserstein distance, the sufficient condition~\eqref{ndmc} improves to
\begin{equation}\label{ndmc2}
n\geq C(d + \ln(\epsilon^{-1})).
\end{equation}
However, 
the required number of flops for HMC is instead~$\tilde{O}(n^{4/3}d\epsilon^{-2/3})$. 
This complexity includes the assumed~$O(nd)$ flops required for a gradient evaluation. 
In addition, for Poisson regression, 
we show that if~\eqref{ndmc} is strengthened to
\begin{equation}\label{ndmc3}
n\geq Cde^{C(\ln(d/\epsilon))^{1/2}},
\end{equation}
then we have the same~$\tilde{O}(nd)$ complexity for the Gibbs sampler, but in this case,~$\tilde{O}$ hides sub-polynomial factors. These results are summarized in Section~\ref{lrg2}. 

Our Gaussian design assumption serves to ensure that covariates are directionally well-distributed, which is also required in the aforementioned analysis on LA and Gaussian VI to obtain~\eqref{ndlv}. 
In the case of logistic regression with Gaussian priors, 
our results improve on established MCMC analysis by removing the dependence on the usual condition number (see~\eqref{kqdeq} below for a definition), which is characteristically linear in~$n$ for fixed~$d$, as we elaborate on in Section~\ref{lrg} below. 
For example, in regimes where~$n=\Omega(d^2)$ as in~\eqref{ndlv}, analysis in the literature yields an extra~$d$ factor through the condition number compared to the~$n=\Theta(d)$ regime. 
Our bounds instead depend on local condition numbers, which we show to be~$O(1)$ under data-generating assumptions. This improvement shows that the computational complexity of MCMC is generically competitive in terms of scaling in~$n,d$ with that of LA and Gaussian VI for logistic regression with Gaussian priors, in particular under~\eqref{ndlv}. 

From a practical viewpoint, our results show that in large-data settings with~$n$ not as large as~$d^2$, MCMC can be as computationally efficient as LA and Gaussian VI, without the possible loss of accuracy in the latter. This cost competitiveness persists for~$n$ larger than~$d^2$. 
On the other hand, 
our results do not provide a full comparison of the numerical constants. Moreover, we caution that 
our analysis concerns the mixing time/burn-in required under feasible initializations. For estimating integrals/expectations, this implies a complexity which is competitive in~$n,d$ with that of LA/VI, but a comprehensive comparison of the practical scaling jointly in~$n,d,\epsilon$ lies outside the scope of the paper.

We obtain our announced results by analyzing the behaviour of MCMC 
when strong convexity is assumed only in a bounded region, and developing conditions for neglectability outside of this region. 
In case of HMC, this analysis builds on well-known coupling techniques in strongly convex settings (see e.g.~\cite{bourabee2022unadjusted,gouraud2023hmc}) by deriving bounds on the probability and growth of trajectories escaping the bounded region. In case of the Gibbs sampler, we leverage recent results in~\cite{ascolani2024e} and use a maximal coupling argument alongside concentration properties of the conditional distributions. 
This analysis on MCMC dynamics is complemented by results on the posterior distribution in GLMs with canonical link functions. Under our data-generating assumptions and given~\eqref{ndmc}-\eqref{ndmc3}, we provide high-probability statements about local strong convexity and on measure concentration~\cite{MR1849347} into bounded regions, as demanded by the statements on HMC and Gibbs.

\subsection{Further connections to literature}

Our broad strategy is similar to~\cite{altmeyer2024p} and references within, namely~\cite{MR4828853,MR4721029}. The author of~\cite{altmeyer2024p} obtains polynomial sampling guarantees for Bayesian targets with only local convexity and smoothness, but without particular attention on optimality in the dependence on~$n,d$ (cf.~\cite[Theorem~3.7]{altmeyer2024p}) and with focus on the infinite dimensional setting (with a corresponding Gaussian sieve prior, see~(3.1) therein). 
Our work uses a stochastic dominance argument for our concentration statements, leveraging the (non-strong) log-concavity of the likelihood (without prior) in GLMs with canonical link functions. This accommodates 
different priors and requires instead~\eqref{ndmc}-\eqref{ndmc3} in producing the announced complexity bounds. More comparison with~\cite{altmeyer2024p,MR4828853,MR4721029} is made in Remark~\ref{ralt} below. We also mention here~\cite{MR4303259,MR3020951,N2e6} for related local approaches to obtain mixing times.

Recently, data augmentation algorithms have been analyzed in~\cite{ascolani2025m,lee2025f} to obtain nonasymptotic mixing time bounds in particular for logistic and probit regression with Gaussian priors. The authors of~\cite{lee2025f} compare their results with known bounds on unadjusted Langevin algorithm (ULA)~\cite{MR3960927}, underdamped Langevin Monte Carlo~\cite{NEURIPS2019_eb86d510,pmlr-v195-zhang23a}, the Metropolis adjusted Langevin algorithm of~\cite{MR4763254} (MALA) and show an improvement in mixing time, by considering a Gaussian prior setting with an~$O(nd)$ condition number, see Appendix~B therein. 
Their data-generating assumptions for this purpose are somewhat weaker than those considered in the present work; for example, BvM theorems cannot be expected to hold in the~$n\rightarrow\infty$ limit without further conditions. We note that 
their (strongly convex) assumptions on the prior do not necessarily accommodate the Zellner priors considered in~\cite[Section~6.2]{MR3641401},~\cite[Section~5]{durmus2018highdimensional} in terms of scaling, see Section~\ref{lrg} below. 
The authors of~\cite{ascolani2025m} consider data-generating assumptions more similar to (in fact a particular case of) ours, as well as Zellner priors with arbitrary scaling, see Section~2.3 therein. In both works, the resulting complexities are similar to those derived for the Gibbs sampler in~\cite{ascolani2024e} under the same respective assumptions. In particular, the complexity is~$O(\kappa)$ as~$\kappa\rightarrow\infty$, where~$\kappa$ is the condition number. 
As mentioned, our results remove dependence on this usual condition number~$\kappa$ of the negative log-likelihood and replace it with local condition numbers, which are~$O(1)$ as~$n,d\rightarrow\infty$ subject to~\eqref{ndmc}-\eqref{ndmc3} and under data-generating assumptions. They do not incur the extra~$O(nd)$ condition number dependence of~\cite{lee2025f}, given the difference in assumptions, nor the extra~$O(n/d)$ factor in~\cite[Section~2.3\footnote{For Corollary~2.5 therein, the~$O(n/d)$ factor arises for~$g$ as in Remark~\ref{remz}\ref{remz1} below.}]{ascolani2025m}. 

There are a number of quantitative works in the literature on stochastic gradient MCMC~\cite{pmlr-v139-zou21b} and higher order Langevin methods~\cite{dang2025highorderlangevinmontecarlo,mahajan2025p,NEURIPS2019_eb86d510}. These works take advantage of the sum structure and additional smoothness assumptions respectively in the log-likelihood in order to obtain improved bounds under strong convexity. It's not clear whether our analysis (targeting~$\pi$ rather than a surrogate as described above) 
can be extended to these methods to obtain any improvement in complexity. For example, in our proofs for HMC, the integration time is required to be large relative to the step-size, which excludes ULA, see Remark~\ref{remw}\ref{remw2}. On the other hand, arguments based on~$s$-conductance~\cite{MR4763254,andrieu2023we} appear compatible with our arguments on measure concentration. Therefore it seems possible to obtain analogous results for Random Walk Metropolis and MALA (we note the existing works~\cite{MR2533478,MR4758128} in the BvM regime), though we do not pursue that here. 

\subsection{Organization of the paper}
In Section~\ref{OGLM}, we state the announced results for GLMs rigorously. 
In Section~\ref{concensec}, we present general curvature conditions for the local concentration of a probability measure. In Sections~\ref{HMC} and~\ref{gibsec}, total variation and Wasserstein bounds are given for HMC and the Gibbs sampler respectively targeting a probability measure~$\hat{\pi}$ on~$\mathbb{R}^d$ that exhibits local curvature and global smoothness. 
In Section~\ref{csglm}, the curvature and smoothness conditions from the previous sections are verified with high probability for posteriors in GLMs under data-generating assumptions. Moreover, in that section, conditions are given for existence and boundedness of a critical point with high probability in the posterior of GLMs. These conditions are verified in particular for linear, logistic and Poisson regression therein. Section~\ref{poisec} is dedicated to Poisson regression, where the second derivative of the log-posterior grows exponentially, and where the previous analysis fails to yield the announced condition~$n\gtrsim d$. 

\paragraph*{Notation} The closed Euclidean ball of radius~$r$ centered at~$x\in\R^d$ is denoted~$B_r(x)$ with~$B_{\infty}(x) := \R^d$ and~$B_r:=B_r(0)$. The set of natural numbers including~$0$ is denoted~$\N $. 
The space of~$\mathbb{R}^{d\times d}$ positive definite matrices is denoted~$\mathbb{S}^{d\times d}$. 
The Hessian matrix of a function~$f\in C^2(\R^d)$ is denoted~$D^2f$. Unless otherwise stated, an element~$v\in\R^d$ is interpreted as a column vector. The identity matrix is denoted~$I_d\in\R^{d\times d}$. The minimum and maximum real eigenvalue of a matrix~$M$ are denoted~$\lambda_{\min}(M),\lambda_{\max}(M)$ respectively. The condition number of a matrix~$M$ is denoted~$\kappa(M) = \lambda_{\max}(M)/\lambda_{\min}(M)\in[-\infty,\infty]$. The notation~$\abs{\cdot}$ denotes the Euclidean norm for vectors and the operator norm for matrices acting under the Euclidean norm. The Gaussian density on~$\R^d$ with mean~$\mu\in\R^d$ and 
covariance~$\Sigma\in\R^{d\times d}$ 
is denoted~$\varphi_{\mu,\Sigma}:\R^d\rightarrow(0,\infty)$, namely~$\varphi_{\mu,\Sigma}(x) = (2\pi)^{-\frac{d}{2}}(\det(\Sigma))^{-\frac{1}{2}} e^{-\frac{1}{2}(x-\mu)^{\top}\Sigma^{-1}(x-\mu)} $ for all~$x\in\R^d$. 
The indicator function of a set~$S$ is denoted~$\mathds{1}_S$. 
The standard mollifier is denoted~$\varphi:\R^d\rightarrow[0,\infty)$, given by~$\varphi(x) = e^{-1/(1-\abs{x}^2)}/\int_{B_1}e^{-1/(1-\abs{y}^2)}dy$ for all~$x\in B_1$ and~$\varphi(x) = 0$ otherwise. For any~$\bar{\epsilon}>0$, the associated function~$\varphi_{\bar{\epsilon}}:\R^d\rightarrow[0,\infty)$ is defined by~$\varphi_{\bar{\epsilon}}(x) = \bar{\epsilon}^{-d}\varphi(x/\bar{\epsilon})$ for all~$x\in\R^d$. We adopt the convention that~$0\cdot\infty=0/0=0$. 
The uniform distribution on an interval~$[t_1,t_2]$ for~$t_1,t_2\in\mathbb{R}$ with~$t_1\leq t_2$ is denoted by~$\textrm{Unif}(t_1,t_2)$ and that on a set~$G$ is denoted by~$\textrm{Unif}(G)$. 
The geometric distribution 
with success probability~$\rho\in[0,1]$ is the distribution of the number of Bernoulli trials in an infinite i.i.d. sequence before and excluding the first success. For any function~$g:\mathbb{R}^d\rightarrow\mathbb{R}$,~$x=(x_1,\dots,x_d)\in\mathbb{R}^d$ and~$i\in\mathbb{N}\cap[1,d]$, we denote by~$x_{-i}$ the vector~$(x_1,\dots,x_{i-1},x_{i+1},\dots,x_d)$ and by~$g_i(\cdot,x_{-i})$ the function~$\mathbb{R}\ni y_i\mapsto g(x_1,\dots,x_{i-1},y_i,x_{i+1},\dots,x_d)$. For any probability measures~$\mu,\nu$ on a common measure space with~$\sigma$-algebra~$\mathcal{F}$, we denote~$\textrm{TV}(\mu,\nu)=\sup_{A\in\mathcal{F}}\abs{\mu(A)-\nu(A)}$. For any~$v\in\mathbb{R}^d$, we denote~$\abs{v}_{\infty}=\max_i\abs{v_i}$. For any~$q\in(0,1)\cup(1,\infty)$ and probability distributions~$\mu,\nu$, let~$\mathcal{R}_q(\mu||\nu)$ be the R\'enyi divergence of order~$q$ given by~$\mathcal{R}_q(\mu||\nu)=(q-1)^{-1}\ln\int(d\mu/d\nu)^qd\nu$ if~$\mu\ll\nu$ and~$\mathcal{R}_q(\mu||\nu)=\infty$ otherwise. 
Moreover, for any distributions~$\mu,\nu$ on~$\mathbb{R}^d$, let~$\mathcal{R}_1(\mu||\nu)=\lim_{q\rightarrow1^-}\mathcal{R}_q(\mu||\nu)$, which satisfies~$\mathcal{R}_1(\mu||\nu)=\textrm{KL}(\mu|\nu)$ by Theorem~5 in~\cite{MR3225930}, where
\begin{equation*}
\textrm{KL}(\mu|\nu) := \begin{cases}
\int_{\mathbb{R}^d}\ln(d\mu/d\nu)d\mu &\textrm{if }\mu\ll \nu\\
\infty &\textrm{otherwise}.
\end{cases}
\end{equation*}
We denote by~$W_2$ the~$L^2$ Wasserstein distance with respect to the Euclidean metric, defined for probability measures~$\mu,\nu$ on~$\mathbb{R}^d$ by
\begin{equation*}
W_2(\mu,\nu) = \bigg(\inf_{\pi\in\Pi(\mu,\nu)}\int_{\mathbb{R}^d\times\mathbb{R}^d}\abs{x-y}^2\pi(dx,dy)\bigg)^{1/2},
\end{equation*}
where~$\Pi(\mu,\nu)$ denotes the set of couplings between~$\mu$ and~$\nu$, namely the set of probability measures on~$\mathbb{R}^d\times\mathbb{R}^d$ with~$\pi(B\times \mathbb{R}^d) = \mu(B)$ and~$\pi(\mathbb{R}^d\times B)=\nu(B)$ for all Borel subsets~$B\subset\mathbb{R}^d$. 

\section{Overview for GLMs}\label{OGLM}

In this section, we provide an overview of our results for GLMs with Gaussian design. 
Firstly, in Section~\ref{lrg}, we revisit Bayesian logistic regression with Gaussian priors, with the goal of clarifying the scaling in~$n,d$ of the condition number under data-generating assumptions. In Section~\ref{lrg2}, we present precise statements corresponding to the announced results. 

We define a GLM with canonical link and random design. 
Let~$P_X$ be a probability distribution on~$\R^d$ with density~$p_X$ w.r.t.\ the Lebesgue measure and 
let~$(S,\mathcal{F},\eta)$ be a measure space with~$S\subseteq \mathbb{R}$. 
Let~$A\in C^2(\R )$ be convex. 
Let~$\phi>0$ and let~$\pmb{c}:S\rightarrow\mathbb{R}$ be a function. 
Throughout, we assume
that for any~$\theta\in\R^d$, 
the function
\begin{equation*}
S\times \mathbb{R}^d\ni (y,x)\mapsto 
e^{\phi^{-1}\cdot(y\theta^{\top}x - A(\theta^{\top}x)) + \pmb{c}(y)} p_X(x)
\end{equation*}
is measurable and integrable jointly and conditionally given~$x$ with integrals equal to one. 
Let~$P_{\theta}$ be a probability distribution on~$S \times\R^d$ with density
\begin{equation}\label{pitdef}
p_{\theta}(y,x) = e^{\phi^{-1}\cdot(y\theta^{\top}x - A(\theta^{\top}x)) + \pmb{c}(y)} p_X(x).
\end{equation}
Let~$\theta^*\in\mathbb{R}^d$ and let~$(Y_i,X_i)_{i\in\mathbb{N}}$ be an i.i.d. sequence of~$\mathbb{R}\times\mathbb{R}^d$-valued r.v.'s such that~$(Y_i,X_i)\sim P_{\theta^*}$. For any~$n\in\N $, 
denote~$Z^{(n)} = (Y_i,X_i)_{i\in[1,n]\cap\N }$. We also denote by~$X$ the~$\mathbb{R}^{n\times d}$ matrix with its~$i^{\textrm{th}}$ row given by~$X_i$. 

Let~$\pi\in C^2(\mathbb{R}^d)$. The function~$\pi$ hereafter denotes the 
prior density on~$\theta$. We allow~$\pi$ to be possibly improper, random (we have in mind only through~$X_i$, see e.g.~\eqref{zpl} below) and/or depend on~$n,d$. 

Assume that
the (random) function~$\pi(\cdot|Z^{(n)}):\R^d\rightarrow[0,\infty)$ defined by 
\begin{equation}\label{picdef}
\pi(\theta|Z^{(n)}) = \frac{\pi(\theta)\prod_{i=1}^n p_{\theta}(Y_i,X_i) }{\int_{\R^d}\pi(\theta)\prod_{i=1}^n p_{\theta}(Y_i,X_i) d\theta} \qquad\forall \theta\in\R^d.
\end{equation}
a.s.\ defines the density for a probability measure on~$\mathbb{R}^d$. 
We denote this (random) probability measure by~$\Pi(\cdot|Z^{(n)})$, called the posterior. 
Throughout, we assume~$\pi(\theta|Z^{(n)})>0$ for all~$\theta\in\mathbb{R}^d$.

\subsection{Condition number of logistic regression}\label{lrg}

Mixing time analysis for MCMC algorithms on strongly convex targets is often concerned with the dependence on condition number, in conjunction with dependence on dimension and error tolerance. State-of-the-art results achieve a linear dependence~\cite{MR4763254,ascolani2024e,gouraud2023hmc} on the condition number, with~\cite{altschuler2025s} achieving slightly better-than-linear dependence, but at the cost of worsened dimension dependence. We highlight in this section that a typical scaling for the condition number in logistic regression with Gaussian priors is~$n/d$. In particular in the~$n=\Theta(d^{1+\alpha})$ regime, where~$\alpha \in[0,\infty)$, the condition number is~$\Theta(d^{\alpha})$. Therefore in the regime~\eqref{ndlv} where LA and Gaussian VI are valid, existing MCMC analysis incur (at least) an additional~$d$ factor in complexity compared to first-order optimization methods. 

The logistic regression case is
\begin{align}\label{loge}
A(z) &= \ln(1+e^z),& \phi &= 1,& \pmb{c}(y)&= 0, & S&=\{0,1\}
\end{align}
for all~$z\in\mathbb{R}$,~$y\in S$ in~\eqref{pitdef}, with~$\eta$ the counting measure. We consider Gaussian priors on~$\theta$. 
Let~$Q$ be a~$\mathbb{S}^{d\times d}$-valued random element, to be the precision matrix of the Gaussian prior. 
Let~$n\in\mathbb{N}$ and consider the logistic regression negative log-density~$U:\Omega\times\mathbb{R}^d\rightarrow\mathbb{R}$ given by
\begin{equation*}
U(\theta) 
=\frac{\theta^{\top}Q\theta}{2}-\sum_{i=1}^n (Y_iX_i^{\top}\theta
- A(X_i^{\top}\theta)) 
=\frac{\theta^{\top}Q\theta}{2}-\sum_{i=1}^n (Y_iX_i^{\top}\theta
- \ln(1+e^{X_i^{\top}\theta})).
\end{equation*}
for all~$\theta\in\mathbb{R}^d$. 
Let~$\tilde{U}$ be the prior-preconditioned version of~$U$, i.e.
\begin{equation*}
\tilde{U}(\theta)=U(Q^{-1/2}\theta)
\end{equation*}
for all~$\theta\in\R^d$. Then, we have the following (folklore/well-known) expression for the condition number of $\tilde U$. 
Similar statements can be made about the scaling of the condition number of~$U$ rather than that of~$\tilde{U}$. 
Namely, if~$Q$ is suitably scaled, then both condition numbers scale like~$n/d$ as~$n,d\rightarrow\infty$ subject to~$n\geq Cd$ for a large enough constant~$C$. 
Note that the below Proposition~\ref{kqd} is an equality, rather than an upper bound, and holds for general~$Q$. 
We provide a proof in Appendix~\ref{ovp} for self-containedness (see also~\cite[Section~6.2]{MR3641401}).
\begin{prop}\label{kqd}
The condition number of $\tilde{U}$ is a.s.\ given by
\begin{equation}\label{kqdeq}
\kappa(\tilde{U}):=\frac{\sup_{\theta\in\R^d}\lambda_{\max}(D^2\tilde{U}(\theta))}{\inf_{\theta\in\R^d}\lambda_{\min}(D^2\tilde{U}(\theta))}=1+\frac{\lambda_{\max}(Q^{-1/2}X^{\top}XQ^{-1/2})}{4}.
\end{equation}
In particular, 
\begin{enumerate}[label=(\roman*)]
\item if~$Q = 3d(n\pi^2)^{-1}X^{\top}X\in\mathbb{S}^{d\times d}$, then~$\kappa(\tilde{U})=1+(\pi^2/12)n/d$. \item If instead~$Q=cI_d$ for some~$c>0$, then~$\kappa(\tilde{U}) = 1+c^{-1}\lambda_{\max}(X^{\top}X)/4$. 
\end{enumerate}
\end{prop}
\begin{remark}\label{remz}
\begin{enumerate}[label=(\roman*)]
\item \label{remz1} The choice~$Q = 3d(n\pi^2)^{-1}X^{\top}X$ is the Zellner prior from~\cite[Section~2.1]{MR3256057}, which is used in~\cite[Section~6.2]{MR3641401},~\cite[Section~5]{durmus2018highdimensional}. 
\item \label{remz2} In case~$Q=cI_d$ with~$c>0$, the condition number~$\kappa(\tilde{U})$ depends on~$P_X$. If~$P_X$ is~$N(0,I_d)$, 
which is the scaling of standardized data as~$n,d\rightarrow\infty$, then by Theorem~6.1 in~\cite{MR3967104}, with probability at least~$1-e^{-n/8}$, we have
\begin{equation*}
1+(\sqrt{n}-2\sqrt{d})^2/(16c)\leq \kappa(\tilde{U}) \leq 1+(2\sqrt{n}+\sqrt{d})^2/(4c),
\end{equation*}
which is~$O(n/d)$ if~$c=\Theta(d)$ as~$n,d\rightarrow\infty$ subject to say~$n\geq 5d$. 
If instead~$P_X$ is~$N(0,d^{-1}I_d)$, then the same result yields that it holds with probability at least~$1-e^{-n/8}$ that
\begin{equation*}
1+(\sqrt{n/d}-2)^2/(16c)\leq \kappa(\tilde{U}) \leq 1+(2\sqrt{n/d}+1)^2/(4c),
\end{equation*}
which is~$O(n/d)$ for~$c=\Theta(1)$ under the same regime on~$n,d$.
\end{enumerate}
\end{remark}

\subsection{Main results on GLMs}\label{lrg2}

In this section, we state our results for linear, logistic and Poisson regression. Linear and Poisson regression are the respective cases where
\begin{align}
A(z) &= z^2/2, &\phi&=\sigma^2, &\pmb{c}(y)&=-y^2/(2\sigma^2) - (1/2)\ln(2\pi\sigma^2), & S&=\mathbb{R},\label{line}\\
A(z) &= e^z, & \phi&= 1, & \pmb{c}(y)&= -\ln (y!), & S&= \mathbb{N}. \label{pois}
\end{align}
for all~$z\in\mathbb{R}$,~$y\in S$ in~\eqref{pitdef} and for some~$\sigma>0$, with~$\eta$ being either the Lebesgue or counting measure. The case of logistic regression has been given in~\eqref{loge}. We restrict to Gaussian covariates. 

\begin{assumption}[Gaussian design]\label{A1}
There exists~$\Sigma\in \mathbb{S}^{d\times d}$ with~$X_i\overset{\mathrm{iid}}{\sim}N(0,\Sigma)$. 
\end{assumption}

Next, we specify what is meant by appropriately scaled priors as stated in the introduction. 
Below, recall that~$D^2$ denotes the Hessian operator and~$\abs{\cdot}$ denotes the operator norm.
\begin{assumption}[Appropriately scaled priors]\label{asp}
There exist~$C_{\pi}\geq 0$,~$\rho_{\pi}\in[0,1)$ such that
\begin{equation*}
\mathbb{P}\bigg(\sup_{\theta\in\mathbb{R}^d}\abs{D^2 \ln\pi(\theta)}\leq C_{\pi}d\bigg) \geq 1-\rho_{\pi}.
\end{equation*}
\end{assumption}
Under Assumption~\ref{A1}, the value~$C_{\pi}$ 
in Assumption~\ref{asp} 
is allowed to depend on~$n,d,\Sigma$, but we have in mind that it is a constant with respect to~$n$ and depends on~$d$ only through~$\lambda_{\max}(\Sigma)$, in which case the announced results follow from the statements below. 
For example:
\begin{itemize}
\item Under Assumption~\ref{A1} and~$n\geq d$, the prior
\begin{equation}\label{zpl}
\pi(\theta) = \exp\bigg(-\frac{3d}{2n\pi^2}\abs{X\theta}^2\bigg)
\end{equation}
(see Remark~\ref{remz}\ref{remz1}), satisfies Assumption~\ref{asp} with~$C_{\pi}=27\lambda_{\max}(\Sigma)/\pi^2$ and~$\rho_{\pi}=e^{-n/2}$ by Theorem~6.1 in~\cite{MR3967104}.
\item The Gaussian prior
\begin{equation*}
\pi(\theta) = \exp(-\theta^{\top}Q\theta/2)
\end{equation*}
for~$Q\in\mathbb{S}^{d\times d}$ satisfies Assumption~\ref{asp} with~$C_{\pi}=\lambda_{\max}(Q)/d$ and~$\rho_{\pi}=0$. 
Under Assumption~\ref{A1} with for example~$\Sigma=I_d$, if~$Q=dI_d$, then we have~$C_{\pi}=1=\lambda_{\max}(\Sigma)$ and the results below yield the announced conclusions. Otherwise for~$\Sigma=d^{-1}I_d$, if~$Q=I_d$, then we have~$C_{\pi}=d^{-1}=\lambda_{\max}(\Sigma)$.
\item For any degree of freedom~$\nu >0$, the Student-$t$ prior is represented by
\begin{equation}\label{zs1}
\pi(\theta) = (1+\nu^{-1}\theta^{\top}Q\theta)^{-(\nu+d)/2}
\end{equation}
for~$Q\in\mathbb{S}^{d\times d}$. Under Assumption~\ref{A1}, if~$Q = \sigma n^{-1}X^{\top}X$ for some absolute constant~$\sigma>0$ and if~$n\geq d$, then~$n^{-1}X^{\top}X$ behaves like~$\Sigma$ with high probability in the sense of Theorem~6.1 in~\cite{MR3967104}, so that 
this prior satisfies Assumption~\ref{asp} with~$C_{\pi} = \bar{\sigma}\lambda_{\max}(\Sigma)\kappa(\Sigma)$ for some constant~$\bar{\sigma}>0$ depending only on~$\nu$ and with~$\rho_{\pi}=e^{-n/2}$. 
If instead~$Q = \sigma dn^{-1}X^{\top}X$ with again~$n\geq d$, which converges to the Gaussian counterpart~\eqref{zpl} as~$\nu\rightarrow\infty$ (if~$\sigma=3/\pi^2$), then Assumption~\ref{asp} is satisfied only with~$C_{\pi}=\bar{\sigma}d\lambda_{\max}(\Sigma)\kappa(\Sigma)$ (and~$\rho_{\pi}=e^{-n/2}$). This scaling of~$C_{\pi}$ with~$d$ will not yield the same~$n\gtrsim d$ conditions as announced, instead producing~$n\gtrsim d^2$ conditions. On the other hand, for such~$Q$, Assumption~\ref{asp} is recovered with~$C_{\pi} = \bar{\sigma}\lambda_{\max}(\Sigma)\kappa(\Sigma)$ if~$\nu=\Omega(d)$.
\item For any degree of freedom~$\nu >0$, the independent Student-$t$ prior is represented by
\begin{equation}\label{zs2}
\pi(\theta) = \prod_{i=1}^d (1+\nu^{-1}Q_i\theta_i^2 )^{-(\nu+1)/2},
\end{equation}
where~$Q_i>0$ for all~$i\in[1,d]\cap\mathbb{N}$.
The prior~\eqref{zs2} satisfies Assumption~\ref{asp} with~$C_{\pi} = (1+\nu^{-1})\max_iQ_i$ and~$\rho_{\pi}=0$. 
\item A flat prior~$\pi(\theta)=1$ satisfies Assumption~\ref{asp} with~$C_{\pi}=\rho_{\pi}=0$.
\end{itemize}

\subsubsection{Linear and logistic regression}
Before presenting our main results on MCMC, the following Lemma~\ref{mainmap} asserts the existence of a critical point~$\theta_{\textrm{map}}\in\mathbb{R}^d$ of~$\pi(\cdot|Z^{(n)})$ near the data-generating parameter~$\theta^*$. It will be necessary to make assumptions on the initialization of the MCMC algorithms in relation to~$\theta_{\textrm{map}}$.

\begin{lemma}\label{mainmap}
Let Assumptions~\ref{A1} and~\ref{asp} hold, denote by~$B_0$ the event in Assumption~\ref{asp} and assume~$0\in\argmax_{\theta\in\mathbb{R}^d} \pi(\theta)$. Assume~$\abs{\theta^*}^2\lambda_{\max}(\Sigma)$ and~$C_{\pi}/\lambda_{\max}(\Sigma)$ are bounded above by an absolute constant~$C^*>0$. 
Let~$\delta\in(0,1)$. 
For linear and logistic regression, 
there exist~$C>0$, which depends only on~$\sigma$ (polynomially) for linear regression and is an absolute constant for logistic regression, and an event~$B\subset B_0$ with~$\mathbb{P}(B)\geq 1-\delta-\rho_{\pi}$ such that if
\begin{equation*}
n\geq 
\begin{cases}C
\kappa(\Sigma)^2
(d+\ln(\delta^{-1}))(1+\ln(\kappa(\Sigma)\delta^{-1}d))&\textrm{for linear regression}\\
\textcolor{black}{C\kappa(\Sigma)^2(d+\ln(\delta^{-1}))}&\textrm{for logistic regression},
\end{cases}
\end{equation*}
then on~$B$ there exists a unique~$
\theta_{\textrm{map}}\in\mathbb{R}^d$ satisfying~$
\abs{\theta_{\textrm{map}}-\theta^*}
\leq (\phi/\lambda_{\max}(\Sigma))^{1/2}$ and
\begin{equation*}
\nabla \ln\pi(\cdot|Z^{(n)})|_{\theta_{\textrm{map}}} = 0.
\end{equation*}
\end{lemma}
\begin{remark}
\begin{enumerate}[label=(\roman*)]
\item 
The boundedness assumption on~$\abs{\theta^*}^2\lambda_{\max}(\Sigma)$ is analogous for example to~\cite[equation~(2)]{MR3984492}. 
It is a natural condition reflecting for example in logistic regression that the curvature of the log-likelihood near~$\theta^*$ can decay exponentially with (typical values of)~$\abs{\theta^*\cdot X_i}$. 
In any case, our analysis is explicit in the dependence on~$\abs{\theta^*}^2\lambda_{\max}(\Sigma)$. 
\item The boundedness assumption on~$C_{\pi}/\lambda_{\max}(\Sigma)$ has been verified for a number of priors and~$\Sigma$ just after Assumption~\ref{asp}. Again our analysis is explicit in the dependence on~$C_{\pi}/\lambda_{\max}(\Sigma)$. 
\end{enumerate}
\end{remark}

Our first main result on the Gibbs sampler given by Algorithm~\ref{alg1} 
is the following Theorem~\ref{gibbs1}. We will state the result under the setting and notations of Lemma~\ref{mainmap}. 
Theorem~\ref{gibbs1} applies, for a target accuracy~$\epsilon\in(0,1)$, when the number of iterations~$N$ satisfies
\begin{equation}\label{ngm}
N = \ceil{18\bar{C}d\kappa(\Sigma)\big(\ln[3\epsilon^{-1}]+\ln[
1+ d\ln(18d\bar{C}\kappa(\Sigma))]\big)},\\
\end{equation}
where~$L_{\textrm{sc}},\bar{C}$ are given by
\begin{subequations}\label{hmcsetc}
\begin{align}
L_{\textrm{sc}} &= ((d/n)\cdot C_{\pi}/\lambda_{\max}(\Sigma) + 9\phi^{-1}),\label{cbb}\\
\bar{C} &= \bigg(\inf_{z\in B_{2(C^*)^{1/2}+4}}A''(z)\bigg)^{-1}
\cdot 12\phi L_{\textrm{sc}}
\label{cba}
\end{align}
\end{subequations}
and~$C^*$ is the absolute constant from Lemma~\ref{mainmap}. 
Note that~$N=\tilde{O}(d)$ as~$n,d,\epsilon^{-1}\rightarrow\infty$, so that if the conditional draws in Algorithm~\ref{alg1} have complexity~$O(n)$, then the total complexity to obtain an~$\epsilon$-accurate sample is~$\tilde{O}(nd)$ as stated in the introduction. This~$O(n)$ complexity for the conditional draw is valid for all of the priors (with matrices~$Q$ as mentioned) stated after Assumption~\ref{asp}, where either the covariance structure is determined by~$X^{\top}X$, or it is element-wise. In these cases, the part of the computation that scales with~$d$ can be stored and reused from iteration to iteration, as explicitly promoted in~\cite{luu2025g}. 

In addition, a so-called feasible initialization is assumed, that is, 
\begin{equation}\label{fea}
\theta_0 = \theta_{\textrm{map}} + W \qquad \textrm{on~$B$ (from Lemma~\ref{mainmap}),}
\end{equation}
where~$W$ is an~$\mathbb{R}^d$-valued r.v. independent of~$Z^{(n)}$ and the draws in Algorithm~\ref{alg1}, such that~$W\sim N(0,(9L_{\textrm{sc}}\lambda_{\max}(\Sigma)n)^{-1}I_d)$.
Our analysis allows the distribution of~$\theta_0$ to instead be centered only close to~$\theta_{\textrm{map}}$ (and not necessarily at~$\theta_{\textrm{map}}$), but we forgo this generalization in favour of simplicity.

\begin{theorem}\label{gibbs1}
Assume the setting and notations of Lemma~\ref{mainmap}. 
Let~$\epsilon\in(0,1)$. 
For linear and logistic regression, there exists~$C>0$, which depends only on~$\sigma$ (polynomially) for linear regression and is an absolute constant for logistic regression, such that if
\begin{equation}\label{ngib}
n\geq C
\big(
\kappa(\Sigma)\ln(
\epsilon^{-1}) + 
\kappa(\Sigma)^2(d+\ln(\delta^{-1}))(1+\ln(\kappa(\Sigma)\delta^{-1}d))\big),
\end{equation}
and if in Algorithm~\ref{alg1}, the number of iterations~$N$ is given by~\eqref{ngm}-\eqref{hmcsetc}, 
then on~$B$ (from Lemma~\ref{mainmap}), 
given~$\theta_0$ satisfies~\eqref{fea}, it holds that
\begin{equation}\label{weg}
\textrm{TV}(\nu_N^{\textrm{Gibbs}},\pi(\cdot|Z^{(n)}))\leq \epsilon,
\end{equation}
where~$\nu_N^{\textrm{Gibbs}}$ denotes the conditional distribution of~$\theta_N$ in Algorithm~\ref{alg1} given~$Z^{(n)}$. 
\end{theorem}

We state next our main result on HMC, given by Algorithm~\ref{alg2} 
for linear and logistic regression. 
For a target accuracy~$\epsilon\in(0,1)$, Algorithm~\ref{alg2} will be considered with the following parameters:
\begin{subequations}\label{hmcset}
\begin{align}
&N = \ceil{35\bar{C}\kappa(\Sigma)\ln(178n/(L_{\textrm{sc}}\epsilon^2))},\label{hmcset1}\\
&h = \ceil{8\sqrt{2}(n/(L_{\textrm{sc}}\epsilon^2))^{1/3}(1+35\bar{C}\kappa(\Sigma))^{2/3}}^{-1}(8 L_{\textrm{sc}}\lambda_{\max}(\Sigma)n)^{-1/2},\label{hmcset2}\\
&T = (8L_{\textrm{sc}}\lambda_{\max}(\Sigma)n)^{-1/2}-h,\\
&\abs{\theta_0 - \theta_{\textrm{map}}}\leq (L_{\textrm{sc}}\lambda_{\max}(\Sigma))^{-1/2} \qquad\textrm{on~$B$ (from Lemma~\ref{mainmap})},\label{hmcset4}
\end{align}
\end{subequations}
where the constants~$L_{\textrm{sc}},\bar{C}>0$ are given by~\eqref{hmcsetc}. 
\begin{theorem}\label{mainhmc}
Assume the setting and notations of Lemma~\ref{mainmap}. 
Let~$\epsilon\in(0,1)$. 
For linear and logistic regression, there exists~$C>0$, which depends only on~$\sigma$ (polynomially) for linear regression and is an absolute constant for logistic regression, such that if 
\begin{align}\label{nhmc}
n&\geq C
\kappa(\Sigma)^2
\big((d+\ln(\delta^{-1}))\cdot \tilde{C}+\kappa(\Sigma)(\ln(\epsilon^{-1})+\ln\kappa(\Sigma))\big),\\
\textrm{where}\qquad 
\tilde{C}&:=
\begin{cases}
1+\ln(\kappa(\Sigma)\delta^{-1}d))&\textrm{for linear regression}\\
1&\textrm{for logistic regression},
\end{cases}\nonumber
\end{align}
then on~$B$ (from Lemma~\ref{mainmap}), 
given~\eqref{hmcset} with~\eqref{hmcsetc}, 
it holds that 
\begin{equation}\label{wep}
(\lambda_{\max}(\Sigma)n)^{1/2} W_2(\nu_N^{\textrm{HMC}},\pi(\cdot|Z^{(n)}))\leq \epsilon,
\end{equation}
where~$\nu_N^{\textrm{HMC}}$ denotes the conditional distribution of~$\theta_N$ in Algorithm~\ref{alg2} given~$Z^{(n)}$. 
\end{theorem}
\begin{remark}
\begin{enumerate}[label=(\roman*)]
\item The rescaling factor on the Wasserstein distance in~\eqref{wep} corresponds to the square root of the smoothness constant of target log-density. 
This rescaling factor appears in the literature in the form of~$\hat{m}^{1/2}$ for a global convexity constant~$\hat{m}$ (see e.g.~\cite{MR4763254}), which is not available here for non-Gaussian priors. 
\item The number of gradient evaluations required to implement Algorithm~\ref{alg2} with~\eqref{hmcset} (assuming~$\theta_0$ satisfying~\eqref{hmcset4} is available) is~$NT/h$, which is~$\tilde{O}(\kappa(\Sigma)^{5/3}n^{1/3}\epsilon^{-2/3})$ as~$n,d,\kappa(\Sigma)\rightarrow\infty$ subject to~\eqref{nhmc}. 
Assuming that the complexity of evaluating a gradient of the log-density in terms of flops is~$O(nd)$, the total complexity in terms of flops is~$\tilde{O}(\kappa(\Sigma)^{5/3}n^{4/3}d\epsilon^{-2/3})$. 
\end{enumerate}
\end{remark}

\subsubsection{Poisson regression}\label{poimse}
In case of Poisson regression, complexity bounds are given for the Gibbs sampler, but not HMC. Here, we incur an additional super-polylogarithmic but sub-polynomial factor in~$d,\kappa(\Sigma)$, both in the condition on~$n$ and in the number of iterations required in the sampler. 

Analogous to Lemma~\ref{mainmap}, we first state the existence of a critical point~$\theta_{\textrm{map}}\in\mathbb{R}^d$ of~$\pi(\cdot|Z^{(n)})$ near~$\theta^*$ for Poisson regression. 
\begin{lemma}\label{mainmap2}
Let Assumptions~\ref{A1} and~\ref{asp} hold 
and assume~$0\in\argmax_{\theta\in\mathbb{R}^d} \pi(\theta)$. Assume~$\abs{\theta^*}^2\lambda_{\max}(\Sigma)$ and~$C_{\pi}/\lambda_{\max}(\Sigma)$ are bounded above by an absolute constant. 
Let~$\delta\in(0,1)$. 
For Poisson regression, 
there exist an absolute constant~$C\geq 1$ and an event~$B
$ with~$\mathbb{P}(B)\geq 1-\delta-\rho_{\pi}$ such that if
\begin{equation*}
n\geq C\kappa(\Sigma)^2
(d+\ln(\delta^{-1})) e^{C(\ln(\kappa(\Sigma)d\delta^{-1}))^{1/2}},
\end{equation*}
then on~$B$ there exists a unique~$
\theta_{\textrm{map}}\in\mathbb{R}^d$ satisfying~$
\abs{\theta_{\textrm{map}}-\theta^*}
\leq (\lambda_{\max}(\Sigma))^{-1/2}$ and
\begin{equation*}
\nabla \ln\pi(\cdot|Z^{(n)})|_{\theta_{\textrm{map}}} = 0.
\end{equation*}
\end{lemma}

Under the setting and notations of Lemma~\ref{mainmap2}, the Gibbs sampler will be considered with a number of iterations~$N\in\mathbb{N}$ and a feasible initialization
\begin{equation}\label{fea2}
\theta_0=\theta_{\textrm{map}} + \bar{W} \qquad\textrm{on~$B$ (from Lemma~\ref{mainmap2})},
\end{equation}
where~$\bar{W}$ is an~$\mathbb{R}^d$-valued r.v. independent of~$Z^{(n)}$ and the draws in Algorithm~\ref{alg1}, such that 
\begin{equation}\label{lry2}
\bar{W}\sim N(0,\tilde{L}^{-1}I_d),\qquad \tilde{L}= C e^{C(\ln(nN/(\epsilon\delta)))^{1/2}} \cdot n\lambda_{\max}(\Sigma),
\end{equation}
for an absolute constant~$C\geq1$. 
The number of iterations~$N$ will be 
\begin{equation}\label{nqw}
N=
\ceil{C\kappa(\Sigma) d e^{C(\ln(\kappa(\Sigma)dn\delta^{-1}\epsilon^{-1}))^{1/2} }},
\end{equation}
which is~$\tilde{O}(d)$ as~$d,n,\delta^{-1},\epsilon^{-1}\rightarrow\infty$. For technical reasons, we will assume that the prior~$\pi$ is deterministic (it does not depend on the data~$Z^{(n)}$).

\begin{theorem}\label{mainpoi}
Assume the setting and notations of Lemma~\ref{mainmap2}. Suppose~$\pi$ is deterministic. 
Let~$\epsilon\in(0,1)$. 
For Poisson regression, there exists an absolute constant~$C\geq 1$ such that if
\begin{equation}\label{nsd0}
n\geq C\kappa(\Sigma)^2de^{C(\ln(\kappa(\Sigma)d\delta^{-1}\epsilon^{-1}))^{1/2}}
\end{equation}
and if in Algorithm~\ref{alg1}, the number of iterations is~$N$ given by~\eqref{nqw} 
and the feasible initialization~\eqref{fea2}-\eqref{lry2} is used, 
then~\eqref{weg} holds with probability at least~$1-\delta$, where~$\nu_N^{\textrm{Gibbs}}$ is the conditional distribution of~$\theta_N$ in Algorithm~\ref{alg1} given~$Z^{(n)}$.
\end{theorem}

\section{Concentration under curvature profiles}\label{concensec}
We begin our technical considerations related to the proof techniques. In this section, we study a general probability measure~$\hat{\pi}$ on~$\mathbb{R}^d$ with an unnormalized density~$e^{-U}$ for~$U:\mathbb{R}^d\rightarrow\mathbb{R}$. In particular we don't assume~$\hat{\pi}=\pi(\cdot|Z^{(n)})$. 
For 
many arguments in this paper, we require a concentration property on some~$\hat{\pi}$, under the local curvature Condition~\ref{cond:curvature3} below on~$U$. This property controls the contributions of the tails in the distribution, which allows to neglect regions with poor curvature, and to show this is the content of the present section. 
The main concentration property to be established for~$\hat{\pi}$ under Condition~\ref{cond:curvature3} is 
Lemma~\ref{tvv}. 

\begin{condition}[Curvature profile without prior]\label{cond:curvature3}
$U,U_0\in C^1(\mathbb{R}^d)$,~$\inf_{\mathbb{R}^d}U_0 >-\infty$ 
and there exists non-increasing~$m_0:[0,\infty)\rightarrow[0,\infty)$ 
such that~$\int_0^{\infty}m_0\neq 0$ 
and~$(\theta/\abs{\theta})\cdot(\nabla U(\theta) - \nabla U_0(\theta))\geq \int_0^{\abs{\theta}}m_0$ 
for all~$\theta\in\mathbb{R}^d\setminus\{0\}$. 
\end{condition}
Condition~\ref{cond:curvature3} can be thought of as a positive lower bound on the Hessian of a `prior-less' negative log-likelihood (but note that it is weaker: for example it only concerns radial curvature). 
Note that the intermediate value theorem implies in the condition that
\begin{equation*}
\nabla U(0) - \nabla U_0(0)=0, 
\end{equation*}
namely~$U-U_0$ admits a minimum at~$0$. 
The introduction of~$U_0$ in Condition~\ref{cond:curvature3} is motivated by our considerations on Bayesian inference, but our statements below in this section will be about~$U$ as opposed to~$U-U_0$. In the context of~\eqref{picdef}, 
a characteristic choice (for the sake of intuition) is
\begin{equation*}
U(\theta) - U_0(\theta) = -\ln (\textstyle \prod_{i=1}^n p_{\theta + \theta_{\textrm{mle}}}(Y_i,X_i)),\qquad 
U_0(\theta) = -\ln\pi(\theta + \theta_{\textrm{mle}}) 
\end{equation*}
for a critical point~$\theta_{\textrm{mle}}=\theta_{\textrm{map}}$ as in Lemmata~\ref{mainmap} and~\ref{mainmap2} applied with a flat prior~$\pi=1$. However, it will be convenient to instead set~$U_0$ as a version of~$-\ln\pi$ that is flattened only on a ball around the origin (as in the proof of Corollary~\ref{nvc}). 
There is no assumption made on the form of~$U_0$ in Condition~\ref{cond:curvature3}, but its choice will affect the bounds appearing in the results below. We make more comments about~$U_0$ in Remark~\ref{remmn} below.

First, we state the following Lemma~\ref{lemma:stoch_ord}, which provides a sufficient condition for stochastic ordering among two one-dimensional probability densities. A proof of Lemma~\ref{lemma:stoch_ord} can be found in Lemma~5 of~\cite[Supplementary Material]{MR4589070}, where the underlying space is~$\mathbb{R}$, but the proof follows (almost) verbatim for~$[0,\infty)$. 
\begin{lemma}\label{lemma:stoch_ord}
Let $f$ and $\bar{f}$ be two pdf's on $[0,\infty)$ such that $\bar{f}(x)/f(x)$ is non-decreasing in $x$.
Then $\int g(x)f(x)dx\leq \int g(x)\bar{f}(x)dx$ for every non-decreasing $g:[0,\infty)\to\R$.
\end{lemma}

Next, we show in Lemma~\ref{tvw} below that under Condition~\ref{cond:curvature3}, in order to obtain concentration properties for the prior-less probability measure~$\hat{\pi}_{-0}$ given by~$\hat{\pi}_{-0}(dx)\propto\exp(-(U(x)-U_0(x)))dx$, it suffices to have concentration properties for the negative log-density that has exactly the curvature profile~$m_0$.
\begin{lemma}\label{tvw}
Assume Condition~\ref{cond:curvature3} 
and 
denote by~$\hat{\pi}_{-0},\bar{\pi}_{-0}$ the probability measures on~$\mathbb{R}^d$ given by~$\hat{\pi}_{-0}(dx)\propto\exp(-(U(x)-U_0(x)))dx$ and
\begin{equation}\label{pba}
\bar{\pi}_{-0}(dx)\propto\exp\bigg(
-\int_0^{\abs{x}}\int_0^tm_0(s)dsdt\bigg)dx.
\end{equation}
Let~$g:[0,\infty)\to[0,\infty)$ be a non-decreasing function. 
It holds that
\begin{equation*}
\int_{\mathbb{R}^d} g(\abs{x}) \hat{\pi}_{-0}(dx) \leq \int_{\mathbb{R}^d} g(\abs{x}) \bar{\pi}_{-0}(dx). 
\end{equation*}
\end{lemma}
\begin{proof}
Let~$f,\bar{f}:[0,\infty)\rightarrow[0,\infty)$ be the pdf's of~$|X|$ for~$X$ with distribution~$\hat{\pi}_{-0}$ and~$\bar\pi_{-0}$ respectively.
The proof consists in showing that $\bar{f}$ stochastically dominates $f$ using Lemma~\ref{lemma:stoch_ord}. 
Let~$v:[0,\infty)\rightarrow[0,\infty)$ 
be given by~$v(r)=\int_{0}^{r}\int_{0}^tm_0(s)ds dt$. 
Using standard change of variable with polar coordinates one has, for every~$r\in[0,\infty)$,
\begin{align*}
&-\ln f(r)+\ln \bar{f}(r)\\
&\quad=-\ln \E[\exp(-(U(rW) -U_0(rW)))]+\ln \E[\exp(-v(|rW|))]+c_1\\
&\quad= -\ln \E[\exp( -(U(rW)-U_0(rW)))] - v(r)
+ c_1,
\end{align*}
where $c_1\in\R$ is a constant independent of $r$ and~$W\sim \textrm{Unif}(\mathbb{S}^{d-1})$. Let~$\tilde{U}:[0,\infty)\rightarrow\mathbb{R}$ be given by~$\tilde U(r)=-\ln \E[\exp(-(U(rW)-U_0(rW)))]$. 
We have for any~$r\in[0,\infty)$ that
\begin{align*}
\tilde U'(r)&=\frac{\E[((W\cdot\nabla )U(rW) - (W\cdot\nabla )U_0(rW)) \exp(-U(rW)+U_0(rW))]}{\E[\exp(-U(rW)+U_0(rW))]}\\
&\geq 
\min_{w\in \mathbb{S}^{d-1}}
((w\cdot\nabla )U(rw) - (w\cdot\nabla )U_0(rw))\\
&\geq \textstyle \int_0^rm_0(t)dt,
\end{align*}
where the second inequality follows from Condition~\ref{cond:curvature3}. 
It follows that $-\ln f(r)+\ln \bar{f}(r)$, and thus also $\bar{f}(r)/f(r)$, is non-decreasing in $r$.
The result then follows from Lemma~\ref{lemma:stoch_ord}.
\end{proof}

Having shown in Lemma~\ref{tvw} that it suffices to work with~$\bar{\pi}_{-0}$ given by~\eqref{pba}, it remains to 
relate~$\hat{\pi}$ to~$\hat{\pi}_{-0}$ and to 
obtain the desired property for~$\bar{\pi}_{-0}$. Here, we replace~$m_0$ with~$m_0(r)\mathds{1}_{[0,r]}$ for arbitrary~$r\in[0,\infty)$, so that~$\bar{\pi}_{-0}$ is explicit. 

\begin{lemma}\label{tvv}
Assume Condition~\ref{cond:curvature3} and~$\int_{\mathbb{R}^d}e^{-U+U(0)}\geq \int_{\mathbb{R}^d}e^{-L\abs{\cdot}^2/2}$ for some~$L>0$. 
Let~$r\geq 0$, denote~$c = m_0(r)$, assume~$c>0$ and 
\begin{equation}\label{eq0}
r\sqrt{
c/d}\geq 3.
\end{equation}
The probability measure~$\hat{\pi}$ 
satisfies for any~$i\in\mathbb{N}$ that
\begin{align}
\int_{\mathbb{R}^d\setminus B_r}\abs{x}^i\hat{\pi}(dx) &< \bigg(\sup_{\mathbb{R}^d\setminus B_r} e^{-U_0}\bigg)\Big( \inf_{B_r}e^{-U_0} \Big)^{-1}\cdot \varrho\bigg(1+\varrho\bigg(\frac{L}{
c}\bigg)^{d/2}\,\bigg)\bigg(6r\bigg(1+\frac{i}{d}\bigg)\bigg)^i,\label{tvveq1}\\
\textrm{with}\quad\varrho&:= e^{2d-3
cr^2/8 }.
\end{align}
\end{lemma}
\begin{remark}\label{remmn}
\begin{enumerate}[label=(\roman*)]
\item \label{pieo}
In the context of GLMs, if the prior is log-concave, then it would be convenient to take~$U_0=0$ in Lemma~\ref{tvv}. The inclusion of the case~$U_0\neq0$ accommodates heavy-tailed priors, where the toy potential constructed from the curvature~$m_0$ (see the proof of Lemma~\ref{tvv}) 
cannot account for negative curvature.
\item \label{comn} Again for GLMs, to give some intuition for how the quantities in~\eqref{tvveq1} scale, the characteristic square radius~$
cr^2$ will (ideally) be set to scale like the number of datapoints~$n$ (up to logarithmic terms), so that under~$n\gtrsim d$, 
the right-hand side of~\eqref{tvveq1} is exponentially small in~$n$ given that the other factors on the right-hand side of~\eqref{tvveq1} are well-behaved. 
Note that taking~$
cr^2=\Theta(n)$ together with~$n\gtrsim d$ is consistent with condition~\eqref{eq0}. 
\end{enumerate}
\end{remark}
\begin{proof}[Proof of Lemma~\ref{tvv}]
We have for any~$i\in\mathbb{N}$ that
\begin{equation}\label{uda}
\frac{\int_{\mathbb{R}^d\setminus B_r}|x|^ie^{-U(x)}dx}{\int_{\mathbb{R}^d}e^{-U(x)}dx} \leq 
\bigg(\sup_{\mathbb{R}^d\setminus B_r}e^{-U_0}\bigg) \cdot A_1 \cdot A_2,
\end{equation}
where
\begin{equation*}
A_1 := \frac{\int_{\mathbb{R}^d\setminus B_r}|x|^ie^{-U(x)+U_0(x)}dx}{\int_{\mathbb{R}^d}e^{-U(x) + U_0(x)}dx},\qquad A_2 := \frac{\int_{\mathbb{R}^d}e^{-U(x) + U_0(x)}dx}{\int_{\mathbb{R}^d}e^{-U(x)}dx},
\end{equation*}
and~$A_1,A_2 <\infty$ by Condition~\ref{cond:curvature3}. 
Let~$f:\R^d\rightarrow[0,\infty)$ be given by
\begin{equation}\label{fdef}
f(x) = \begin{cases}
c\abs{x}^2/2 &\textrm{if } \abs{x} \leq r,\\
cr^2/2 + 
cr(\abs{x} - r) &\textrm{otherwise}.
\end{cases}
\end{equation}
By Lemma~\ref{tvw} with~$m_0$ replaced by~$m_0(r)\mathds{1}_{[0,r]}$, the quantity~$A_1$ satisfies
\begin{equation}\label{a1e}
A_1 \leq \frac{\int_{\mathbb{R}^d\setminus B_r}|x|^i e^{-f(x)}dx}{\int_{\mathbb{R}^d}e^{-f}} \leq \frac{\int_{\mathbb{R}^d\setminus B_r}|x|^i e^{-f(x)}dx}{\int_{B_r}e^{-f}}.
\end{equation}
On the other hand, for~$A_2$, by Condition~\ref{cond:curvature3} and intermediate value theorem, we have~$\nabla U(0) - \nabla U_0(0) = 0$ and~$e^{-U+U_0}\leq e^{-f}\cdot e^{-U(0)+U_0(0)}$. 
Therefore by our assumption on~$\int e^{-U+U(0)}$, the quantity~$A_2$ satisfies
\begin{align}
A_2 &\leq \frac{\int_{\mathbb{R}^d\setminus B_r} e^{-U+U_0} }{\int_{\mathbb{R}^d} e^{-U}} + \frac{\int_{B_r} e^{-U+U_0}}{\int_{B_r} e^{-U}}\nonumber\\
&\leq \frac{e^{U_0(0)}\int_{\mathbb{R}^d\setminus B_r} e^{-U+U_0+U(0)-U_0(0)} }{\int_{\mathbb{R}^d} e^{-U+U(0)}} + \frac{\int_{B_r} e^{-U+U_0}}{\inf_{B_r}e^{-U_0}\int_{B_r} e^{-U+U_0}}\nonumber\\
&\leq \frac{e^{U_0(0)}\int_{\mathbb{R}^d\setminus B_r} e^{-f} }{\int_{\mathbb{R}^d} e^{-L|\cdot|^2/2}} + \Big(\inf_{B_r}e^{-U_0}\Big)^{-1}.\label{a2e}
\end{align}
It remains to estimate~$\int_{\mathbb{R}^d\setminus B_r}|\cdot|^ie^{-f}$ and~$\int_{B_r} e^{-f}$. 
By Theorem~5.2 in~\cite{MR1849347} (Gaussians satisfy log-Sobolev inequalities (LSI)) and Theorem~5.3 in~\cite{MR1849347} with~$F = \abs{\cdot}$ (concentration for Lipschitz functions under LSI), then using~\eqref{eq0}, 
we have
\begin{align}%
\int_{B_r}e^{-f} &= \bigg(\frac{2\pi}{
c}\bigg)^{d/2}\int_{B_r} \varphi_{0,
c^{-1}I_d}\nonumber\\
&\geq \bigg(\frac{2\pi}{
c}\bigg)^{d/2} \bigg(1-\exp\bigg(-\frac{
c}{2}\bigg(r-\sqrt{\frac{d}{
c}}\,\bigg)^2\bigg)\bigg)\nonumber\\
&\geq (2\pi/
c)^{d/2}/2, \label{eqs}
\end{align}
where we have used~$\int_{\R^d}\abs{x} \varphi_{0,
c^{-1}I_d}(x)dx \leq \sqrt{d/
c}$ by Jensen's inequality,
then~$d\geq 1$ and~\eqref{eq0} 
for the last inequality. 
On the other hand, it holds for any~$i\in\N $ that
\begin{equation}\label{eq1}
\int_{\R^d\setminus B_r} \abs{x}^ie^{-f(x)}dx
= \frac{d\pi^{d/2}}{\Gamma(d/2+1)} \int_r^{\infty}e^{-\frac{
cr^2}{2} - 
cr(x-r)} x^{d-1+i} dx.
\end{equation}
For the integral on the right-hand side, we have by integrating by parts
\begin{equation}
I_d:=\int_r^{\infty} e^{-
crx} x^{d-1+i}dx = \frac{r^{d-1+i}}{
cr} e^{-
cr^2} + \frac{d-1+i}{
cr}\int_r^{\infty} e^{-
crx} x^{d-2+i} dx 
\end{equation}
and so on to obtain
\begin{equation}\label{eq2}
I_d = e^{-
cr^2}\sum_{j=1}^{d+i}\frac{r^{d-j+i}(d-1+i)!}{(
cr)^j(d-j+i)!}.
\end{equation}
By~\eqref{eq0} and then the inequality~$j!\leq j^{j+1}e^{1-j}$ for~$j\in\N \setminus\{0\}$, it follows
\begin{align*}
I_d &\leq  e^{-
cr^2}r^{d+i}\sum_{j=1}^{d+i}\frac{(d-1+i)!}{(9d)^j(d-j+i)!}\nonumber\\
&\leq e^{-
cr^2}r^{d+i}\bigg(\frac{d+i}{d}\bigg)^{\!\!d+i}\ \sum_{j=1}^{d+i}\frac{(d-1+i)!}{9^jj^j(d-j+i)!}\nonumber\\
&\leq e^{-
cr^2}r^{d+i}\bigg(\frac{d+i}{d}\bigg)^{\!\!d+i}\  \sum_{j=1}^{d+i}\frac{je^{1-j}}{9^j }\begin{pmatrix} d+i\, \\ j\, \end{pmatrix}.
\end{align*}
Since the function~$[1,\infty)\ni x\mapsto xe^{-x}/9^x$ is decreasing, we have
\begin{equation}\label{eq3}
I_d \leq e^{-
cr^2}(2r)^{d+i}(1+i/d)^{d+i}/9.
\end{equation}
By~\cite[equation~(8)]{MR2853479} (and~$\Gamma(3/2)=\sqrt{\pi}/2$), we have~$\Gamma(d/2+1)> \sqrt{\pi}(d/(2e))^{d/2}$, so that substituting~\eqref{eq3} into~\eqref{eq1} yields
\begin{align*}
\int_{\R^d\setminus B_r}\abs{x}^ie^{-f(x)}dx &< \frac{d\pi^{d/2}(2e)^{d/2}}{\sqrt{\pi}}\cdot \frac{e^{-
cr^2/2}2^{d+i}}{9}\bigg(r\bigg(1+\frac{i}{d}\bigg)\bigg)^{\!d+i}\ \cdot d^{-d/2}\\
&= \frac{de^{-
cr^2/2}}{9\sqrt{\pi}} \bigg(\frac{2e\pi }{d}\bigg)^{d/2}\cdot\bigg(2r\bigg(1+\frac{i}{d}\bigg)\bigg)^{\!d+i},
\end{align*}
which implies by~$d/(9\sqrt{\pi})\leq e^{d/(9\sqrt{\pi})}\leq e^{d/15}$ and~$(1+i/d)^d\leq (e^{i/d})^d=e^i$ 
that
\begin{equation}\label{eda}
\int_{\R^d\setminus B_r}\abs{\cdot}^ie^{-f} < e^{-
cr^2/2}\bigg(\frac{8\pi e^{17/15}r^2}{d}\bigg)^{d/2}\cdot \bigg(2er\bigg(1+\frac{i}{d}\bigg)\bigg)^{\!i}.
\end{equation}
Substituting~\eqref{eqs} 
and~\eqref{eda} into our estimates~\eqref{a1e},~\eqref{a2e} for~$A_1$ and~$A_2$ above yields
\begin{align*}
A_1&\leq 
2e^{-
cr^2/2}\bigg(\frac{
4cr^2e^{17/15}}{d}\bigg)^{d/2}\cdot \bigg(2er\bigg(1+\frac{i}{d}\bigg)\bigg)^{\!i},\\
A_2 &\leq e^{U_0(0)}\cdot 
\bigg(\frac{2\pi}{L}\bigg)^{-d/2}\cdot e^{-
cr^2/2}\bigg(\frac{8\pi e^{17/15}r^2}{d}\bigg)^{d/2} + \Big(\inf_{B_r}e^{-U_0}\Big)^{-1}\\
&\leq \Big(\inf_{B_r}e^{-U_0}\Big)^{\!-1} \bigg(1+ 
\bigg(\frac{L}{
c}\bigg)^{d/2}\!\!\!\!\cdot e^{-
cr^2/2}\bigg(\frac{
4cr^2e^{17/15}}{d}\bigg)^{d/2}\,\bigg),
\end{align*}
from which, together with~\eqref{uda}, the inequality~$
cr^2/d \leq e^{
cr^2/(4d)}$ (which holds by~\eqref{eq0}) and~$2(4e^{17/15})^{d/2}=\exp(\ln2 + (d/2)\ln(4e^{17/15}))\leq \exp(d(\ln2 +(1/2)\ln(4e^{17/15})))\leq e^{2d}$, 
the assertion follows. 
\end{proof}

\section{Hamiltonian Monte Carlo}\label{HMC}

We analyze in this section unadjusted HMC iterates, with randomized midpoint or Verlet
integrator as in~\cite{chak2024r}, applied to a target negative log-density~$U$. Here, we also work with a general probability measure~$\hat{\pi}$ on~$\mathbb{R}^d$ with an unnormalized density~$e^{-U}$ for~$U:\mathbb{R}^d\rightarrow\mathbb{R}$, rather than with the posterior~$\pi(\cdot|Z^{(n)})$. 
Throughout, the variable~$u\in\{0,1\}$ will be used to determine whether the Hamiltonian trajectory is simulated using a Verlet integrator~($u=0$) or a randomized midpoint integrator~($u=1$) as in~\cite{bourabee2022unadjusted}.

Let~$(\theta_k)_{k\in\N },(\bar{\theta}_k)_{k\in\N }$ 
be sequences of~$\R^d$-valued r.v.'s defined inductively as follows, with~$\theta_0,\bar{\theta}_0$ being arbitrary r.v.'s. 
Let~$h,T>0$ be such that~$T/h\in\N $. 
Let~$u\in\{0,1\}$ and let~$(\hat{u}_k)_{k\in\mathbb{N}} = ((u_{k,i})_{i\in\mathbb{N}\cap[0,T/h)})_{k\in\mathbb{N}}$ 
be an i.i.d. sequence independent of~$\theta_0,\bar{\theta}_0$ such that for any~$k\in\mathbb{N}$,~$(u_{k,i})_{i\in\mathbb{N}\cap[0,T/h)}$ is an i.i.d. sequence of r.v.'s with~$u_{k,i}\sim \textrm{Unif}(0,h)$ for all~$i\in\mathbb{N}\cap[0,T/h)$. 
Let~$(W_k)_{k\in\N }$ be an i.i.d. sequence of r.v.'s independent of~$\theta_0$,~$\bar{\theta}_0$,~$(\hat{u}_k)_k$ such that~$W_k\sim N(0,I_d)$ for all~$k\in\N$. 
For any~$k\in\mathbb{N}$, 
let~$(q_{k,i},p_{k,i})_{i\in[0,T/h]\cap\mathbb{N}}$ 
be such that
\begin{subequations}\label{qpdef}
\begin{align}
q_{k,i+1} &= q_{k,i} + hp_{k,i} - (h^2/2)\nabla U(q_{k,i} + uu_{k,i}p_{k,i}),\label{qdef}\\
p_{k,i+1} &= p_{k,i} - (1/2)(1+u)h\nabla U(q_{k,i} + uu_{k,i}p_{k,i})  - (1/2)(1-u)h\nabla U(q_{k,i+1}),\label{pdef}
\end{align}
\end{subequations}
with~$(q_{k,0},p_{k,0}) = (\theta_k,W_k)$ 
and let~$(\bar{q}_{k,i},\bar{p}_{k,i})_{i\in[0,T/h]\cap\mathbb{N}} $
be given by~\eqref{qpdef} with~$q_{k,j},p_{k,j}$ replaced by~$\bar{q}_{k,j},\bar{p}_{k,j}$ for any~$j$ 
and with~$(\bar{q}_{k,0},\bar{p}_{k,0}) = (\bar{\theta}_k,W_k)$. 
For any~$k\in\N $, the sequences~$(\theta_k)_k$ and~$(\bar{\theta}_k)_k$ are then given by
\begin{equation*}
\theta_{k+1} = q_{k,T/h},\qquad\bar{\theta}_{k+1}=\bar{q}_{k,T/h}.
\end{equation*}
For any~$k\in\N $, let~$\delta_{\theta_0}P^k,\delta_{\bar{\theta}_0}\bar{P}^k$ denote the distribution of~$\theta_k$ and~$\bar{\theta}_k$ respectively. 

In addition to Condition~\ref{cond:curvature3}, which is well-suited to separating the effect of a possibly nonconvex factor (thought of as the prior) in the density on concentration properties of the complete distribution~$e^{-U}/\int_{\mathbb{R}^d}e^{-U}$, we will use the following Condition~\ref{cond:curvature2}. 
\begin{condition}[Curvature]\label{cond:curvature2}
$U\in C^2(\R^d)$, $\nabla U(0) = 0$ and there exist non-increasing~$m:[0,\infty]\rightarrow\mathbb{R}$ such that
$\inf_{\theta\in B_r}\lambda_{\min}(D^2U(\theta))\geq m(r)$ for all~$r\in[0,\infty]$. 
\end{condition}
Condition~\ref{cond:curvature2} prescribes the curvature profile of the complete negative log-density~$U$ in the sense of local Euclidean balls (as opposed to the restricted convexity of~\cite{MR3025133} or that of~\cite{MR3190851}). In comparison to Condition~\ref{cond:curvature3}, the function~$m$ is allowed to take negative values, in contrast to~$m_0$, which is consistent with the possibly negative curvature of~$U_0$. Condition~\ref{cond:curvature2} is better suited for the analysis below based on point evaluations of~$U$ and not of~$U-U_0,U_0$. We make more comments on the compatibility between Conditions~\ref{cond:curvature3} and~\ref{cond:curvature2} in Remark~\ref{cc}\ref{cc3}, in particular on the consistency between~$\nabla U(0)-\nabla U_0(0)=0$ (which is a consequence of Condition~\ref{cond:curvature3}) and~$\nabla U(0)=0$ as assumed in Condition~\ref{cond:curvature2}.  In the regression context, for the radius~$r$ that will be fixed later and appropriately scaled priors,~$m(r)$ and~$m_0(r)$ scale in the same way with~$n,d$ and are both positive. 

We will also use the following (standard) smoothness condition.
\begin{condition}[$L$-smoothness]\label{cond:smooth}
$U\in C^2(\mathbb{R}^d)$ and there exist $L>0$ 
such that~$|D^2U(\theta)|\leq L$ 
for all~$\theta\in 
\mathbb{R}^d$.
\end{condition}

The next Lemma~\ref{nelem} establishes 
\begin{itemize}
\item 
in~\eqref{nelem2} a Lyapunov/drift-type inequality for the iterates in the HMC trajectory and
\item in~\eqref{nelem3} a conditional contraction property for synchronously coupled iterations of HMC with a nonexpansion property otherwise.
\end{itemize}
The property~\eqref{nelem2} will be used to control the norm of the HMC iterates, in order to obtain that the conditional contraction~\eqref{nelem3} with a local curvature constant occurs with good probability.

\begin{lemma}\label{nelem}
Assume that~$U$ satisfies Conditions~\ref{cond:curvature2},~\ref{cond:smooth} 
and that~$L(T+h)^2\leq 1/8$. 
Let~$k\in\mathbb{N}$ and denote
\begin{subequations}\label{dne}
\begin{align}
\tilde{\theta}_k&= \textstyle\max_{s\in[0,T]}(\abs{\theta_k}^2/2 + \abs{\theta_k + 2sW_k}^2/2 )^{\frac{1}{2}},\\ 
\tilde{\theta}_k'&= \textstyle\max_{s\in[0,T]}(\abs{\bar{\theta}_k}^2/2 + \abs{\bar{\theta}_k + 2sW_k}^2/2 )^{\frac{1}{2}},
\end{align}
\end{subequations}
Moreover, let~$\bar{m}:[0,\infty]\rightarrow\mathbb{R}$ and~$\bar{c}$ be given by~$\bar{m}(\infty)=m(\infty)$ and
\begin{equation}\label{ccdef}
\bar{m}(r) = \int_0^1m(tr)dt,\qquad\bar{c} = \begin{cases}
1/16&\textrm{if }u=0\\
5/16&\textrm{if }u=1.
\end{cases}
\end{equation}
For 
any~$\tilde{r}\in[0,\infty]$ with~$\bar{m}(\tilde{r})>0$, it holds a.s.\ that
\begin{equation}
\abs{\theta_{k+1}} \leq 
\begin{cases}
e^{-\bar{c}\bar{m}(\tilde{r})T^2}\tilde{\theta}_k &\textrm{if } e^{\frac{5}{16}}\tilde{\theta}_k \leq \tilde{r}\\
e^{\frac{5}{2}LT(T+h)}\tilde{\theta}_k &\textrm{otherwise}
\end{cases}\label{nelem2}
\end{equation}
and if~$m(\tilde{r})>0$ then it holds a.s.\ that
\begin{equation}
\abs{\theta_{k+1} - \bar{\theta}_{k+1}}^2
\leq \begin{cases}
e^{-2\bar{c}m(\tilde{r})T^2} \abs{\theta_k - \bar{\theta}_k}^2 &\textrm{if }e^{\frac{5}{16}}\max(\tilde{\theta}_k,\tilde{\theta}_k') \leq \tilde{r},\\
e^{5LT(T+h)} \abs{\theta_k - \bar{\theta}_k}^2&\textrm{otherwise}.
\end{cases}\label{nelem3}
\end{equation}
\end{lemma}
\begin{remark} \label{remw}
\begin{enumerate}[label=(\roman*)]
\item
In Lemma~\ref{nelem} and more generally in the rest of this Section, the negative log-density~$U$ is treated as deterministic, but in the GLM context,~$e^{-U}$ is the posterior~$\theta\mapsto\pi(\theta+\theta_{\textrm{map}}|Z^{(n)})$ from~\eqref{picdef} with~$\theta_{\textrm{map}}$ being a critical point of~$\pi(\cdot|Z^{(n)})$, which is random and independent of~$(W_k),\theta_0,\bar{\theta}_0,\hat{u}$ and the r.v.'s in the proof. Clearly, this independence allows the same conclusions to hold given the event that~$\theta\mapsto-\ln\pi(\theta+\theta_{\textrm{map}}|Z^{(n)})$ satisfies the assumptions in place of~$U$. 
\item \label{remw2}
In utilizing~\eqref{nelem2}, we will restrict to the events where~$\abs{W_k}$ is well bounded, so that for example~$\abs{sW_k}\leq CT$ for all~$s\in[0,T]$ and some constant~$C>0$. It turns out that this order~$T$ perturbation (as~$T\rightarrow0$) from~$\tilde{\theta}_k$ in~\eqref{nelem2} is too large for arbitrarily small~$T$, because it is larger than the order~$T^2$ contraction arising from the strong convexity. Setting~$T$ to be overly small would be necessary if for example~$T=h$ (which is the case of ULA if also~$u=0$) to enforce the asymptotic bias (Wasserstein distance between an invariant measure of the algorithm and the target) to be small. In other words, the application of Lemma~\ref{nelem} below is valid essentially only for HMC if~$T/h\gg 1$. 
\end{enumerate}
\end{remark}

\begin{proof}[Proof of Lemma~\ref{nelem}]
Let~$\|\cdot\|_M:\mathbb{R}^{2d}\rightarrow[0,\infty)$ be given for any~$\bar{x},\bar{v}\in\mathbb{R}^d$ by
\begin{align}
\|(\bar{x}^{\top},\bar{v}^{\top})^{\top}\|_M^2 &=
\abs{\bar{x}}^2 + 2T\bar{x}\cdot\bar{v} + 2T^2\abs{\bar{v}}^2 
\nonumber\\
&= 
\abs{\bar{x}}^2/2 + \abs{\bar{x} + 2T\bar{v}}^2/2. 
\label{Mndef}
\end{align}

We prove firstly the second cases 
in~\eqref{nelem2} and~\eqref{nelem3} 
by showing that they hold irrespective of whether~$e^{5/16}\tilde{\theta}_k > \tilde{r}$ or~$e^{5/16}\max(\tilde{\theta}_k,\tilde{\theta}_k') > \tilde{r}$ holds. To give some intuition before we proceed, note that these assertions control the growth of the trajectory~$\theta_{k+1}$ and the square difference~$\abs{\theta_{k+1}-\bar{\theta}_{k+1}}^2$ relative to the initial~$\theta_k,\bar{\theta}_k,W_k$. 
These estimates are somewhat natural expansion estimates based on the bounded Hessian Condition~\ref{cond:smooth}. The proof here gathers some estimates on this growth for each individual step/half-step in the integrators, which have been obtained elsewhere.

In case of~\eqref{nelem2}, 
for any~$k\in\mathbb{N}$, 
we prove by induction that it holds a.s.\ for any~$i\in[0,T/h]\cap\mathbb{N}$ that
\begin{equation}\label{kfx}
\bigg\|
\begin{pmatrix}q_{k,i} \\  (1-ih/T)p_{k,i}\end{pmatrix}\bigg\|_M
\leq 
\bigg(\prod_{j=0}^{i-1} e^{5Lh(T -jh) }\bigg)
\bigg\|\begin{pmatrix}
q_{k,0}\\p_{k,0}
\end{pmatrix}\bigg\|_M.
\end{equation}
The base case~$i=0$ is clear. Assume~\eqref{kfx} holds with~$i=j\in[0,T/h)\cap \mathbb{N}$. 
Consider first~$u=0$ (Verlet integrator). In this~$u=0$ case, we leverage expansion results already derived in~\cite{chak2024r}. In particular, we split the Verlet integrator (see~\cite[(1.1)]{chak2024r}) into two parts (one part being half-step velocity update then half-step position update, and the next part being half position then half velocity), and show that each of these two parts can only expand by a certain factor in~$\|\cdot\|_M$ norm, up to a change in velocity prefactor. The latter (chronological) part of integrator is dealt with first. 
By definition of the Verlet integrator~\eqref{qpdef} ($u=0$), we have
\begin{align}
&\bigg\|\begin{pmatrix}q_{k,j+1} \\  (1-(j+1)h/T)p_{k,j+1}\end{pmatrix}\bigg\|_M \nonumber\\
&\quad= \bigg\|\begin{pmatrix}\tilde{q}' + (h/2)\tilde{p}' \\  (1-(j+1)h/T)(\tilde{p}'-(h/2)\nabla U(\tilde{q}'+(h/2)\tilde{p}'))\end{pmatrix}\bigg\|_M,\label{fjqx}
\end{align}
where~$(\tilde{q}',\tilde{p}')$ is the position-velocity vector after a half velocity then half position update from~$(q_{k,j},p_{k,j})$, defined by
\begin{subequations}\label{qbcx}
\begin{align}
\tilde{q}' &:= q_{k,j}+(h/2)p_{k,j} - (h/2)^2
\nabla U(q_{k,j}),\label{qbdx1}\\
\tilde{p}' &:= p_{k,j}-(h/2)
\nabla U(q_{k,j}).\label{qbdx2}
\end{align}
\end{subequations}
Note that the vector within the norm on right-hand side of~\eqref{fjqx} is itself (up to the velocity prefactor) the half velocity then half position update from~$(\tilde{q}',\tilde{p}')$. Of course, this is the full Verlet step, which is what appears on the left-hand side. 
By Proposition~4.9 in~\cite{chak2024r} with the substitutions~$(\bar{q}',\bar{p}') = (\tilde{q}',\tilde{p}')$,~$\bar{q}=\bar{p}=0$ and
\begin{equation}
\eta_0 = 1- (j+1/2)h/T,\quad\eta_1 = 1- (j+1)h/T,\quad  b(\cdot,\theta)=\nabla U \quad \forall \theta,\label{exh1}
\end{equation}
where~$b$ 
on the left hand side in the last equation of~\eqref{exh1} refers to the function appearing in Proposition~4.9 in~\cite{chak2024r}, 
the right-hand side of~\eqref{fjqx} satisfies
\begin{align}
&\bigg\|\begin{pmatrix}\tilde{q}' + (h/2)\tilde{p}' \\  (1-(j+1)h/T)(\tilde{p}'-(h/2)\nabla U(\tilde{q}'+(h/2)\tilde{p}'))\end{pmatrix}\bigg\|_M\nonumber\\
&\quad\leq \bigg(1+4Lh\bigg(T-\bigg(j+\frac{1}{2}\bigg)h\bigg)\bigg)^{\!\frac{1}{2}}\bigg\|\begin{pmatrix}\tilde{q}' \\  (1-(j+1/2)h/T)\tilde{p}'\end{pmatrix}\bigg\|_M,\label{jfqx2}
\end{align}
where we have used an explicit calculation for~$j=T/h-1$ which does not satisfy the~$\eta_0-\eta_1\leq \eta_1$ condition in Proposition~4.9 in~\cite{chak2024r} with~\eqref{exh1}. 
Moreover, by definitions~\eqref{qbcx} of~$\tilde{q}'$,~$\tilde{p}'$, 
the norm on the right-hand side of~\eqref{jfqx2} is
\begin{align}
&\bigg\|\begin{pmatrix}\tilde{q}' \\  (1-(j+1/2)h/T)\tilde{p}'\end{pmatrix}\bigg\|_M\nonumber\\
&\quad= \bigg\|\begin{pmatrix}q_{k,j}+(h/2)p_{k,j} - (h/2)^2\nabla U(q_{k,j})\\  (1-(j+1/2)h/T)(p_{k,j}-(h/2)\nabla U(q_{k,j}))\end{pmatrix}\bigg\|_M. \label{jfqx3}
\end{align}
By Proposition~4.8 in~\cite{chak2024r} with the substitutions
\begin{subequations}\label{exh3}
\begin{align}
&(\bar{q}',\bar{p}') = (q_{k,j},p_{k,j}),\qquad \bar{q}=\bar{p}=0,\label{exh3a}\\
&\eta_0 = 1-jh/T,\qquad\eta_1 = 1-(j+1/2)h/T,\qquad b(\cdot,\theta) = \nabla U\quad\forall \theta,\label{exh3b}
\end{align}
\end{subequations}
the first term on the right-hand side of~\eqref{jfqx3} satisfies
\begin{align}
&\bigg\|\begin{pmatrix}q_{k,j}+(h/2)p_{k,j} - (h/2)^2\nabla U(q_{k,j})\\  (1-(j+1/2)h/T)(p_{k,j}-(h/2)\nabla U(q_{k,j}))\end{pmatrix}\bigg\|_M\nonumber\\
&\quad \leq (1+6Lh(T-jh))^{\frac{1}{2}}\bigg\|\begin{pmatrix}q_{k,j}\\  (1-jh/T)p_{k,j}\end{pmatrix}\bigg\|_M.\label{jfqx4}
\end{align}
Gathering the inductive assumption~\eqref{kfx} with~$i=j$,~\eqref{fjqx},~\eqref{jfqx2},~\eqref{jfqx3},~\eqref{jfqx4}, and noting that the prefactors in front of the norms satisfy
\begin{align*}
\bigg(1+4Lh\bigg(T-\bigg(j+\frac{1}{2}\bigg)h\bigg)\bigg)^{\!\frac{1}{2}} &\leq \exp(2Lh(T-(j+1/2)h)),\\
(1+6Lh(T-jh))^{\frac{1}{2}}&\leq \exp(3Lh(T-jh))
\end{align*}
respectively, 
we obtain~\eqref{kfx} with~$i=j+1$, hence~\eqref{kfx} for all~$i\in[0,T/h]\cap\mathbb{N}$. 
For the~$u=1$ (randomized midpoint) case, the arguments and conclusions above apply with 
Proposition~\ref{RvRth} 
instead of Propositions~4.8,~4.9 in~\cite{chak2024r}, where we take the substitutions~\eqref{exh3a} and
\begin{equation}\label{exh4}
\bar{u} = u_{k,j},\quad \eta_0 = 1-jh/T,\quad \eta_1 = 1-(j+1)h/T,\quad \bar{b}(\cdot,\xi) = \nabla U \quad \forall\xi.
\end{equation}
Note that in this~$u=1$ case, there was no splitting of the integrator step into two, as in the~$u=0$ case. This is reflected in the definitions~\eqref{exh4} of the velocity prefactors~$\eta_0,\eta_1$ (cf.~\eqref{exh1},~\eqref{exh3b}). 

In both~$u\in\{0,1\}$ cases, by taking~$i=T/h$ in~\eqref{kfx},
we have a.s. that
\begin{equation*}
\abs{q_{k,T/h} }
= \bigg\|
\begin{pmatrix}q_{k,T/h} \\ 0\end{pmatrix}\bigg\|_M
\leq 
\bigg(\prod_{j=0}^{T/h-1} e^{5Lh(T -jh) }\bigg)
(\abs{q_{k,0}}^2/2 +\abs{q_{k,0}+2Tp_{k,0}}^2/2)^{1/2}.
\end{equation*}
Therefore, since
\begin{equation*}
\textstyle\sum_{j=0}^{T/h-1}Lh(T-jh) = LT^2 - (1/2)Lh^2(T/h-1)T/h = (1/2)LT(T+h),
\end{equation*} 
we have a.s.\ that
\begin{equation*}
\abs{q_{k,T/h} } \leq 
e^{\frac{5}{2}LT(T+h)}
(\abs{q_{k,0}}^2/2 +\abs{q_{k,0}+2Tp_{k,0}}^2/2)^{1/2},
\end{equation*}
which implies the second case in~\eqref{nelem2}.

In case of~\eqref{nelem3}, 
for any~$k\in\mathbb{N}$, 
we prove by induction that it holds a.s.\ for any~$i\in[0,T/h]\cap\mathbb{N}$ that
\begin{equation}
\bigg\|
\begin{pmatrix}q_{k,i} - \bar{q}_{k,i} \\  (1-ih/T)(p_{k,i} - \bar{p}_{k,i})\end{pmatrix}\bigg\|_M^2
\leq 
\bigg(\prod_{j = 0}^{i-1}e^{10Lh(T - jh)}\bigg)
\bigg\|\begin{pmatrix}
q_{k,0}-\bar{q}_{k,0}\\p_{k,0}-\bar{p}_{k,0}
\end{pmatrix}\bigg\|_M^2.\label{kf}
\end{equation}
The strategy is the same as for~\eqref{nelem2}, except we apply the results in~\cite{chak2024r} without setting one of the position-velocity pairs to zero. 
The base case~$i=0$ is clear. Assume~\eqref{kf} holds with~$i=j\in[0,T/h)\cap \mathbb{N}$. 
Consider first~$u=0$. 
It holds a.s.\ that
\begin{align}
&
\bigg\|\begin{pmatrix}q_{k,j+1} - \bar{q}_{k,j+1} \\  (1-(j+1)h/T)(p_{k,j+1} - \bar{p}_{k,j+1})\end{pmatrix}\bigg\|_M^2
\nonumber\\
&\quad=
\bigg\|\begin{pmatrix}\tilde{q}-\tilde{q}' + \frac{h}{2}(\tilde{p} - \tilde{p}') \\  (1-(j+1)h/T)(\tilde{p} - \tilde{p}'-\frac{h}{2}(\nabla U(\tilde{q}+\frac{h}{2}\tilde{p}) - \nabla U(\tilde{q}'+\frac{h}{2}\tilde{p}')))\end{pmatrix}\bigg\|_M^2,\label{fjq}
\end{align}
where~$(\tilde{q}',\tilde{p}')$ is defined in~\eqref{qbcx} and
\begin{equation}\label{qbd2}
\tilde{q}:= \bar{q}_{k,j}+(h/2)\bar{p}_{k,j} - (h/2)^2\nabla U(\bar{q}_{k,j}),\qquad \tilde{p}:=\bar{p}_{k,j}-(h/2)\nabla U(\bar{q}_{k,j}).
\end{equation}
By Proposition~4.9 in~\cite{chak2024r} with the substitutions~$(\bar{q},\bar{p},\bar{q}',\bar{p}') = (\tilde{q},\tilde{p},\tilde{q}',\tilde{p}')$ 
and~\eqref{exh1}, 
the right-hand side of~\eqref{fjq} satisfies 
\begin{align}
&\bigg\|\begin{pmatrix}\tilde{q}-\tilde{q}' + \frac{h}{2}(\tilde{p} - \tilde{p}') \\  (1-(j+1)h/T)(\tilde{p} - \tilde{p}'-\frac{h}{2}(\nabla U(\tilde{q}+\frac{h}{2}\tilde{p}) - \nabla U(\tilde{q}'+\frac{h}{2}\tilde{p}')))\end{pmatrix}\bigg\|_M^2\nonumber\\
&\quad\leq \bigg(1+4Lh\bigg(T-\bigg(j+\frac{1}{2}\bigg)h\bigg)\bigg)\bigg\|\begin{pmatrix}\tilde{q}-\tilde{q}' \\  (1-(j+1/2)h/T)(\tilde{p} - \tilde{p}')\end{pmatrix}\bigg\|_M^2.\label{jfq2}
\end{align}
Moreover, by definitions~\eqref{qbcx},~\eqref{qbd2} of~$(\tilde{q},\tilde{p},\tilde{q}',\tilde{p}')$, 
the norm on the right-hand side of~\eqref{jfq2} satisfies
\begin{align}
&
\bigg\|\begin{pmatrix}\tilde{q}-\tilde{q}' \\  (1-(j+1/2)h/T)(\tilde{p} - \tilde{p}')\end{pmatrix}\bigg\|_M^2
\nonumber\\
&\quad\leq \bigg\|\begin{pmatrix}q_{k,j} - \bar{q}_{k,j} +\frac{h}{2}(p_{k,j}-\bar{p}_{k,j}) - \frac{h^2}{4}(\nabla U(q_{k,j}) - \nabla U(\bar{q}_{k,j}))\\  (1-(j+1/2)h/T)(p_{k,j} - \bar{p}_{k,j} -\frac{h}{2}(\nabla U(q_{k,j}) - \nabla U(\bar{q}_{k,j})))\end{pmatrix}\bigg\|_M^2.\label{jfq3}
\end{align}
By Proposition~4.8 in~\cite{chak2024r} with the substitutions~$(\bar{q},\bar{p},\bar{q}',\bar{p}') = (\bar{q}_{k,j},\bar{p}_{k,j},q_{k,j},p_{k,j})$ and~\eqref{exh3b}, 
the right-hand side of~\eqref{jfq3} satisfies
\begin{align}
&\bigg\|\begin{pmatrix}q_{k,j} - \bar{q}_{k,j} +\frac{h}{2}(p_{k,j}-\bar{p}_{k,j}) - \frac{h^2}{4}(\nabla U(q_{k,j}) - \nabla U(\bar{q}_{k,j}))\\  (1-(j+1/2)h/T)(p_{k,j} - \bar{p}_{k,j} -\frac{h}{2}(\nabla U(q_{k,j}) - \nabla U(\bar{q}_{k,j})))\end{pmatrix}\bigg\|_M^2\nonumber\\
&\quad \leq (1+6Lh(T-jh))\bigg\|\begin{pmatrix}q_{k,j} - \bar{q}_{k,j}\\  (1-jh/T)(p_{k,j} - \bar{p}_{k,j})\end{pmatrix}\bigg\|_M^2.\label{jfq4}
\end{align}
Gathering the inductive assumption~\eqref{kf} with~$i=j$,~\eqref{fjq},~\eqref{jfq2},~\eqref{jfq3},~\eqref{jfq4}, 
we obtain~\eqref{kf} with~$i=j+1$, hence~\eqref{kf} for all~$i\in[0,T/h]\cap\mathbb{N}$. 
For the~$u=1$ (randomized midpoint) case, the arguments above apply with 
Proposition~\ref{RvRth} 
instead of Propositions~4.8,~4.9 in~\cite{chak2024r}, where we take the substitutions~$(\bar{q},\bar{p},\bar{q}',\bar{p}') = (q_{k,j},p_{k,j},\bar{q}_{k,j},\bar{p}_{k,j})$ and~\eqref{exh4}. 
In both~$u\in\{0,1\}$ cases, by taking~$i=T/h$ in~\eqref{kf}, 
by resolving the product factor on the right-hand side of~\eqref{kf} in the same way as before, 
and by~$p_{k,0}=\bar{p}_{k,0}=W_k$, 
we have a.s.\ that
\begin{align*}
&
\abs{q_{k,T/h} - \bar{q}_{k,T/h} }^2
\\
&\quad\leq 
e^{5LT(T+h)}
(\abs{q_{k,0} - \bar{q}_{k,0}}^2/2 + \abs{q_{k,0} - \bar{q}_{k,0} + 2T(p_{k,0} - \bar{p}_{k,0})}^2/2)
\\
&\quad=
e^{5LT(T+h)}
\abs{\theta_k - \bar{\theta}_k}^2, 
\end{align*}
which implies the second case in~\eqref{nelem3}.

It remains to prove for any~$k\in\mathbb{N}$ and~$\tilde{r}\in[0,\infty]$ with~$\bar{m}(\tilde{r})>0$ (and with~$m(\tilde{r})>0$ resp.) the assertion~\eqref{nelem2} (and~\eqref{nelem3} resp.) 
given~$e^{5/16}\tilde{\theta}_k 
\leq \tilde{r}$ (and~$e^{5/16}\max(\tilde{\theta}_k,\tilde{\theta}_k') 
\leq \tilde{r}$ resp.). 
These are (almost) contraction bounds relative to~$\theta_k,\bar{\theta}_k,W_k$, conditional on the values of the latter. The broad strategy is also essentially to apply results from~\cite{chak2024r} (at least in the~$u=0$ Verlet case). First, we use the arguments above leading to second case of~\eqref{nelem2} to control the points at which~$\nabla U$ is evaluated in the trajectory. These points are shown to live in a fixed ball (of radius~$\tilde{r}$), so that we can infer the amount of (local) positive curvature available in~$U$ (from Condition~\ref{cond:curvature2}) to contract the trajectories. 

Fix~$k\in\mathbb{N}$,~$\tilde{r}\in[0,\infty]$ with~$\bar{m}(\tilde{r})>0$ and assume~$e^{5/16}\tilde{\theta}_k
\leq\tilde{r}$. 
First note that in the above arguments for~\eqref{kfx}, since they apply for any~$T\in (0,(8L)^{-1/2}]\cap h\mathbb{N}$, we have that for any~$k\in\mathbb{N}$,~$\tilde{T}\in (0,T]\cap h\mathbb{N}$ and~$i\in[0,\tilde{T}/h]\cap\mathbb{N}$, it holds a.s.\ that
\begin{equation}\label{kf2}
\bigg\|\begin{pmatrix}q_{k,i}\\(1-ih/\tilde{T})p_{k,i}
\end{pmatrix}\bigg\|_{M,\tilde{T}}\leq e^{\frac{5}{2}LT(T+h)}
\bigg\|
\begin{pmatrix}
q_{k,0}\\p_{k,0}
\end{pmatrix}
\bigg\|_{M,\tilde{T}},
\end{equation}
where~$\|\cdot\|_{M,\tilde{T}}$ is defined in the same way as~$\|\cdot\|_M$ except with~$T$ replaced by~$\tilde{T}$. 
For any~$s\in[0,T)\cap h\mathbb{N}$, by 
taking~$\tilde{T}=s+h$ in~\eqref{kf2} with~$i=s/h$, 
by~\eqref{Mndef} (with~$T$ replaced by~$\tilde{T} = s+h$), we have 
a.s.\ that
\begin{align}
&(\abs{q_{k,s/h}}^2/2 + \abs{q_{k,s/h} + 2hp_{k,s/h}}^2/2)^{\frac{1}{2}}\nonumber\\
&\quad\leq e^{\frac{5}{2}LT(T+h)} (\abs{q_{k,0}}^2/2 + \abs{q_{k,0} + 2(s+h)p_{k,0}}^2/2)^{\frac{1}{2}}. \label{rnf}
\end{align}
For any~$s\in[0,T]\cap h\mathbb{N}$, by instead taking~$\tilde{T} = s$ in~\eqref{kf2} with~$i=s/h$ (or by definitions for~$s=0$), it holds a.s.\ that
\begin{align}
\abs{q_{k,s/h}} 
&\leq e^{\frac{5}{2}LT(T+h)} (\abs{q_{k,0}}^2/2 + \abs{q_{k,0} + 2sp_{k,0}}^2/2)^{1/2} 
\nonumber\\
& \leq e^{\frac{5}{2}LT(T+h)}\tilde{\theta}_k 
\nonumber\\
&\leq\tilde{r},\label{ja2}
\end{align}
where we have used the assumption~$L(T+h)^2\leq 1/8$. 
Moreover, for any~$s\in[0,T)\cap h\mathbb{N}$ and any~$\bar{u}\in[0,h]$, it holds a.s.\ that
\begin{align}
&\abs{q_{k,s/h}+\bar{u}p_{k,s/h}}\nonumber\\
&\quad\leq \abs{q_{k,s/h}/2} + \abs{q_{k,s/h}/2 + \bar{u}p_{k,s/h}}\nonumber\\
&\quad\leq \abs{q_{k,s/h}/2} + \max(\abs{q_{k,s/h}/2}, \abs{q_{k,s/h}/2+ hp_{k,s/h}})\nonumber\\
&\quad= \max(\abs{q_{k,s/h}}, \abs{q_{k,s/h}/2} + \abs{q_{k,s/h}/2+ hp_{k,s/h}}).\label{gjh}
\end{align}
By~\eqref{rnf} and~$L(T+h)^2\leq 1/8$, the second argument on the right-hand side of~\eqref{gjh} satisfies
\begin{align*}
&\abs{q_{k,s/h}/2} + \abs{q_{k,s/h}/2+ hp_{k,s/h}}\\
&\quad = (1/\sqrt{2})(\abs{q_{k,s/h}}/\sqrt{2} + \abs{q_{k,s/h}+ 2hp_{k,s/h}}/\sqrt{2})\\
&\quad \leq (\abs{q_{k,s/h}}^2/2 + \abs{q_{k,s/h} + 2hp_{k,s/h}}^2/2)^{1/2}\\
&\quad \leq e^{\frac{5}{16}}\tilde{\theta}_k 
\\
&\quad \leq \tilde{r}.
\end{align*}
Note that by~\eqref{ja2}, the first argument on the right-hand side of~\eqref{gjh} is also bounded above a.s.\ by~$e^{5/16}\tilde{\theta}_k
$ and by~$\tilde{r}$, so that for any~$s\in[0,T)\cap h\mathbb{N}$ and~$\bar{u}\in[0,h]$ we have
\begin{equation}\label{ja1}
\abs{q_{k,s/h}+\bar{u}p_{k,s/h}} \leq e^{5/16}\tilde{\theta}_k
\leq \tilde{r}.
\end{equation} 
The inequalities~\eqref{ja2} and~\eqref{ja1} bound the points on~$\mathbb{R}^d$ at which gradients are evaluated in the trajectories for~$u=0$ and~$u=1$ respectively. This fact is used in the following without further mention. 

We proceed with the contraction bounds, first for a single position-velocity pair. 
In case~$u=1$ (randomized midpoint), we will apply Proposition~\ref{lbaR} with the substitutions~\eqref{exh3a},~\eqref{exh4}. Note that in Proposition~\ref{lbaR}, the vector in the left-hand side of~\eqref{lbaeqR} corresponds exactly with the difference between one step of the integrator~\eqref{qpdef} with~$u=1$ applied to~$(\bar{q},\bar{p})$ and to~$(\bar{q}',\bar{p}')$, and~\eqref{exh3a} sets~$(\bar{q},\bar{p})=0$ (which stays as zero after the integrator step~\eqref{qpdef} by~$\nabla U(0)=0$). In particular, the vector on the left-hand side of~\eqref{lbaeqR} is~$-(q_{k,j+1}^{\top},p_{k,j+1}^{\top})$ under~\eqref{exh3a}. Moreover, the contraction constant~$c$ in~\eqref{cdef} satisfies
\begin{align*}
c&=\frac{(c_0\eta_0+c_1\eta_1)\hat{m}h^2}{4(\eta_0-\eta_1)}\\
&= \frac{\big((3/2)(1-jh/T) + (1-(j+1)h/T)\big)\cdot \tilde{m}h^2}{4h/T}\\
&= \big((5/2)(1-jh/T) - h/T\big)\cdot \tilde{m}hT/4\\
&\geq (5/4)(2-(2j+1)h/T) \cdot \tilde{m}hT/4\\
&= (5/16)(2-(2j+1)h/T) \cdot \tilde{m}hT,
\end{align*}
where, by~\eqref{ja1} and Condition~\ref{cond:curvature2},
\begin{equation}\label{mmb}
\tilde{m} := 2\lambda_{\min}\bigg(\int_0^1 D^2U(t(q_{k,j} + u_{k,j}p_{k,j}))dt\bigg) \geq 2\int_0^1m(\tilde{r}t)dt = 2\bar{m}(\tilde{r}).
\end{equation}
Note also that~$\|\cdot\|_M$ defined in~\eqref{Mndef} coincides~$\|\cdot\|_M$ defined with~\eqref{cMdef} under~\eqref{exh4}. 
Therefore, this application of Proposition~\ref{lbaR} yields 
for any~$j\in[0,T/h)\cap\N $ 
that
\begin{align}
&\bigg\| \begin{pmatrix} q_{k,j+1}\\(1-(j+1)h/T)p_{k,j+1} \end{pmatrix} \bigg\|_M \nonumber\\
&\quad\leq \bigg(1-\frac{5\tilde{m}(2-(2j+ 1)h/T)hT}{16}\bigg)^{\!\!\frac{1}{2}}\bigg\| \begin{pmatrix} q_{k,j} \\(1-jh/T)p_{k,j}\end{pmatrix} \bigg\|_M.\label{kd22}
\end{align}
By induction in~$j$ on~\eqref{kd22}, since~\eqref{mmb} implies
\begin{equation*}
\prod_{j=0}^{T/h-1}\bigg(1-\frac{5\tilde{m}(2-(2j+ 1)h/T)hT}{16}\bigg)^{\!\!\frac{1}{2}}\leq \prod_{j=0}^{T/h-1}\exp\bigg(-\frac{5}{16}\bar{m}(\tilde{r})(2-(2j+1)h/T)hT\bigg),
\end{equation*}
and since in the product exponent on the right-hand side we have
\begin{equation*}
\textstyle\sum_{j=0}^{T/h-1}(2-(2j+1)h/T)hT= 2T^2-2\cdot(1/2)(T/h-1)T/h\cdot h^2 - hT =T^2,
\end{equation*}
inequality~\eqref{kd22} and the definition~\eqref{Mndef} for~$M$ imply
\begin{equation*}
\abs{\theta_{k+1}}=\abs{q_{k,T/h}}\leq e^{-(5/16)\bar{m}(\tilde{r})T^2}\tilde{\theta}_k .
\end{equation*}
This proves~\eqref{nelem2} for~$u=1$. In case~$u=0$, the proof is the same, except we use the~$u=0$ case in~\eqref{kfx} to control the points at which gradients are evaluated and we use the corresponding contraction results in~\cite{chak2024r} (Propositions~4.1,~4.2 therein). Note that these contraction propositions assume global convexity. However, local convexity (as in Condition~\ref{cond:curvature2}) is in fact sufficient. Indeed, an inspection of the proofs of Propositions~4.1,~4.2 in~\cite{chak2024r} (and~(2.3) in Remark~2.1 therein) shows that, under the notation therein 
with~$R=0$,~$b(\cdot,\theta)=\nabla U$ for all~$\theta$, the conclusions hold with~$m/2$ therein replaced with 
\begin{equation*}
\tilde{m}' 
:= \lambda_{\min}\bigg(\int_0^1D^2 U(\bar{q}'+t(\bar{q}-\bar{q}')) dt\bigg),
\end{equation*}
as long as~$\tilde{m}'>0$ (because the contraction calculation only uses the left-hand side of~(2.3) in Remark~2.1 therein). 
Since we have an upper bound on the size of each point in the trajectory (via~\eqref{kfx} as in the~$u=1$ case), we have that~$\tilde{m}'\geq \bar{m}(\tilde{r})>0$ under~$(\bar{q}',\bar{p}') = (\tilde{q}',\tilde{p}')$ given by~\eqref{qbcx}. Concretely, 
Proposition~4.2 in~\cite{chak2024r} may be applied with the substitutions~$R=0$,~$(\bar{q}',\bar{p}') = (\tilde{q}',\tilde{p}')$ given by~\eqref{qbcx},~$\bar{q}=\bar{p}=0$ and~\eqref{exh1}, then 
Proposition~4.1 therein 
with~$R=0$,~\eqref{exh3} and with~$m/2$ replaced by~$\tilde{m}'\geq \bar{m}(\tilde{r})$ 
to obtain~\eqref{nelem2} for~$u=0$. 

For the last assertion~\eqref{nelem3}, similar arguments as for~\eqref{nelem2} apply. 
We sketch the changes for~$u=1$ as follows. 
By definition(s)~\eqref{qpdef} of~$q_{k,i},\bar{q}_{k,i},p_{k,i},\bar{p}_{k,i}$, for any~$j\in[0,T/h)\cap\N $,
it holds a.s.\ that
\begin{equation}\label{djp}
\bigg\|\begin{pmatrix}q_{k,j+1} - \bar{q}_{k,j+1} \\  (1-(j+1)h/T)(p_{k,j+1} - \bar{p}_{k,j+1})\end{pmatrix}\bigg\|_M^2
=
\bigg\|\begin{pmatrix}q^{\dagger}\\ (1-(j+1)h/T)p^{\dagger}\end{pmatrix}\bigg\|_M^2,
\end{equation}
where
\begin{align*}
q^{\dagger} &:= q_{k,j}-\bar{q}_{k,j} + h(p_{k,j}-\bar{p}_{k,j}) \\
&\quad- (h^2/2)(\nabla U(q_{k,j} + u_{k,j}p_{k,j}) - \nabla U(\bar{q}_{k,j}+ u_{k,j}\bar{p}_{k,j})),\\
p^{\dagger} &:= p_{k,j}- \bar{p}_{k,j}-h(\nabla U(q_{k,j}+u_{k,j}p_{k,j}) - \nabla U(\bar{q}_{k,j}+u_{k,j}\bar{p}_{k,j})).
\end{align*}
By Proposition~\ref{lbaR}, using the same arguments as for~\eqref{kd22} to estimate the contraction~$c$, the norm on the right-hand side of~\eqref{djp} satisfies
\begin{align}
&\bigg\| \begin{pmatrix} q^{\dagger}\\(1-(j+1)h/T)p^{\dagger}\end{pmatrix} \bigg\|_M^2 \nonumber\\
&\quad\leq \bigg(1-\frac{5m(\tilde{r})(2-(2j+1)h/T)hT}{8}\bigg)\bigg\| \begin{pmatrix} q_{k,j} - \bar{q}_{k,j}  \\(1-jh/T) ( p_{k,j} -  \bar{p}_{k,j}) \end{pmatrix} \bigg\|_M^2,\label{kd2p}
\end{align}
so that by induction in~$j$, using again the same calculations immediately following~\eqref{kd22}, we have
\begin{equation*}
\abs{\theta_{k+1} - \bar{\theta}_{k+1}}^2 = \abs{q_{k,T/h} - \bar{q}_{k,T/h}}^2 \leq e^{-(5/8)m(\tilde{r})T^2}\abs{\theta_k - \bar{\theta}_k}^2. 
\end{equation*}
Again, Propositions~4.1,~4.2 in~\cite{chak2024r} instead of Proposition~\ref{lbaR} 
may be considered instead to obtain both~$u\in\{0,1\}$ cases. 
\end{proof}

The following Lemma~\ref{conmain} asserts a concentration property for HMC trajectories, under the main 
condition~$\hat{n}_r\geq d$, for~$\hat{n}_r$ given in~\eqref{nrdef} and~$r>0$. The parameter~$r$ decides the quantity~$\mathbb{E}[\abs{\theta_N}^2\mathds{1}_{\{\abs{\theta_N}>r\}}]$ to be controlled in the conclusion~\eqref{concon}. 
For this control to be strong enough (to have an exponentially small right-hand side),~$r$ is to be chosen such that HMC trajectories indeed do not appear outside of this radius with significant probability. 
If~$\bar{m}(3r/2)$ is proportional to~$L$ and~$LT^2$ is fixed to an absolute constant as~$d\rightarrow\infty$, then~$\hat{n}_r$ is proportional to~$Lr^2$. Thus, in that setting, making the assumption~$\hat{n}_r\geq d$ is to assume that the square radius~$Lr^2$ scales at least like~$d$. Lemma~\ref{conmain} then asserts that this scaling on~$r$ is indeed such that HMC trajectories do not appear outside of~$B_r$ with significant probability. 
Note that this scaling on~$Lr^2$ is sharp in terms of~$d$ by concentration of the Gaussian increments~$W_k$ around the sphere of radius~$\sqrt{d}$, in the sense that if~$Lr^2=o(d)$, then~$\mathbb{E}[\abs{\theta_N}^2\mathds{1}_{\{\abs{\theta_N}>r\}}]$ cannot generically be  exponentially (in~$d$) small as~$d\rightarrow\infty$. 
This critical scaling on~$Lr^2$ in terms of~$d$ points to a concentration of trajectories in some range of radii~$r$, mirroring the concentration of (target) measure~\cite{MR1849347}; 
the scaling also appears characteristically in nonconvex sampling leading to dichotomous results~\cite{chak2025t}. 

In the regression context, it will be shown that it is sensible to choose~$r$ such that~$Lr^2$ scales like the number of datapoints~$n$. 
Lemma~\ref{conmain} then asserts that~$\mathbb{E}[\abs{\theta_N}^2\mathds{1}_{\{\abs{\theta_N}>r\}}]$ is exponentially-in-$n$ small, up to the initial states being concentrated in~$B_r$, where the condition~$\hat{n}_r\geq d$ is that~$n\geq Cd$ for large enough~$C>0$.
\begin{lemma}\label{conmain}
Assume the settings and notations of Lemma~\ref{nelem}. Assume moreover that there exist~$r_1\geq 0$,~$c_i:[r_1,\infty)\rightarrow[0,\infty)$ for~$i\in\{0,2\}$ 
such that
\begin{equation}\label{c0s}
\mathbb{E}[\abs{\theta_0}^i\mathds{1}_{\abs{\theta_0}>r}] \leq c_i(r)\qquad \forall r\geq r_1.
\end{equation}
For any~$\bar{N}\in\mathbb{N}$ and~$\bar{r}\in [r_1,\infty]$, let
\begin{equation}\label{drdef}
\delta_{\bar{N},\bar{r}} = 2c_2(\bar{r}) + 4T^2c_0(\bar{r})\bar{N}^2d +( 2\bar{N}(\bar{r}^2 + c_2(\bar{r})) + 4T^2\bar{N}^3(d^2+2d)^{1/2} )e^{-(\sqrt{\hat{n}_{\bar{r}}} - \sqrt{d})^2/4}
\end{equation}
and
\begin{equation}\label{nrdef}
\hat{n}_{\bar{r}} = (\bar{c}\bar{r}T\bar{m}(3\bar{r}/2))^2,
\end{equation}
where recall~$\bar{c}$ is given by~\eqref{ccdef}. 
Let~$r\in[r_1,\infty]$ be such that~$\bar{m}(3r/2)>0$. 
For any~$N\in\mathbb{N}\setminus\{0\}$, if~$\hat{n}_r \geq d$, then it holds that
\begin{equation}\label{concon}
\mathbb{E}[\abs{\theta_N}^2\mathds{1}_{\{\abs{\theta_N}> r\}}] \leq e^{5NLT(T+h)}\delta_{N,r}.
\end{equation}
\end{lemma}
\begin{proof}
Fix~$r\geq r_1$ and~$k\in\mathbb{N}$. 
For any~$\bar{r}>0$, let
\begin{equation*}
\mathcal{K}_{k,\bar{r}} = \{ \abs{W_k} \leq \sqrt{\hat{n}_{\bar{r}}} \}.
\end{equation*}
Note first that
\begin{align}
\tilde{\theta}_k &= \textstyle\max_{s\in[0,T]}(\abs{\theta_k}^2/2 + \abs{\theta_k+2sW_k}^2/2)^{1/2}\nonumber\\
&\leq (\abs{\theta_k}^2/2 + (\abs{\theta_k}+2T\abs{W_k})^2/2)^{1/2}\nonumber\\
&\leq (\abs{\theta_k}^2 + 2T\abs{\theta_k}\abs{W_k} + 2T^2\abs{W_k}^2)^{1/2}.\label{Mpr}
\end{align}
Therefore, 
we have 
on the event~$\mathcal{K}_{k,r}\cap\{\abs{\theta_k}\leq r\}$, by~$\bar{m}\leq L$ and~$LT^2\leq 1/8$, that
\begin{align}
e^{5/16}\tilde{\theta}_k
&\leq re^{5/16} \big( 1 + 2T^2\bar{c}\bar{m}(3r/2) + 2T^4\bar{c}^2(\bar{m}(3r/2))^2\big)^{1/2}\nonumber\\
&\leq re^{5/16}( 1+ 1/10 + 1/200)^{1/2}\nonumber\\
&\leq 3r/2.\label{Mppr}
\end{align}
Therefore by~\eqref{nelem2} in Lemma~\ref{nelem} with~$\tilde{r}=3r/2$ 
and the definition~\eqref{dne} of~$\tilde{\theta}_k$, 
for any~$r>0$ with~$\bar{m}(3r/2)>0$, it holds on the event~$\mathcal{K}_{k,r}\cap\{\abs{\theta_k}\leq r\}$ a.s.\ that
\begin{equation}\label{s1s2}
\abs{\theta_{k+1}} \leq e^{-\bar{c}\bar{m}(3r/2)T^2}\tilde{\theta}_k.
\end{equation}
On the event~$\mathcal{K}_{k,r}\cap\{\abs{\theta_k}\leq r\}$, inequality~\eqref{s1s2} and the first inequality in~\eqref{Mppr} (which yields by~$(1+x+x^2/2)^{1/2}\leq \exp(x/2)$ valid for all~$x\geq 0$, applied to~$x=2T^2\bar{c}\bar{m}(3r/2)$, the inequality~$\tilde{\theta}_k\leq re^{T^2\bar{c}\bar{m}(3r/2)}$) imply
\begin{equation*}
\abs{\theta_{k+1}} \leq r.
\end{equation*}
In particular, for any~$r>0$ with~$\bar{m}(3r/2)>0$ and~$k\in\mathbb{N}\setminus\{0\}$, it holds that
\begin{equation}\label{jf2}
\bar{\mathcal{K}}_{k,r}:=\{\abs{\theta_0}\leq r\}\cap\cap_{l=0}^{k-1}\mathcal{K}_{l,r} \subset \cap_{l=0}^k\{\abs{\theta_l}\leq r\}.
\end{equation}
On the other hand, for any~$r>0$ with~$\bar{m}(3r/2)>0$ and~$N\in\mathbb{N}\setminus\{0\}$, we prove by induction that it holds for any~$i\in[0,N]\cap\mathbb{N}$ that
\begin{equation}\label{som}
\mathbb{E}[\abs{\theta_N}^2\mathds{1}_{\Omega\setminus \bar{\mathcal{K}}_{N,r}}] \leq e^{5iLT(T+h)}\mathbb{E}[(\abs{\theta_{N-i}} + \sqrt{2}T\textstyle\sum_{k=N-i}^{N-1}\abs{W_k})^2\mathds{1}_{\Omega\setminus \bar{\mathcal{K}}_{N,r}}].
\end{equation}
Let~$j\in[1,N]\cap\mathbb{N}$ and assume~\eqref{som} holds for~$i=j-1$. By~\eqref{nelem2} in Lemma~\ref{nelem},~\eqref{Mpr}, 
it holds a.s.\ that
\begin{equation*}
\abs{\theta_{N-j+1}} \leq e^{\frac{5}{2}LT(T+h)}
\tilde{\theta}_{N-j} \leq e^{\frac{5}{2}LT(T+h)}(\abs{\theta_{N-j}}+ \sqrt{2}T\abs{W_{N-j}}),
\end{equation*}
which may be substituted into the inductive assumption to obtain~\eqref{som} with~$i=j$. In particular,~\eqref{som} holds with~$i=N$. Therefore, by definition~\eqref{jf2} of~$\bar{\mathcal{K}}_{N,r}$, it holds for any~$r>0$ with~$\bar{m}(3r/2)>0$ and~$N\in\mathbb{N}\setminus\{0\}$ that
\begin{align}
&\mathbb{E}[\abs{\theta_N}^2
\mathds{1}_{\Omega\setminus\bar{\mathcal{K}}_{N,r}}] \nonumber\\
&\quad\leq e^{5NLT(T+h)}\mathbb{E}\bigg[\bigg(\abs{\theta_0} + \sqrt{2}T\sum_{k=0}^{N-1} \abs{W_k} \bigg)^2\mathds{1}_{\Omega\setminus\bar{\mathcal{K}}_{N,r}}\bigg]\nonumber\\
&\quad\leq e^{5NLT(T+h)} \mathbb{E}\bigg[\bigg(2\abs{\theta_0}^2 +4T^2\bigg(\sum_{k=0}^{N-1} \abs{W_k} \bigg)^{\!2}\,\bigg)\bigg(\mathds{1}_{\Omega\setminus\{\abs{\theta_0} \leq r\}} + \sum_{k=0}^{N-1}\mathds{1}_{\Omega\setminus \mathcal{K}_{k,r}}\bigg)\bigg].\label{jf20}
\end{align}
We treat the various terms in the expansion of the right-hand side of~\eqref{jf20}. 
By~\eqref{c0s}, it holds for any~$r\geq r_1$ and~$N\in\mathbb{N}\setminus\{0\}$ that
\begin{equation}\label{jf3}
\mathbb{E}\bigg[\bigg(2\abs{\theta_0}^2 + 4T^2\bigg(\sum_{k=0}^{N-1} \abs{W_k} \bigg)^{\!2}\,\bigg)\mathds{1}_{\Omega\setminus\{\abs{\theta_0} \leq r\}} \bigg] \leq 2c_2(r) + 4T^2c_0(r)N^2d.
\end{equation}
Moreover, if~$r>0$ satisfies~$\hat{n}_r\geq d$, then by standard concentration of the multivariate Gaussian, namely Theorem~5.3 in~\cite{MR1849347} with~$F=\abs{\cdot}$ and the inequality~$\mathbb{E}[\abs{W_k}]\leq \sqrt{d}$, it holds for any~$r\geq r_1$ and~$N\in\mathbb{N}\setminus\{0\}$ that
\begin{equation}\label{jf4}
\mathbb{E}\bigg[2\abs{\theta_0}^2  \sum_{k=0}^{N-1}\mathds{1}_{\Omega\setminus \mathcal{K}_{k,r}}\bigg] = 2N\mathbb{E}[\abs{\theta_0}^2] \mathbb{P}( \Omega\setminus \mathcal{K}_{0,r}) \leq 2N(r^2 + c_2(r))e^{-(\sqrt{\hat{n}_r} - \sqrt{d})^2/2}.
\end{equation}
For the remaining term on the right-hand side of~\eqref{jf20}, note that, applying again Theorem~5.3 in~\cite{MR1849347}, 
it holds for any~$k\in\mathbb{N}$,~$r>0$ and~$i\in\{0,1,2\}$ with~$\hat{n}_r\geq d$ that
\begin{equation*}
\mathbb{E}[\abs{W_k}^i\mathds{1}_{\Omega\setminus\mathcal{K}_{k,r}}] \leq (\mathbb{E}[\abs{W_k}^{2i}])^{1/2}(\mathbb{P}(\abs{W_k}> \hat{n}_r^{1/2}))^{1/2}\leq (d^2+2d)^{i/4} e^{-(\sqrt{\hat{n}_r}-\sqrt{d})^2/4},
\end{equation*}
which implies for any~$N\in\mathbb{N}\setminus\{0\}$ that
\begin{equation}\label{jf5}
\mathbb{E}\bigg[ 4T^2\bigg(\sum_{k=0}^{N-1} \abs{W_k} \bigg)^{\!2}\,\bigg( \sum_{k=0}^{N-1}\mathds{1}_{\Omega\setminus \mathcal{K}_{k,r}}\bigg)\bigg] \leq 4T^2N^3(d^2+2d)^{1/2}e^{-(\sqrt{\hat{n}_r} - \sqrt{d})^2/4}.
\end{equation}
Substituting~\eqref{jf3},~\eqref{jf4},~\eqref{jf5} into~\eqref{jf20} then using~\eqref{jf2} yields 
the assertion. 
\end{proof}
Before proceeding with the contraction result, the next Lemma~\ref{conprop0} provides the basic argument for neglecting the trajectories that have escaped a region of local curvature. 

\begin{lemma}\label{conprop0}
Assume all of the settings and notations of Lemma~\ref{conmain}. 
For any~$k\in\N $,~$r_3\geq 7T\sqrt{d}$ such that~$m(r_3)>0$, 
it holds that
\begin{align}
\mathbb{E}[\abs{\theta_{k+1} -\bar{\theta}_{k+1}}^2] &\leq e^{-2\bar{c}m(r_3)T^2}\mathbb{E}[\abs{\theta_k - \bar{\theta}_k}^2] +21LT^2\big(4\mathbb{E}[(\abs{\theta_k}^2 +\abs{ \bar{\theta}_k}^2)\nonumber\\
&\quad\cdot\mathds{1}_{\{\abs{\theta_k}>r_3/2\}\cup\{\abs{\bar{\theta}_k}>r_3/2\}}] +  r_3^2e^{-(r_3/(7T) - \sqrt{d})^2/2}\big).
\label{conpropeq}
\end{align}
\end{lemma}
\begin{proof}
Fix~$k\in\mathbb{N}$. 
For any~$\tilde{r}\geq 0$, let
\begin{equation*}
\hat{\mathcal{E}}_{k,\tilde{r}} = \{\mathbb{E}[\abs{\theta_{k+1}-\bar{\theta}_{k+1}}^2|\theta_k,\bar{\theta}_k,W_k]\leq e^{-2\bar{c}m(\tilde{r})T^2}\abs{\theta_k-\bar{\theta}_k}^2\}.
\end{equation*} 
By~\eqref{nelem3} in Lemma~\ref{nelem} and~\eqref{Mpr}, it holds for any~$\bar{r}>0$ with~$m(\bar{r})>0$ that
\begin{align*}
\hat{\mathcal{E}}_{k,\bar{r}} &\supset \{ e^{5/16}\tilde{\theta}_k \leq \bar{r}\}\cap  \{ e^{5/16}\tilde{\theta}_k'\leq \bar{r}\} \\
&\supset 
\{ \abs{\theta_k} + \sqrt{2}T\abs{W_k}  \leq e^{-5/16}\bar{r}\} \cap  \{ \abs{\bar{\theta}_k} + \sqrt{2}T\abs{W_k} \leq e^{-5/16}\bar{r}\},
\end{align*}
where 
for example
\begin{equation*}
\{ \abs{\theta_k} + \sqrt{2}T\abs{W_k} \leq e^{-5/16}\bar{r}\} \supset \{\abs{\theta_k} \leq \bar{r}/2\}\cap \{T\abs{W_k}\leq \bar{r}/7\}.
\end{equation*}
Therefore we have for any such~$\bar{r}>0$ that
\begin{equation}\label{esu}
\hat{\mathcal{E}}_{k,\bar{r}} \supset \{\abs{\theta_k}\leq \bar{r}/2\}\cap \{\abs{\bar{\theta}_k}\leq \bar{r}/2\} \cap \{T\abs{W_k}\leq \bar{r}/7\},
\end{equation}
where, by Gaussian concentration, namely Theorem~5.3 in~\cite{MR1849347} with~$F = \abs{\cdot}$ and the inequality~$\mathbb{E}[\abs{W_k}]\leq \sqrt{d}$,  if~$\bar{r}\geq7T\sqrt{d}$, then it holds that
\begin{equation}\label{kk2}
\mathbb{P}(T\abs{W_k}> \bar{r}/7) \leq \mathbb{P}(\abs{W_k}> \mathbb{E}[\abs{W_k}] + \bar{r}/(7T) - \sqrt{d}) \leq e^{-(\bar{r}/(7T) - \sqrt{d})^2/2}.
\end{equation}
Fix~$r_3\geq 7T\sqrt{d}$ (as in the assertion). 
We estimate the contraction by considering the event~$\hat{\mathcal{E}}_{k,r_3}$ and its complement separately. By definition of~$\hat{\mathcal{E}}_{k,r_3}$ and~\eqref{nelem3}, 
it holds 
that 
\begin{align}
&\mathbb{E}[\abs{\theta_{k+1}-\bar{\theta}_{k+1}}^2]\nonumber\\
&\quad\leq \mathbb{E}[e^{-2\bar{c}m(r_3)T^2}\abs{\theta_k-\bar{\theta}_k}^2\mathds{1}_{\hat{\mathcal{E}}_{k,r_3}}] + e^{10LT^2}\mathbb{E}[\abs{\theta_k-\bar{\theta}_k}^2\mathds{1}_{\Omega\setminus\hat{\mathcal{E}}_{k,r_3}}] \nonumber\\
&\quad= e^{-2\bar{c}m(r_3)T^2}\mathbb{E}[\abs{\theta_k-\bar{\theta}_k}^2] + (e^{10LT^2}-e^{-2\bar{c}m(r_3)T^2})\mathbb{E}[\abs{\theta_k-\bar{\theta}_k}^2\mathds{1}_{\Omega\setminus\hat{\mathcal{E}}_{k,r_3}}].\label{hp1}
\end{align}
For the coefficient in the second term on the right-hand side of~\eqref{hp1}, by~$LT^2\leq 1/8$, we have~$e^{10LT^2}-1 = e^{(10/8)\cdot 8LT^2}-1 \leq (e^{10/8}-1)\cdot 8LT^2$ and~$1-e^{-2\bar{c}m(r_3)T^2}\leq 2\bar{c}m(r_3)T^2\leq 2\bar{c}LT^2$, so that, recalling the definition~\eqref{ccdef} of~$\bar{c}$, 
\begin{equation}\label{jfr}
e^{10LT^2}-e^{-2\bar{c}m(r_3)T^2}\leq 21LT^2. 
\end{equation} 
By~\eqref{esu}, 
the second expectation on the right-hand side of~\eqref{hp1} may be bounded as
\begin{align*}
&\mathbb{E}[\abs{\theta_k-\bar{\theta}_k}^2\mathds{1}_{\Omega\setminus\hat{\mathcal{E}}_{k,r_3}}] \\
&\quad\leq 2\mathbb{E}[(\abs{\theta_k}^2+\abs{\bar{\theta}_k}^2)(\mathds{1}_{\{\abs{\theta_k}>r_3/2\}\cup\{\abs{\bar{\theta}_k}>r_3/2\}} +\mathds{1}_{\{T\abs{W_k}>r_3/7\}} )],
\end{align*}
where
\begin{align*}
&\mathbb{E}[(\abs{\theta_k}^2+\abs{\bar{\theta}_k}^2)\mathds{1}_{\{T\abs{W_k}>r_3/7\}} ] \\
&\quad= \mathbb{E}[(\abs{\theta_k}^2+\abs{\bar{\theta}_k}^2)\mathds{1}_{\{T\abs{W_k}>r_3/7\}} (\mathds{1}_{\{\abs{\theta_k}>r_3/2\}\cup\{\abs{\bar{\theta}_k}>r_3/2\}} \\
&\qquad+ \mathds{1}_{\{\abs{\theta_k}\leq r_3/2\}\cap\{\abs{\bar{\theta}_k}\leq r_3/2\}} ) ]\\
&\quad\leq \mathbb{E}[(\abs{\theta_k}^2+\abs{\bar{\theta}_k}^2)\mathds{1}_{\{\abs{\theta_k}>r_3/2\}\cup\{\abs{\bar{\theta}_k}>r_3/2\}}  + 2(r_3/2)^2\mathds{1}_{\{T\abs{W_k}>r_3/7\}} ],
\end{align*}
so that 
\begin{align*}
\mathbb{E}[\abs{\theta_k-\bar{\theta}_k}^2\mathds{1}_{\Omega\setminus\hat{\mathcal{E}}_{k,r_3}}] 
&\leq 4\mathbb{E}[(\abs{\theta_k}^2+\abs{\bar{\theta}_k}^2)\mathds{1}_{\{\abs{\theta_k}>r_3/2\}\cup\{\abs{\bar{\theta}_k}>r_3/2\}} ] \\
&\quad+ 4(r_3/2)^2\mathbb{P}(\{T\abs{W_k}>r_3/7\} ).
\end{align*}
Together with~\eqref{kk2},~\eqref{hp1},~\eqref{jfr} and~$L(T+h)^2\leq 1/8$, the proof concludes.
\end{proof}

Theorem~\ref{wama} below is our main contraction result for HMC and the Wasserstein distance.

\begin{theorem}\label{wama}
Assume the settings and notations of 
Lemma~\ref{conmain}. 
Let~$\hat{\pi}$ be the measure on~$\R^d$ given by~$\hat{\pi}(dx)= e^{-U(x)}dx/\int e^{-U}$. 
Assume that there exist~$\hat{r}_1\geq 0$,~$\hat{c}_i:[\hat{r}_1,\infty]\rightarrow[0,\infty)$ for~$i\in\{0,2\}$ such that
\begin{equation}\label{c0s2}
\textstyle\int_{\R^d\setminus B_r} \abs{x}^i\hat{\pi}(dx) \leq \hat{c}_i(r)\qquad \forall r\geq \hat{r}_1.
\end{equation}
Let
\begin{equation}\label{ldef}
\bar{\lambda} = \begin{cases}
1/17&\textrm{if }u=0\\
1/4&\textrm{if }u=1. 
\end{cases}
\end{equation}
Let~$(r_3,r_4,r_5)\in[\max(2(r_1\vee\hat{r}_1),7T\sqrt{d}),\infty]\times[r_1,\infty]\times[\hat{r}_1,\infty]$ 
satisfy~$m(r_3)>0$,~$\bar{m}(3r_4/2)>0$, 
let~$\bar{\epsilon}>0$ and let
\begin{equation}\label{Nsdef1}
N_*= 
0\vee \frac{2\ln(\bar{\epsilon}^{-1}W_2(\delta_{\theta_0},\hat{\pi}))}{\bar{\lambda}m(r_3)T^2}.\\
\end{equation}
If~$\min(\hat{n}_{r_3/2},\hat{n}_{r_4}) \geq d$ and~$N\geq N_*$, then it holds that
\begin{align}
W_2(\delta_{\theta_0}P^N,\hat{\pi})&\leq \bar{\epsilon} +\bigg(\frac{23\delta_*L}{\bar{\lambda} m(r_3)}\bigg)^{\frac{1}{2}} + \bigg(1+\frac{1}{\bar{\lambda}m(r_3)T^2}\bigg)e_h,\label{wamaeq}
\end{align}
where
\begin{subequations}\label{wes}
\begin{align}
\delta_*&= 4e^{5NL(T+h)^2}(\delta_{N,r_3/2} + \delta_{N,r_4} + 4r_5^2r_3^{-2}\delta_{N,r_3/2})+ 4r_4^2\hat{c}_0(r_3/2) \nonumber\\
&\quad+ 4\hat{c}_2(r_5) + 4\hat{c}_2(r_3/2) + r_3^2e^{-(r_3/(7T)-\sqrt{d})^2/2},\label{dsdef}\\
e_h&=\begin{cases}
2 (d/L + \textstyle\int_{\mathbb{R}^d}\abs{\cdot}^2d\hat{\pi})^{1/2}L^{1/2}h&\textrm{if }u=0,\\
71(\sqrt{d/L} + (\textstyle\int_{\mathbb{R}^d}\abs{\cdot}^2d\hat{\pi})^{1/2})L^{3/4}h^{3/2}&\textrm{if }u=1.
\end{cases}\label{wamacon}
\end{align}
\end{subequations}
\end{theorem}
\begin{remark}
The term~$\delta_*$ corresponds to the (small) probability that the trajectories escape Euclidean balls~$B_{r_3},B_{3r_4/2},B_{r_5}$. 
The term~$e_h$ corresponds to the asymptotic bias of the algorithm; it appears in the same way, in particular with the same rates in~$h$, as in previous analysis (see e.g.~\cite{chak2024r}).
\end{remark}
\begin{proof}[Proof of Theorem~\ref{wama}]
For any~$l\in\mathbb{N}$, 
let~$\hat{\pi}\bar{P}^l$ denote the distribution of~$\theta_l$ given~$\theta_0\sim\hat{\pi}$. 
By the triangle and Young's inequality, it holds for any~$l\in\N $ and~$C>0$ that
\begin{equation}\label{wl1}
(W_2(\delta_{\theta_0}P^{l+1},\hat{\pi}))^2 \leq (1+C)(W_2(\delta_{\theta_0}P^{l+1}, \hat{\pi}\bar{P}))^2 + (1+1/C) (W_2(\hat{\pi}\bar{P},\hat{\pi}))^2.
\end{equation}
Moreover, by Lemma~\ref{conmain},~$\delta_{\theta_0}P^l$ has finite second moment for any~$l\in\mathbb{N}$. 
By~\eqref{c0s2} and again Lemma~\ref{conmain},~$\hat{\pi}\bar{P}$ also has finite second moment. Therefore by a standard optimal transport result, namely~\cite[Theorem~4.1 with~$c(x,y)=\abs{x-y}^2$,~$a=b=0$]{MR2459454}, for any~$l\in\mathbb{N}$, there exists a coupling~$(\vartheta_l,\hat{\vartheta}_l)$ 
such 
that~$\vartheta_l,\hat{\vartheta}_l$ have respective distributions~$\delta_{\theta_0}P^l,\hat{\pi}$ and
\begin{equation}\label{wdep}
(W_2(\delta_{\theta_0}P^l,\hat{\pi}))^2 = \mathbb{E}[\abs{\vartheta_l - \hat{\vartheta}_l}^2].
\end{equation}
For any~$l\in\mathbb{N}$, setting~$\bar{\theta}_0 = \hat{\vartheta}_l$ then using Lemma~\ref{conprop0} (with~$k=0$,~$(\theta_k)_{k\in\N }$ replaced by~$(\vartheta_{l+k})_{k\in\mathbb{N}}$, and~$(\hat{u}_k)_k,(W_k)_{k\in\N }$ therein time-shifted, so that for example~$W_0$ is replaced by~$W_l$) and~\eqref{wdep}, the first square term on the right-hand side of~\eqref{wl1} satisfies 
\begin{align}
&(W_2(\delta_{\theta_0}P^{l+1}, \hat{\pi}\bar{P}))^2\nonumber \\
&\quad\leq \mathbb{E}[\abs{q_{l,T/h} - \bar{q}_{0,T/h}}^2]\nonumber\\
&\quad\leq e^{-2\bar{c}m(r_3)T^2}(W_2(\delta_{\theta_0}P^l,\hat{\pi}))^2 \nonumber\\
&\qquad+ 21LT^2\big(4\mathbb{E}[(\abs{\vartheta_l}^2+\abs{\hat{\vartheta}_l}^2)\mathds{1}_{\{\abs{\vartheta_l}>\frac{r_3}{2}\}\cup\{\abs{\hat{\vartheta}_l}>\frac{r_3}{2}\}}] + r_3^2e^{-(\frac{r_3}{7T} - \sqrt{d})^2/2}\big).
\label{mr1}
\end{align}
We estimate the first expectation on the right-hand side of~\eqref{mr1}. For the indicator function therein, note that~$\mathds{1}_{A\cup B}\leq \mathds{1}_A+\mathds{1}_B$ for any events~$A,B$. It holds that
\begin{align*}
\mathbb{E}[\abs{\vartheta_l}^2\mathds{1}_{\{\abs{\hat{\vartheta}_l}>r_3/2\}}]
&= \mathbb{E}[\abs{\vartheta_l}^2\mathds{1}_{\{\abs{\hat{\vartheta}_l}>r_3/2\}}(\mathds{1}_{\abs{\vartheta_l}\leq r_4}+\mathds{1}_{\abs{\vartheta_l}> r_4})]\\
&\leq r_4^2\mathbb{P}(\abs{\hat{\vartheta}_l}>r_3/2) + \mathbb{E}[\abs{\vartheta_l}^2\mathds{1}_{\abs{\vartheta_l}> r_4}],
\end{align*}
which implies, by~\eqref{c0s2} and Lemma~\ref{conmain}, 
that
\begin{equation}\label{mt1}
\mathbb{E}[\abs{\vartheta_l}^2\mathds{1}_{\{\abs{\hat{\vartheta}_l}>r_3/2\}}] \leq r_4^2\hat{c}_0(r_3/2) + e^{5NL(T+h)^2}\delta_{l,r_4}.
\end{equation}
Along the same lines, it holds
for any~$l\in\N $ that
\begin{equation*}
\mathbb{E}[\abs{\hat{\vartheta}_l }^2\mathds{1}_{\{\abs{\vartheta_l}>r_3/2\}}] \leq r_5^2\mathbb{P}(\abs{\vartheta_l}>r_3/2) + \hat{c}_2(r_5),
\end{equation*}
which, by Lemma~\ref{conmain}, implies 
\begin{align}
\mathbb{E}[\abs{\hat{\vartheta}_l }^2\mathds{1}_{\{\abs{\vartheta_l}>r_3/2\}}] &\leq r_5^2(r_3/2)^{-2}\mathbb{E}[\abs{\vartheta_l}^2\mathds{1}_{\{\abs{\vartheta_l}>r_3/2\}}] + \hat{c}_2(r_5)\nonumber\\
&\leq r_5^2(r_3/2)^{-2}e^{5NL(T+h)^2}\delta_{l,r_3/2} + \hat{c}_2(r_5).\label{mt2}
\end{align}
Gathering~\eqref{mt1},~\eqref{mt2} and using again Lemma~\ref{conmain} as well as~\eqref{c0s2}, it holds 
for any~$l\in\N $ 
that
\begin{equation}\label{mt3}
\mathbb{E}[(\abs{\vartheta_l}^2+\abs{\hat{\vartheta}_l}^2)\mathds{1}_{\{\abs{\vartheta_l}>r_3/2\}\cup\{\abs{\hat{\vartheta}_l}>r_3/2\}}] \leq \hat{\delta}(l,r_3,r_4,r_5),
\end{equation}
where~$\hat{\delta}:\N \times[2(r_1\vee \hat{r}_1),\infty]\times[r_1,\infty]\times[\hat{r}_1,\infty]\rightarrow[0,\infty)$ is given for any~$(k,\tilde{r}_3,\tilde{r}_4,\tilde{r}_5)\in\N \times[2(r_1\vee \hat{r}_1),\infty]\times[r_1,\infty]\times[\hat{r}_1,\infty]$ by
\begin{align*}
\hat{\delta}(k,\tilde{r}_3,\tilde{r}_4,\tilde{r}_5) &= e^{5NL(T+h)^2}\delta_{k,\tilde{r}_3/2} + \tilde{r}_4^2\hat{c}_0(\tilde{r}_3/2) + e^{5NL(T+h)^2}\delta_{k,\tilde{r}_4}  \\
&\quad+ \tilde{r}_5^2(\tilde{r}_3/2)^{-2}e^{5NL(T+h)^2}\delta_{k,\tilde{r}_3/2}+ \hat{c}_2(\tilde{r}_5) + \hat{c}_2(\tilde{r}_3/2).
\end{align*}
By definition~\eqref{drdef} of~$\delta_{\cdot,\cdot}$, 
it holds for any~$N\in\N $,~$(\tilde{r}_3,\tilde{r}_4,\tilde{r}_5)\in[2(r_1\vee \hat{r}_1),\infty]\times[r_1,\infty]\times[\hat{r}_1,\infty]$ and~$k\in[0,N]\cap\N $ that~$\hat{\delta}(k,\tilde{r}_3,\tilde{r}_4,\tilde{r}_5)\leq \hat{\delta}(N,\tilde{r}_3,\tilde{r}_4,\tilde{r}_5)$. 
Therefore, substituting~\eqref{mt3} into~\eqref{mr1}, 
it holds 
that
\begin{align}
&(W_2(\delta_{\theta_0}P^{l+1}, \hat{\pi}\bar{P}))^2\nonumber \\
&\quad\leq e^{-2\bar{c}m(r_3)T^2}(W_2(\delta_{\theta_0}P^l,\hat{\pi}))^2 \nonumber\\
&\qquad+ 21LT^2\big[4\hat{\delta}(N,r_3,r_4,r_5) +  r_3^2e^{-(r_3/(7T) - \sqrt{d})^2/2}\big].
\label{jri}
\end{align}
Note that the quantity in the square brackets in~\eqref{jri} is equal to~$\delta_*$ given by~\eqref{dsdef}. 
For any~$N\in\N $, 
substituting~\eqref{jri} into~\eqref{wl1} with~$C=\bar{c}m(r_3)T^2$ and applying induction yields that
\begin{align*}
&(W_2(\delta_{\theta_0}P^N,\hat{\pi}))^2\\
&\quad\leq e^{-\bar{c}Nm(r_3)T^2}(W_2(\delta_{\theta_0},\hat{\pi}))^2 + \bigg(\sum_{i=0}^{\infty}e^{-\bar{c}im(r_3)T^2}\,\bigg)\\
&\qquad\cdot\big[21\delta_*e^{\bar{c}m(r_3)T^2}LT^2
+ (1+1/(\bar{c}m(r_3)T^2))(W_2(\hat{\pi}\bar{P},\hat{\pi}))^2 \big].
\end{align*}
We have~$\sum_{i=0}^{\infty}e^{-\bar{c}im(r_3)T^2}=1/(1-e^{-\bar{c}m(r_3)T^2})\leq1/(\bar{\lambda}m(r_3)T^2)$, where~$\bar{\lambda}$ is given by~\eqref{ldef} (and recall definition~\eqref{ccdef} of~$\bar{c}$). Therefore, for any~$\bar{\epsilon}>0$ and~$N\geq N_*$ (defined in~\eqref{Nsdef1}), 
we obtain
\begin{align}
&(W_2(\delta_{\theta_0}P^N,\hat{\pi}))^2\nonumber\\
&\quad\leq \bar{\epsilon}^2 + (23\delta_* /\bar{\lambda})(L/m(r_3))
\nonumber\\
&\qquad 
+ (1/(\bar{\lambda}m(r_3)T^2))(1+1/(\bar{c}m(r_3)T^2))(W_2(\hat{\pi}\bar{P},\hat{\pi}))^2.\label{kno0}
\end{align}
It remains to give a suitable bound for~$W_2(\hat{\pi}\bar{P},\hat{\pi})$. Note that estimating this quantity is relatively standard for bounding the asymptotic bias in inexact MCMC, see for example~\cite{durmus2023asymptotic}. We will make use of calculations already done in the literature where possible. 
In case~$u=1$ (randomized midpoint integrator), 
by Lemma~7 in~\cite{bourabee2022unadjusted} (note that~$\nabla U(0)=0$ suffices in place of A.1 therein), it holds that
\begin{align}
W_2(\hat{\pi}\bar{P},\hat{\pi}) 
&\leq 71(\mathbb{E}[\abs{W_0}] + \sqrt{L}(\hat{\pi}(\abs{\cdot}^2))^{1/2})L^{1/4}h^{3/2}\nonumber\\
&\leq 71(\sqrt{d} + \sqrt{L}
(\hat{\pi}(\abs{\cdot}^2))^{1/2}
)L^{1/4}h^{3/2},\label{kn0}
\end{align}
where~$\hat{\pi}(\abs{\cdot}^2)$ denotes~$\int_{\mathbb{R}^d}\abs{\cdot}^2d\hat{\pi}$. 
In case~$u=0$ (Verlet integrator), let~$[0,T]\times \R^d\times\R^d\ni (t,x,v)\mapsto (\hat{q}_t(x,v),\hat{p}_t(x,v))$ be such that for any~$x,v\in\R^d$, the function~$[0,T]\ni t\mapsto(\hat{q}_t,\hat{p}_t) = (\hat{q}_t(x,v),\hat{p}_t(x,v))$ is the solution to the exact Hamiltonian flow
\begin{equation*}
d\hat{q}_t/dt = \hat{p}_t,\qquad d\hat{p}_t/dt = -\nabla U(\hat{q}_t)
\end{equation*}
with~$(\hat{q}_0,\hat{p}_0) = (x,v)$. 
For any~$i\in[0,T/h]\cap\N $ and any measure~$\tilde{\pi}$ on~$\R^{2d}$, let~$\tilde{\pi}P_h^i,\tilde{\pi}\hat{P}_h^i$ denote the respective distributions of the~$\R^{2d}$-valued 
r.v.'s~$(q_{0,i},p_{0,i})$ and~$(\hat{q}_{ih}(\theta_0,W_0),\hat{p}_{ih}(\theta_0,W_0))$ 
both given~$\theta_0$ distributed according to the marginal of~$\tilde{\pi}$ over the first~$d$ coordinates. The notation~$P_h$ should not be confused with~$P$ or~$\bar{P}$, which we recall from just after~\eqref{qpdef} to denote full transitions with~$T/h$ steps of step-size~$h$. 
Let~$\gamma$ denote the standard Gaussian distribution on~$\R^d$. 
Moreover, let~$W_{2,e}$ denote the~$L^2$ Wasserstein distance w.r.t.\ the distance~$\R^{2d}\times\R^{2d} \ni ((x,v),(y,w))\mapsto \sqrt{\abs{x-y}^2 + \abs{v-w}^2/L}$. 
It holds that
\begin{align}
W_2(\hat{\pi}\bar{P},\hat{\pi}) 
&\leq W_{2,e}((\hat{\pi}\times \gamma)P_h^{T/h},(\hat{\pi}\times \gamma)\hat{P}_h^{T/h})\nonumber\\
&\leq W_{2,e}((\hat{\pi}\times \gamma)P_h^{T/h},(\hat{\pi}\times \gamma)\hat{P}_h^{T/h-1}P_h)\nonumber\\
&\quad + W_{2,e}((\hat{\pi}\times \gamma)\hat{P}_h^{T/h-1}P_h,(\hat{\pi}\times \gamma)\hat{P}_h^{T/h}).\label{pl1}
\end{align}
For the first term on the right-hand side of~\eqref{pl1}, by Lemma~\ref{veli} (taking square roots on both sides in the inequality therein), it holds for any~$i\in[1,T/h]\cap\N $ that
\begin{equation}\label{pz1}
W_{2,e}((\hat{\pi}\times \gamma)P_h^i,(\hat{\pi}\times \gamma)\hat{P}_h^{i-1}P_h) \leq (1+2\sqrt{L}h)W_{2,e}((\hat{\pi}\times \gamma)P_h^{i-1},(\hat{\pi}\times \gamma)\hat{P}_h^{i-1}).
\end{equation}
For any~$x,v\in\mathbb{R}^d$, 
let~$(\tilde{q}_i,\tilde{p}_i)_{i\in[0,T/h]\cap\mathbb{N}}=(\tilde{q}_i(x,v),\tilde{p}_i(x,v))_{i\in[0,T/h]\cap\mathbb{N}}$ be given inductively by~\eqref{qpdef} with~$q_{k,i},p_{k,i}$ replaced by~$\tilde{q}_i,\tilde{p}_i$ for all~$i$ and with~$\tilde{q}_0=x$ and~$\tilde{p}_0=v$. 
For the second term on the right-hand side of~\eqref{pl1}, from the proof of the first assertion to Lemma~29 in~\cite{MR4309974} (reading off the bounds, under the notation of~\cite{MR4309974}, for~$\abs{\hat{x}_t-y_t}$ and~$\abs{\hat{v}_{\delta}-\tilde{w}_{\delta}}$ therein\footnote{Note that there is a small inconsequential sign error in the definition of~$\tilde{w}_{\delta}$.} w.r.t.~$Z$), 
it holds for any~$x,v\in\R^d$ that 
\begin{align*}
\abs{\tilde{q}_h(x,v) - \hat{q}_h(x,v)}&\leq e^{Lh^2}(Lh^3/6)(\abs{v} + (h/2)\abs{\nabla U(x)})\\
&\leq (e^{1/8}/(6\sqrt{8}))Lh^2(\abs{v}/\sqrt{L} + (\sqrt{L}h/2)\abs{x}),\\
\abs{\tilde{p}_h(x,v) - \hat{p}_h(x,v)}/\sqrt{L} &\leq (e^{Lh^2}(L^2h^4/24) + Lh^2) (\abs{v} + (h/2)\abs{\nabla U(x)})/\sqrt{L}\\
&\leq e^{1/8}(1/192 + 1)Lh^2(\abs{v}/\sqrt{L} + (\sqrt{L}h/2)\abs{x}),
\end{align*}
so that, by using Young's inequality to obtain~$(\abs{v}/\sqrt{L} + (\sqrt{L}h/2)\abs{x})^2\leq (1+Lh^2/4)(\abs{v}^2/L + \abs{x}^2)$, we have
\begin{align*}
&\abs{\tilde{q}_h(x,v) - \hat{q}_h(x,v)}^2 + \abs{\tilde{p}_h(x,v) - \hat{p}_h(x,v)}^2/L\\
&\quad \leq (e^{1/4}/288 + e^{1/4}(193/192)^2)(Lh^2)^2(1+Lh^2/4)(\abs{v}^2/L + \abs{x}^2)\\
&\quad\leq 2(Lh^2)^2(\abs{v}^2/L+\abs{x}^2).
\end{align*}
Therefore, for any~$i\in[1,T/h]\cap\N $, since~$(\hat{\pi}\times\gamma)\hat{P}_h^{i-1}$ is the same distribution as~$\hat{\pi}\times\gamma$, it holds that
\begin{align}
W_{2,e}((\hat{\pi}\times \gamma)\hat{P}_h^{i-1}P_h,(\hat{\pi}\times \gamma)\hat{P}_h^{i}) &\leq \sqrt{2}Lh^2 \bigg( \int_{\R^d} \frac{\abs{v}^2}{L} \gamma(dv) + \int_{\R^d}\abs{x}^2 \hat{\pi}(dx)\bigg)^{\!\frac{1}{2}}\nonumber\\
&\leq \sqrt{2}Lh^2(d/L + 
\hat{\pi}(\abs{\cdot}^2)
)^{1/2}.\label{pz2}
\end{align} 
Substituting~\eqref{pz1} and~\eqref{pz2} into~\eqref{pl1}, then repeating the arguments~$T/h$ times yields
\begin{equation*}
W_2(\hat{\pi}\bar{P},\hat{\pi}) \leq \big(\textstyle\sum_{i=0}^{T/h-1} (1+2\sqrt{L}h)^i\big) \cdot \sqrt{2}Lh^2(d/L + 
\hat{\pi}(\abs{\cdot}^2)
)^{1/2},
\end{equation*}
where the inequality~$\sum_{i=0}^{T/h-1} (1+2\sqrt{L}h)^i \leq (1+2\sqrt{L}h)^{T/h-1}\sum_{i=0}^{\infty}(1+2\sqrt{L}h)^{-i}\leq e^{2\sqrt{L}(T-h)}(1-(1+2\sqrt{L}h)^{-1})^{-1}\leq  e^{2\sqrt{L}T}/(2\sqrt{L}h) \leq e^{2/\sqrt{8}}/(2\sqrt{L}h)$ implies
\begin{equation}\label{kno2}
W_2(\hat{\pi}\bar{P},\hat{\pi}) \leq 2\sqrt{L}h (d/L + 
\hat{\pi}(\abs{\cdot}^2)
)^{1/2}.
\end{equation}
Thus the proof 
concludes by substituting~\eqref{kn0} or~\eqref{kno2} into~\eqref{kno0} in the respective cases of~$u\in\{0,1\}$. 
\end{proof}

In the following Corollary~\ref{cor1}, we interpret Theorem~\ref{wama} in the regression context, and write down the resulting complexities. In particular, we collect our assumptions for a smoothness constant~$L$ (and strong convexity profile~$m$) that varies with a parameter~$n\in\mathbb{N}$ and with~$d$ such that
\begin{equation*}
L \propto \lambda n
\end{equation*}
for some scale parameter~$\lambda>0$ possibly depending on~$d$. This~$n$ is the number of datapoints in the regression context. 
In that setting, we have in mind for example~$\lambda=1$ if~$X_i\sim N(0,I_d)$, but~$\lambda=1/d$ if~$X_i\sim N(0,d^{-1}I_d)$; more comments are made in Remark~\ref{cc}\ref{lre} below. 
We fixate on~$u=1$, that is HMC with a randomized midpoint integrator, in order to obtain the improved dependence on the target accuracy~$\epsilon$. 

We fix the integration time~$T$ to saturate the condition~$L(T+h)^2\leq 1/8$ (a lower bound on~$T$ of this order is necessary, see Remark~\ref{remw}\ref{remw2}). We assume a deterministic starting position~$\theta_0$ (with distance~$O(\lambda^{-1/2})$ from the minimum of~$U$). 
As with the previous statements in this Section~\ref{HMC}, we assume the notations for HMC from the beginning of the section.

\begin{corollary}\label{cor1}
Let~$U_0:\mathbb{R}^d\rightarrow\mathbb{R}$. 
Assume~$U,U_0$ satisfy Conditions~\ref{cond:curvature3},~\ref{cond:smooth},~\ref{cond:curvature2}. 
Let~$n,N\in\mathbb{N}\setminus\{0\}$, let~$L_{\textrm{sc}},\lambda>0$ and assume~$L=L_{\textrm{sc}}\lambda n$. Let~$\bar{m}:[0,\infty]\rightarrow\mathbb{R}$ be given by~\eqref{ccdef}. Let~$\hat{\pi}$ be the measure on~$\R^d$ given by~$\hat{\pi}(dx)= e^{-U(x)}dx/\int e^{-U}$. 
Denote
\begin{equation}\label{kaps}
\kappa = \frac{L}{m\big(\frac{8}{3(L_{\textrm{sc}}\lambda)^{1/2}}\big)},\qquad
\kappa_0 = \frac{L}{m_0\big(\frac{4}{3(L_{\textrm{sc}}\lambda)^{1/2}}\big)},\qquad
\bar{\kappa} = \frac{L}{\bar{m}\big(\frac{2}{(L_{\textrm{sc}}\lambda)^{1/2}}\big)}
\end{equation}
and assume~$\kappa,\bar{\kappa},\kappa_0\in(0,\infty)$. 
For any~$\epsilon\in(0,(L_{\textrm{sc}}\lambda)^{-1/2}]$, if 
\begin{align}
&n \geq \max\bigg(50\bar{\kappa} ^2 \bigg[d^{\frac{1}{2}} + 2\bigg(\frac{5N}{8}
+\ln\bigg(82800\bigg(\frac{32}{9}+\frac{N^2}{30}\bigg)\frac{N\kappa}{L_{\textrm{sc}}\lambda\epsilon^2}\bigg)\bigg)^{\!\!\frac{1}{2}}\bigg]^2, \frac{81\kappa_0d}{16}, \nonumber\\
&\qquad 
\frac{3\kappa_0}{2}\bigg[\!\ln\!\bigg(\frac{4\cdot18^2\cdot82800\cdot\kappa }{(3/4)^2L_{\textrm{sc}}\lambda\epsilon^2}\bigg(\sup_{\mathbb{R}^d\setminus B_{r_*}} \!\! e^{-U_0} \!\bigg)\Big(\inf_{B_{r_*}}e^{-U_0}\Big)^{-1}\bigg)\! + d(2+\ln( \kappa_0^{\frac{1}{2}}))\bigg]\bigg),\label{coreq3}
\end{align}
where~$r_* = (4/3)(L_{\textrm{sc}}\lambda)^{-1/2}$, hold 
and if the algorithmic parameters satisfy~$u = 1$,~$T+h=(8L)^{-1/2}$,~$\theta_0\in\mathbb{R}^d$ with~$\abs{\theta_0}\leq (4/3)(L_{\textrm{sc}}\lambda)^{-1/2}$ and
\begin{align}
&N \geq 35\kappa 
\ln(178/(L_{\textrm{sc}}\lambda\epsilon^2)),\label{coreq5}\\
&\sqrt{L}h \leq \frac{(L_{\textrm{sc}}\lambda\epsilon^2)^{1/3}(1+35\kappa )^{-2/3}}{32},\label{coreq4}
\end{align}
then it holds that~$W_2(\delta_{\theta_0}P^N,\hat{\pi}) \leq \epsilon$.
\end{corollary}
\begin{remark}\label{cc}
\begin{enumerate}[label=(\roman*)]
\item \label{cc1} The algorithm with parameters as in Corollary~\ref{cor1} requires 
at most~$NT/h$ gradient evaluations. 
For the scale invariant desideratum~$\epsilon=\hat{\epsilon}/L^{1/2} = \hat{\epsilon}/(L_{\textrm{sc}}\lambda n)^{1/2}$ for some~$\hat{\epsilon}>0$, if~$L_{\textrm{sc}}$ is an absolute constant, then~$NT/h$ scales like
\begin{equation}\label{nth}
\frac{NT}{h} = \frac{N\sqrt{L}T}{\sqrt{L}h} = \tilde{O}\bigg(\frac{\kappa ^{5/3}n^{1/3}}{\hat{\epsilon}^{2/3}} \bigg).
\end{equation}
If~$n\propto d$, 
then~\eqref{nth} scales like~$d^{1/3}$ in terms of dimension, which is the complexity attained in~\cite{bourabee2022unadjusted} in the globally strongly convex case. 
The slightly worsened~$n^{1/3}$ dependence in~\eqref{nth} is a consequence of our local curvature and concentration considerations. 
\item \label{lre} The rescaling~$L=L_{\textrm{sc}}\lambda n$
and Conditions~\ref{cond:curvature3},~\ref{cond:smooth},~\ref{cond:curvature2} are motivated by our results on GLMs with~$X_i\sim N(0,\Sigma)$ for some~$\Sigma\in\mathbb{S}^{d\times d}$ in Section~\ref{csglm}. 
The results therein show for~$\lambda=\lambda_{\max}(\Sigma)$
that~$m_0(c_1\lambda^{-1/2})$ and, under suitable scaling assumptions,~$L,m(c_2\lambda^{-1/2}),\bar{m}(c_3\lambda^{-1/2})$, 
all scale like~$n\bar{\lambda}$ for constants~$c_1,c_2,c_3>0$ as~$d,n\rightarrow\infty$ and with~$\bar{\lambda}\in\{\lambda_{\max}(\Sigma),\lambda_{\min}(\Sigma)\}$. In particular, we have in mind for the quantities
\begin{equation*}
L_{\textrm{sc}},m_{\textrm{sc}}(c_1(L_{\textrm{sc}}\lambda)^{-1/2}),\bar{m}_{\textrm{sc}}(c_2(L_{\textrm{sc}}\lambda)^{-1/2}),
m_{0,\textrm{sc}}(c_1(L_{\textrm{sc}}\lambda)^{-1/2}),\kappa ,\bar{\kappa},\kappa_0
\end{equation*}
to be scale invariant, in the sense that they are constant as~$d,n\rightarrow\infty$. 
\item \label{cc3} We comment on the reasonableness of the assumption that~$U,U_0$ satisfy both Conditions~\ref{cond:curvature3} and~\ref{cond:curvature2}. Condition~\ref{cond:curvature2} has been used as the underlying assumption for our results on HMC in this Section~\ref{HMC}; Condition~\ref{cond:curvature3} is used in order to assert the negligibility of~$\delta_*$ as defined in~\eqref{dsdef}, where we use the concentration Lemma~\ref{tvv}. 
The main point to note here is that Condition~\ref{cond:curvature3} implies~$\nabla (U-U_0)|_0=0$, but the curvature Condition~\ref{cond:curvature2} in this Section~\ref{HMC} asserts~$\nabla U(0)=0$. In the regression context, this would correspond to~$U_0$ having a critical point at the maximum a posteriori. If~$U_0$ satisfies~$\pi(dx)\propto e^{-U_0(x)}dx$ (it is the negative log-density of the prior), then 
this is not reasonable. 
For this reason, when we apply Corollary~\ref{cor1} later in the proof of the main Theorem~\ref{mainhmc}, we will set~$U_0$ to be a flattened version of the negative log-density of the prior. 
Alternatively, one could assume instead that~$U$ satisfies Condition~\ref{cond:curvature2} and~$U(\cdot+\hat{x})$ satisfies Condition~\ref{cond:curvature3} (in place of~$U$ therein) for some~$\hat{x}\in\mathbb{R}^d$ ($\hat{x}$ being the difference between the maximum a posteriori and the maximum likelihood estimator). 
\end{enumerate}

\end{remark}
\begin{proof}
Below we use (possibly without further mention), 
by~\eqref{coreq4}, that
\begin{equation}\label{lth}
\sqrt{L}T 
= \bigg(1-\frac{\sqrt{L}h}{\sqrt{L}(T+h)}\bigg)\sqrt{L}(T+h) \geq \bigg(1-\frac{1/72}{1/\sqrt{8}}\bigg)\sqrt{L}(T+h).
\end{equation}
The assertion follows from Theorem~\ref{wama} with
\begin{equation}\label{reqs}
r_1=\hat{r}_1=r_4 = r_5 = \frac{r_3}{2} = \frac{4}{3\sqrt{L_{\textrm{sc}}\lambda}},\qquad \bar{\epsilon} = \epsilon/5. 
\end{equation}
Firstly, note that by~\eqref{reqs} and (the second argument in)~\eqref{coreq3}, we have~$\hat{r}_1^2m_0(\hat{r}_1) = (4/3)^2L^{-1}nm_0(\hat{r}_1) = (4/3)^2n\kappa_0^{-1}$ and
\begin{equation*}
\hat{r}_1 \sqrt{m_0(\hat{r}_1)/d\,} = (4/3)L^{-1/2}(n/d)^{1/2} \sqrt{m_0(\hat{r}_1)\,} = (4/3)(n/d)^{1/2} \kappa_0^{-1/2}\geq 3.
\end{equation*}
Therefore by Condition~\ref{cond:curvature3} and Lemma~\ref{tvv}, it holds for any~$r\geq \hat{r}_1$ that
\begin{equation*}
\int_{\mathbb{R}^d\setminus B_r} \abs{x}^i \hat{\pi}(dx) \leq \int_{\mathbb{R}^d\setminus B_{\hat{r}_1}} \abs{x}^i \hat{\pi}(dx) \leq C_i\hat{r}_1^i\cdot e^{2d-cn}(1+\kappa_0^{d/2}e^{2d-cn}),
\end{equation*}
where
\begin{align*}
C_i =\bigg(\sup_{\mathbb{R}^d\setminus B_{\hat{r}_1}} e^{-U_0} \bigg)\Big(\inf_{B_{\hat{r}_1}}e^{-U_0}\Big)^{-1}(6(1+i/d))^i, \qquad c = \frac{2\kappa_0^{-1}}{3}.
\end{align*}
By the last argument in~\eqref{coreq3}, we have
\begin{align}
1+\kappa_0^{d/2}e^{2d-cn}&\leq 2\nonumber\\
\textrm{and so}\quad\textstyle \int_{\mathbb{R}^d\setminus B_
r} \abs{x}^i\hat{\pi}(dx) &\leq \begin{cases}
\epsilon^2\kappa^{-1}\hat{r}_1^{-2}/82800&\textrm{if }i=0,\\
\epsilon^2\kappa^{-1}/82800&\textrm{if }i=2.
\end{cases}\label{coreq1}
\end{align}
With this, we verify next the assumptions of Theorem~\ref{wama}. By definition~\eqref{nrdef} of~$\hat{n}_{\cdot}$ and using~$L=L_{\textrm{sc}}\lambda n$,~\eqref{lth},~\eqref{kaps}, we have
\begin{equation}\label{nhex}
\hat{n}_{r_3/2} = \hat{n}_{r_4} = \bigg(\frac{5}{16}\bigg)^2\cdot\frac{16}{9}\cdot\frac{1}{L_{\textrm{sc}}\lambda}\cdot\bigg(\frac{T}{T+h}\bigg)^2\cdot\frac{1}{8L}\cdot \bigg(\bar{m}\bigg(\frac{2}{\sqrt{L_{\textrm{sc}}\lambda}}\bigg)\bigg)^2 \geq \frac{n}{50\bar{\kappa} ^2},
\end{equation}
so that~$\min(\hat{n}_{r_3/2},\hat{n}_{r_4})\geq d$ by~\eqref{coreq3}. 
By the assumption~$\abs{\theta_0}\leq (4/3)(L_{\textrm{sc}}\lambda)^{-1/2}$, condition~\eqref{c0s} holds with~$r_1$ as in~\eqref{reqs} and
\begin{equation}\label{c020}
c_0=c_2=0.
\end{equation}
By~\eqref{coreq1}, condition~\eqref{c0s2} holds with~$\hat{r}_1$ as in~\eqref{reqs} and~$\hat{c}_i$ 
given by
\begin{equation}\label{jo2}
\hat{c}_i(r) =  
\begin{cases}
\epsilon^2\kappa^{-1}\hat{r}_1^{-2}/82800&\textrm{if }i=0,\\
\epsilon^2\kappa^{-1}/82800&\textrm{if }i=2.
\end{cases}
\end{equation}
Therefore, by definition~\eqref{Nsdef1} of~$N_*$, the inequality
\begin{align}
(W_2(\delta_{\theta_0},\hat{\pi}))^2 &\leq 2\textstyle\int_{\R^d}\abs{x}^2\delta_{\theta_0}(dx) + 2\textstyle\int_{\R^d}\abs{x}^2\hat{\pi}(dx) \nonumber\\
&\leq 2(r_1^2+c_2(r_1) +\hat{r}_1^2+\hat{c}_2(\hat{r}_1)),\label{w2a}
\end{align}
the assumption~$L(T+h)^2=1/8$,~$\epsilon\in(0,(L_{\textrm{sc}}\lambda)^{-1/2}]$,~\eqref{reqs} and~\eqref{ldef}, it holds that
\begin{align*}
N_*& \leq \frac{\ln(2\bar{\epsilon}^{-2}(r_1^2+ c_2(r_1) +\hat{r}_1^2 + \hat{c}_2(\hat{r}_1)))}{\bar{\lambda}m(r_3)T^2}\\
&\leq 4\kappa (LT^2)^{-1}\ln(2\epsilon^{-2}\cdot25 \cdot(2(4/3)^2/(L_{\textrm{sc}}\lambda) +\epsilon^2/82800))\\
&\leq 4\kappa (1-\sqrt{8}/72)^{-2}\cdot 8\ln(2\epsilon^{-2}\cdot25 \cdot(2(4/3)^2 +1/82800)/(L_{\textrm{sc}}\lambda))\\
&\leq 35\kappa \ln(178/(L_{\textrm{sc}}\lambda\epsilon^2)),
\end{align*}
which implies~$N\geq N_*$ by the assumption~\eqref{coreq5} on~$N$. 
By definition~\eqref{dsdef} of~$\delta_*$ and using~\eqref{jo2}, it holds that
\begin{equation}\label{dej}
\delta_* = 12e^{(5/8)N}\delta_{N,r_4} 
+ 12\epsilon^2\kappa^{-1}/82800
+ 4r_4^2e^{-(2r_4/(7T) - \sqrt{d})^2/2},
\end{equation}
For the first term on the right-hand side of~\eqref{dej}, by definition~\eqref{drdef} of~$\delta_{\cdot,\cdot}$,~\eqref{c020} and~\eqref{coreq3} (so that~$(d^2+2d)^{1/2}\leq d+1\leq 2d\leq n/15$ by~$\bar{\kappa} \geq 1$), it holds that
\begin{align*}
\delta_{N,r_4} &\leq ((32/9)N/(L_{\textrm{sc}}\lambda) + N^3(d^2+2d)^{1/2}/(2L_{\textrm{sc}}\lambda n))e^{-(\sqrt{\hat{n}_{r_4}} - \sqrt{d})^2/4}\\
&\leq (32/9 + N^2/30)(N/(L_{\textrm{sc}}\lambda))e^{-(\sqrt{\hat{n}_{r_4}} - \sqrt{d})^2/4},
\end{align*}
where by~\eqref{nhex} and (the first argument in)~\eqref{coreq3} we have
\begin{equation}\label{jo1}
\delta_{N,r_4} \leq \epsilon^2e^{-(5/8)N}\kappa ^{-1}/82800 \\
\end{equation}
For the last term on the right-hand side of~\eqref{dej}, using again (the first argument in)~\eqref{coreq3} together 
with~$\bar{\kappa}\geq 1$ 
and~$N\geq 1$, which imply
\begin{equation*}
n\geq \bigg(d^{\frac{1}{2}} + \bigg[2\ln\bigg(\frac{147200\kappa }{3L_{\textrm{sc}}\lambda\epsilon^2}\bigg)\bigg]^{\!\frac{1}{2}\,}\bigg)^2,
\end{equation*}
it holds that
\begin{align}
r_4^2e^{-(2r_4/(7T) - \sqrt{d})^2/2} &\leq (16/9)(L_{\textrm{sc}}\lambda)^{-1}e^{-((8\sqrt{8}/21)\sqrt{n}-\sqrt{d})^2/2}\nonumber \\
&\leq \epsilon^2\kappa ^{-1}/27600.\label{jo3}
\end{align}
Gathering~\eqref{jo1},~\eqref{jo3} and substituting into~\eqref{dej} yields
\begin{equation*}
\delta_*\leq \epsilon^2\kappa ^{-1}/2300,
\end{equation*}
so that
\begin{equation}\label{wma2}
(92\delta_*\kappa)^{\frac{1}{2}} \leq \epsilon/5.
\end{equation}
By definition~\eqref{wamacon} of~$e_h$, it holds that
\begin{equation*}
e_h \leq 71\big(\sqrt{d/L} + ((16/9)(L_{\textrm{sc}}\lambda)^{-1} + 
\hat{c}_2(\hat{r}_1)
)^{1/2}\big)L^{3/4}h^{3/2}.
\end{equation*}
By~\eqref{coreq3} 
and~$\bar{\kappa}\geq 1$, we have~$d\leq n/30$ thus~$\sqrt{d/(L_{\textrm{sc}}\lambda n)}\leq \sqrt{1/(30L_{\textrm{sc}}\lambda)}$. 
Together with~\eqref{reqs},~\eqref{jo2} and~$\epsilon\leq (L_{\textrm{sc}}\lambda)^{-1/2}$, this implies
\begin{equation}\label{wma3}
e_h \leq 71(1/\sqrt{30}+ \sqrt{16/9+1/82800})(L_{\textrm{sc}}\lambda)^{-1/2}L^{3/4}h^{3/2}.
\end{equation}
Substituting~\eqref{coreq4} into~\eqref{wma3} yields~$e_h\leq (3/5)(1+35\kappa)^{-1}\epsilon$. Substituting this,~$\bar{\epsilon}= \epsilon/5$ and~\eqref{wma2} into~\eqref{wamaeq} concludes the proof.
\end{proof}

\section{Gibbs sampler}\label{gibsec}
In this section, we bound the total variation error of the Gibbs sampler, where the negative log-density of the target is allowed to have diminishing-to-negative curvature (Condition~\ref{cond:curvature2}). We will treat exploding gradients in the Poisson regression case separately in Section~\ref{poisec}.

As in the previous sections, let~$U:\mathbb{R}^d\rightarrow\mathbb{R}$ be the negative log-density of a target probability measure~$\hat{\pi}(dx)\propto e^{-U(x)}dx$. 
For any~$x = (x_1,\dots,x_d)\in\mathbb{R}^d$ and~$i\in[1,d]\cap\mathbb{N}$, let~$x_{-i} = (x_1,\dots,x_{i-1},x_{i+1},\dots,x_d)\in\mathbb{R}^{d-1}$, let~$\hat{\pi}_i(\cdot|x_{-i})$ be the conditional distribution of~$y_i$ given~$y_{-i}=x_{-i}$ under~$y\sim\hat{\pi}$ and let~$F_{x,i}:\mathbb{R}\rightarrow[0,1]$ be the cumulative distribution function of~$\hat{\pi}_i(\cdot|x_{-i})$. 
Let~$(i_k)_{k\in\mathbb{N}}$ be an i.i.d. sequence of~$[1,d]\cap\mathbb{N}$-valued r.v.'s such that~$i_k$ is uniformly drawn for all~$k$. Let~$(u_k)_{k\in\mathbb{N}}$ be an i.i.d. sequence of~$[0,1]$-valued r.v.'s independent of~$(i_k)_k$ such that~$u_k\sim \textrm{Unif}(0,1)$ for all~$k$. In this section, we set~$(\theta_k)_{k\in\mathbb{N}}$ to be~$\mathbb{R}^d$-valued r.v.'s given inductively by
\begin{equation*}
(\theta_{k+1})_i = 
\begin{cases}
F_{\theta_k,i_k}^{-1}(u_k) &\textrm{if } i=i_k,\\
(\theta_k)_i &\textrm{if } i\neq i_k
\end{cases}
\end{equation*}
with~$\theta_0$ independent of~$(i_k)_k,(u_k)_k$.   
For any~$k\in\mathbb{N}\setminus\{0\}$, denote~$\delta_{\theta_0},\delta_{\theta_0}P^k$ to be the distributions of~$\theta_0,\theta_k$.

The sequence~$(\theta_k)_k$ is that of the iterates of an exact random-scan Gibbs sampler. The initial~$\theta_0$ will be assumed to be the feasible start~$N(0,\bar{L}^{-1}I_d)$ for 
some~$\bar{L}>0$ 
to be set to some smoothness constant larger than~$L$ from Condition~\ref{cond:smooth} (but the same order).
Note that we will assume~$\nabla U(0)=0$, so this~$\theta_0$ is a Gaussian around an optimizer. Our analysis can also be adapted to accommodate~$\theta_0\sim N(\theta_{\textrm{pre}},\bar{L}^{-1}I_d)$ with~$\abs{\theta_{\textrm{pre}}}=O(\sqrt{d/\bar{L}})$, with the same effect at least in the regression context. We omit this extension in favour of simplicity.

Before presenting our main result for the Gibbs sampler, we state the following Lemma~\ref{nfo}. This lemma bounds explicitly the probability with which~$\abs{\theta_k}$ stays in a ball around~$0$. It relies in essence on the fact that warmness is preserved through invariant kernels.
\begin{lemma}\label{nfo}
Assume~$U,U_0$ satisfy Conditions~\ref{cond:curvature3},~\ref{cond:smooth} and~$0\in\argmin_{\mathbb{R}^d}U_0$. 
Let~$r\geq 0$ and assume~$r^2m_0(r)\geq 9d$. Let~$\bar{L}\geq L$ and~$\theta_0\sim N(0,\bar{L}^{-1}I_d)$. It holds for any~$k\in\mathbb{N}$ that 
\begin{equation*}
\mathbb{P}(\abs{\theta_k}\geq r)\leq 5(\bar{L}/m_0(r))^{d/2}\cdot e^{2d-3m_0(r)r^2/8}.
\end{equation*}
\end{lemma}
As mentioned in previous sections, we are interested in the regression context in the case where~$m_0(r)r^2=\Omega(n)$ (and~$\bar{L}/m_0(r)=O(1)$). In this case, Lemma~\ref{nfo} controls the escape probability of the Gibbs sampler given~$n\gtrsim d$. More generally, for an~$m$-strongly log-concave target, Lemma~\ref{nfo} controls the escape probability of the Gibbs iterates outside of the characteristic radius~$r=\tilde{\Omega}(\sqrt{d/m})$, which is the same control as for Gaussian distributions.
\begin{proof}[Proof of Lemma~\ref{nfo}]
We first show by induction for any~$k\in\mathbb{N}$ that the distribution~$\delta_{\theta_0}P^k$ is absolutely continuous with density~$\nu_k$ satisfying
\begin{equation}\label{ipe}
\nu_k\leq e^{-U+U(0)}\cdot (2\pi/\bar{L})^{-d/2}
\end{equation}
almost everywhere.
The base case~$k=0$ holds by~$\theta_0\sim N(0,\bar{L}^{-1}I_d)$, Condition~\ref{cond:smooth},~$\nabla U(0)=0$ (which follows by~$\nabla U_0(0)=0$ and Condition~\ref{cond:curvature3}, in which the intermediate value theorem gives~$\nabla U(0)-\nabla U_0(0)=0$) and the assumption~$\bar{L}\geq L$. 
Fix~$k\in\mathbb{N}$ and suppose the~$k^{\textrm{th}}$ case holds. 
By Lemma~2.1(a) in~\cite{ascolani2024e}, the disintegration theorem and the inductive assumption,~$\delta_{\theta_0}P^{k+1}$ is a mixture of absolutely continuous measures (so it is absolutely continuous) and it has a 
density~$\nu_{k+1}
$ 
of the form
\begin{equation}\label{vld2}
\nu_{k+1}
(x)=\frac{1}{d}\sum_{i=1}^d\frac{e^{-U(x)}}{\int_{\mathbb{R}}e^{-U_i(y,x_{-i})}dy} \cdot \int_{\mathbb{R}} (\nu_k)_i(y,x_{-i})dy.
\end{equation}
Moreover, substituting~\eqref{ipe} into~\eqref{vld2} yields~\eqref{ipe} with~$k$ replaced by~$k+1$, which concludes the induction. \\
\indent We proceed to control the right-hand side of~\eqref{ipe}. Let~$f:\mathbb{R}^d\rightarrow[0,\infty)$ be given by~\eqref{fdef} with~$c=m_0(r)$, and let~$\mu$ be the probability measure with density, also denoted~$\mu$, satisfying~$\mu\propto e^{-f}$. 
By Condition~\ref{cond:curvature3} and the assumption~$0\in\argmin U_0$, the right-hand side of~\eqref{ipe} satisfies
\begin{align}
e^{-U+U(0)}\cdot (2\pi/\bar{L})^{-d/2} &\leq \textstyle e^{-(U-U_0)+U(0)}\cdot \sup_{\mathbb{R}^d}e^{-U_0} \cdot (2\pi/\bar{L})^{-d/2}\nonumber\\
&\leq \textstyle e^{-(f+U(0)-U_0(0))+U(0)}\cdot e^{-U_0(0)} \cdot (2\pi/\bar{L})^{-d/2}\nonumber\\
&= e^{-f} \cdot (2\pi/\bar{L})^{-d/2}\nonumber\\
&= \mu \textstyle\int_{\mathbb{R}^d}e^{-f} \cdot (2\pi/\bar{L})^{-d/2}.\label{las}
\end{align}
By our concentration result applied to~$e^{-f}$, namely Lemma~\ref{tvv} (with~$U=f$,~$U_0=0$ and~$L=m_0(r)$ therein), we have
\begin{align*}
\textstyle\int_{\mathbb{R}^d} e^{-f} &= \textstyle\int_{B_r}e^{-m_0(r)\abs{x}^2/2}dx + \mu(\mathbb{R}^d\setminus B_r)\cdot \int_{\mathbb{R}^d} e^{-f}\\
&\leq \textstyle (2\pi/m_0(r))^{d/2} + \varrho\cdot(1+ \varrho)\cdot \int_{\mathbb{R}^d} e^{-f},
\end{align*}
where~$\varrho=e^{2d-3m_0(r)r^2/8}$. 
Therefore, using the assumption~$r^2m_0(r)\geq 8d$ to obtain~$\varrho(1+\varrho)\leq (1/e)(1+1/e)< 2/3$, it holds that
\begin{equation}\label{mrdef}
\int_{\mathbb{R}^d}e^{-f}\cdot (2\pi/\bar{L})^{-d/2}\leq 3\bigg(\frac{\bar{L}}{m_0(r)}\bigg)^{d/2} =: M.
\end{equation}
Combining~\eqref{mrdef} with~\eqref{las} and~\eqref{ipe} yields~$\nu_k\leq M\mu$ 
for all~$k\in\mathbb{N}$. 
Consequently, by Lemma~\ref{tvv} (with the same substitutions as before), we have
\begin{equation*}
\nu_k(\mathbb{R}^d\setminus B_r) \leq M\mu (\mathbb{R}^d\setminus B_r) \leq M\varrho(1+\varrho)\leq (3/2)M\varrho
\end{equation*}
for all~$k\in\mathbb{N}$, which gives the assertion.
\end{proof}

The next Theorem~\ref{gm} is the main result of the section. 

\begin{theorem}\label{gm}
Assume the notations and settings of Lemma~\ref{nfo}, and assume~$U$ satisfies Condition~\ref{cond:curvature2}. Assume~$m(2r)>0$ and~$\bar{L}\geq 5L+4m(2r)$. 
It holds for any~$k\in\mathbb{N}$ that
\begin{align}
\textrm{TV}(\delta_{\theta_0}P^k,\hat{\pi})&\leq \tilde{\delta}_{k,r}+\bigg(1-\frac{m(2r)}{(5L+4m(2r)) d}\bigg)^{\!k/2} \bigg( 
1 + \frac{d}{4}\ln\bigg(\frac{2d\bar{L}}{m(2r)}\bigg)\bigg)^{\!1/2}\nonumber\\
&\quad+\bigg(\sup_{\mathbb{R}^d\setminus B_r} e^{-U_0}\bigg)\Big( \inf_{B_r}e^{-U_0} \Big)^{-1}\cdot \varrho\bigg(1+\varrho\bigg(\frac{L}{m_0(r)}\bigg)^{d/2}\bigg),\label{mab}\\
\textrm{where}\qquad\varrho &= e^{2d-3m_0(r)r^2/8}\nonumber\\
\tilde{\delta}_{k,r}&= 5(k+1)(\bar{L}/m_0(r))^{d/2}\cdot e^{2d-3m_0(r)r^2/8}.\label{inpe}
\end{align}
\end{theorem}
\begin{remark}\label{ralt} 
In the proof of Theorem~\ref{gm}, we modify the negative log-density by adding a quadratic outside a ball. 
Alternatively, it is possible to use in the proof of Theorem~\ref{gm} the surrogate studied in~\cite{altmeyer2024p,MR4828853,MR4721029}. More generally, the strategy to consider sampling dynamics targeting a surrogate distribution~\cite{MR4828853,MR4721029} seems to also be applicable to our settings for both HMC and Gibbs. 
We note however that in practice, it is more convenient to run MCMC dynamics on the true target, rather than a surrogate. 
Moreover, MCMC targeting the surrogate distribution is limited by the error between~$\pi$ and the surrogate in the long-time limit for fixed~$n,d$. 
Analysis for MCMC targeting the true distribution in this context is also available in~\cite{altmeyer2024p}, 
where a comparison argument is used to connect ULA to an Ornstein-Uhlenbeck process. 
However, this comparison argument for ULA is not applicable to HMC nor Gibbs, which have improved (known) complexities (cf.~\cite[Theorem~3.7]{altmeyer2024p}). 
\end{remark}
The broad strategy for the proof of Theorem~\ref{gm} is as follows. 
We construct a modified~$\bar{\pi}(\cdot|Z^{(n)})$, which is globally strongly log-concave. 
We then construct a Gibbs trajectory targeting~$\bar{\pi}(\cdot|Z^{(n)})$, which enjoys the entropy contraction of~\cite{ascolani2024e}. Thereafter, we maximally couple this trajectory with the original Gibbs trajectory targeting~$\pi(\cdot|Z^{(n)})$, using Lemma~\ref{nfo} to show that they stick together often enough.

\begin{proof}[Proof of Theorem~\ref{gm}]
Let~$Q_0,Q:\mathbb{R}^d\rightarrow[0,\infty)$ be given for any~$x\in\mathbb{R}^d$ by
\begin{align*}
Q_0(x) &= \begin{cases}
0 &\textrm{if }\abs{x}<5r/4,\\
2(L+m(2r))(\abs{x}-5r/4)^2 &\textrm{if }\abs{x}\geq 5r/4,
\end{cases}\\
Q &= \varphi_{r/4}\ast Q_0,
\end{align*}
which satisfies
\begin{equation*}
D^2Q_0(x) = 2(L+m(2r))\big((2-5r/(2\abs{x}))I_d + (5r/2)xx^{\top}/\abs{x}^3\big) \qquad\forall x\in \mathbb{R}^d\setminus B_{5r/4},
\end{equation*}
and so~$\inf_{\mathbb{R}^d\setminus B_{7r/4}} \lambda_{\min}(D^2Q_0) \geq L + m(2r)$. Thus
\begin{equation}
\inf_{x\in\mathbb{R}^d\setminus B_{2r}}\lambda_{\min} (D^2Q(x)) \geq L+m(2r),\quad\sup_{x,y\in\mathbb{R}^d}\abs{\nabla Q(x)-\nabla Q(y)}\leq 4(L+m(2r))\abs{x-y}.
\end{equation}
By Conditions~\ref{cond:curvature2},~\ref{cond:smooth}, 
the function~$\bar{U} := U + Q$ admits an~$(5L + 4m(2r))$-Lipschitz gradient and it is~$m(2r)$-strongly convex with~$\nabla \bar{U}(0)=0$. 
Let~$\bar{\pi}$ be the measure on~$\mathbb{R}^d$ given by~$\bar{\pi}(dx) = e^{-\bar{U}(x)}dx/\int_{\mathbb{R}^d} e^{-\bar{U}}$. 
Let~$\hat{\theta}_0=\bar{\theta}_0=\theta_0$ and let~$(\hat{\theta}_k)_{k\in\mathbb{N}},(\bar{\theta}_k)_{k\in\mathbb{N}}$ be given inductively as follows. For any~$k\in\mathbb{N}$ assume~$\hat{\theta}_k,\bar{\theta}_k$ are given and 
for any~$x\in\mathbb{R}^d$ and any~$i\in[1,d]\cap\mathbb{N}$, let~$\bar{F}_{x,i}:\mathbb{R}\rightarrow[0,1]$ be the cumulative distribution function of the conditional distribution of~$y_i$ given~$y_{-i}=x_{-i}$ under~$y\sim \bar{\pi}$. 
For any~$\hat{\theta},\bar{\theta}\in\mathbb{R}^d$ and~$i\in[1,d]\cap\mathbb{N}$, 
let~$(\xi_k^{\hat{\theta},\bar{\theta},i},\bar{\xi}_k^{\hat{\theta},\bar{\theta},i})$ be the maximal coupling as constructed in the proof of Theorem~7.3 in~\cite[Chapter~3]{MR1741181} of~$(F_{\hat{\theta},i}^{-1}(u_k),\bar{F}_{\bar{\theta},i}^{-1}(u_k))$ 
and let~$\hat{\theta}_{k+1},\bar{\theta}_{k+1}$ be given for any~$i\in[1,d]\cap\mathbb{N}$ by
\begin{equation*}
((\hat{\theta}_{k+1})_i,(\bar{\theta}_{k+1})_i) = 
\begin{cases}
(\xi_k^{\hat{\theta}_k,\bar{\theta}_k,i_k},\bar{\xi}_k^{\hat{\theta}_k,\bar{\theta}_k,i_k}) &\textrm{if } i=i_k,\\
((\hat{\theta}_k)_i,(\bar{\theta}_k)_i) &\textrm{if } i\neq i_k.
\end{cases}
\end{equation*}
Note that measurability of~$((\hat{\theta}_{k+1})_i,(\bar{\theta}_{k+1})_i)$ (w.r.t. the probability space) follows from the coupling construction (see also~\cite[Lemma~4.22]{MR4226142}).
Again for any~$k\in\mathbb{N}\setminus\{0\}$ we denote by~$\delta_{\bar{\theta}_0}=\delta_{\hat{\theta}_0} = N(0, \bar{L}^{-1}I_d)$,~$\delta_{\bar{\theta}_0}\bar{P}^k$,~$\delta_{\hat{\theta}_0}\hat{P}^k$ to be 
the distributions of~$\bar{\theta}_0=\hat{\theta}_0$ and~$\bar{\theta}_k,\hat{\theta}_k$ respectively. 
Note that for any~$k\in\mathbb{N}$,~$\hat{\theta}_k$ has the same law as~$\theta_k$. 

By Lemma~2.4 and Theorem~3.2 both in~\cite{ascolani2024e}, it holds for any~$k\in\mathbb{N}$ that
\begin{equation}\label{KLcon}
\textrm{KL}(\delta_{\bar{\theta}_0}\bar{P}^k| \bar{\pi})\leq \bigg(1-\frac{m(2r)}{(5L+4m(2r)) d}\bigg)^{\!k} \textrm{KL}(\delta_{\bar{\theta}_0}| \bar{\pi}).
\end{equation}
For the right-hand side of~\eqref{KLcon}, by the weak triangle inequality (Lemma~2.2.23 in~\cite{chewi2023o}) and monotonicity (Theorem~3 in~\cite{MR3225930}) of R\'enyi divergences, it holds that
\begin{equation}\label{klk}
\textrm{KL}(\delta_{\bar{\theta}_0}| \bar{\pi})\leq \mathcal{R}_2(\delta_{\bar{\theta}_0}|| \bar{\pi}) 
=\mathcal{R}_{2}(N(0,\bar{L}^{-1}I_d)||\bar{\pi}).
\end{equation}
By Lemma~31 in~\cite{pmlr-v178-c}, monotonicity of R\'enyi divergences and Lemma~30 in~\cite{MR4309974}, 
the right-hand side of~\eqref{klk} satisfies
\begin{equation}\label{hoi2}
\mathcal{R}_{2}(N(0,\bar{L}^{-1}I_d)||\bar{\pi})\leq 2 + (d/2)\ln(2d\bar{L}/m(2r)).
\end{equation}
Substituting~\eqref{hoi2} into~\eqref{klk}, then the result into~\eqref{KLcon} yields by Pinsker's inequality for any~$k\in\mathbb{N}$ that
\begin{equation}\label{tvbR}
\textrm{TV}(\delta_{\bar{\theta}_0}\bar{P}^k, \bar{\pi})\leq R_k,
\end{equation}
where
\begin{equation*}
R_k :=\bigg(1-\frac{m(2r)}{(5L+4m(2r)) d}\bigg)^{\!k/2} \bigg( 
1 + \frac{d}{4}\ln\bigg(\frac{2d\bar{L}}{m(2r)}\bigg)\bigg)^{\!1/2}.
\end{equation*}
\indent Next, we prove by induction it holds for any~$t\in\mathbb{N}$ that
\begin{equation}\label{cek1}
\hat{\mathcal{E}}_t  := \cap_{l\in[0,t]\cap\mathbb{N}}\{\abs{\hat{\theta}_l}\leq r\}\subset \{\bar{\theta}_t=\hat{\theta}_t\}.
\end{equation}
In words, we prove that if~$(\hat{\theta}_k)_k$, the trajectory targeting~$\pi(\cdot|Z^{(n)})$, stays inside~$B_r$, then the maximally coupled~$(\bar{\theta}_k)_k$ also stays in the same ball. We do this by considering the conditional densities from which Gibbs draws and showing that those for~$(\hat{\theta}_k)_k$ are always bounded above by those of~$(\bar{\theta}_k)_k$ in~$B_r$, and so trajectories remain coupled as long as~$\hat{\theta}_k\in B_r$. 
Fix~$t\in\mathbb{N}$ and assume~\eqref{cek1}. The conditional distribution of~$(\bar{\theta}_{t+1})_{i_t} = \bar{\xi}_t^{\hat{\theta}_t,\bar{\theta}_t,i_t}$ given~$\bar{\theta}_t$ a.s.\ has density~$e^{-\bar{U}_{i_t}(x,(\bar{\theta}_t)_{-i_t})}dx/\int_{\mathbb{R}}e^{-\bar{U}_{i_t}(\cdot,(\bar{\theta}_t)_{-i_t})}$. Moreover, it holds on the event~$\hat{\mathcal{E}}_t = \hat{\mathcal{E}}_t\cap\{\bar{\theta}_t=\hat{\theta}_t\}$ 
that
\begin{equation*}
\bar{U}_{i_t}(x,(\bar{\theta}_t)_{-i_t}) = U_{i_t}(x,(\hat{\theta}_t)_{-i_t}) \qquad\forall x\in \{y\in\mathbb{R}:\abs{y}^2+\abs{(\bar{\theta}_t)_{-i_t}}^2\leq r^2\}
\end{equation*} 
and by~$Q\geq 0$ that
\begin{equation*}
\int_{\mathbb{R}}e^{-\bar{U}_{i_t}(\cdot,(\bar{\theta}_t)_{-i_t})}\leq \int_{\mathbb{R}}e^{-U_{i_t}(\cdot,(\hat{\theta}_t)_{-i_t})},
\end{equation*}
which imply
\begin{equation*}
\frac{e^{-\bar{U}_{i_t}(x,(\bar{\theta}_t)_{-i_t})}}{\int_{\mathbb{R}}e^{-\bar{U}_{i_t}(\cdot,(\bar{\theta}_t)_{-i_t})}}\geq \frac{e^{-U_{i_t}(x,(\hat{\theta}_t)_{-i_t})}}{\int_{\mathbb{R}}e^{-U_{i_t}(\cdot,(\hat{\theta}_t)_{-i_t})}}\qquad \forall x\in \{y\in\mathbb{R}:\abs{y}^2+\abs{(\bar{\theta}_t)_{-i_t}}^2\leq r^2\}.
\end{equation*}
the right-hand side of which is the density of the conditional distribution of~$(\hat{\theta}_{t+1})_{i_t} = \xi_t^{\hat{\theta}_t,\bar{\theta}_t,i_t}$ given~$\hat{\theta}_t, i_t$. Therefore, by construction of the maximal coupling~\cite[Chapter~3, Theorem~7.3]{MR1741181}, 
it holds on the event~$\{\abs{\xi_t^{\hat{\theta}_t,\bar{\theta}_t,i_t}}^{2} + \abs{(\bar{\theta}_t)_{-i_t}}^2\leq r^2\}\cap\hat{\mathcal{E}}_t$ that~$\xi_t^{\hat{\theta}_t,\bar{\theta}_t,i_t}=\bar{\xi}_t^{\hat{\theta}_t,\bar{\theta}_t,i_t}$, which implies
\begin{equation}
\hat{\mathcal{E}}_{t+1} = \hat{\mathcal{E}}_{t+1}\cap\{\bar{\theta}_t=\hat{\theta}_t\} \subset \{\xi_t^{\hat{\theta}_t,\bar{\theta}_t,i_t}=\bar{\xi}_t^{\hat{\theta}_t,\bar{\theta}_t,i_t}\}\cap\{\bar{\theta}_t=\hat{\theta}_t\} = \{\bar{\theta}_{t+1}=\hat{\theta}_{t+1}\},
\end{equation}
and this is~\eqref{cek1} with~$t$ replaced by~$t+1$. 
For any~$k\in\mathbb{N}$, it follows by~\eqref{tvbR} that
\begin{align}
\textrm{TV}(\delta_{\hat{\theta}_0}\hat{P}^k, \bar{\pi}) &\leq \textrm{TV}(\delta_{\hat{\theta}_0}\hat{P}^k, \delta_{\bar{\theta}_0}\bar{P}^k) + \textrm{TV}(\delta_{\bar{\theta}_0}\bar{P}^k, \bar{\pi})\nonumber\\
&\leq 1-\mathbb{P}(\bar{\theta}_k=\hat{\theta}_k) + R_k\nonumber\\
&\leq 1-\mathbb{P}(\hat{\mathcal{E}}_k) + R_k.\label{tvb3o}
\end{align}
Since~$\hat{\theta}_l$ has the same law as~$\theta_l$ for all~$l\in\mathbb{N}$, Lemma~\ref{nfo} implies
\begin{equation*}
\mathbb{P}(\hat{\mathcal{E}}_k)\geq 1-5(k+1)(\bar{L}/m_0(r))^{d/2}\cdot e^{2d-3m_0(r)r^2/8} =: 1-\tilde{\delta}_{k,r},
\end{equation*}
Substituting into~\eqref{tvb3o} and noting again that~$\hat{\theta}_k$ has the same law as~$\theta_k$ for all~$k\in\mathbb{N}$ yields
\begin{equation}
\textrm{TV}(\delta_{\theta_0}P^k, \hat{\pi}) \leq \textrm{TV}(\delta_{\theta_0}P^k, \bar{\pi}) +  \textrm{TV}(\bar{\pi},\hat{\pi})\leq \tilde{\delta}_{k,r}+R_k+ \textrm{TV}(\bar{\pi},\hat{\pi}).
\end{equation}
It remains to bound~$\textrm{TV}(\bar{\pi},\hat{\pi})$. Let~$\psi,\bar{\psi}$ be~$\mathbb{R}^d$-valued r.v.'s with distributions~$\hat{\pi},\bar{\pi}$ respectively. Let~$(\psi',\bar{\psi}')$ be the maximal coupling of~$(\psi,\bar{\psi})$ as constructed in the proof of Theorem~7.3 in~\cite[Chapter~3]{MR1741181}. Clearly, the probability density function~$e^{-U}/\int_{\mathbb{R}^d}e^{-U}$ of~$\psi$ is less than or equal to that~$e^{-\bar{U}}/\int_{\mathbb{R}^d}e^{-\bar{U}}$ of~$\bar{\psi}$ on~$B_r$, therefore we have~$\mathbb{P}(\psi'=\bar{\psi}')\geq \hat{\pi}(B_r) = 1- \hat{\pi}(\mathbb{R}^d\setminus B_r)$. On the other hand, Lemma~\ref{tvv} 
implies
\begin{equation*}
\hat{\pi}(\mathbb{R}^d\setminus B_r) < \bigg(\sup_{\mathbb{R}^d\setminus B_r} e^{-U_0}\bigg)\Big( \inf_{B_r}e^{-U_0} \Big)^{-1}\cdot \bar{\varrho}\bigg(1+\bar{\varrho}\bigg(\frac{L}{m_0(r)}\bigg)^{d/2}\bigg),
\end{equation*}
where~$\bar{\varrho} = e^{2d-3m_0(r)r^2/8}$,
so that the proof concludes by~$\textrm{TV}(\bar{\pi},\hat{\pi})=1-\mathbb{P}(\psi'= \bar{\psi}')$. 
\end{proof}

We state complexity guarantees with emphasis on scalings motivated by the regression context, see Remark~\ref{cc}\ref{lre}.

\begin{corollary}\label{rak}
Assume the settings and notations of Theorem~\ref{gm} with~$\bar{L}=9L$. Let~$L_{\textrm{sc}},\lambda>0$. Denote
\begin{align}
r_* &= d^{1/2}((L_{\textrm{sc}}\lambda d)^{-1/2}+(\pi L)^{-1/2} 
),\label{rdt}\\
\kappa &= L/m(2r_*),\qquad \kappa_0 = L/m_0(r_*),\nonumber\\
c_U &=e^{-\min_{\mathbb{R}^d}U_0}(\textstyle\inf_{B_{r_*}}e^{-U_0})^{-1}\label{cud}
\end{align}
and assume~$\kappa,\kappa_0\in(0,\infty),r_*^2m_0(r_*)\geq 9d$. For any~$n,N\in\mathbb{N}\setminus\{0\}$ with~$n\geq d$ and any~$\epsilon\in(0,1)$, if the conditions~$L=L_{\textrm{sc}}\lambda n$,
\begin{align}
n&\geq (8/3)\kappa_0\big[ 2d +(d/2)\ln(5\kappa_0) + \ln(30N\epsilon^{-1}c_U)\big],
\label{nbh}\\
N&\geq 18\kappa d\ln\big(3\epsilon^{-1}\big(
1 + (d/4)\ln(18d\kappa)\big)\big)\label{nbd}
\end{align}
hold, then it holds that~$\textrm{TV}(\delta_{\theta_0}P^N,\hat{\pi}) \leq \epsilon$. 
\end{corollary}
\begin{proof}
We bound the terms on the right-hand side of~\eqref{mab} with~$k=N$ and~$r=r_*$ given by~\eqref{rdt}. We start by addressing the first term on this right-hand side, namely~$\tilde{\delta}_{N,r_*}$. We have
\begin{align*}
\tilde{\delta}_{N,r_*} &= 5(N+1)(\bar{L}/m_0(r_*))^{d/2}\cdot e^{2d-3m_0(r_*)r_*^2/8}\\
&\leq 10N(5\kappa_0)^{d/2}\cdot e^{2d-(3/8)\kappa_0^{-1}L_{\textrm{sc}}\lambda n(L_{\textrm{sc}}\lambda)^{-1}}\\
&= 10N(5\kappa_0)^{d/2}\cdot e^{2d-(3/8)\kappa_0^{-1}n}\\
&\leq\epsilon/3,
\end{align*}
where we have used~\eqref{nbh} to obtain the last line. 
Next, we address the second term on the right-hand side of~\eqref{mab}. By~\eqref{nbd} and the inequality~$1-x\leq e^{-x}$ for~$x\in(0,\infty)$, we have
\begin{equation*}
(1-1/(9\kappa d))^{N/2}\big(
1 +(d/4)\ln(18d\kappa)\big)^{1/2}\leq \epsilon/3.
\end{equation*}
Lastly, we address the last term on the right-hand side of~\eqref{mab}. Note that~\eqref{nbh} implies the two inequalities
\begin{equation*}
(3/8)\kappa_0^{-1}n
\geq 2d+\ln(\kappa_0^{d/2})\\
\end{equation*}
and
\begin{equation*}
(3/8)\kappa_0^{-1}n
\geq 2d + \ln(6\epsilon^{-1}c_U).
\end{equation*}
Therefore by~\eqref{rdt} and again~$L=L_{\textrm{sc}}\lambda n$, denoting~$\varrho_* = e^{2d-3m_0(r_*)r_*^2/8}$, we have
\begin{equation*}
\kappa_0^{d/2}\varrho_* = \kappa_0^{d/2} e^{2d-(3/8)\kappa_0^{-1}Lr_*^2}\leq \kappa_0^{d/2}e^{2d-(3/8)\kappa_0^{-1}n
}\leq 1
\end{equation*}
and (using also the inequalities just derived)
\begin{equation*}
c_U\cdot \varrho_*(1+\kappa_0^{d/2}\varrho_*) \leq 2c_U\varrho_* 
\leq \epsilon/3.
\end{equation*}
Substituting these bounds into~\eqref{mab} with~$k=N$ and~$r=r_*$ concludes the proof.
\end{proof}

\section{Properties of GLMs}\label{csglm}
In this section, we give assumptions on the data-generating model in GLMs under which, with exponentially-in-$n$ high probability, the negative log-density of the posterior without prior exhibits local curvature and smoothness of the same order in~$d,n$ as~$d,n\rightarrow\infty$. 
Together with suitable assumptions on the prior, these results imply local curvature and smoothness of the negative log-density including the prior. In addition, in Section~\ref{map}, we provide conditions under which there exists (random)~$\theta_{\textrm{map}}\in\mathbb{R}^d$ such that the gradient of the negative log-density evaluated at~$\theta_{\textrm{map}}$ is zero with high probability. Together with the curvature results, this~$\theta_{\textrm{map}}$ is the maximum a posteriori (but this alone is not important in our analysis throughout, so we forgo a rigorous proof). 

In all of this section, we assume the setting at the beginning of Section~\ref{OGLM}.
\subsection{Curvature}\label{cursec}
We address local curvature first. Similar results on random matrices~\cite[Section~6]{MR3967104} and in the context of sparse recovery~\cite{MR3612870} have been previously established, which we use to prove the results below. 

In the following, we use from~\cite[Section~3.3]{MR3931734} the Definition~3.10 for VC-dimension. We have not found an explicit proof of the following Lemma~\ref{VClem} in the literature, but a similar claim is stated (without proof) for example in the beginning of the proof of Corollary~2.5 in~\cite{MR3612870}, see also Lemma~2.6.19 in~\cite{MR4628026}. 
\begin{lemma}\label{VClem}
For any~$c\geq 0$, 
the VC-dimensions of the classes
\begin{equation*}
\{ \mathds{1}_{\{x\in\R^d:\abs{v^{\top}x}< c\}} : v\in\R^d \},\qquad \{ \mathds{1}_{\{x\in\R^d:\abs{v^{\top}x}> c\}} : v\in\R^d \}
\end{equation*}
are at most~$10(d+1)$.
\end{lemma}
In the following proof, we use from Sections~3.2 and~3.3 in~\cite{MR3931734} the Definitions~3.6 and~3.10 for the growth function and VC-dimension respectively.

\begin{proof}[Proof of Lemma~\ref{VClem}]
By Theorem~B in~\cite{MR532837}, the VC-dimension of the class of all half spaces is~$d+1$, where a half space is a set of the form~$\{x\in\R^d: v^{\top}x \geq c\}$ for~$v\in\mathbb{S}^{d-1}$ and~$c\in\R $. 
Moreover, by Lemma~2.6.19(i) in~\cite{MR4628026}, the same holds when the inequality is replaced by a strict inequality. 
Therefore, by Sauer's lemma~\cite[Corollary~3.18]{MR3931734}, the growth function~$\Pi_H:\N \rightarrow\N $ of the class
\begin{equation*}
H:= \{ \mathds{1}_{\{x\in\R^d:v^{\top}x > c\}} : v\in\mathbb{S}^{d-1},c\in\R \}
\end{equation*}
satisfies
\begin{equation}\label{pim}
\Pi_H(m) \leq (em/(d+1))^{d+1}\qquad \forall m\geq d+1.
\end{equation}
In particular, for any~$\bar{m}\in\N $, any finite subset~$S\subset\R^d$ of size~$\bar{m}$ and any~$v_1\in\mathbb{S}^{d-1}$,~$c_1\in\R $, since the set
\begin{equation}\label{SS}
S \cap \{ x\in\R^d:v_1^{\top}x < c_1 \} =: \{ \hat{s}_1,\dots,\hat{s}_k \}
\end{equation}
is at most size~$\bar{m}$, the size of the set
\begin{equation*}
\{(h(\hat{s}_1),\dots,h(\hat{s}_k)) : h\in H\}
\end{equation*}
is at most~$\max(2^d,(e\bar{m}/(d+1))^{d+1})$. Moreover, for any~$\bar{m}\in\N $ and any finite subset~$S\subset\R^d$ of size~$\bar{m} $,~\eqref{pim} implies that there are at most~$\max(2^d,(e\bar{m}/(d+1))^{d+1})$ distinct subsets of~$S$ of the form~\eqref{SS} for~$v_1\in\mathbb{S}^{d-1}$,~$c_1\in\R $. Gathering, for any~$\bar{m}\in\N $ and any set~$S=\{s_1,\dots,s_{\bar{m}}\}\subset\R^d$ of size~$\bar{m}$, the size of the set
\begin{equation}\label{hs}
\{(h(s_1),\dots,h(s_{\bar{m}})):h\in\{\mathds{1}_{\{x\in\R^d:v_1^{\top}x> c_1,v_2^{\top}x< c_2\}}:v_1,v_2\in\mathbb{S}^{d-1},c_1,c_2\in\R \}\}
\end{equation}
is at most~$(\max(2^d,(e\bar{m}/(d+1))^{d+1}))^2$. Clearly for any such~$\bar{m}$,~$S$ and~$c\geq 0$, the set~\eqref{hs} is a superset of 
\begin{equation*}
\{(h(s_1),\dots,h(s_{\bar{m}})):h\in\{\mathds{1}_{\{x\in\R^d:\abs{v^{\top}x}< c\}}:v\in\mathbb{S}^{d-1}\}\},
\end{equation*}
so that the size of this set has the same upper bound. It is readily verified that the largest~$\bar{m}\in\N $ such that~$(\max(2^d,(e\bar{m}/(d+1))^{d+1}))^2 \geq 2^{\bar{m}}$ is at most~$10(d+1)$, which concludes the proof for the first class. The argument for the second class is the same.
\end{proof}

Proposition~\ref{p1} below is our main result for the local curvature of GLMs. We extend the proof strategy of Corollary~2.5 in~\cite{MR3612870} in order to accommodate nonconstant~$A''$ with~$A''(x)\rightarrow 0$ as~$\abs{x}\rightarrow\infty$. Similar curvature results were established in~\cite[Lemma~4]{MR4009715} for Gaussian design and~\cite[Theorems~5 and~6]{chardon2024f} for logistic regression. Our assumptions here allow for non-Gaussian designs, but since we focus on the Gaussian case in Section~\ref{map}, we do not elaborate any further on this generality.
\begin{prop}\label{p1}
Let~$\hat{\theta}\in\mathbb{R}^d$. 
Assume that there exist~$w_1,w_2\geq 0$ and~$\beta_1,\beta_2>0$ such that
\begin{subequations}\label{xwn}
\begin{align}
\inf_{v\in\mathbb{S}^{d-1}}\mathbb{P}((v^{\top}X_1)^2<w_1) &\geq \beta_1,\label{xwn1}\\
\inf_{v\in\mathbb{S}^{d-1}}\mathbb{P}((v^{\top}X_1)^2>w_2) &\geq \beta_2.\label{xwn2}
\end{align}
\end{subequations}
For any~$\bar{c}_1,\bar{c}_2\in(1/2,1)$ and~$r>0$, there exist absolute constants~$c_1,c_2>0$ such that if~$n\geq c_1(d+1)/\min(\beta_1,\beta_2)^2$, then it holds with probability at least~$1-\exp(-c_2\beta_1^2n) -\exp(-c_2\beta_2^2n)$ that
\begin{equation}\label{p1eq}
\inf_{\theta\in B_r(\hat{\theta})}\lambda_{\min}\bigg(\sum_{i=1}^n A''(\theta^{\top}X_i)X_iX_i^{\top}\bigg) \geq \big(\bar{c}_1\beta_1+\bar{c}_2\beta_2-1\big)w_2n\inf_{y\in B_{r^{\dagger}}}A''(y),
\end{equation}
where~$r^{\dagger}=(\abs{\hat{\theta}}+r)w_1^{1/2}$.
\end{prop}
\begin{remark}\label{p1r}
We make inequalities~\eqref{xwn} and the coefficient on the right-hand side of~\eqref{p1eq} explicit for mean-zero Gaussian~$X_i$ and~$\bar{c}_1=\bar{c}_2=19/20$. 
Suppose Assumption~\ref{A1} holds. 
Since~$v^{\top} X_1$ is Gaussian with variance~$v^{\top}\Sigma v$, inequalities~\eqref{xwn} hold with
\begin{equation}\label{p1req}
w_1=4\lambda_{\max}(\Sigma),\quad\beta_1=\Phi(2)-\Phi(-2),\qquad w_2=4\lambda_{\min}(\Sigma)/9,\quad \beta_2=2\Phi(-2/3), 
\end{equation}
where~$\Phi$ denotes the cumulative distribution function of a standard Gaussian in one dimension. In this case, we have that~$w_2(19\beta_1/20+19\beta_2/20-1)>\lambda_{\min}(\Sigma)/6$ is a constant independent of~$n,d$ as~$n,d\rightarrow\infty$ (except possibly through~$\lambda_{\min}(\Sigma)$).
\end{remark}
\begin{proof}[Proof of Proposition~\ref{p1}]
Fix~$\bar{c}_1,\bar{c}_2\in(1/2,1)$ and let
\begin{align*}
\mathcal{F}_1 &= \{ f:\R^d\rightarrow\R  : (\exists v\in\mathbb{S}^{d-1} \textrm{ s.t. } f(x) = (w_1 - (v^{\top}x)^2)\vee 0 \ \forall x\in\R^d) \},\\
\mathcal{F}_2 &= \{ f:\R^d\rightarrow\R  : (\exists v\in\mathbb{S}^{d-1} \textrm{ s.t. } f(x) = (v^{\top}x)^2 \ \forall x\in\R^d ) \}.
\end{align*}
By Lemma~2.3 in~\cite{MR3612870} 
with~$u=0$~($u=w_2$ resp.) and~$\mathcal{F}=\mathcal{F}_1$~($\mathcal{F}=\mathcal{F}_2$ resp.), 
together with Lemma~\ref{VClem}, there exist absolute constants~$c_1,c_2>0$ such that if~$n\geq c_1 (d+1)/\min(\beta_1,\beta_2)^2$, then it holds with probability at least~$1-\exp(-c_2\beta_1^2n)-\exp(-c_2\beta_2^2n)$ that
\begin{align}
\inf_{v\in\mathbb{S}^{d-1}} \abs{\{ i\in\{1,\dots,n\}: (v^{\top}X_i)^2 < w_1 \}} &= \inf_{f\in\mathcal{F}_1} \abs{\{ i\in\{1,\dots,n\}:\abs{f(X_i)} > 0 \}}\nonumber\\
&\geq \bar{c}_1\beta_1 n\label{feq1}
\end{align}
and
\begin{align}
\inf_{v\in\mathbb{S}^{d-1}} \abs{\{ i\in\{1,\dots,n\}: (v^{\top}X_i)^2 > w_2 \}} &= \inf_{f\in\mathcal{F}_2} \abs{\{ i\in\{1,\dots,n\}:\abs{f(X_i)} > w_2 \}} \nonumber\\
&\geq \bar{c}_2\beta_2 n,\label{feq2}
\end{align}
where the constants on the right-hand sides of~\eqref{feq1} and~\eqref{feq2} have been increased from~$\beta_1 n/2$,~$\beta_2 n/2$ as appearing in Lemma~2.3 in~\cite{MR3612870}, which is possible by changing the absolute constants~$c_1,c_2$ in its proof. 
Since~\eqref{feq1} implies for any~$r>0$ that
\begin{equation*}
\inf_{\theta \in B_r(\hat{\theta})} \Big|\Big\{i\in\{1,\dots,n\}:A''(\theta^{\top}X_i)\geq \inf_{y\in B_{r^{\dagger}}}A''(y) \Big\} \Big|\geq \bar{c}_1\beta_1 n,
\end{equation*}
where~$r^{\dagger}$ is defined in the assertion, the proof concludes.
\end{proof}

\subsection{Smoothness}\label{smosec}

We state a high probability upper bound for the eigenvalues of the relevant Hessians in GLMs. 
Proposition~\ref{eas} is a straightforward application of Theorem~6.1 in~\cite{MR3967104} for Gaussian~$X_i$. 

Proposition~\ref{eas} does not directly accommodate Poisson regression, but it will be used in its analysis to obtain the announced results.

\begin{prop}\label{eas}
Let~$\hat{\theta}\in\mathbb{R}^d$ and assume~$X_i\sim N(0,\Sigma)$. 
Assume there exists~$C_A>0$ such that~$A''\leq C_A$. It holds with probability at least~$1-e^{-n/2}$ that
\begin{equation*}
\sup_{v\in\mathbb{S}^{d-1},\theta \in \mathbb{R}^d}\sum_{i=1}^nA''(\theta^{\top}X_i)(v^{\top}X_i)^2\leq C_A \big(2\sqrt{n\lambda_{\max}(\Sigma)} + \sqrt{\textrm{Tr}(\Sigma)}\,\big)^2.
\end{equation*}
\end{prop}
\begin{proof}
We have
\begin{equation*}
\sup_{v\in\mathbb{S}^{d-1},\theta \in \mathbb{R}^d}\sum_{i=1}^nA''(\theta^{\top}X_i)(v^{\top}X_i)^2\leq C_A \sup_{v\in\mathbb{S}^{d-1}}\sum_{i=1}^n(v^{\top}X_i)^2.
\end{equation*}
By Theorem~6.1 in~\cite{MR3967104}, 
it holds with probability at least~$1-e^{-n/2}$ 
that
\begin{equation*}
\sup_{\bar{u}\in\mathbb{S}^{d-1}}\sum_{i=1}^n(\bar{u}^{\top}X_i)^2 \leq \big(2\sqrt{n\lambda_{\max}(\Sigma)} + \sqrt{\textrm{Tr}(\Sigma)}\,\big)^2. \qedhere
\end{equation*}
\end{proof}

We also provide the following Proposition~\ref{pep}, which allows~$X_i$ that are not Gaussian, using a Chernoff bound and an~$\epsilon$-net covering, though we do not expand on this setting past Proposition~\ref{pep}. 

\begin{prop}\label{pep}
Let~$\hat{\theta}\in\mathbb{R}^d$. 
For any~$t>0 $,~$r\in[0,\infty]$ and~$v\in\mathbb{S}^{d-1}$, let
\begin{equation}\label{etrv}
\hat{e}_{t,r,v} = \mathbb{E}\Big[e^{t\sup_{\theta\in B_r(\hat{\theta})}A''(\theta^{\top}X_1) (v^{\top}X_1)^2}\Big].
\end{equation}
For any~$r\in[0,\infty]$ and~$a\in\R $, 
it holds with probability at least
\begin{equation*}
1-\inf_{\epsilon\in(0,1)}\bigg(1+\frac{2}{\epsilon}\bigg)^d \bigg(\sup_{v\in\mathbb{S}^{d-1}}\inf_{t>0}e^{-ta(1-\epsilon)^2}\hat{e}_{t,r,v}\bigg)^n
\end{equation*}
that
\begin{equation*}
\sup_{v\in\mathbb{S}^{d-1},\,\theta\in B_r(\hat{\theta})} \sum_{i=1}^n A''(\theta^{\top}X_i)(v^{\top}X_i)^2 < na.
\end{equation*}
\end{prop}
\begin{proof} 
For any~$\epsilon>0$, by Lemma~5.2 in~\cite{MR2963170}, 
there exists an~$\epsilon$-net, namely a finite set~$S_{\epsilon}\subset\mathbb{S}^{d-1}$ of size at most~$(1+2/\epsilon)^d$ such that 
for any~$v\in\mathbb{S}^{d-1}$, there exists~$\hat{v}\in S_{\epsilon}$ satisfying~$\abs{v-\hat{v}}\leq \epsilon$. 
For any~$\epsilon>0$, let~$\hat{v}_{\epsilon}:\mathbb{S}^{d-1}\rightarrow S_{\epsilon}$ be such that
\begin{equation}\label{vhd}
\abs{v-\hat{v}_{\epsilon}(v)}\leq \epsilon
\end{equation}
for all~$v\in\mathbb{S}^{d-1}$. 
For 
any~$\theta\in\R^d$ 
and~$i\in[1,n]\cap\N $, denote
\begin{equation}\label{vhd2}
\hat{X}_{\theta,i} = (A''(\theta^{\top}X_i))^{\frac{1}{2}} X_i.
\end{equation}
For any~$v\in\mathbb{S}^{d-1}$,~$\theta\in\R^d$ and~$\epsilon,c>0$,
it holds that
\begin{align*}
\sum_{i=1}^n (v^{\top}\hat{X}_{\theta,i})^2 &\leq \sum_{i=1}^n \bigg((1+c)( (\hat{v}_{\epsilon}(v))^{\top} \hat{X}_{\theta,i} )^2 + (1+1/c)((v-\hat{v}_{\epsilon}(v))^{\top}\hat{X}_{\theta,i} )^2\bigg)\\
&\leq (1+c)\sum_{i=1}^n((\hat{v}_{\epsilon}(v))^{\top}\hat{X}_{\theta,i})^2 + \epsilon^2(1+1/c)\sup_{w\in\mathbb{S}^{d-1}} \sum_{i=1}^n(w^{\top}\hat{X}_{\theta,i} )^2.
\end{align*}
Taking the supremum over~$v$ and~$\theta$ on both sides, it holds for any~$\epsilon,c>0$ and~$r\in[0,\infty]$ satisfying~$\epsilon^2(1+1/c)<1$ that
\begin{align}
\sup_{v\in\mathbb{S}^{d-1},\theta\in B_r(\hat{\theta})}\sum_{i=1}^n(v^{\top}\hat{X}_{\theta,i})^2 &\leq \frac{1+c}{1-\epsilon^2(1+1/c)}\sup_{v\in\mathbb{S}^{d-1},\,\theta\in B_r(\hat{\theta})}\sum_{i=1}^n((\hat{v}_{\epsilon}(v))^{\top}\hat{X}_{\theta,i} )^2 \nonumber\\
&\leq \frac{1+c}{1-\epsilon^2(1+1/c)}\max_{\bar{v}\in S_{\epsilon}}\sum_{i=1}^n\sup_{\theta\in B_r(\hat{\theta})}(\bar{v}^{\top}\hat{X}_{\theta,i} )^2.\label{pa2}
\end{align}
Optimizing the fraction on the right-hand side of~\eqref{pa2} w.r.t.~$c$ (for fixed~$\epsilon\in(0,1)$) suggests the choice~$c=\epsilon/(1-\epsilon)$, which is indeed valid since~$c=\epsilon/(1-\epsilon)>\epsilon^2/(1-\epsilon^2)$ for~$\epsilon\in(0,1)$ and this is equivalent to~$\epsilon^2(1+1/c)<1$. 

On the other hand, for any~$\epsilon>0$,~$\bar{v}\in S_{\epsilon}$ and~$r\in[0,\infty]$, by a Chernoff bound, it holds for any~$a\in\R $ that
\begin{equation}\label{pe1}
\mathbb{P}\bigg(\sum_{i=1}^n\sup_{\theta\in B_r(\hat{\theta})} (\bar{v}^{\top}\hat{X}_{\theta,i})^2 \geq na\bigg) \leq \Big(\inf_{t>0} e^{-ta}\hat{e}_{t,r,\bar{v}}\Big)^n.
\end{equation}
Since there are at most~$(1+2/\epsilon)^d$ elements in~$S_{\epsilon}$,~\eqref{pe1} implies 
\begin{align}
\mathbb{P}\bigg( \max_{\bar{v}\in S_{\epsilon}} \sum_{i=1}^n \sup_{\theta\in B_r(\hat{\theta})} (\bar{v}^{\top}\hat{X}_{\theta,i})^2 \geq na\bigg)
&= \mathbb{P}\bigg( \bigcup_{\bar{v}\in S_{\epsilon}} \bigg\{\sum_{i=1}^n \sup_{\theta\in B_r(\hat{\theta})} (\bar{v}^{\top}\hat{X}_{\theta,i})^2 \geq na\bigg\}\bigg) \nonumber\\
&\leq \bigg(1+\frac{2}{\epsilon}\bigg)^d \bigg(\sup_{v\in\mathbb{S}^{d-1}}\inf_{t>0} e^{-ta}\hat{e}_{t,r,v}\bigg)^n.\label{pa3}
\end{align}
Combining~\eqref{pa2} with~$c=\epsilon/(1-\epsilon)$ and~\eqref{pa3} with~$a$ replaced by~$a(1-\epsilon)^2$ concludes the proof.
\end{proof}

\subsection{Maximum a posteriori}\label{map}
Using the curvature 
from Section~\ref{cursec}, 
we provide conditions for the existence of a critical point in the log-density in GLMs, and bound its distance from the true data-generating parameter~$\theta^*$. Similar results for logistic regression with Gaussian or flat prior 
were given in Section~5.1 of~\cite{MR4804813} and 
in~\cite{chardon2024f} (see also references therein). It may also be shown that this critical point is indeed the maximum a posteriori (see~\cite[Lemma~3]{chardon2024f}), but we forgo this part because it is not important for our results. Of course, our assumptions permit flat priors, which yields corresponding results for the maximum likelihood estimator.

First, in the next Lemma~\ref{lgb}, we control the size of the gradient of the log-likelihood without prior. 
Together with suitable assumptions on the prior, this yields with high probability a bound on the size of the gradient of the log-likelihood including prior. 
This bound on the gradient combined with local curvature (in a region) from Section~\ref{cursec} yields below our main Theorem~\ref{tte} of the section. 

Our proof of Lemma~\ref{lgb} shares similarities with that of Lemma~5.2 in~\cite{MR4804813} in that both proofs use Chernoff-type bounds, but the proof of Lemma~\ref{lgb} differs in a delicate way to accommodate linear and Poisson regression.

\begin{lemma}\label{lgb}
Let Assumption~\ref{A1} hold. 
Assume there exist constants~$\rho_0,\rho_1>0$ and~$y^*>0$ 
such that 
\begin{subequations}\label{ypy}
\begin{align}
&\quad\mathbb{P}(\abs{Y_1-A'(\theta^*\cdot X_1)}\geq y^*) \leq \rho_0, \quad\textrm{and}\label{ypy0}\\
\textrm{either} &\quad \mathbb{E}[\abs{Y_1-A'(\theta^*\cdot X_1)}\mathds{1}_{\abs{Y_1-A'(\theta^*\cdot X_1)}\geq y^*}|X_1] \leq \rho_1\quad\textrm{a.s.,}\label{ypy1}\\
\textrm{or } &\quad \mathbb{E}[\abs{Y_1-A'(\theta^*\cdot X_1)}^2\mathds{1}_{\abs{Y_1-A'(\theta^*\cdot X_1)}\geq y^*}] \leq \rho_1^2.\label{ypy2}
\end{align}
\end{subequations}
Let~$a\geq 64\rho_1(\lambda_{\max}(\Sigma))^{1/2}$. It holds that
\begin{equation}\label{ypc}
\mathbb{P}\bigg( \bigg|\sum_{i=1}^n (Y_i - A'(\theta^*\cdot X_i))X_i\bigg| \geq na \bigg) \leq n\rho_0 + e^{d\ln(5) - nt^*},
\end{equation}
where
\begin{equation}\label{tsd}
t^*= a\cdot(\lambda_{\max}(\Sigma))^{-1/2}\,\min\bigg(
\frac{\min(a\cdot(\lambda_{\max}(\Sigma))^{-1/2},15)}{30(y^*)^2},\frac{1}{4}
\bigg).
\end{equation}
\end{lemma}
\begin{proof}
Let
\begin{equation*}
A^* = \cap_{i=1}^n 
\{\abs{Y_i-A'(\theta^*\cdot X_{i})}<y^*\}.
\end{equation*}
It holds 
that
\begin{align}
&\mathbb{P}\big(\textstyle \sup_{s\in\mathbb{S}^{d-1}}
\sum_{i=1}^n (Y_i - A'(\theta^*\cdot X_i))X_i\cdot s \geq an \big)\nonumber\\
&\quad= \mathbb{P}\big(\big\{\textstyle
\sup_{s\in\mathbb{S}^{d-1}}
\sum_{i=1}^n (Y_i - A'(\theta^*\cdot X_i))X_i\cdot s \geq an \big\} \cap A^*\big)\nonumber\\
&\qquad+ \mathbb{P}\big(\big\{\textstyle 
\sup_{s\in\mathbb{S}^{d-1}}
\sum_{i=1}^n (Y_i - A'(\theta^*\cdot X_i))X_i\cdot s \geq an \big\} \cap (\Omega\setminus A^*)\big)\nonumber\\
&\quad\leq \mathbb{P}\big(\big\{\textstyle
\sup_{s\in\mathbb{S}^{d-1}}
\sum_{i=1}^n (Y_i - A'(\theta^*\cdot X_i))X_i\cdot s \geq an \big\} \cap A^*\big) + \mathbb{P}\big(\Omega\setminus A^*\big),\label{psq}
\end{align}
For the last term on the right-hand side of~\eqref{psq}, by~\eqref{ypy0}, it holds that
\begin{equation}\label{pso}
\mathbb{P}(\Omega\setminus A^*) \leq n\rho_0.
\end{equation}
For the first term on the right-hand side of~\eqref{psq}, note that by Lemma~5.2 in~\cite{MR2963170}, there exists 
a~$1/2$-net~$S_{1/2}$ of~$\mathbb{S}^{d-1}$ with~$\abs{S_{1/2}}=5^d$. Therefore there exists a mapping~$\hat{s}:\mathbb{S}^{d-1}\rightarrow S_{1/2}$ with~$\abs{\hat{s}(s)-s}\leq 1/2$ for all~$s\in\mathbb{S}^{d-1}$. It holds for any~$s\in\mathbb{S}^{d-1}$ that
\begin{equation*}
\textstyle
\sum_{i=1}^n (Y_i - A'(\theta^*\cdot X_i))X_i\cdot s = \sum_{i=1}^n (Y_i - A'(\theta^*\cdot X_i))X_i\cdot (\hat{s}(s)+(s-\hat{s}(s) )),
\end{equation*}
so that 
by taking supremum over~$s\in\mathbb{S}^{d-1}$ on both sides, we have
\begin{equation*}
\sup_{s\in\mathbb{S}^{d-1}}
\sum_{i=1}^n (Y_i - A'(\theta^*\cdot X_i))X_i\cdot s \leq  2\sup_{s\in S_{1/2}}\sum_{i=1}^n (Y_i - A'(\theta^*\cdot X_i))X_i\cdot s.
\end{equation*}
In particular, the first term on the right-hand side of~\eqref{psq} satisfies
\begin{align}
&\mathbb{P}\big(\big\{\textstyle
\sup_{s\in\mathbb{S}^{d-1}}
\sum_{i=1}^n (Y_i - A'(\theta^*\cdot X_i))X_i\cdot s \geq an \big\} \cap A^*\big) \nonumber\\
&\quad \leq \mathbb{P}\big(\big\{\textstyle
\sup_{s\in S_{1/2}}
\sum_{i=1}^n (Y_i - A'(\theta^*\cdot X_i))X_i\cdot s \geq an/2 \big\} \cap A^*\big)\nonumber\\
&\quad \leq 5^d\textstyle
\sup_{s\in S_{1/2}}\mathbb{P}\big(\big\{
\sum_{i=1}^n (Y_i - A'(\theta^*\cdot X_i))X_i\cdot s \geq an/2 \big\} \cap A^*\big).\label{qlq}
\end{align}
Let~$t
>0$. 
Since~$t,a>0$ are both strictly positive, 
it holds for any~$s\in \mathbb{S}^{d-1}$ that
\begin{align}
&\mathbb{P}\big(\big\{\textstyle\sum_{i=1}^n (Y_i - A'(\theta^*\cdot X_i))X_i\cdot s \geq an/2 \big\} \cap A^*\big) \nonumber\\
&\quad\leq e^{-tna/2}\mathbb{E}\big[\mathds{1}_{A^*}\cdot \exp\big(
t\textstyle\sum_{i=1}^n (Y_i - A'(\theta^*\cdot X_i))X_i\cdot s
\big)\big] \nonumber\\
&\quad\leq e^{-tna/2}\big(\mathbb{E}\big[ 
\mathds{1}_{\{\abs{Y_1-A'(\theta^*\cdot X_1)}<y^*\}}
\cdot\exp(t (Y_1 - A'(\theta^*\cdot X_1))X_1\cdot s)\big]\big)^n.\label{bnc}
\end{align}
We estimate the expectation on the right-hand side. 
Denote
\begin{equation*}
\bar{Y} = Y_1- A'(\theta^*\cdot X_1).
\end{equation*}
For any~$s\in\mathbb{S}^{d-1}$, 
note that it holds on the event~$\abs{\bar{Y}}<y^*$ that
\begin{align}
&e^{t\bar{Y}X_1\cdot s}-1 - t\bar{Y}X_1\cdot s\nonumber\\
&\quad= (e^{t\bar{Y}X_1\cdot s} - 1 - t\bar{Y}X_1\cdot s)(\mathds{1}_{\{X_1\cdot s\geq 0\}} + \mathds{1}_{\{X_1\cdot s< 0\}} )\nonumber\\
&\quad\leq (e^{ty^*X_1\cdot s} - 1 - ty^*X_1\cdot s)\mathds{1}_{\{X_1\cdot s\geq 0\}} + (e^{-ty^*X_1\cdot s} - 1 + ty^*X_1\cdot s)\mathds{1}_{\{X_1\cdot s< 0\}}\nonumber \\
&\quad\leq e^{ty^*X_1\cdot s} + e^{-ty^*X_1\cdot s} -2.\label{rep}
\end{align}
For any~$s\in\mathbb{S}^{d-1}$, the expectation on the right-hand side of~\eqref{bnc} satisfies 
\begin{equation*}
\mathbb{E}\big[\mathds{1}_{\{\abs{\bar{Y}}<y^*\}}\cdot\exp\big( t\bar{Y} X_1\cdot s
\big) \big] \leq 1+ E_1 + E_2,
\end{equation*}
where, by~\eqref{rep}, 
\begin{align*}
E_1 &:= \mathbb{E}\Big[\mathds{1}_{\{\abs{\bar{Y}}<y^*\}} \cdot (e^{t\bar{Y}X_1\cdot s} - 1 - t\bar{Y}X_1\cdot s)\Big] \\
&\leq \mathbb{E}[e^{ty^*X_1\cdot s} + e^{-ty^*X_1\cdot s} - 2 \big] \\
&= 2e^{(ty^*)^2s^{\top}\Sigma s/2} -2
\end{align*}
and, by the property~$\mathbb{E}[Y_1|X_1]=A'(\theta^*X_i)$ of exponential families and in case of~\eqref{ypy1}, 
\begin{align*}
E_2&:=\mathbb{E}\big[\mathds{1}_{\{\abs{\bar{Y}}<y^*\}}\cdot t\bar{Y}X_1\cdot s\big]\\
&= t\mathbb{E}\big[ \big( \mathbb{E}\big[\bar{Y}\big|X_1\big] - \mathbb{E}\big[\mathds{1}_{\{\abs{\bar{Y}}\geq y^*\}} \cdot \bar{Y}\big|X_1\big] \big) X_1\cdot s\big]\\
&\leq t\mathbb{E}\big[\mathbb{E}\big[\mathds{1}_{\{\abs{\bar{Y}}\geq y^*\}} \cdot \abs{\bar{Y}}\big|X_1 \big] \abs{X_1\cdot s}\big]\\
&\leq t\rho_1\mathbb{E}[\abs{X_1\cdot s}]\\
&\leq t\rho_1 (s^{\top}\Sigma s)^{1/2}. 
\end{align*}
By similar arguments,~$E_2\leq t\rho_1(s^{\top}\Sigma s)^{1/2}$ also holds in case of~\eqref{ypy2}. 
Substituting the last three inequalities into~\eqref{bnc}, then the result into~\eqref{qlq}, then into~\eqref{psq} and~\eqref{pso} yields
\begin{align}
&\mathbb{P}\big(\textstyle \sup_{s\in\mathbb{S}^{d-1}}
\abs{\sum_{i=1}^n (Y_i - A'(\theta^*\cdot X_i))X_i\cdot s} \geq an \big)\nonumber\\
&\leq n\rho_0 + 5^d e^{-tna/2}(1+2(e^{(ty^*)^2\lambda_{\max}(\Sigma)/2}-1) + t\rho_1(\lambda_{\max}(\Sigma))^{1/2})^n\nonumber\\
&\leq n\rho_0 + 5^d \exp(-tna/2 +2n(e^{(ty^*)^2\lambda_{\max}(\Sigma)/2}-1) + tn\rho_1(\lambda_{\max}(\Sigma))^{1/2}).\label{wal}
\end{align}
By setting~$t=4t^*
/a$, with~$t^*$ given by~\eqref{tsd}, 
we have
\begin{align*}
&2n(e^{(ty^*)^2\lambda_{\max}(\Sigma)/2}-1) \\
&\quad\leq 2n\bigg(\exp\bigg(t\cdot(\lambda_{\max}(\Sigma))^{1/2}\min\bigg(\frac{a}{15(\lambda_{\max}(\Sigma))^{1/2}},1\bigg)\bigg)-1\bigg) \\
&\quad\leq 2n(e-1)\cdot ta/15.
\end{align*}
Together with the assumption~$a\geq 64\rho_1(\lambda_{\max}(\Sigma))^{1/2}$, inequality~\eqref{wal} implies~\eqref{ypc}.
\end{proof}

The next Proposition~\ref{vll} verifies condition~\eqref{ypy} for linear, logistic and Poisson regression.
\begin{prop}\label{vll}
Let~$\sigma>0$. 
For linear and logistic regression, namely the respective cases where~\eqref{line} and~\eqref{loge} hold 
for all~$z\in\mathbb{R}$,~$y\in S$ and where~$\eta$ is the Lebesgue and counting measure, condition~\eqref{ypy} is satisfied for any~$\epsilon\in(0,1)$ with~$\rho_0=\rho_1/\phi^{1/2}=\epsilon$ and respectively
\begin{align*}
y^* &= \big(2\sigma^2\ln(
2/\epsilon)\big)^{1/2},\\
y^* &= 2.
\end{align*}
If in addition Assumption~\ref{A1} holds, then for Poisson regression, where~\eqref{pois} holds for all~$z\in\mathbb{R}$,~$y\in S$ and where~$\eta$ is the counting measure, condition~\eqref{ypy} is satisfied for any~$\epsilon\in(0,1)$ with~$\rho_0=\rho_1=\epsilon$ and
\begin{equation*}
y^* = \max\big(e^2\exp\!\big((2\abs{\theta^*}^2\lambda_{\max}(\Sigma)\ln(2/\epsilon^*))^{1/2}\big), \ln(2/\epsilon^*)\big),
\end{equation*}
where
\begin{equation}\label{see}
\epsilon^* := \epsilon^4\exp(-2(\theta^*)^{\top}\Sigma\theta^*)/6. 
\end{equation}
\end{prop}
\begin{proof}
For linear regression,~$Y_1$ is conditionally Gaussian with mean~$A'(\theta^*\cdot X_1)=\theta^*\cdot X_1$ and variance~$\sigma^2$. 
Therefore for any~$\epsilon\in(0,1)$, condition~\eqref{ypy0} holds with~$\rho_0=\epsilon$ and~$y^*$ satisfying
\begin{equation}\label{ysl}
y^* \geq  
(2\sigma^2 \ln(2/\epsilon))^{1/2}, 
\end{equation}
where we have used standard Gaussian concentration results, see e.g.~\eqref{psp} below. 
In addition, for any~$y>0$, 
the left-hand side of~\eqref{ypy1} is
\begin{align*}
\mathbb{E}[\abs{Y_1-\mathbb{E}[Y_1|X_1]}\mathds{1}_{\abs{Y_1-\mathbb{E}[Y_1|X_1]}\geq y}|X_1] &= (2\sigma^2/\pi)^{1/2} e^{-y^2/(2\sigma^2)},
\end{align*}
so that taking~$y=y^*$ as given in the statement yields both~\eqref{ypy0} and~\eqref{ypy1} with~$\rho_0=\rho_1/\sigma=\epsilon$ for any~$\epsilon\in(0,1)$. 

For logistic regression,~$Y_1$ takes values in~$[0,1]$, so~\eqref{ypy} holds with any~$y^*>1$ and~$\rho_0=\rho_1=0$.

For Poisson regression, Lemma~\ref{poiy} asserts for any~$\bar{\epsilon}\in(0,1)$ and any~$y\geq y_{\bar{\epsilon}}^*$ given by~\eqref{ypa0} that~$\mathbb{P}(Y_1\geq y)\leq \bar{\epsilon}$. 
Note that Lemma~\ref{poiy} depends only on the properties of the Poisson distribution and~$N(0,\Sigma)$, and not on those of the regression likelihood, so there is no circular reasoning. 
Therefore, using also the form of~$y_{\bar{\epsilon}}^*$ in~\eqref{ypa0}, it holds for 
any such~$\bar{\epsilon},y$ that
\begin{align*}
\mathbb{P}(\abs{Y_1-e^{\theta^*\cdot X_1}} \geq y) 
&= \mathbb{P}(\{\abs{Y_1-e^{\theta^*\cdot X_1}} \geq y\}\cap\{e^{\theta^*\cdot X_1}< y\}) \\
&\quad+ \mathbb{P}(\{\abs{Y_1-e^{\theta^*\cdot X_1}} \geq y\}\cap\{e^{\theta^*\cdot X_1} \geq y\})\\
&\leq \mathbb{P}(Y_1 \geq y+e^{\theta^*\cdot X_1}) + \mathbb{P}(e^{\theta^*\cdot X_1}\geq  y)\\
&\leq \mathbb{P}(Y_1 \geq y) + \mathbb{P}(\theta^*\cdot X_1 \geq \ln y)\\
&\leq \bar{\epsilon} + e^{-(\ln y)^2/(2\abs{\theta^*}^2\lambda_{\max}(\Sigma))}\\
&\leq \bar{\epsilon} + \bar{\epsilon}/2.
\end{align*}
In particular, for any~$\epsilon\in(0,1)$, condition~\eqref{ypy0} holds with~$\rho_0=\epsilon$ and~$y^*\geq y_{2\epsilon/3}^*$, where the last inequality is satisfied if~$y^*= y_{\epsilon^*}^*$ for~$\epsilon^*$ given by~\eqref{see}, as assumed, because~$\epsilon^*\leq 2\epsilon/3$. We proceed to show~\eqref{ypy2}. 
For any~$\bar{\epsilon}\in(0,1)$ and~$y\geq y_{2\bar{\epsilon}/3}^*$, by the expression for the fourth centered moment of the Poisson distribution~\cite[(5.20)]{MR467977}, 
the left-hand side in~\eqref{ypy2} may be bounded as in
\begin{align*}
&\mathbb{E}[\abs{Y_1 - \mathbb{E}[Y_1|X_1]}^{2}\mathds{1}_{\abs{Y_1- \mathbb{E}[Y_1|X_1]}\geq y}] \\
&\quad\leq (\mathbb{E}[\abs{Y_1 - \mathbb{E}[Y_1|X_1]}^{4}])^{1/2}\cdot \mathbb{P}(\abs{Y_1- \mathbb{E}[Y_1|X_1]}\geq y)^{1/2}\\
&\quad\leq (\mathbb{E}[e^{\theta^*\cdot X_1} + 3e^{2\theta^*\cdot X_1}])^{1/2}\cdot \bar{\epsilon}^{1/2}\\
&\quad=(e^{(\theta^*)^{\top}\Sigma\theta^*/2} + 3e^{2(\theta^*)^{\top}\Sigma\theta^*})^{1/2}\cdot \bar{\epsilon}^{1/2}\\
&\quad\leq 2e^{(\theta^*)^{\top}\Sigma\theta^*}\cdot \bar{\epsilon}^{1/2}.
\end{align*}
Therefore for any~$\epsilon\in(0,1)$, by setting~$\bar{\epsilon} = (3/2)\epsilon^*$ given by~\eqref{see}, 
condition~\eqref{ypy2} holds with~$\rho_1=\epsilon$ and~$y^*=y_{\epsilon^*}^*$, which is the assertion.
\end{proof}

With Lemma~\ref{lgb} and having verified condition~\eqref{ypy} in Proposition~\ref{vll}, we are ready to give the main technical result of the section. In the rest of this section, we will consider the slightly more general density~$\bar{\pi}(\cdot|Z^{(n)}):\mathbb{R}^d\rightarrow (0,\infty)$ given for any~$\theta\in\mathbb{R}^d$ by
\begin{equation}\label{piba}
\bar{\pi}(\theta|Z^{(n)}) = \frac{\pi(\theta)\prod_{i=1}^n e^{\phi^{-1}\cdot(Y_i\theta^{\top}X_i - \bar{A}(\theta^{\top}X_i))}}{\int_{\mathbb{R}^d}\pi(\bar{\theta})\prod_{i=1}^n e^{\phi^{-1}\cdot(Y_i\bar{\theta}^{\top}X_i - \bar{A}(\bar{\theta}^{\top}X_i))}d\bar{\theta}},
\end{equation}
where~$\bar{A}\in C^2(\mathbb{R})$ is a convex function such that there exists~$\hat{r}\in(0,\infty]$ with
\begin{subequations}\label{aba}
\begin{align}
\bar{A}(z)&=A(z) & & \forall z\in(-\infty,\hat{r})\\
\bar{A}''(z)&\in [0,A''(\hat{r})] & & \forall z\in[\hat{r},\infty).
\end{align}
\end{subequations}
In case~$\bar{A}=A$ and~$\hat{r}=\infty$, we have~$\bar{\pi}(\cdot|Z^{(n)})=\pi(\cdot|Z^{(n)})$ and the results below can be read as such. We allow for the more general~$\bar{A}$ because it will be useful for our considerations on Poisson regression. 
Note that we do not change the distribution of~$(Y_i,X_i)$ with~$\bar{A}$. 
\begin{theorem}\label{tte}
Let Assumptions~\ref{A1} and~\ref{asp} hold and let~$B_0$ denote the event in Assumption~\ref{asp}. 
Assume~$\nabla\ln\pi(0)=0$. 
Assume there exist constants~$\rho_0,\rho_1>0$ and~$y^*>0$ such that condition~\eqref{ypy} holds. 
Let~$c_1,c_2$ be the absolute constants from Proposition~\ref{p1} with~$\bar{c}_1=\bar{c}_2=19/20$. 
Let~$\hat{r}^* = 2\abs{\theta^*}(\lambda_{\max}(\Sigma))^{1/2}+ 2\phi^{1/2}$. 
Assume
\begin{align}
\rho_1&\leq
\hat{\rho}(\hat{r}^*):=\frac{\phi^{1/2}}{512\kappa(\Sigma)}\cdot \inf_{z\in B_{\hat{r}^*}}A''(z)\label{r1b}\\
n &\geq \max\bigg(4c_1(d+1),\frac{42\phi C_{\pi}d}{\lambda_{\min}(\Sigma)\inf_{B_{\hat{r}^*}}A''},\frac{56\phi^{1/2} \kappa(\Sigma)C_\pi d\abs{\theta^*}}{(\lambda_{\max}(\Sigma))^{1/2}\inf_{B_{\hat{r}^*}}A''}\bigg).\label{r2b}
\end{align}
Let~$\bar{A}\in C^2(\mathbb{R})$ be a convex function such that there exists~$\hat{r}\in[\hat{r}^*,\infty]$ satisfying~\eqref{aba}. 
There exists an event~$B\subset B_0$ with probability at least
\begin{align}
&1-\exp(-c_2(\Phi(2)-\Phi(-2))^2n) - \exp(-4c_2(\Phi(-2/3))^2n)-\rho_{\pi} - n\rho_0\nonumber\\
&\quad-\exp(d\ln(5) - nt^*) - n\exp(-\hat{r}^2/(2\abs{\theta^*}^2\lambda_{\max}(\Sigma))),\label{spo}
\end{align}
where
\begin{equation}\label{tsd2}
t^*= \frac{\phi^{1/2}}{8\kappa(\Sigma)} \inf_{B_{\hat{r}^*}}A''\cdot \min\bigg(
\frac{\min(\phi^{1/2}(8\kappa(\Sigma))^{-1}\inf_{B_{\hat{r}^*}}A'',15)}{30(y^*)^2},\frac{1}{4}
\bigg),
\end{equation}
such that on~$B$ there exists a unique~$
\theta_{\textrm{map}}\in B_{(\phi/\lambda_{\max}(\Sigma))^{1/2}}(\theta^*)$ satisfying
\begin{equation*}
\nabla\ln\bar{\pi}(\cdot|Z^{(n)})|_{\theta_{\textrm{map}}}=0,
\end{equation*}
where~$\bar{\pi}(\cdot|Z^{(n)}):\mathbb{R}^d\rightarrow (0,\infty)$ is given by~\eqref{piba}. 
\end{theorem}
\begin{proof}
Throughout the proof, we denote~$\tau=(\phi/\lambda_{\max}(\Sigma))^{1/2}$ where convenient. 
By Remark~\ref{p1r}, Proposition~\ref{p1} holds with~\eqref{p1req}. Let~$w_1,w_2,\beta_1,\beta_2$ be given by~\eqref{p1req}. 
In particular, if~$n\geq c_1(d+1)/\min(\beta_1,\beta_2)^2$, then there exists an event~$\mathcal{E}_1$ 
with~$\mathbb{P}(\mathcal{E}_1)\geq 1-\exp(-c_2\beta_1^2n) - \exp(-c_2\beta_2^2n)$ such that it holds on~$\mathcal{E}_1$ 
that
\begin{equation*}
\inf_{\theta\in B_{\tau}(\theta^*)}\lambda_{\min}\bigg(\phi^{-1}\sum_{i=1}^n \bar{A}''(\theta^{\top}X_i)X_iX_i^{\top}\bigg) \geq \frac{n\phi^{-1}\lambda_{\min}(\Sigma)\inf_{z\in B_{\hat{r}^*}}A''(z)}{6}.
\end{equation*}
Therefore if
\begin{equation*}
n\geq 
42\phi C_{\pi}d/\Big(\lambda_{\min}(\Sigma)\inf_{z\in B_{\hat{r}^*}}A''(z)\Big),
\end{equation*}
then by Assumption~\ref{asp} there exists an event~$\mathcal{E}_1'\subset B_0$ 
with~$\mathbb{P}(\mathcal{E}_1')\geq 1-\exp(-c_2\beta_1^2n) - \exp(-c_2\beta_2^2n) - \rho_{\pi}$ such that 
\begin{equation}\label{d2p}
\inf_{\theta\in B_{\tau}(\theta^*)}\lambda_{\min}\Big(-D^2\ln 
\bar{\pi}(\cdot|Z^{(n)})
\big|_{\theta}\Big) \geq \frac{n\phi^{-1}\lambda_{\min}(\Sigma)\inf_{z\in B_{\hat{r}^*}}A''(z)}{7}.
\end{equation}
Inequality~\eqref{d2p} implies for any~$\theta\in B_{\tau}(\theta^*)$ that
\begin{align}
&\big(-\nabla\ln
\bar{\pi}(\cdot|Z^{(n)})
\big|_{\theta} + \nabla\ln
\bar{\pi}(\cdot|Z^{(n)})
\big|_{\theta^*}\big)^{\top}(\theta-\theta^*)\nonumber\\
&\quad= -\bigg(\int_0^1 D^2\ln
\bar{\pi}(\cdot|Z^{(n)})
\big|_{\theta^*+t(\theta-\theta^*)}dt (\theta-\theta^*)\bigg)^{\top}(\theta-\theta^*)\nonumber\\
&\quad\geq \big(n\phi^{-1}\lambda_{\min}(\Sigma)\textstyle\inf_{z\in B_{\hat{r}^*}}A''(z)/7\big)\abs{\theta-\theta^*}^2.\label{kah}
\end{align}

On the other hand, 
by~\eqref{psp} (which only depends on Assumption~\ref{A1}), it holds that
\begin{equation*}
\mathbb{P}(\cup_{i=1}^n\{\theta^*\cdot X_i\geq \hat{r}\})\leq ne^{-\hat{r}^2/(2\abs{\theta^*}^2\lambda_{\max}(\Sigma))}.
\end{equation*}
Thus it holds with probability at least~$1-ne^{-\hat{r}^2/(2\abs{\theta^*}^2\lambda_{\max}(\Sigma))}$ that 
\begin{equation}\label{equ}
\phi^{-1}\sum_{i=1}^n (Y_i - \bar{A}'(\theta^*\cdot X_i))X_i =  \phi^{-1}\sum_{i=1}^n (Y_i - A'(\theta^*\cdot X_i))X_i.
\end{equation}
By~\eqref{r1b}, Lemma~\ref{lgb} with
\begin{equation*}
a= \frac{(\phi\lambda_{\max}(\Sigma))^{1/2}}{8\kappa(\Sigma)}\cdot \inf_{B_{\hat{r}^*}}A'',
\end{equation*}
together with~\eqref{equ}, implies there exists an event~$\mathcal{E}_2$ 
with~$\mathbb{P}(\mathcal{E}_2)\geq 1- n\rho_0 - e^{d\ln(5) - nt^*} - ne^{-\hat{r}^2/(2\abs{\theta^*}^2\lambda_{\max}(\Sigma))}$, where~$t^*$ is given by~\eqref{tsd2} 
such that it holds on~$\mathcal{E}_2$ that 
\begin{align}
\bigg|\phi^{-1}\sum_{i=1}^n (Y_i - \bar{A}'(\theta^*\cdot X_i))X_i\bigg| 
&= 
\bigg|\phi^{-1}\sum_{i=1}^n (Y_i - A'(\theta^*\cdot X_i))X_i\bigg|\nonumber \\
&< \frac{\phi^{-1/2}(\lambda_{\max}(\Sigma))^{1/2}}{8\kappa(\Sigma)}\cdot n\inf_{B_{\hat{r}^*}}A''.\label{ypc2}
\end{align}
Moreover, the assumption~$\nabla \ln\pi(0)=0$ and Assumption~\ref{asp} yields on~$B_0$ that
\begin{equation*}
\abs{\nabla \ln\pi(\theta)}\leq C_{\pi}d\abs{\theta}\qquad\forall\theta\in\mathbb{R}^d.
\end{equation*}
Together with~\eqref{ypc2}, we have that if
\begin{equation*}
n\geq \frac{\phi^{1/2}\cdot \kappa(\Sigma)}{(1/7-1/8)\inf_{B_{\hat{r}^*}}A''}\cdot \frac{C_\pi d\abs{\theta^*}}{(\lambda_{\max}(\Sigma))^{1/2}},
\end{equation*}
then it also holds on~$\mathcal{E}_1'\cap\mathcal{E}_2\subset B_0$ that
\begin{align}
\big|\nabla \ln
\bar{\pi}(\cdot|Z^{(n)})
|_{\theta^*}\big| &< \frac{\phi^{-1/2}(\lambda_{\max}(\Sigma))^{1/2}}{8\kappa(\Sigma)} \cdot n \inf_{B_{\hat{r}^*}}A'' + C_{\pi}d\abs{\theta^*}\nonumber\\
&< \frac{\phi^{-1/2}(\lambda_{\max}(\Sigma))^{1/2}}{7\kappa(\Sigma)}\cdot n \inf_{ B_{\hat{r}^*}}A''.\label{jca}
\end{align}

Combining~\eqref{kah} 
and~\eqref{jca} yields on~$\mathcal{E}_1'\cap\mathcal{E}_2$ and for any~$\theta\in B_{\tau}(\theta^*)$ that
\begin{align}
-\nabla\ln
\bar{\pi}(\cdot|Z^{(n)})
\big|_{\theta}\cdot (\theta-\theta^*) &> \big(n\phi^{-1}\lambda_{\min}(\Sigma)\textstyle\inf_{B_{\hat{r}^*}}A''/7\big)\abs{\theta-\theta^*}^2\nonumber\\
&\quad-\phi^{-1/2} (\lambda_{\max}(\Sigma))^{1/2}(7\kappa(\Sigma))^{-1}n\textstyle\inf_{ B_{\hat{r}^*}}A''\abs{\theta-\theta^*}. \label{cga}
\end{align}
The function~$-\ln
\bar{\pi}(\cdot|Z^{(n)})
:B_{\tau}(\theta^*)\rightarrow \mathbb{R}$ is continuous, thus there exists~$\hat{\theta}\in B_{\tau}(\theta^*)$ 
such that~$-\ln
\bar{\pi}(\hat{\theta}|Z^{(n)})
= \inf_{\theta\in B_{\tau}(\theta^*)}(-\ln
\bar{\pi}(\theta|Z^{(n)}))
$. 
Suppose for a contradiction that~$\hat{\theta}$ is on the boundary of~$B_{\tau}(\theta^*)$, namely~$\abs{\hat{\theta}-\theta^*}=\tau = (\phi/\lambda_{\max}(\Sigma))^{1/2}$. Inequality~\eqref{cga} implies
\begin{equation*}
-\nabla\ln
\bar{\pi}(\cdot|Z^{(n)})
\big|_{\hat{\theta}}\cdot (\hat{\theta}-\theta^*) > 0.
\end{equation*}
By Taylor's theorem on the function~$\mathbb{R}\ni t\mapsto-\ln
\bar{\pi}
(\hat{\theta} + t(\theta^*-\hat{\theta})
|Z^{(n)}
)$ at~$t=0$, there exists~$\theta\in B_{\tau}(\theta^*)$ such that~$-\ln
\bar{\pi}(\theta|Z^{(n)})
<-\ln
\bar{\pi}(\hat{\theta}|Z^{(n)})
$, which contradicts the defining property of~$\hat{\theta}$. Thus~$\hat{\theta}$ belongs in the interior of~$B_{\tau}(\theta^*)$, and it must attain zero gradient by the twice differentiability of~$\ln
\bar{\pi}(\cdot|Z^{(n)})
$. 
Uniqueness follows from strong convexity~\eqref{d2p}. 
\end{proof}
\begin{remark}
By inspection of the proof of Theorem~\ref{tte}, and setting (e.g.)~$
\theta_{\textrm{map}}=0$ outside of~$B$, it follows from the measurable version of Berge's maximum theorem~\cite[Theorem~18.19]{MR2378491} that~$
\theta_{\textrm{map}}$ is an r.v. (that it is measurable). We use this throughout without further mention.
\end{remark}

In the following Corollary~\ref{nvc}, we combine Theorem~\ref{tte} (in the linear, logistic and Poisson regression cases with Proposition~\ref{vll}) with the local curvature result Proposition~\ref{p1} (in the Gaussian~$P_X$ case with Remark~\ref{p1r}) around~$\theta_{\textrm{map}}$, in order to show the concentration of~$\bar{\pi}(\cdot|Z^{(n)})$ around~$\theta_{\textrm{map}}$ using the results of Section~\ref{concensec}. 
This conclusion will require a condition on~$n$ that depends on~$\sup_{B_{\hat{r}}}A''$ in~\eqref{aba}, see~\eqref{nrc} below. For linear and logistic regression, this dependence is not problematic, since~$A''$ is globally bounded, so the result readily applies to the true posterior~$\pi(\cdot|Z^{(n)})$ with~$\hat{r}=\infty$. For Poisson regression, we will use the concentration part of Corollary~\ref{nvc} only with a carefully chosen~$\hat{r}<\infty$ (essentially the right-hand side of~\eqref{rsd} below), which will suffice for the conclusions in Section~\ref{poisec}.
\begin{corollary}\label{nvc}
Let Assumptions~\ref{A1} and~\ref{asp} hold and let~$B_0$ denote the event in Assumption~\ref{asp}. 
Assume~$0\in\argmax\pi$. 
Let~$R>0$ be given by
\begin{equation}\label{rhdef2}
R=\abs{\theta^*} + 2(\phi/\lambda_{\max}(\Sigma))^{1/2}.
\end{equation}
Let~$\bar{A}\in C^2(\mathbb{R})$ and~$\hat{r}\in [4R(\lambda_{\max}(\Sigma))^{1/2},\infty]$ be such that~\eqref{aba} and~$\sup_{B_{\hat{r}}}A''<\infty$ hold. 
For each of the linear, logistic and Poisson regression cases, given by~\eqref{line} for some~$\sigma>0$,~\eqref{loge} and~\eqref{pois} respectively, let~$c^*:(0,\infty)^2\rightarrow\mathbb{R}$ be respectively given by
\begin{subequations}\label{porp}
\begin{align}
c^*(\bar{c},\bar{\delta}) &= \bar{c}\max\big(1,\abs{\theta^*}(\lambda_{\max}(\Sigma))^{1/2},(1\vee \sigma)\ln(\kappa(\Sigma)+n/\bar{\delta}),\sigma^{-1}\big) \label{porlin}\\
c^*(\bar{c},\bar{\delta}) &= 
\bar{c}e^{\bar{c}\abs{\theta^*}(\lambda_{\max}(\Sigma))^{1/2}},\\
c^*(\bar{c},\bar{\delta}) &= \bar{c}e^{\bar{c}\abs{\theta^*}^2\lambda_{\max}(\Sigma)}\cdot e^{\bar{c}\abs{\theta^*}(\lambda_{\max}(\Sigma)\ln(\kappa(\Sigma)))^{1/2}}\cdot e^{\bar{c}\abs{\theta^*}(\lambda_{\max}(\Sigma)\ln(n/\bar{\delta}))^{1/2}}\!\!.\label{por}
\end{align}
\end{subequations}
Let
\begin{equation*}
c^{\dagger} = \max\bigg( (\kappa(\Sigma))^2,
\kappa(\Sigma)\bigg(1\vee\ln\bigg(\frac{\phi C_{\pi}}{\lambda_{\min}(\Sigma)} + \sup_{B_{\hat{r}}} A''\cdot\kappa(\Sigma) \bigg)\bigg),\frac{\phi C_{\pi}}{\lambda_{\min}(\Sigma)}\bigg)
\end{equation*}
and let~$\epsilon,\delta\in(0,1)$. 
In each linear, logistic, Poisson case, there exists an absolute constant~$c_3>0$ and an event~$B\subset B_0$ with~$\mathbb{P}(B)\geq 1-\delta-\rho_{\pi}-n\exp(-\hat{r}^2/(2\abs{\theta^*}^2\lambda_{\max}(\Sigma)))$ such that if
\begin{align}
n&\geq  c^*(c_3,\delta)\cdot c^{\dagger}\cdot\big(d+\ln(\delta^{-1})+\ln(\epsilon^{-1}) \big),\label{nrc}
\end{align}
then on~$B$ there exists a unique~$
\theta_{\textrm{map}}\in B_{(\phi/\lambda_{\max}(\Sigma))^{1/2}}(\theta^*)$ with~$
\nabla\ln\bar{\pi}(\cdot|Z^{(n)})|_{\theta_{\textrm{map}}}=0$
and it holds on~$B$ that
\begin{equation*}
\int_{\mathbb{R}^d\setminus B_{(\phi/\lambda_{\max}(\Sigma))^{1/2}}(\theta_{\textrm{map}})} \bar{\pi}(\bar{\theta}|Z^{(n)}) d\bar{\theta}\leq \epsilon,
\end{equation*}
where~$\bar{\pi}(\cdot|Z^{(n)}):\mathbb{R}^d\rightarrow (0,\infty)$ is given by~\eqref{piba}. 
\end{corollary}
\begin{remark}
For existence and uniqueness of~$\theta_{\textrm{map}}$, we will not require~$\sup_{B_{\hat{r}}}A''<\infty$, and the condition~\eqref{nrc} on~$n$ can be weakened, see Lemma~\ref{uuc}. This will be used in Section~\ref{poisec}.
\end{remark}

Before proceeding with the full proof of Corollary~\ref{nvc}, in which we will use our concentration result (Lemma~\ref{tvv}), we present two functions~$U,U_0$ that we will check the assumptions of Lemma~\ref{tvv} on (and subsequently apply Lemma~\ref{tvv} to, in Corollary~\ref{nvc}). We also check that~$U$ satisfies Condition~\ref{cond:curvature2} for later use. This is Lemma~\ref{uuc} below. The function~$U$ is the (space-translated) negative log-density of the full posterior. The function~$U_0$ is a flattened version of the negative log-density of the prior. We first define a smooth and appropriately-scaled cut-off function~$\chi$, which will be used to do this flattening in a ball around the origin. The ball will be large enough to encompass a region of concentration around~$\theta_{\textrm{map}}$ (to be defined by Theorem~\ref{tte}).

We use this flattened~$U_0$ instead of~$-\ln\pi$ in order to control the prefactors related to~$e^{-U_0}$ in~\eqref{tvveq1}, which would otherwise grow like~$e^{Cd}$ in general under Assumption~\ref{asp} compared to the~$e^{-cn/d}$ control (to be shown) in the rest of the factors in~\eqref{tvveq1} of Lemma~\ref{tvv}. We use this flattened~$U_0$ also instead of a completely flat~$U_0=0$ in order to accommodate the fact that~$U-U_0$ needs to have non-negative curvature everywhere to satisfy Condition~\ref{cond:curvature3} in Lemma~\ref{tvv}, which would not necessarily be the case under heavy-tailed priors.

Let~$R$ be given by~\eqref{rhdef2}. Recall the notation~$\varphi_{\bar{\epsilon}}$ for the scaled mollifier and let~$\chi:\mathbb{R}^d\rightarrow[0,1]$ be the smooth cut-off function given for any~$x\in\mathbb{R}^d$ by
\begin{equation}\label{chidef}
\chi(x) = \int_{\mathbb{R}^d}\varphi_{R/2}(x-y)\mathds{1}_{B_{3R/2}}(y)dy,
\end{equation}
which satisfies~$\sup_{\mathbb{R}^d}\abs{\nabla\chi}\leq \int_{\mathbb{R}^d}\abs{\nabla \varphi_{R/2}(y)}dy\leq \int_{\mathbb{R}^d}(2/R)\abs{\nabla \varphi(y)}dy= (2/R)\int_{B_1}\frac{2\abs{y}\varphi(y)}{1-\abs{y}^2}dy\leq  c_0R^{-1}$ for an absolute constant~$c_0>0$ and similarly~$\sup_{\mathbb{R}^d}\abs{D^2\chi}\leq c_0'R^{-2}$ for an absolute constant~$c_0'>0$. 
\begin{lemma}\label{uuc}
Let Assumptions~\ref{A1} and~\ref{asp} hold and let~$B_0$ denote the event in Assumption~\ref{asp}. 
Assume~$0\in\argmax\pi$. 
Let~$R>0$ be given by~\eqref{rhdef2}. 
Let~$\bar{A}\in C^2(\mathbb{R})$ and~$\hat{r}\in [4R(\lambda_{\max}(\Sigma))^{1/2},\infty]$ be such that~\eqref{aba} holds. 
Let~$\delta\in(0,1)$. 
In each linear, logistic, Poisson case, there exists an absolute constant~$\hat{c}_3>0$ 
and an event~$B\subset B_0$ with
\begin{equation*}
\mathbb{P}(B)\geq 1-\delta-\rho_{\pi}-n\exp(-\hat{r}^2/(2\abs{\theta^*}^2\lambda_{\max}(\Sigma)))
\end{equation*}
such that if
\begin{align}
n\geq c^*(\hat{c}_3,\delta) \cdot
\max\bigg((\kappa(\Sigma))^2,\frac{\phi C_{\pi}}{\lambda_{\min}(\Sigma)}\bigg)\cdot(d+\ln(\delta^{-1}))\label{nmw}
\end{align}
holds, where~$c^*$ is given by~\eqref{porp} in each respective case, then the following statements hold on~$B$.
\begin{enumerate}[label=(\roman*)]
\item 
There exists a unique~$
\theta_{\textrm{map}}\in B_{(\phi/\lambda_{\max}(\Sigma))^{1/2}}(\theta^*)$ with
\begin{equation*}
\nabla\ln\bar{\pi}(\cdot|Z^{(n)})|_{\theta_{\textrm{map}}}=0,
\end{equation*}
\item For~$\chi$ given by~\eqref{chidef} with~\eqref{rhdef2}, the functions~$U,U_0:\mathbb{R}^d\rightarrow\mathbb{R}$ given by
\begin{subequations}\label{uuo}
\begin{align}
U(\theta)&=\bar{U}(\theta):=\textstyle -\ln\pi(\theta+\theta_{\textrm{map}}) -\phi^{-1}\sum_{i=1}^n 
(Y_i(\theta+\theta_{\textrm{map}})^{\top}X_i - \bar{A}((\theta+\theta_{\textrm{map}})^{\top}X_i)),\label{uuo1}\\
U_0(\theta)&=\bar{U}_0(\theta):=-\ln\pi(0)+(1-\chi(\theta+\theta_{\textrm{map}}))\cdot(-\ln\pi(\theta+\theta_{\textrm{map}}) +\ln\pi(0))\label{uuo2}
\end{align}
\end{subequations}
satisfy Conditions~\ref{cond:curvature3},~\ref{cond:smooth},~\ref{cond:curvature2} with
\begin{subequations}\label{molm}
\begin{align}
m_0 &= (\hat{c}(R)n/2)\mathds{1}_{[0,(\phi/\lambda_{\max}(\Sigma))^{1/2}]},\nonumber\\
\textrm{where}\qquad \hat{c}(R) &:= (6\phi)^{-1}\lambda_{\min}(\Sigma)
\cdot \inf_{B_{2R(\lambda_{\max}(\Sigma))^{1/2}}}A'',\label{chdef}\\
\textrm{and}\qquad L &= C_{\pi}d+\phi^{-1}\cdot 9n\lambda_{\max}(\Sigma)\cdot \sup_{B_{\hat{r}}}A'',\label{lcdef}\\
m &= m_0 - C_{\pi}d\mathds{1}_{((\phi/\lambda_{\max}(\Sigma))^{1/2},\infty]}\label{mcdef}
\end{align}
\end{subequations}
whenever~$L<\infty$.
\end{enumerate}
\end{lemma}
\begin{remark}\label{rel}
It is straightforward to verify that~\eqref{nrc} with large enough~$c_3$ is stronger than~\eqref{nmw}. In particular, the setting of Corollary~\ref{nvc} is stronger than that of Lemma~\ref{uuc}, which will be used in the proof of Corollary~\ref{nvc}.
\end{remark}
\begin{proof}[Proof of Lemma~\ref{uuc}]
First, we check the existence of~$\theta_{\textrm{map}}$ using Theorem~\ref{tte}. Subsequently, we check local strong convexity (by bounds on the Hessian of terms from~$\bar{U}-\bar{U}_0$) around the `true' origin (which will be enough for local strong convexity around~$\theta_{\textrm{map}}$ over a smaller ball, in particular Condition~\ref{cond:curvature3}). Lastly, we will check Conditions~\ref{cond:curvature2},~\ref{cond:smooth}. All of these verifications will be valid only on some events in the probability space. We take~$B$ to be the intersection of these events and estimate its probability at the end of the proof. 

We first use Proposition~\ref{vll} to verify condition~\eqref{ypy}, which is required in Theorem~\ref{tte}. 
Let 
\begin{equation*}
\bar{\rho}_1 = (512\kappa(\Sigma))^{-1}\cdot 
\phi^{1/2}\inf_{B_{2\abs{\theta^*}(\lambda_{\max}(\Sigma))^{1/2}+2\phi^{1/2}}}A''.
\end{equation*}
By Proposition~\ref{vll} 
(with~$\epsilon = \bar{\rho}_1/\phi^{1/2}$ therein), 
either~\eqref{ypy1} or~\eqref{ypy2} is satisfied with~$\rho_1=\bar{\rho}_1$ and~$y^*$ given by the expressions in Proposition~\ref{vll} with~$\epsilon=\bar{\rho}_1/\phi^{1/2}$. Define~$\bar{y}^*$ to be these respective expressions. We have in particular in each linear, logistic and Poisson case respectively that
\begin{align*}
\bar{y}^*=\bar{y}_{\textrm{lin}}^* &:= \big(2\sigma^2\ln(1024\kappa(\Sigma))\big)^{1/2},\\
\bar{y}^*=\bar{y}_{\textrm{log}}^* &:= 2,\\
\bar{y}^*\leq\bar{y}_{\textrm{poi}}^* &:= \bar{c}_3\exp\Big(\bar{c}_3 \Big(\abs{\theta^*}^2\lambda_{\max}(\Sigma) + \sqrt{\abs{\theta^*}^2\lambda_{\max}(\Sigma)\ln\kappa(\Sigma)}\,\Big)\Big)
\end{align*}
for some absolute constant~$\bar{c}_3\geq 1$ (which follows after some manipulation on the expression for~$y^*$ in Proposition~\ref{vll} for the Poisson case, using~$A(z)=e^z$).
Moreover by again Proposition~\ref{vll} (this time with~$\epsilon = \delta/(2n)$ therein), 
condition~\eqref{ypy0} is satisfied with~$\rho_0=\delta/(2n)$ and~$y^*$ given as in Proposition~\ref{vll} with~$\epsilon=\delta/(2n)$. Define~$\hat{y}^*$ to be the respective expressions. 
In particular, in each case, we have respectively
\begin{align*}
\hat{y}^* = \hat{y}_{\textrm{lin}}^* &:= \big(2\sigma^2\ln(4n/\delta)\big)^{1/2},\\
\hat{y}^* = \hat{y}_{\textrm{log}}^* &:= 2,\\
\hat{y}^* \leq \hat{y}_{\textrm{poi}}^* &:= \bar{c}_3'\exp\Big(\bar{c}_3' \Big(\abs{\theta^*}^2\lambda_{\max}(\Sigma) + \sqrt{\abs{\theta^*}^2\lambda_{\max}(\Sigma)\ln(2n/\delta)}\,\Big)\Big)
\end{align*}
for some absolute constant~$\bar{c}_3'\geq 1$. These two applications of Proposition~\ref{vll} yield that condition~\eqref{ypy} is satisfied with~$\rho_0=\delta/(2n)$,~$\rho_1=\bar{\rho}_1$ and~$y^*=\max(\hat{y}^*,\bar{y}^*)$. 
Therefore by Theorem~\ref{tte}, there exists an absolute constant~$\hat{c}_3'\geq 1$ and an event~$\tilde{B}'\subset B_0$ 
such that 
if~\eqref{nmw} (with~$\hat{c}_3$ replaced by~$\hat{c}_3'$) 
holds, 
then on~$\tilde{B}'$ there exists a unique~$
\theta_{\textrm{map}}\in B_{(\phi/\lambda_{\max}(\Sigma))^{1/2}}(\theta^*)$ satisfying~$
\nabla \ln\bar{\pi}(\cdot|Z^{(n)})|_{\theta_{\textrm{map}}}=0$. 
Moreover,~$\tilde{B}'$ has probability at least~\eqref{spo}, where~$c_1,c_2$ are absolute constants,~$\rho_0=\delta/(2n)$ and~$t^*$ is given by~\eqref{tsd2} with~$\hat{r}^* = 2\abs{\theta^*}(\lambda_{\max}(\Sigma))^{1/2} + 2\phi^{1/2}$ and~$y^* = \max(\hat{y}^*,\bar{y}^*)$. 
Note that this~$t^*$ satisfies in the respective linear, logistic and Poisson cases
\begin{subequations}\label{tsf}
\begin{align}
t^* &= t_{\textrm{lin}}^*:=\frac{\sigma}{8\kappa(\Sigma)}\min\bigg(
\frac{\min(\sigma(8\kappa(\Sigma))^{-1},15)}{30\max((\bar{y}_{\textrm{lin}}^*)^2,(\hat{y}_{\textrm{lin}}^*)^2)}
,\frac{1}{4}\bigg),\\
t^* &\geq t_{\textrm{log}}^*:=\bigg(\frac{e^{-2\abs{\theta^*}(\lambda_{\max}(\Sigma))^{1/2}}}{C\kappa(\Sigma)}\bigg)^2,\\
t^* &\geq t_{\textrm{poi}}^*:=\bigg(\frac{\exp(-2\abs{\theta^*}(\lambda_{\max}(\Sigma))^{1/2}- 2)}{8\kappa(\Sigma)}\bigg)^2\cdot 
\frac{1}{30(\max(\bar{y}_{\textrm{poi}}^*,\hat{y}_{\textrm{poi}}^*))^2}
\end{align}
\end{subequations}
for some generic absolute constant~$C>0$, 
where in the logistic case~\eqref{loge} we have used~$\inf_{B_{r^{\dagger}}}A''\geq C^{-1}e^{-r^{\dagger}}$ for any~$r^{\dagger}>0$. 
This~$\theta_{\textrm{map}}$ is the random variable we use in~\eqref{uuo}. In particular, note that given~\eqref{nmw} (with large enough~$\hat{c}_3$) and on~$\tilde{B}'$, we have~$\nabla(\bar{U}-\bar{U}_0)|_0=\nabla \bar{U}|_0=0$ (where we have used the definition~\eqref{rhdef2} of~$R$ to get that~$\bar{U}_0=-\ln\pi(0)$ near~$0$), which is necessary for Conditions~\ref{cond:curvature3},~\ref{cond:curvature2}. 

We proceed onto Condition~\ref{cond:curvature3}. Note that since~$\nabla (\bar{U}-\bar{U}_0)|0=\nabla \bar{U}|_0=0$, it suffices to lower bound the Hessian of~$\bar{U}-\bar{U}_0$ w.r.t~$m_0$ as in~\eqref{molm}. By definition~\eqref{uuo} of~$\bar{U},\bar{U}_0$, for any~$\theta\in\mathbb{R}^d$ we have
\begin{equation}\label{htm}
D^2\bar{U}(\theta)-D^2\bar{U}_0(\theta) = D^2F(\theta+\theta_{\textrm{map}}) + \phi^{-1} \textstyle\sum_{i=1}^n\bar{A}''((\theta+\theta_{\textrm{map}})^{\top} X_i) X_iX_i^{\top},
\end{equation}
where the function~$F:\mathbb{R}^d\rightarrow\mathbb{R}$ is given for any~$x\in\mathbb{R}^d$ by
\begin{equation}\label{fcdef}
F(x) = \chi(x)\cdot (-\ln\pi(x)+\textstyle\sup_{\mathbb{R}^d}\ln\pi)
\end{equation}
for~$\chi$ as in~\eqref{chidef}, where we have used the assumption~$0\in\argmax\pi$. 
We lower bound the sum on the right-hand side of~\eqref{htm}, then upper bound the first term on the right-hand side. 
By Proposition~\ref{p1} and Remark~\ref{p1r} with~$\hat{\theta}=0$,~$\bar{c}_1=\bar{c}_2=19/20$ and~$r=R$, there exist absolute constants~$c_1',c_2'>0$ and an event~$\tilde{B}$ with~$\mathbb{P}(\tilde{B})\geq 1-2e^{-c_2'n}$ such that 
if~$n\geq c_1'd$ (which is satisfied given~\eqref{nmw} for large enough~$\hat{c}_3$), then it holds on~$\tilde{B}$ that
\begin{equation}\label{dbr}
\inf_{\theta\in B_R}
\phi^{-1}\cdot \lambda_{\min}\bigg(\sum_{i=1}^n 
\bar{A}''(\theta^{\top}X_i) 
X_iX_i^{\top}\bigg) 
\geq 
\hat{c}(R)n,
\end{equation}
where~$\hat{c}(R)$ is defined in~\eqref{chdef}. We apply Proposition~\ref{p1} and Remark~\ref{p1r} a second time, with the larger radius~$r=2R$ and the same~$\hat{\theta},\bar{c}_1,\bar{c}_2$. This application yields that there exist absolute constants~$c_1'',c_2''>0$ and an event~$\tilde{B}^{\dagger}$ with~$\mathbb{P}(\tilde{B}^{\dagger})\geq 1-2e^{-c_2''n}$ such that if~$n\geq c_1''d$, then it holds on~$\tilde{B}^{\dagger}$ that
\begin{equation}\label{dbrd}
\inf_{\theta\in B_{2R}}
\phi^{-1}\cdot \lambda_{\min}\bigg(\sum_{i=1}^n 
\bar{A}''(\theta^{\top}X_i) 
X_iX_i^{\top}\bigg) 
\geq 
\frac{\phi^{-1}\lambda_{\min}(\Sigma)
\inf_{z\in B_{4R(\lambda_{\max}(\Sigma))^{1/2}}}A''}{6}\cdot n.
\end{equation}
Moreover,~$F$ satisfies, by the assumption~$0\in\argmax\pi$, on the event~$B_0$ (for any~$x\in B_{2R}$ thus any~$x\in\mathbb{R}^d$) that
\begin{align}
\abs{D^2 F(x)} &\leq \abs{D^2 \chi(x)(-\ln\pi(x)+\textstyle\sup_{\mathbb{R}^d}\ln\pi)} + 2\abs{\nabla\chi(x)}\abs{\nabla\ln\pi(x)} + C_{\pi}d\nonumber\\
&\leq c_0'R^{-2}\cdot 4C_{\pi}dR^2 + 2c_0R^{-1}\cdot 2C_{\pi}dR + C_{\pi}d\nonumber\\
&\leq \bar{c}_0C_{\pi}d,\label{dfc}
\end{align}
where~$c_0,c_0'$ are absolute constants (see just after~\eqref{chidef}) and~$\bar{c}_0:=1+4c_0+4c_0'$. 
Note also that on the event~$\tilde{B}'$, by definition~\eqref{rhdef2} of~$R$ and by~$\theta_{\textrm{map}}\in B_{(\phi/\lambda_{\max}(\Sigma))^{1/2}}(\theta^*)$, we have~$B_{(\phi/\lambda_{\max}(\Sigma))^{1/2}}(\theta_{\textrm{map}})\subset B_R$. 
Substituting~\eqref{dbr} and~\eqref{dfc} into~\eqref{htm}, on the event~$\tilde{B}\cap\tilde{B}^{\dagger}\cap\tilde{B}'$, 
we have that if~$n\geq 2\bar{c}_0C_{\pi}d/\hat{c}(R)$ (which follows from~\eqref{nmw} for large enough~$\hat{c}_3$, where we have used the definitions~\eqref{chdef},~\eqref{rhdef2} of~$\hat{c}(R)$ and~$R$ and the forms of~$A$ in all of the cases), then
\begin{align}
\inf_{\theta\in B_{(\phi/\lambda_{\max}(\Sigma))^{1/2}}}\lambda_{\min}(D^2\bar{U}(\theta)-D^2\bar{U}_0(\theta)) &\geq \inf_{\theta\in B_R}\lambda_{\min}(D^2\bar{U}(\theta-\theta_{\textrm{map}})-D^2\bar{U}_0(\theta-\theta_{\textrm{map}}))\nonumber\\
&\geq \hat{c}(R)n - \bar{c}_0C_{\pi}d\nonumber\\
&\geq \hat{c}(R)n/2.\label{csd1}
\end{align}
Moreover, substituting~\eqref{dbrd} and~\eqref{dfc} into~\eqref{htm}, on the event~$\tilde{B}\cap\tilde{B}^{\dagger}\cap\tilde{B}'$, we have that if
\begin{equation*}
n\geq \bar{c}_0C_{\pi}d\cdot 6\phi \bigg(\lambda_{\min}(\Sigma)\inf_{z\in B_{4R(\lambda_{\max}(\Sigma))^{1/2}}}A''\bigg)^{-1},
\end{equation*}
(which again follows from~\eqref{nmw} for large enough~$\hat{c}_3$), then
\begin{equation}\label{csd2}
\inf_{\theta\in B_{2R}}\lambda_{\min}(D^2\bar{U}(\theta-\theta_{\textrm{map}})-D^2\bar{U}_0(\theta-\theta_{\textrm{map}}))\geq 0.
\end{equation}
In addition, by definitions~\eqref{fcdef},~\eqref{chidef} of~$F,\chi$, on the event~$\tilde{B}'$, we have~$F=0$ on~$\mathbb{R}^d\setminus B_{2R}$, which implies by~\eqref{htm} that~$\inf_{\mathbb{R}^d\setminus B_{2R}} (D^2\bar{U}(\theta-\theta_{\textrm{map}})-D^2\bar{U}_0(\theta-\theta_{\textrm{map}})) \geq 0$. Together with~\eqref{csd1} and~\eqref{csd2}, this implies on~$\tilde{B}\cap\tilde{B}^{\dagger}\cap\tilde{B}'$ that if~\eqref{nmw} holds with large enough~$\hat{c}_3$, then 
Condition~\ref{cond:curvature3} is satisfied 
with~\eqref{uuo} and~$m_0$ as in~\eqref{molm}.

On the other hand, by the assumption~\eqref{aba} on~$\bar{A},\hat{r}$, our smoothness Proposition~\ref{eas} and Assumption~\ref{asp} imply there exists an event~$\tilde{B}''$ with~$\mathbb{P}(\tilde{B}'')\geq 1-e^{-n/2}$ such that if~$n\geq d$ then it holds on~$\tilde{B}''\cap B_0$ that the smoothness Condition~\ref{cond:smooth} is satisfied with~$U$ as in~\eqref{uuo1} and 
\begin{align*}
L &= C_{\pi}d + \textstyle\sup_{B_{\hat{r}}}A'' \cdot \phi^{-1}\cdot(2\sqrt{n\lambda_{\max}(\Sigma)} + \sqrt{\textrm{Tr}(\Sigma)})^2\\
&\leq C_{\pi}d + \textstyle\sup_{B_{\hat{r}}}A'' \cdot \phi^{-1}\cdot 9n\lambda_{\max}(\Sigma).
\end{align*}

Lastly, considering again~\eqref{dbr} and Assumption~\ref{asp} together, if
\begin{equation*}
n\geq \frac{12\phi C_{\pi} d}{\lambda_{\min}(\Sigma)}\cdot\bigg(\inf_{B_{2R(\lambda_{\max}(\Sigma))^{1/2}}} A''\bigg)^{-1}
\end{equation*}
(which is satisfied given~\eqref{nmw} for large enough~$\hat{c}_3$, where we have used that~$A$ is constant for linear regression and~$\phi=1$ otherwise), then~$\hat{c}(R)n-C_{\pi}d\geq \hat{c}(R)n/2$ and on~$\tilde{B}\cap\tilde{B}'$ Condition~\ref{cond:curvature2} is satisfied with~$m$ given by~\eqref{mcdef}.

It remains to estimate (further) the probability~$\mathbb{P}(\tilde{B}\cap\tilde{B}^{\dagger}\cap\tilde{B}'\cap\tilde{B}'')\geq 1-2e^{-c_2'n}-2e^{-c_2''n}-(1-P^{\dagger})-e^{-n/2}$, where~$P^{\dagger}$ is defined by~\eqref{spo} with absolute constants~$c_1,c_2$,~$\rho_0=\delta/(2n)$ and~$t^*$ satisfying~\eqref{tsf} in the respective cases. 
We check that~\eqref{nmw} with large enough~$\hat{c}_3$ implies
\begin{equation*}
P^{\dagger}\geq 1-\rho_{\pi}-5\delta/6-n\exp(-\hat{r}^2/(2\abs{\theta^*}^2\lambda_{\max}(\Sigma))).
\end{equation*}
This, together with~$n\geq C\ln(\delta^{-1})$ (recall~$C$ is a generic absolute constant), would imply~$\mathbb{P}(\tilde{B}\cap\tilde{B}^{\dagger}\cap\tilde{B}'\cap\tilde{B}'')\geq 1-\delta - \rho_{\pi} - n\exp(-\hat{r}^2/(2\abs{\theta^*}^2\lambda_{\max}(\Sigma)))$ as required. 
By~\eqref{tsf}, condition~\eqref{nmw} implies
\begin{equation*}
n\geq C(t^*)^{-1}(d+\ln(\delta^{-1})),
\end{equation*}
which subsequently implies 
\begin{equation*}
\delta/2 + \exp(d\ln(5) - nt^*) \leq 2\delta/3,
\end{equation*}
and this concludes by definition of~$P^{\dagger}$ and by~$\rho_0=\delta/(2n)$ (using also~$n\geq C\ln(\delta^{-1})$).
\end{proof}

\begin{proof}[Proof of Corollary~\ref{nvc}]
Existence and uniqueness of~$
\theta_{\textrm{map}}$ hold by Lemma~\ref{uuc} (see also Remark~\ref{rel}), so it remains to check the concentration property using our Lemma~\ref{tvv}. 
In the sequel, we assume the notations in the statement of Lemma~\ref{uuc}. Having checked Conditions~\ref{cond:curvature3} and~\ref{cond:smooth} in Lemma~\ref{uuc} for~$U,U_0$ therein, we check~\eqref{eq0} before applying Lemma~\ref{tvv}. 
If (in particular\footnote{The extraneous~$\ln$ factor and term in~\eqref{n3c} are used subsequently to control the right-hand side of~\eqref{pcol}.})
\begin{equation}\label{n3c}
n\geq \frac{16}{3}\cdot\frac{\lambda_{\max}(\Sigma)}{\hat{c}(R)\phi}\cdot \bigg(d\cdot\bigg[1+\ln\bigg(\frac{ 8C_{\pi} + \phi^{-1}\cdot 72\textstyle\sup_{B_{\hat{r}}}A'' \cdot\lambda_{\max}(\Sigma)}{\hat{c}(R)}\bigg)\bigg] + \ln(2\epsilon^{-1})\bigg)
\end{equation}
(which is satisfied given~\eqref{nrc} for large enough~$c_3$), 
then~\eqref{eq0} holds (with~$r=(\phi/\lambda_{\max}(\Sigma))^{1/2}$ and~$c=\hat{c}(R)n/2$ therein). Thus we have verified all of the conditions required to apply Lemma~\ref{tvv} to~$U,U_0$ as in~\eqref{uuo}, with~$i=0$
and~$r=(\phi/\lambda_{\max}(\Sigma))^{1/2}$,~$c=\hat{c}(R)n/2$. Note that, by definitions~\eqref{uuo2},~\eqref{rhdef2} of~$\bar{U}_0$ and~$R$, the assumption~$0\in\argmax_{\mathbb{R}^d}\pi$ and the assertion in Lemma~\ref{uuc} that~$\theta_{\textrm{map}}\in B_{(\phi/\lambda_{\max}(\Sigma))^{1/2}}$ imply
\begin{equation*}
\textstyle\sup_{\mathbb{R}^d}e^{-\bar{U}_0}=\inf_{B_{(\phi/\lambda_{\max}(\Sigma))^{1/2}}}e^{-\bar{U}_0} = \pi(0).
\end{equation*}
This application of Lemma~\ref{tvv} yields 
that 
if~\eqref{nrc} holds with large enough~$c_3$, then by Lemma~\ref{tvv}, 
it holds on~$B$ (from Lemma~\ref{uuc}) that 
\begin{align}
&\int_{\mathbb{R}^d\setminus B_{(\phi/\lambda_{\max}(\Sigma))^{1/2}}(\theta_{\textrm{map}})}\bar{\pi}(\bar{\theta}|Z^{(n)}) d\bar{\theta}\nonumber\\
&\quad\leq 
\varrho\bigg(1+\varrho\bigg(\frac{C_{\pi}d + \phi^{-1}\cdot 9\textstyle\sup_{B_{\hat{r}}}A'' \cdot n\lambda_{\max}(\Sigma)}{\hat{c}(R)n/2}\bigg)^{d/2}\,\bigg),\label{pcol}
\end{align}
where~$\varrho = e^{2d-3\hat{c}(R)n(\phi/\lambda_{\max}(\Sigma))/16}$. 
Further, if~\eqref{n3c} holds (also using~$L/(\hat{c}(R)n/2)\geq 1$ from definitions~\eqref{lcdef},~\eqref{chdef}, which implies that the first~$\ln$ expression in~\eqref{n3c} is no less than~$\ln(4L/(\hat{c}(R)n/2))$), 
then we have
\begin{equation*}
n\geq \frac{16}{3}\cdot \frac{\lambda_{\max}(\Sigma)}{\hat{c}(R)\phi}\cdot d\cdot \bigg(2 +\frac{1}{2}\ln\bigg(\frac{C_{\pi} + \phi^{-1}\cdot 9\textstyle\sup_{B_{\hat{r}}}A'' \cdot \lambda_{\max}(\Sigma)}{\hat{c}(R)/2}\bigg)\bigg),
\end{equation*}
which, assuming~$n\geq d$, implies
\begin{equation*}
\varrho\bigg(\frac{C_{\pi}d + \phi^{-1}\cdot 9\textstyle\sup_{B_{\hat{r}}}A'' \cdot n\lambda_{\max}(\Sigma)}{\hat{c}(R)n/2}\bigg)^{d/2}\leq 1.
\end{equation*}
Similarly (also using the inequality just derived), condition~\eqref{n3c} implies
\begin{equation*}
\varrho\bigg(1+\varrho\bigg(\frac{C_{\pi}d + \phi^{-1}\cdot 9\textstyle\sup_{B_{\hat{r}}}A'' \cdot n\lambda_{\max}(\Sigma)}{\hat{c}(R)n/2}\bigg)^{d/2}\,\bigg) \leq 2\varrho \leq \epsilon.
\end{equation*}
Thus on~$B$ the left-hand side of~\eqref{pcol} is bounded above by~$\epsilon$ given~\eqref{nrc} with large enough~$c_3$. This concludes the proof.
\end{proof}

\section{Poisson regression}\label{poisec}

It turns out that to obtain an~$n\geq Cd$ condition in the Poisson regression case, the results above and variations thereof relying on well-behavedness inside a Euclidean ball do not suffice. In the proofs of the statements below, we give arguments specifically to neglect the regions~$\cup_{i=1}^n\{\theta\in\mathbb{R}^d:\theta^{\top}X_i\geq r\}\subset\{\theta\in\mathbb{R}^d:\abs{X\theta}_{\infty}\geq r\}$ for some~$r>0$ where the negative log-posterior and its gradient are exponentially large. 
The value~$r$ will be required to scale square-root-logarithmically with respect to key parameters, which will produce super-polylogarithmic, but sub-polynomial, (e.g. smoothness) constants and the announced complexities.

In all of this section, we assume the GLM setting and notation from the beginning of Section~\ref{OGLM}. The Poisson regression case is~\eqref{pois} 
and we will focus on normally distributed covariates. 
The main results are Theorem~\ref{poiwei}, in which concentration of some densities in the complement of the aforementioned regions is shown, and Theorem~\ref{poigib}, in which sample and iteration complexities for the Gibbs sampler targeting~$\pi(\cdot|Z^{(n)})$ are given.

\subsection{Preliminary results and concentration}
We begin with a basic lemma based on the properties of the Gaussian and Poisson distributions. It is worth noting that~\eqref{psp} below holds for a fixed~$\theta\in B_r$, and so it cannot be used to infer properties over larger sets of~$\theta$ (for example through~$\epsilon$-nets) without possibly incurring additional dimension dependence.

\begin{lemma}\label{poiy}
Let Assumption~\ref{A1} hold. 
For~$i\in\mathbb{N}$ and for any~$r,\bar{r}\geq 0$,~$\theta\in B_r$, it holds that
\begin{equation}\label{psp}
\mathbb{P}(\theta^{\top}X_i\geq \bar{r})
\leq e^{-\bar{r}^2/(2r^2\lambda_{\max}(\Sigma))}.
\end{equation}
Assume additionally~\eqref{pois} for all~$z\in\mathbb{R}$,~$y\in S$ and that~$\eta$ is the counting measure, namely~$Y_i$ is conditionally Poisson with parameter~$\exp(\theta^*\cdot X_i)$ for~$i\in\mathbb{N}$. 
For any~$\epsilon\in(0,1)$, 
it holds for~$i\in\mathbb{N}$ and for any
\begin{equation}\label{ypa0}
y\geq y_{\epsilon}^*:= \max\big(e^2e^{(2\abs{\theta^*}^2\lambda_{\max}(\Sigma)\ln(2/\epsilon))^{1/2}}, \ln(2/\epsilon)\big) 
\end{equation}
that 
\begin{equation}\label{ypa}
\mathbb{P}(Y_i< y) \geq 1-\epsilon.
\end{equation}
\end{lemma}
\begin{proof}
Let~$i\in\mathbb{N}$. 
For any~$r\geq 0$ and any~$\theta\in B_r\setminus\{0\}$, by Theorems~5.2,~5.3 in~\cite{MR1849347} with~$F=\theta^{\top}X_i/\abs{\theta}$, it holds for any~$\bar{r}\geq 0$ that
\begin{equation*}
\mathbb{P}(\theta^{\top}X_i\geq \bar{r}\abs{\theta})\leq e^{-\bar{r}^2/(2\lambda_{\max}(\Sigma))},
\end{equation*}
which implies~\eqref{psp}. 
Moreover, it holds by assumption that~$(Y_i,X_i)\sim P_{\theta^*}$, where the density of~$P_{\theta^*}$ is given by~\eqref{pitdef} with~\eqref{pois}. Namely~$Y_i$ is conditionally Poisson with parameter~$\exp(\theta^*\cdot X_i)$. Therefore by~\cite[Theorem~5.4]{MR3674428} and~\eqref{psp}, 
it holds for any~$\bar{r}\geq 0$ and~$y\in(e^{\bar{r}},\infty)$ that 
\begin{align*}
\mathbb{P}(Y_i \geq y) &= \mathbb{E}[\mathbb{E}[\mathds{1}_{\{Y_i\geq y\}}|X_i]] \\
&\leq \mathbb{E}[\mathbb{E}[\mathds{1}_{\{Y_i\geq y\}}|X_i]
(\mathds{1}_{
\theta^*\cdot X_i\geq \bar{r}
}+\mathds{1}_{
\theta^*\cdot X_i < \bar{r}
})]\\
&\leq \mathbb{P}(\theta^*\cdot X_i\geq \bar{r}) + \mathbb{E}[\mathbb{E}[\mathds{1}_{\{Y_i\geq y\}}|X_i]\mathds{1}_{
\theta^*\cdot X_i< \bar{r}
}]\\
&\leq e^{-\bar{r}^2/(2\abs{\theta^*}^2\lambda_{\max}(\Sigma))} + e^{(1+\bar{r}-\ln(y))y}.
\end{align*}
In particular, for any~$\epsilon\in(0,2)$, 
taking~$\bar{r} 
= \big(2\abs{\theta^*}^2\lambda_{\max}(\Sigma)\ln(2/\epsilon)\big)^{1/2}$ yields~\eqref{ypa} for any~$y\geq \max(e^{\bar{r}+2},\ln(2/\epsilon))$. 
\end{proof}

We write down the 
alternate (to~$A(z) = e^z$) function~$\bar{A}$ that produces a well-behaved density through~\eqref{piba}. For some~$\hat{r}\in(0,\infty]$, consider~$\bar{A}:\mathbb{R}\rightarrow\mathbb{R}$ given by
\begin{equation}\label{aba2}
\bar{A}(z)= \begin{cases}
e^z &\textrm{if }z\in(-\infty,\hat{r}),\\
e^{\hat{r}}\big(5/6 + 3(z-\hat{r})/2+\max(0,(1-z+\hat{r})^3/6)\big)
&\textrm{if }z\in[\hat{r},\infty).
\end{cases}
\end{equation}
Note that such~$\bar{A}\in C^2(\mathbb{R})$ is convex and satisfies~\eqref{aba} with~$A(z)=e^z$. Moreover~$\bar{A}=\exp(\cdot)$ if~$\hat{r}=\infty$, so we recover the Poisson regression~$A$ in that case.
Note that~$\bar{A},\bar{A}',\bar{A}''$ satisfy~$\bar{A}(z)\leq e^z$ for all~$z\in\mathbb{R}$,~$\sup_{\mathbb{R}}\bar{A}'=3e^{\hat{r}}/2$ and~$\sup_{\mathbb{R}}\bar{A}''=e^{\hat{r}}$.

The following Theorem~\ref{poiwei} asserts that the alternate posterior~\eqref{piba} is concentrated on well-behaved regions, in which the gradient from~$A$ is controlled.

\begin{theorem}\label{poiwei}
Let Assumptions~\ref{A1} and~\ref{asp} hold and let~$B_0$ denote the event in Assumption~\ref{asp}. 
Assume~\eqref{pois} for all~$z\in\mathbb{R}$,~$y\in S$ and that~$\eta$ is the counting measure. 
Assume~$0\in\argmax \pi$. 
Let~$r>0$. 
Let~$\hat{r}\in[r,\infty]$ and let~$\bar{A},\bar{\pi}(\cdot|Z^{(n)})$ be 
given by~\eqref{aba2} and~\eqref{piba}. 
Let~$\epsilon,\delta\in(0,1)$. If
\begin{equation}\label{rsd}
r\geq r^*(\epsilon,\delta) := \max\big(4\ln(2y_{\delta/(5n)}^*), (\abs{\theta^*}(\lambda_{\max}(\Sigma))^{1/2} + 2) (8\ln(20n/(\epsilon\delta)))^{1/2}\big),
\end{equation}
where~$y_{\delta/(5n)}^*$ is given by~\eqref{ypa0}, hold then the following statements hold.
\begin{enumerate}[label=(\roman*)]
\item \label{pot0} 
There exists an absolute constant~$c_4\geq1$ such that if
\begin{equation*}
n\geq c^*(c_4,\delta)\max\big((\kappa(\Sigma))^2,C_{\pi}(\lambda_{\min}(\Sigma))^{-1}\big)(d+\ln(\delta^{-1})),
\end{equation*}
where~$c^*$ is given by~\eqref{por}, 
then there exists an event~$B\subset B_0$ with~$\mathbb{P}(B)\geq 1-\delta - \rho_{\pi}$ on which there exists a unique~$
\theta_{\textrm{map}}\in B_{(\lambda_{\max}(\Sigma))^{-1/2}}(\theta^*)$ with~$
\nabla \ln\bar{\pi}(\cdot|Z^{(n)})|_{\theta_{\textrm{map}}}=0$. 
\item \label{pot1} 
Suppose~$\pi$ is deterministic and~$\hat{r}<\infty$. Let~$R$ be given by~\eqref{rhdef2}. 
There exists an absolute constant~$c_4'\geq 1$ and an event~$B'\subset B$ (from part~\ref{pot0}) with~$\mathbb{P}(B')\geq 1-\delta
$ such that if~\eqref{nrc} (with~$c^*$ given by~\eqref{por} and with~$c_3$ replaced by~$c_4'$) holds, then it holds on~$B'$ that 
\begin{equation}\label{meco}
\int_{\mathbb{R}^d\setminus B_R}\bar{\pi}(\bar{\theta}|Z^{(n)})
d\bar{\theta}\leq \epsilon,\qquad\int_{\cup_{i=1}^n\{\theta\in\mathbb{R}^d:\theta^{\top}X_i\geq r\}} 
\bar{\pi}(\bar{\theta}|Z^{(n)})
d\bar{\theta}
\leq \epsilon.
\end{equation}
\end{enumerate}
\end{theorem}
\begin{proof}
Throughout the proof, for~$\epsilon,\delta\in(0,1)$ as in the statement, let
\begin{equation}\label{epde}
\hat{\epsilon} 
=\epsilon/3,\qquad 
\hat{\delta}
=\delta/5.
\end{equation} 
By Lemma~\ref{uuc}, 
if~\eqref{nmw} holds with~$\delta$ replaced by~$\hat{\delta}$ (which is satisfied under our condition on~$n$ in part~\ref{pot0} for large enough~$c_4$), 
then there exists an event~$\bar{B}^{\dagger}\subset B_0$ with
\begin{equation}\label{pbp1}
\mathbb{P}(\bar{B}^{\dagger})\geq 1-\hat{\delta} - \rho_{\pi} -n\exp(-\hat{r}^2/(2\abs{\theta^*}^2\lambda_{\max}(\Sigma)))
\end{equation}
such that on~$\bar{B}^{\dagger}$ there exists a unique~$
\theta_{\textrm{map}}\in B_{(\lambda_{\max}(\Sigma))^{-1/2}}(\theta^*)$ with~$
\nabla \ln \bar{\pi}(\cdot|Z^{(n)})|_{\theta_{\textrm{map}}}=0$. 
Setting~$B=\bar{B}^{\dagger}$ (for part~\ref{pot0} of the assertion), by our assumption~$\hat{r}\geq r\geq r^*(\epsilon,\delta)$, we have~$\mathbb{P}(B)\geq 1-\delta-\rho_{\pi}$, which shows part~\ref{pot0}.

We proceed with part~\ref{pot1}. 
By Corollary~\ref{nvc}, if~\eqref{nrc} holds with~$\delta,\epsilon$ replaced by~$\hat{\delta},\hat{\epsilon}$, then it holds on~$\bar{B}^{\dagger}$ (which is the same event as that in Corollary~\ref{nvc} by inspection of its proof) that
\begin{equation}\label{pcol2}
\int_{\mathbb{R}^d\setminus B_R} \bar{\pi}(\bar{\theta}|Z^{(n)}) d\bar{\theta} \leq \int_{\mathbb{R}^d\setminus B_{(\lambda_{\max}(\Sigma))^{-1/2}}(\theta_{\textrm{map}})} \bar{\pi}(\bar{\theta}|Z^{(n)}) d\bar{\theta} \leq \hat{\epsilon}.
\end{equation}
Let~$y\geq 1$ be such that~$r\geq 4\ln (2y)$. 
For any~$i\in[1,n]\cap\mathbb{N}$, let
\begin{equation}\label{iig}
I_{i,r,y} = \int_{B_R}\pi(\bar{\theta})\prod_{j=1}^n \exp(Y_j\bar{\theta}^{\top}X_j-\bar{A}(\bar{\theta}^{\top}X_j))\mathds{1}_{\{\theta\in\mathbb{R}^d:\theta^{\top}X_i\geq r\}}(\bar{\theta})
d\bar{\theta}
\cdot\mathds{1}_{\{Y_i<y\}}.
\end{equation}
For any~$i\in\mathbb{N}\cap[1,n]$, 
let~$\mathcal{F}_{-i}$ be the~$\sigma$-algebra generated by~$(Y_j,X_j)_{j\in\mathbb{N}\cap[1,n]\setminus\{i\}}$. 
It holds for any~$i\in\mathbb{N}\cap[1,n]$ a.s.\ that
\begin{align}
\mathbb{E}[I_{i,r,y}|\mathcal{F}_{-i}] 
= \int_{B_R} 
\pi(\bar{\theta})
J_{r,y}(\bar{\theta}) K_{-i}(\bar{\theta})
d\bar{\theta},
\label{tyu}
\end{align}
where~$
J_{r,y},K_{-i}$ are given by 
\begin{align}
J_{r,y}(\bar{\theta}) &= 
\mathbb{E}\big[
\exp(Y_i\bar{\theta}^{\top}X_i-\bar{A}(\bar{\theta}^{\top}X_i))\mathds{1}_{\{\theta\in\mathbb{R}^d:\theta^{\top}X_i\geq r\}}(\bar{\theta})
\mathds{1}_{\{Y_i<y\}}
\big]
\nonumber\\
K_{-i}(\bar{\theta}) &= \textstyle\prod_{j\in\mathbb{N}\cap[1,n]\setminus\{i\}}\exp(Y_j\bar{\theta}^{\top}X_j - 
\bar{A}(\bar{\theta}^{\top}X_j)).\label{kmi}
\end{align}
We analyze~$J_{r,y}$. Note for any~$\bar{y}> 0$, the function~$[\ln \bar{y},\infty)\ni z\mapsto \exp(\bar{y}z-e^z)$ is non-increasing. 
Moreover, since we assumed~$r\geq 4\ln(2y)\geq \ln y$ as well as~$\hat{r}\geq r\geq \ln y$, 
the function~$[\ln y,\infty)\ni z\mapsto 
\exp(yz-\bar{A}(z))$ (recall the definition~\eqref{aba2} of~$\bar{A}$) is also non-increasing. 
Using again~$r\geq \ln y$, 
by~\eqref{psp} 
it holds for any~$\bar{\theta}\in B_R$ that
\begin{equation}\label{nea}
J_{r,y}(\bar{\theta}) \leq  
\exp(yr-e^{r})
\mathbb{E}[
\mathds{1}_{\{\theta\in\mathbb{R}^d:\theta^{\top}X_i\geq r\}}(\bar{\theta})
] \leq \delta_{r,y,R},
\end{equation} 
where
\begin{equation*}
\delta_{r,y,R} := \exp\big(yr-e^{r}-r^2/(2R^2\lambda_{\max}(\Sigma))\big).
\end{equation*}
Substituting this into~\eqref{tyu} yields for 
any~$i\in\mathbb{N}\cap[1,n]$ a.s.\ that
\begin{equation}\label{eie}
\mathbb{E}[I_{i,r,y}|\mathcal{F}_{-i}]
\leq 
\delta_{r,y,R} 
\int_{B_R} 
\pi(\bar{\theta})
K_{-i}(\bar{\theta})
d\bar{\theta}.
\end{equation}
For any such~$i$, 
Markov's inequality and~\eqref{eie} imply
\begin{align*}
\mathbb{P}\bigg(\frac{I_{i,r,y}}{\int_{B_R} 
\pi(\bar{\theta}) 
K_{-i}(\bar{\theta})
d\bar{\theta}
}>\frac{n\delta_{r,y,R}}{\hat{\delta}}\bigg) &\leq \frac{\hat{\delta}}{n\delta_{r,y,R}}\cdot \mathbb{E}\bigg[ \frac{I_{i,r,y}}{\int_{B_R} 
\pi(\bar{\theta})
K_{-i}(\bar{\theta})
d\bar{\theta}
} \bigg]\\
&= \frac{\hat{\delta}}{n\delta_{r,y,R}}\cdot \mathbb{E}\bigg[  \frac{\mathbb{E}[I_{i,r,y}|\mathcal{F}_{-i}]}{\int_{B_R} 
\pi(\bar{\theta})
K_{-i}(\bar{\theta})
d\bar{\theta}
} \bigg]\\
&\leq \hat{\delta}/n. 
\end{align*}
Therefore, for any~$i\in\mathbb{N}\cap[1,n]$, 
if 
\begin{equation}\label{lbf}
r \geq 
(\abs{\theta^*}(\lambda_{\max}(\Sigma))^{1/2} + 2) (2\ln(n^{2}/(\hat{\epsilon}\hat{\delta})))^{1/2}
=
(2R^2\lambda_{\max}(\Sigma)\ln(n^{2}/(\hat{\epsilon}\hat{\delta})))^{1/2},
\end{equation}
then it holds that
\begin{equation}\label{ifb}
\mathbb{P}\bigg(\frac{I_{i,r,y}}{\int_{B_R} 
\pi(\bar{\theta})
K_{-i}(\bar{\theta}) 
d\bar{\theta}
}\leq \frac{\hat{\epsilon}}{n}\cdot \exp(yr-e^r)\bigg) > 1-\hat{\delta}/n.
\end{equation}

We turn to bounding from below the integral~\eqref{iig} without the indicator function appearing in the integrand 
and on the full domain~$\mathbb{R}^d$, namely
\begin{equation}\label{jmi}
I:=\int_{\mathbb{R}^d} \pi(\bar{\theta})\prod_{j=1}^n \exp(Y_j\bar{\theta}^{\top}X_j - \bar{A}(\bar{\theta}^{\top}X_j))
d\bar{\theta}
\cdot\mathds{1}_{\{Y_i<y\}},
\end{equation}
which satisfies for any~$i\in\mathbb{N}\cap[1,n]$ 
a.s. that
\begin{align*}
I &\geq 
\int_{\mathbb{R}^d} \pi(\bar{\theta})\prod_{j=1}^n\exp(Y_j\bar{\theta}^{\top}X_j - \bar{A}(\bar{\theta}^{\top}X_j)) \mathds{1}_{\{\theta\in\mathbb{R}^d:\abs{\theta^{\top}X_i}\leq r/2\}}(\bar{\theta}) 
d\bar{\theta}
\cdot\mathds{1}_{\{Y_i<y\}}\nonumber\\
&\geq
\int_{\mathbb{R}^d} \pi(\bar{\theta})
K_{-i}(\bar{\theta})\exp(Y_{i}\bar{\theta}^{\top}X_{i} - e^{\bar{\theta}^{\top}X_{i}}) \mathds{1}_{\{\theta\in\mathbb{R}^d:\abs{\theta^{\top}X_i}\leq r/2\}}(\bar{\theta}) 
d\bar{\theta}
\cdot\mathds{1}_{\{Y_i<y\}},
\end{align*}
where~$
K_{-i}$ is given 
by~\eqref{kmi}. 
For any~$\bar{y}\in\mathbb{N}\setminus\{0\}$, the function~$[\ln(2\bar{y}),\infty)\ni\bar{r}\mapsto -\bar{y}\bar{r} - e^{-\bar{r}} - (\bar{y}\bar{r}-e^{\bar{r}})$ is non-decreasing. 
Thus for any~$\bar{y}\in\mathbb{N}\setminus\{0\}$ and~$\bar{r}\geq 2\ln (2\bar{y})$, then
\begin{equation*}
-\bar{y}\bar{r}-e^{-\bar{r}} - (\bar{y}\bar{r} - e^{\bar{r}}) \geq -2\bar{y}\ln (2\bar{y})-(2\bar{y})^{-2} - (2\bar{y}\ln(2\bar{y})- (2\bar{y})^2)\geq 0.
\end{equation*}
For~$\bar{y}=0$, we also have~$-e^{-\bar{r}}+e^{\bar{r}}\geq 0$ for~$\bar{r}\geq 0$. Moreover, for any~$\bar{y}\in\mathbb{N}\cup\{0\}$, the function~$\mathbb{R}\ni z\mapsto \bar{y}z - e^z$ is concave, so that, together with the last two sentences, for any~$z\in[-r/2,r/2]$ and any~$\bar{y}\in\mathbb{N}\cup\{0\}$ with~$r\geq 0\vee 4\ln(2\bar{y})$, we have
\begin{equation*}
\bar{y}z-e^z\geq \min(-\bar{y}r/2-e^{-r/2},\bar{y}r/2-e^{r/2})\geq \bar{y}r/2-e^{r/2}\geq -e^{r/2}.
\end{equation*}
Therefore since we assumed~$r/2\geq 0\vee 2\ln(2y)$, it holds for any~$i\in\mathbb{N}\cap[1,n]$ a.s. that
\begin{equation*}
I
\geq 
\exp(-e^{r/2})
\int_{\mathbb{R}^d} 
\pi(\bar{\theta}) K_{-i}(\bar{\theta}) \mathds{1}_{\{\theta\in\mathbb{R}^d:\abs{\theta^{\top}X_i}\leq r/2\}}(\bar{\theta}) 
d\bar{\theta}
\cdot\mathds{1}_{\{Y_i<y\}}.
\end{equation*}
Moreover, by~$r\geq 4\ln(2y)\geq 2\ln(2y)$, we have~$(1/2)e^{r/2}\geq y$, thus~$-e^{r/2}\geq (r/2-e^{r/2})e^{r/2}\geq yr-e^r$, which implies a.s. that
\begin{equation*}
I\geq \exp(yr-e^r)\int_{\mathbb{R}^d} 
\pi(\bar{\theta})
K_{-i}(\bar{\theta}) \mathds{1}_{\{\theta\in\mathbb{R}^d:\abs{\theta^{\top}X_i}\leq r/2\}}(\bar{\theta}) 
d\bar{\theta}
\cdot\mathds{1}_{\{Y_i<y\}}.
\end{equation*}

Together with~\eqref{ifb}, 
for any~$i\in\mathbb{N}\cap[1,n]$, there exists an event~$\bar{B}_i$ with~$\mathbb{P}(\bar{B}_i)\geq 1-\hat{\delta}/n$ such that if~\eqref{lbf} holds, then it holds on~$\bar{B}_i\cap\{Y_i<y\}$
that
\begin{align}
\frac{I_{i,r,y}}{I} &\leq \frac{\hat{\epsilon}
}{n}\cdot \frac{\int_{B_R} 
\pi(\bar{\theta})
K_{-i}(\bar{\theta}) 
d\bar{\theta}
}{\int_{\mathbb{R}^d} 
\pi(\bar{\theta})
K_{-i}(\bar{\theta}) \mathds{1}_{\{\theta\in\mathbb{R}^d:\abs{\theta^{\top}X_i}\leq r/2\}}(\bar{\theta}) 
d\bar{\theta}}.\label{ifi}
\end{align}
The numerator integral on the right-hand side of~\eqref{ifi} satisfies
\begin{align}
&\int_{B_R}
\pi(\bar{\theta})
K_{-i}(\bar{\theta}) 
d\bar{\theta}
\nonumber \\
&\quad= \int_{B_R} 
\pi(\bar{\theta})
K_{-i}(\bar{\theta})\big(\mathds{1}_{\{\theta\in\mathbb{R}^d:\abs{\theta^{\top}X_i}\leq r/2\}}(\bar{\theta}) + \mathds{1}_{\{\theta\in\mathbb{R}^d:\abs{\theta^{\top}X_i}> r/2\}}(\bar{\theta})\big) 
d\bar{\theta}.
\label{ini}
\end{align}
By~\eqref{psp}, the integral associated to the second indicator function on the right-hand side of~\eqref{ini} satisfies
\begin{align*}
&\mathbb{E}\bigg[\int_{B_R} 
\pi(\bar{\theta})
K_{-i}(\bar{\theta})\mathds{1}_{\{\theta\in\mathbb{R}^d:\abs{\theta^{\top}X_i}> r/2\}}(\bar{\theta}) 
d\bar{\theta}
\bigg|\mathcal{F}_{-i}\bigg]\\
&\quad= \int_{B_R} 
\pi(\bar{\theta})
K_{-i}(\bar{\theta})\mathbb{P}(\abs{\bar{\theta}^{\top}X_i}> r/2) 
d\bar{\theta}\\
&\quad\leq 2e^{-r^2/(8R^2\lambda_{\max}(\Sigma))}\int_{B_R} 
\pi(\bar{\theta})
K_{-i}(\bar{\theta}) 
d\bar{\theta}.
\end{align*}
In particular, 
for any~$i\in\mathbb{N}\cap[1,n]$, if
\begin{equation}\label{Rdf}
r\geq (8R^2\lambda_{\max}(\Sigma)\ln(4n/\hat{\delta}))^{1/2},
\end{equation}
then by Markov's inequality, it holds that
\begin{align}
&\mathbb{P}\bigg(\frac{\int_{B_R} 
\pi(\bar{\theta})
K_{-i}(\bar{\theta})\mathds{1}_{\{\theta\in\mathbb{R}^d:\abs{\theta^{\top}X_i}> r/2\}}(\bar{\theta}) 
d\bar{\theta}
}{ \int_{B_R} 
\pi(\bar{\theta})
K_{-i}(\bar{\theta}) 
d\bar{\theta}
} \geq \frac{1}{2}\bigg) \nonumber \\
&\quad\leq 2\mathbb{E}\bigg[\frac{\int_{B_R} 
\pi(\bar{\theta})
K_{-i}(\bar{\theta})\mathds{1}_{\{\theta\in\mathbb{R}^d:\abs{\theta^{\top}X_i}> r/2\}}(\bar{\theta}) 
d\bar{\theta}
}{ \int_{B_R} 
\pi(\bar{\theta})
K_{-i}(\bar{\theta}) 
d\bar{\theta}
}\bigg]\nonumber\\
&\quad\leq 4e^{-r^2/(8R^2\lambda_{\max}(\Sigma))}\nonumber\\
&\quad\leq \hat{\delta}/n.
\end{align}
Therefore by substituting into~\eqref{ini}, for any such~$i$ and if~\eqref{Rdf} holds, then it holds with probability at least~$1-\hat{\delta}/n$ that
\begin{equation*}
\int_{B_R}
\pi(\bar{\theta})
K_{-i}(\bar{\theta}) 
d\bar{\theta}
\leq 2 \int_{B_R} 
\pi(\bar{\theta})
K_{-i}(\bar{\theta}) \mathds{1}_{\{\theta\in\mathbb{R}^d:\abs{\theta^{\top}X_i}\leq r/2\}}(\bar{\theta}) 
d\bar{\theta}.
\end{equation*}
Substituting this into~\eqref{ifi} yields that for 
any~$i\in\mathbb{N}\cap[1,n]$, if~\eqref{lbf} and~\eqref{Rdf} hold, then there exists an event~$\bar{B}_i'$ with
\begin{equation}\label{pbp2}
\mathbb{P}(\bar{B}_i')\geq 1-2\hat{\delta}/n
\end{equation}
such that it holds on~$\bar{B}_i'\cap\{Y_i<y\}$ that
\begin{equation}\label{ifi2}
I_{i,r,y}/I 
\leq 2\hat{\epsilon}/n.
\end{equation}

Recall~$\bar{\pi}(\cdot|Z^{(n)}):\mathbb{R}^d\rightarrow(0,\infty)$ given by~\eqref{piba}. 
By the respective definitions~\eqref{iig},~\eqref{jmi} of~$I_{i,r,y}$ and~$I$, 
it holds on the event~$\cap_{i=1}^n\{Y_i<y\}$ that
\begin{align}
&\int_{\mathbb{R}^d} \bar{\pi}(\bar{\theta}|Z^{(n)}) \mathds{1}_{\cup_{i=1}^n\{\theta\in\mathbb{R}^d:\theta^{\top}X_i\geq r\}}(\bar{\theta}) d\bar{\theta}
\nonumber\\
&\quad\leq 
\int_{\mathbb{R}^d\setminus B_R} \bar{\pi}(\bar{\theta}|Z^{(n)}) d\bar{\theta}
+\int_{B_R} \bar{\pi}(\bar{\theta}|Z^{(n)}) \mathds{1}_{\cup_{i=1}^n\{\theta\in\mathbb{R}^d:\theta^{\top}X_i\geq r\}}(\bar{\theta}) d\bar{\theta}
\nonumber\\
&\quad\leq \int_{\mathbb{R}^d\setminus B_R} \bar{\pi}(\bar{\theta}|Z^{(n)})  d\bar{\theta}
+ \sum_{i=1}^n 
\frac{I_{i,r,y}}{I}.\label{ima}
\end{align}
Substituting~\eqref{pcol2} and~\eqref{ifi2} into~\eqref{ima} yields 
that if~\eqref{lbf},~\eqref{Rdf} and~\eqref{nrc} (with~$\delta,\epsilon$ replaced by~$\hat{\delta},\hat{\epsilon}$) hold, then it holds on~$\bar{B}^{\dagger}\cap \cap_{i=1}^n (\bar{B}_i'\cap \{Y_i<y\})$ that
\begin{equation*}
\int_{\mathbb{R}^d} \bar{\pi}(\bar{\theta}|Z^{(n)}) \mathds{1}_{\cup_{i=1}^n\{\theta\in\mathbb{R}^d:\theta^{\top}X_i\geq r\}}(\bar{\theta}) d\bar{\theta} \leq 
3\hat{\epsilon}=\epsilon,
\end{equation*}
where we recalled the definition~\eqref{epde} of~$\hat{\epsilon}$. 
We set~$y=y_{\delta/(5n)}^*$, and
\begin{equation}\label{bhd}
B'=\bar{B}^{\dagger}\cap \cap_{i=1}^n (\bar{B}_i'\cap \{Y_i<y\}).
\end{equation}
We also set, as per our assumptions,~$\hat{r}\geq r\geq r^*(\epsilon,\delta)$ (defined in~\eqref{rsd}). Note that such~$y,r$ satisfies our above conditions~$r\geq 4\ln(2y) = 4\ln(2y_{\delta/(5n)}^*)$,~\eqref{lbf},~\eqref{Rdf}. Note also that our assumption on~$n$ implies that~\eqref{nrc} with~$\delta,\epsilon$ therein replaced by~$\hat{\epsilon},\hat{\delta}$ is satisfied. Therefore to conclude the proof, 
it remains to estimate~$\mathbb{P}(B')$.
By Lemma~\ref{poiy}, 
it holds for~$y= y_{\hat{\delta}/n}^*$ that
\begin{equation*}
\mathbb{P}(\cap_{i=1}^n\{Y_i<y\})\geq 1-\hat{\delta}=1-\delta/5.
\end{equation*}
Moreover, by the lower bounds~\eqref{pbp1} and~\eqref{pbp2} on~$\mathbb{P}(\bar{B}^{\dagger}),\mathbb{P}(\bar{B}_i')$, and~$\hat{r}\geq r^*(\epsilon,\delta)$, we have
\begin{equation*}
\mathbb{P}(\bar{B}^{\dagger})\geq 1-2\delta/5-\rho_{\pi},\qquad
\mathbb{P}(\cap_{i=1}^n\bar{B}_i')\geq 1-2\delta/5.
\end{equation*}
Since~$\pi$ is deterministic, we have~$\rho_{\pi}=0$. Thus we have~$\mathbb{P}(B')\geq 1-\delta$ as required.
\end{proof}

\subsection{Negligibility of differences in Gibbs dynamics}
In addition to the negligibility of the regions in the target~$\bar{\pi}(\cdot|Z^{(n)})$ where~$A$ and~$\bar{A}$ differ (as obtained in~\eqref{meco}), we also need the negligibility of these regions throughout the Gibbs dynamics (namely that the iterates do not visit them with much probability). This will allow us to (effectively) maximally couple various Gibbs trajectories (similar to the proof of Theorem~\ref{gm}) to give a complexity result for the original Gibbs trajectory targeting the true posterior~$\pi(\cdot|Z^{(n)})$, rather than an altered~$\bar{\pi}(\cdot|Z^{(n)})$.\\ 
\indent We will obtain a bound analogous to the second inequality in~\eqref{meco}, but for the densities~$\nu_k$ of the Gibbs iterates. It is not sufficient to simply apply Lemma~\ref{nfo} (this would increase the sampling complexities). Instead, we recycle the ideas in the proof of Theorem~\ref{poiwei}, which require exchanging an expectation over~$X_i$ (for any particular~$i$) and an integral over~$\mathbb{R}^d$, in order to resolve the indicator function as in~\eqref{nea}. The main problem here is that, despite the form~\eqref{vld2} of the Gibbs densities, we are interested in an integral over~$\mathbb{R}^d$ w.r.t. the previous iterate density~$\nu_k$, which itself depends on~$X_i$, and so this exchange is not valid. To remedy precisely this, we study Gibbs dynamics for leave-one-out targets. These dynamics will be independent of the left-out~$X_i$, thus permitting the aforementioned exchange. The next Lemma~\ref{lop} establishes some basic properties for the leave-one-out Gibbs target we will use, and the subsequent Lemma~\ref{loo} shows that the leave-one-out dynamics are not too far from the full-data dynamics.

We introduce notation for the altered posterior. Obviously, we assume the Poisson case (that is~\eqref{pois} for all~$z\in\mathbb{R},y\in S$ and the uniform measure for~$\eta$) in this section. In addition to leaving out a datapoint and using~$\bar{A}$ as in~\eqref{aba2} in place of~$A$, we strongly convexify the negative log-likelihood everywhere. 
Let~$R$ be given by~\eqref{rhdef2} and~$\Sigma$ be as in Assumption~\ref{A1}. Let
\begin{equation}\label{dbri}
\tilde{c}(2R):= (1/7)\lambda_{\min}(\Sigma)e^{-4R(\lambda_{\max}(\Sigma))^{1/2}}
\end{equation}
and let~$Q_0,Q:\mathbb{R}^d\rightarrow[0,\infty)$ be given for any~$x\in\mathbb{R}^d$ by
\begin{subequations}\label{qua}
\begin{align}
Q_0(x) &= \begin{cases}
0 &\textrm{if }\abs{x}<5R/4,\\
2(C_{\pi}d+\tilde{c}(2R)n)(\abs{x}-5R/4)^2 &\textrm{if }\abs{x}\geq 5R/4,
\end{cases}\\
Q &= \varphi_{R/4}\ast Q_0.
\end{align}
\end{subequations}
Let~$\bar{A}$ be given by~\eqref{aba2} for some~$\hat{r}\in[4R(\lambda_{\max}(\Sigma))^{1/2},\infty)$. 
For any~$i\in[0,n]\cap\mathbb{N}$, let~$\widetilde{\pi}^{(-i)}(\cdot|Z^{(n)}):\mathbb{R}^d\rightarrow\mathbb{R}$ be given for any~$\theta\in\mathbb{R}^d$ by
\begin{subequations}\label{tpid}
\begin{align}
\widetilde{\pi}^{(-i)}(\cdot|Z^{(n)}) &= e^{-\widetilde{U}^{(-i)}}/\textstyle\int_{\mathbb{R}^d}e^{-\widetilde{U}^{(-i)}}\\
\textrm{with}\qquad\widetilde{U}^{(-i)}(\theta) &= Q(\theta)
-\ln\pi(\theta)
-\textstyle\sum_{j\in[1,n]\cap\mathbb{N}\setminus\{i\}}  (Y_j\theta^{\top}X_j - \bar{A}(\theta^{\top}X_j)).\label{utdef}
\end{align}
\end{subequations}
Let~$\widetilde{\pi}^{(-i)}(\cdot|Z^{(n)})$ also denote the probability measure with this density (whenever it exists). Note that~$\widetilde{\pi}^{(-0)}(\cdot|Z^{(n)})$ is the full-data altered (via~$Q$ and~$\bar{A}$) posterior.

The choice of parameters in~\eqref{dbri} and~\eqref{qua} is justified by the following lemma.
\begin{lemma}\label{lop}
Let Assumptions~\ref{A1},~\ref{asp} hold. 
Let~$\pi$ be deterministic (so that~$\rho_{\pi}=0$). 
Let~$\delta\in(0,1)$. 
There exists an absolute constant~$\bar{c}_4\geq 1$ and an event~$B_*$ with~$\mathbb{P}(B_*)\geq 1-\delta$ such that if
\begin{equation}\label{nln}
n\geq \bar{c}_4\max\big(d,C_{\pi}de^{4R(\lambda_{\max}(\Sigma))^{1/2}}/\lambda_{\min}(\Sigma), \ln(\delta^{-1})\big),
\end{equation}
then it holds on~$B_*$ that, for any~$i\in[0,n]\cap\mathbb{N}$,~$\inf_{\mathbb{R}^d}\lambda_{\min}(-D^2\ln\widetilde{\pi}^{(-i)}(\cdot|Z^{(n)}))\geq \tilde{c}(2R)n$ and~$\sup_{\mathbb{R}^d}|D^2\ln\widetilde{\pi}^{(-i)}(\cdot|Z^{(n)})|\leq \widetilde{L}$, where
\begin{equation}\label{lzd}
\widetilde{L} = 5C_{\pi}d + 4\tilde{c}(2R)n + e^{\hat{r}}(2\sqrt{n\lambda_{\max}(\Sigma)} + \sqrt{\textrm{Tr}(\Sigma)})^2.
\end{equation}
\end{lemma}
\begin{proof}
By Proposition~\ref{p1} (and Remark~\ref{p1r}) with~$\hat{\theta}=0$ and~$\bar{c}_1=\bar{c}_2=19/20$, there exist absolute constants~$c_1',c_2'>0$ 
such that 
if~$n\geq c_1'd$ (which is satisfied given~\eqref{nln} for large enough~$\bar{c}_4$), then for any~$i\in[0,n]\cap\mathbb{N}$ there exists an event~$\bar{B}_{i}$ with
\begin{equation}\label{bbr2}
\mathbb{P}(\bar{B}_{i})\geq 1-2e^{-c_2'n}
\end{equation}
such that
\begin{align*}
\inf_{\theta\in B_{2R}}
\lambda_{\min}\bigg(\sum_{j\in[1,n]\cap\mathbb{N}\setminus\{i\}}
\bar{A}''(\theta^{\top}X_j) 
X_jX_j^{\top}\bigg) 
&\geq \frac{(n-1)\lambda_{\min}(\Sigma)
e^{-4R(\lambda_{\max}(\Sigma))^{1/2}}}{6}\\
&= 7\tilde{c}(2R)(n-1)/6.
\end{align*}
Therefore 
by Assumption~\ref{asp}, if in addition
\begin{equation*}
n\geq CC_{\pi}de^{4R(\lambda_{\max}(\Sigma))^{1/2}}/\lambda_{\min}(\Sigma)
\end{equation*}
for a large enough generic absolute constant~$C$ (which is satisfied given~\eqref{nln}), then any~$i\in[0,n]\cap\mathbb{N}$, on~$\bar{B}_{i}$ we have 
\begin{equation}\label{dbr3}
\inf_{\theta\in B_{2R}} \lambda_{\min}\big(-D^2\ln\widetilde{\pi}^{(-i)}(\cdot|Z^{(n)})\big|_{\theta}\big) \geq \frac{n\lambda_{\min}(\Sigma)
e^{-4R(\lambda_{\max}(\Sigma))^{1/2}}}{7} = \tilde{c}(2R)n.
\end{equation}
Moreover, since
\begin{equation}\label{d2q}
D^2Q_0(x) = 2(C_{\pi}d+\tilde{c}(2R)n)\big((2-5R/(2\abs{x}))I_d + (5R/2)xx^{\top}/\abs{x}^3\big) \qquad\forall x\in \mathbb{R}^d\setminus B_{5R/4},
\end{equation}
which gives~$\inf_{\mathbb{R}^d\setminus B_{7R/4}} \lambda_{\min}(D^2Q_0) \geq C_{\pi}d + \tilde{c}(2R)n$ and so
\begin{equation}
\inf_{\mathbb{R}^d\setminus B_{2R}} \lambda_{\min}(D^2Q)\geq C_{\pi}d + \tilde{c}(2R)n,
\end{equation}
by~$\bar{A}''\geq 0$ (from the definition~\eqref{aba2} of~$\bar{A}$) and again Assumption~\ref{asp}, it holds a.s. for any~$i\in[0,n]\cap\mathbb{N}$ that~$\inf_{\mathbb{R}^d\setminus B_{2R}}\lambda_{\min}(-D^2\ln\widetilde{\pi}^{(-i)}(\cdot|Z^{(n)}))\geq \tilde{c}(2R)n$. Together with~\eqref{dbr3}, this implies for any~$i\in[0,n]\cap\mathbb{N}$ on~$\bar{B}_i$ that~$\inf_{\mathbb{R}^d}\lambda_{\min}(-D^2\ln\widetilde{\pi}^{(-i)}(\cdot|Z^{(n)}))\geq \tilde{c}(2R)n$.

On the other hand, by Proposition~\ref{eas}, by~$\sup_{\mathbb{R}^d}\bar{A}''\leq e^{\hat{r}}$ (which follows by definition~\eqref{aba2} of~$\bar{A}$) and by~\eqref{d2q}, for any~$i\in[0,n]\cap\mathbb{N}$, there exists an event~$B_{i}'$ with~$\mathbb{P}(B_{i}')\geq 1-e^{-(n-1)/2}$ such that it holds on~$B_{i}'$ that~$\sup_{\mathbb{R}^d}|D^2\ln\widetilde{\pi}^{(-i)}(\cdot|Z^{(n)})|\leq \widetilde{L}$.

It remains to estimate the probability of~$B_*=(\cap_{i=0}^n\bar{B}_i)\cap(\cap_{i=0}^nB_i')$. By~\eqref{bbr2}, if~$n\geq 1$, then~$\mathbb{P}(\cap_{i=0}^n\bar{B}_i)\geq 1-2(n+1)e^{-c_2'n}\geq 1-4ne^{-c_2'n}\geq 1-8(c_2')^{-1}e^{-c_2'n/2}$. Similarly, if~$n\geq 2$ then~$\mathbb{P}(\cap_{i=0}^n B_i')\geq 1-(n+1)e^{-(n-1)/2}\geq 1-2ne^{-n/4}\geq 1-16e^{-n/8}$. Therefore if~$n\geq \max(1,2(c_2')^{-1}\ln(16/(\delta c_2')),8\ln(32/\delta))$ (which is satisfied given~\eqref{nln}), then it holds that
\begin{equation*}
\mathbb{P}( \cap_{i=0}^n\bar{B}_i\cap \cap_{i=0}^n B_i') \geq 1-\delta.\qedhere
\end{equation*}
\end{proof}

Next, in Lemma~\ref{loo}, we show that on a `high-probability' event, Gibbs iterates targeting~$\widetilde{\pi}^{(-i)}(\cdot|Z^{(n)})$ are not too far from those targeting the full-data~$\widetilde{\pi}^{(-0)}(\cdot|Z^{(n)})$ in R\'enyi divergence, under sensible initializations. The control in R\'enyi divergence, as opposed to say total variation, will be important to achieve our goals as stated before Lemma~\ref{lop}.\\
\indent We define Gibbs iterates targeting~$\widetilde{\pi}^{(-i)}(\cdot|Z^{(n)})$ (as given in~\eqref{tpid}). 
For any~$x\in\mathbb{R}^d$,~$i\in[0,n]\cap\mathbb{N}$ and~$l\in[1,d]\cap\mathbb{N}$, 
let~$\widetilde{F}_{x,l}^{(-i)}
:\mathbb{R}\rightarrow[0,1]$ be the cumulative distribution function of the conditional distribution of~$y_l$ given~$y_{-l}=x_{-l}$ under~$y\sim \widetilde{\pi}^{(-i)}(\cdot|Z^{(n)})$. 
Let~$(i_k)_{k\in\mathbb{N}}$ and~$(u_k)_{k\in\mathbb{N}}$ be the i.i.d. sequences as in the beginning of Section~\ref{gibsec} independent of~$Z^{(n)}$. For any~$i\in[0,n]\cap\mathbb{N}$, let~$(\tilde{\theta}_k^{(-i)})_{k\in\mathbb{N}}$ 
be a sequence 
of~$\mathbb{R}^d$-valued r.v.'s given inductively by
\begin{equation*}
(\tilde{\theta}_{k+1}^{(-i)})_l = 
\begin{cases}
(\widetilde{F}_{\tilde{\theta}_k^{(-i)},i_k}^{(-i)})^{-1}(u_k) &\textrm{if } l=i_k,\\
(\tilde{\theta}_k^{(-i)})_l &\textrm{if } l\neq i_k,
\end{cases}
\end{equation*}
with~$\tilde{\theta}_0^{(-i)}
$ independent of~$(i_k)_k,(u_k)_k$. 
For any~$k\in\mathbb{N}\setminus\{0\}$ and~$i\in[0,n]\cap\mathbb{N}$, denote~$\delta_{\tilde{\theta}_0^{(-i)}},\delta_{\tilde{\theta}_0^{(-i)}}P_{(-i)}^k$ to be the conditional distributions of~$\tilde{\theta}_0^{(-i)},\tilde{\theta}_k^{(-i)}$ given~$Z^{(n)}$. 

\begin{lemma}\label{loo}
Let Assumptions~\ref{A1},~\ref{asp} hold and let~$\pi$ be deterministic. Assume~$0\in\argmax \pi$. Let~$\delta\in(0,1)$ and~$y^{\dagger},r^{\dagger}\geq 0$. 
There exist an absolute constant~$\bar{c}_4'\geq 1$ and an event~$B^{\dagger}$ with~$\mathbb{P}(B^{\dagger})\geq 1-\delta$ such that if
\begin{subequations}\label{rnc}
\begin{align}
\hat{r}&\geq r^*(1/2,\delta/(n+1))\\
n&\geq c^*(\bar{c}_4',\delta/(n+1))\max((\kappa(\Sigma))^2,C_{\pi}(\lambda_{\min}(\Sigma))^{-1})(d+\ln((n+1)/\delta)),\label{rnc2}
\end{align}
\end{subequations}
where~$r^*,c^*$ are given by~\eqref{rsd} and~\eqref{por}, then the following statements hold. 
\begin{enumerate}[label=(\roman*)]
\item \label{loo1} On~$B^{\dagger}$, for any~$i\in[0,n]\cap\mathbb{N}$, there exists a unique~$\theta_{\textrm{map}}^{(-i)}\in B_{(\lambda_{\max}(\Sigma))^{-1/2}}(\theta^*)$ with
\begin{equation*}
\nabla \ln\widetilde{\pi}^{(-i)}(\cdot|Z^{(n)})|_{\theta_{\textrm{map}}^{(-i)}}=0.
\end{equation*}
\item \label{loo2} Let~$\widetilde{L}$ be given by~\eqref{lzd},~$\bar{L}\geq \widetilde{L}$ and let~$(W^{(-i)})_{i=0}^n$ be a sequence of~$\mathbb{R}^d$-valued r.v.'s with~$W^{(-i)}\sim N(0,\bar{L}^{-1}I_d)$ independent of~$(i_k)_k,(u_k)_k,Z^{(n)}$ for all~$i$. 
On the event
\begin{equation}\label{fot}
B^{\dagger}\cap B_* \cap \cap_{i=1}^n \{ Y_i<y^{\dagger}\}\cap\cap_{i=1}^n\cap_{l=1}^d\{\abs{(X_i)_l}\leq r^{\dagger}\}
\end{equation}
(where~$B_*$ is the event from Lemma~\ref{lop}), if for any~$i\in[0,n]\cap\mathbb{N}$, it holds that~$\tilde{\theta}_0^{(-i)}= \theta_{\textrm{map}}^{(-i)} + W^{(-i)}$, then it holds for any~$i\in[1,n]\cap\mathbb{N},k\in\mathbb{N}$ that 
\begin{equation}\label{nja}
\mathcal{R}_2\Big(\delta_{\tilde{\theta}_0^{(-0)}}P_{(-0)}^k\Big|\Big|\delta_{\tilde{\theta}_0^{(-i)}}P_{(-i)}^k\Big) \leq \frac{\bar{L}d(y^{\dagger}r^{\dagger}+3e^{\hat{r}}r^{\dagger}/2)^2}{(\tilde{c}(2R)n)^2} + \frac{2k(y^{\dagger}r^{\dagger}+3e^{\hat{r}}r^{\dagger}/2)^2}{\tilde{c}(2R)n},
\end{equation}
where~$\tilde{c}(2R)$ is given by~\eqref{dbri}.
\end{enumerate}
\end{lemma}
\begin{proof}
Below, for any~$i\in[1,n]\cap\mathbb{N}$ 
and~$\theta\in\mathbb{R}^d$ we denote
\begin{equation}\label{kid}
K_i(\theta) = Y_i\theta^{\top}X_i-\bar{A}(\theta^{\top}X_i).
\end{equation}
We begin by estimating the R\'enyi divergence between the conditional Gibbs updates for the full-data target and those for the leave-one-out target, more specifically~$\hat{\mathcal{R}}_{i,l,\theta}$ defined below. 
We consider the function~$\mathbb{R}\ni\bar{\theta}_l\mapsto\widetilde{\pi}_l^{(-i)}(\bar{\theta}_l,\theta_{-l}|Z^{(n)})$ as a measure on~$\mathbb{R}^d$ with all its mass on the line~$\{\bar{\theta} \in\mathbb{R}^d:\bar{\theta}_{-l} = \theta_{-l}\}$. For any such~$l,\theta$, let~$\lambda_l$ be a measure on~$\mathbb{R}^d$ such that for any Borel set~$A\subset\mathbb{R}^d$,~$\lambda_l(A)$ is given by 
the Lebesgue measure (in one-dimension) of~$\{y\in \mathbb{R}:\theta^{y,l}\in A\}$, with~$\theta^{y,l}=(\theta_1^{y,l},\dots,\theta_d^{y,l})\in\mathbb{R}^d$ given by
\begin{equation}
\theta_j^{y,l} = \begin{cases}
y &\textrm{if } j=l\\
\theta_j &\textrm{otherwise}
\end{cases}
\qquad\forall j\in[1,d]\cap\mathbb{N}.
\end{equation}
By definition of R\'enyi divergence, it holds a.s. for any~$i\in[1,n]\cap\mathbb{N}$,~$l\in[1,d]\cap\mathbb{N}$ and~$\theta\in\mathbb{R}^d$ that
\begin{align*}
\hat{\mathcal{R}}_{i,l,\theta} &:= \mathcal{R}_2\bigg(\frac{
\exp(-\widetilde{U}^{(-0)})\lambda_l}{\widetilde{Z}_l}\bigg|\bigg|\frac{
\exp(-\widetilde{U}^{(-i)})\lambda_l}{\widetilde{Z}_l^{(-i)}}\bigg) \\
&= \ln\int_{\mathbb{R}^d}\exp\big(2K_i(\bar{\theta})\big)\cdot\bigg(\frac{\widetilde{Z}_l^{(-i)}}{\widetilde{Z}_l}\bigg)^2
\cdot\frac{
\exp(-\widetilde{U}^{(-i)})
}{\widetilde{Z}_l^{(-i)}} \lambda_l(d\bar{\theta}),
\end{align*}
where~$\widetilde{U}^{(-i)}$ given by~\eqref{utdef} and
\begin{equation*}
\widetilde{Z}_l = \int_{\mathbb{R}^d}
\exp(-\widetilde{U}^{(-0)})
\lambda_l,
\qquad
\widetilde{Z}_l^{(-i)} = \int_{\mathbb{R}^d}
\exp(-\widetilde{U}^{(-i)})
\lambda_l.
\end{equation*}
Equivalently, defining~$\psi:\{1,2\}\rightarrow\mathbb{R}$ by
\begin{equation*}
\psi(p) = \int_{\mathbb{R}^d}\exp\big(pK_i
\big)\frac{
\exp(-\widetilde{U}^{(-i)})
}{\widetilde{Z}_l^{(-i)}}\lambda_l
\qquad\forall p\in\{1,2\},
\end{equation*}
we have
\begin{equation}\label{r2i}
\hat{\mathcal{R}}_{i,l,\theta}=\ln\psi(2)-2\ln\psi(1). 
\end{equation}
By Lemma~\ref{lop}, if~\eqref{nln} holds (which is satisfied given~\eqref{rnc2} for large enough~$\bar{c}_4'$ by definitions~\eqref{por},~\eqref{rhdef2} of~$c^*,R$), 
then it holds on~$B_*$ (from Lemma~\ref{lop}) that~$\widetilde{U}^{(-i)}$ is~$\tilde{c}(2R)n$-strongly convex everywhere, in particular on the line on which~$\lambda_l$ places its mass. Note that strongly log-concave distributions satisfy log-Sobolev inequalities~\cite[Theorem~5.2]{MR1849347}. Therefore by standard estimates for exponentials of mean-zero Lipschitz functions of r.v.'s with distributions satisfying a log-Sobolev inequality~\cite[(5.8)]{MR1849347}, if~\eqref{rnc2} holds, 
then it holds on~$B_*$ for any~$i\in[1,n]\cap\mathbb{N}$,~$l\in[1,d]\cap\mathbb{N}$ and~$\theta\in\mathbb{R}^d$ 
that
\begin{equation*}
\psi(2)\exp\bigg(-\int_{\mathbb{R}^d}2K_i
\frac{
\exp(-\widetilde{U}^{(-i)})
}{\widetilde{Z}_l^{(-i)}}\lambda_l
\bigg)
\leq \exp\bigg(\frac{2(Y_i\abs{(X_i)_l}+3e^{\hat{r}}\abs{(X_i)_l}/2)^2}{\tilde{c}(2R)n}\bigg),
\end{equation*}
where we have used~\eqref{kid} and~$\sup_{\mathbb{R}}\bar{A}'\leq 3e^{\hat{r}}/2$ by definition~\eqref{aba2} of~$\bar{A}$. Moreover, by Jensen's inequality, we have for any such~$i,l,\theta$ that
\begin{equation*}
\psi(1)\exp\bigg(-\int_{\mathbb{R}^d}K_i
\frac{
\exp(-\widetilde{U}^{(-i)})}{\widetilde{Z}_l^{(-i)}}\lambda_l
\bigg) \geq 1.
\end{equation*}
Substituting the last two inequalities into~\eqref{r2i} yields, under~\eqref{rnc2} 
and on~$B_*$, that
\begin{equation}\label{r2e}
\hat{\mathcal{R}}_{i,l,\theta}\leq \frac{2(Y_i\abs{(X_i)_l}+3e^{\hat{r}}\abs{(X_i)_l}/2)^2}{\tilde{c}(2R)n}.
\end{equation}
\indent Next, we use the estimate~\eqref{r2e} to control the R\'enyi divergence between Gibbs iterates~$\tilde{\theta}_k^{(-0)}$,~$\tilde{\theta}_k^{(-i)}$. For any~$i\in[1,n]\cap\mathbb{N}$ and~$\theta\in\mathbb{R}^d$, let~$P(\cdot,\theta),P^{(-i)}(\cdot,\theta)$ be the measures on~$\mathbb{R}^d$ given respectively by the mixtures of densities
\begin{equation*}
\textstyle d^{-1}\sum_{l=1}^d\exp(-\widetilde{U}^{(-0)})\lambda_l/\widetilde{Z}_l^{(-0)},\qquad d^{-1}\sum_{l=1}^d\exp(-\widetilde{U}^{(-i)})\lambda_l/\widetilde{Z}_l^{(-i)}
\end{equation*}
($P,P^{(-i)}$ are Gibbs update kernels given~$\theta$, satisfying the usual measurability assumptions on a kernel by Lemma~2.1(a) in~\cite{ascolani2024e} and Fubini). For any~$i\in[1,n]\cap\mathbb{N}$ and~$\theta\in\mathbb{R}^d$, 
we also set~$\bar{P}(\cdot,\theta),\bar{P}^{(-i)}(\cdot,\theta)$ to be the measures on~$([1,d]\cap\mathbb{N})\times\mathbb{R}^d$ given for any sets~$(A_l)_{l=1}^d$, with~$A_l\in\mathcal{B}(\mathbb{R}^d)$ for all~$l$, by
\begin{equation*}
\textstyle d^{-1}\sum_{l=1}^d\int_{A_l}\exp(-\widetilde{U}^{(-0)})\lambda_l/\widetilde{Z}_l^{(-0)},\qquad d^{-1}\sum_{l=1}^d\int_{A_l}\exp(-\widetilde{U}^{(-i)})\lambda_l/\widetilde{Z}_l^{(-i)}.
\end{equation*}
The measures~$P(\cdot,\theta),P^{(-i)}(\cdot,\theta)$ are the restrictions of~$\bar{P}(\cdot,\theta),\bar{P}^{(-i)}(\cdot,\theta)$ respectively onto the sub-$\sigma$-algebra~$\{([1,d]\cap\mathbb{N})\times A:A\in\mathcal{B}(\mathbb{R}^d)\}$. Therefore, by data-processing~\cite[Theorem~9]{MR3225930} and the definition of R\'enyi divergence, it holds for any such~$i,\theta$ that
\begin{align}
\exp(\mathcal{R}_2(P(\cdot,\theta)||P^{(-i)}(\cdot,\theta)) )
&\leq \exp(\mathcal{R}_2(\bar{P}(\cdot,\theta)||\bar{P}^{(-i)}(\cdot,\theta)) )\nonumber\\
&= \frac{1}{d}\sum_{l=1}^d \int_{\mathbb{R}^d}\exp(2K_i) \cdot \bigg(\frac{\widetilde{Z}_l^{(-i)}}{\widetilde{Z}_l}\bigg)^2\cdot\frac{\exp(-\widetilde{U}^{(-i)})}{\widetilde{Z}_l^{(-i)}}\lambda_l \nonumber\\
&= \textstyle d^{-1}\sum_{l=1}^d \exp(\hat{\mathcal{R}}_{i,l,\theta}).\label{erb}
\end{align}
Combining this with~\eqref{r2e}, 
under~\eqref{rnc2}, 
it holds on the event~$B_*\cap\mathcal{E}^{\dagger}$, where
\begin{equation}\label{edd}
\mathcal{E}^{\dagger}:=\cap_{i=1}^n\{Y_i<y^{\dagger}\}\cap\cap_{i=1}^n\cap_{l=1}^d\{\abs{(X_i)_l}\leq r^{\dagger}\},
\end{equation}
that
\begin{equation}\label{r3e}
\mathcal{R}_2(P(\cdot,\theta)||P^{(-i)}(\cdot,\theta))  \leq 
\frac{2(y^{\dagger}r^{\dagger} + 3e^{\hat{r}}r^{\dagger}/2)^2}{\tilde{c}(2R)n}.
\end{equation}
\indent To bound the R\'enyi divergence between~$\tilde{\theta}_k^{(-0)}$ and~$\tilde{\theta}_k^{(-i)}$, we also need to bound the R\'enyi divergence between the initial distributions. 
We first establish with Theorem~\ref{poiwei} the existence of the critical points. For any~$i\in[0,n]\cap\mathbb{N}$, we apply Theorem~\ref{poiwei} to the (random) leave-one-out probability distribution~$\bar{\pi}^{(-i)}(\cdot|Z^{(n)})$ given by the unnormalized density~$\mathbb{R}^d\ni\theta\mapsto\pi(\theta)\prod_{j\in[1,n]\cap\mathbb{N}\setminus\{i\}}\exp(K_j(\theta))$ (whenever it exists). Note that since~$\pi$ is deterministic, we have~$\rho_{\pi}=0$. In Theorem~\ref{poiwei}, we set~$\epsilon,\delta,n$ therein to be~$1/2,\delta/(n+1),n-1 \textrm{ (or~$n$)}$ respectively. This application of Theorem~\ref{poiwei} yields that if~\eqref{rnc} holds (with~$n$ replaced by~$n-1$), then there exists an event~$B^{\dagger}$ with~$\mathbb{P}(B^{\dagger})\geq 1-\delta$ on which for any~$i\in[0,n]\cap\mathbb{N}$ there exists a unique~$\theta_{\textrm{map}}^{(-i)}\in B_{(\lambda_{\max}(\Sigma))^{-1/2}}(\theta^*)$ with~$\nabla \ln \bar{\pi}^{(-i)}(\cdot|Z^{(n)})|_{\theta_{\textrm{map}}^{(-i)}}=0$. Since~$Q$ was defined to be zero on a large enough ball, this critical point is also a critical for~$\widetilde{\pi}^{(-i)}(\cdot|Z^{(n)})$; more precisely, by the respective definitions~\eqref{rhdef2},~\eqref{qua},~\eqref{tpid}, of~$R,Q,\widetilde{\pi}^{(-i)}(\cdot|Z^{(n)})$, we also have~$\nabla\ln\widetilde{\pi}^{(-i)}(\cdot|Z^{(n)})|_{\theta_{\textrm{map}}^{(-i)}}=0$ on~$B^{\dagger}$ under~\eqref{rnc}.\\
\indent For any such~$i$, under~\eqref{rnc} 
and on~$B^{\dagger}\cap B_*$, by strong log-concavity of~$\widetilde{\pi}^{(-i)}(\cdot|Z^{(n)})$ (by Lemma~\ref{lop}), we have
\begin{align}
\tilde{c}(2R)n\abs{\theta_{\textrm{map}}^{(-0)}-\theta_{\textrm{map}}^{(-i)}}^2 &\leq \nabla \widetilde{U}^{(-i)}(\theta_{\textrm{map}}^{(-0)})\cdot (\theta_{\textrm{map}}^{(-0)} - \theta_{\textrm{map}}^{(-i)})\nonumber\\
&\leq \abs{\nabla \widetilde{U}^{(-0)}(\theta_{\textrm{map}}^{(-0)}) + \nabla K_i(\theta_{\textrm{map}}^{(-0)}) }\cdot \abs{\theta_{\textrm{map}}^{(-0)} - \theta_{\textrm{map}}^{(-i)}}\nonumber\\
&\leq (Y_i\abs{X_i} + (3e^{\hat{r}}/2)\abs{X_i})\cdot \abs{\theta_{\textrm{map}}^{(-0)} - \theta_{\textrm{map}}^{(-i)}}\label{iad}
\end{align}
where we have used the respective definitions~\eqref{utdef},~\eqref{kid},~\eqref{aba2} of~$\widetilde{U}^{(-i)},K_i,\bar{A}$. 
Inequality~\eqref{iad} implies
\begin{equation*}
\abs{\theta_{\textrm{map}}^{(-0)}-\theta_{\textrm{map}}^{(-i)}}\leq (Y_i\abs{X_i} + (3e^{\hat{r}}/2)\abs{X_i})\cdot (\tilde{c}(2R)n)^{-1}.
\end{equation*}
Consequently, by standard results for R\'enyi divergences~\cite[equation~(10) and Theorem~28]{MR3225930}, it holds for any such~$i$, under~\eqref{rnc} 
and on~$B^{\dagger}\cap B_*$, that
\begin{align*}
\mathcal{R}_2( N(\theta_{\textrm{map}}^{(-0)},\bar{L}^{-1}I_d) || N(\theta_{\textrm{map}}^{(-i)},\bar{L}^{-1}I_d) )
&= \bar{L}\abs{\theta_{\textrm{map}}^{(-i)}-\theta_{\textrm{map}}^{(-0)}}^2 \\
&\leq \bar{L}(Y_i\abs{X_i} + (3e^{\hat{r}}/2)\abs{X_i})^2\cdot (\tilde{c}(2R)n)^{-2}.
\end{align*}
Since on~$\mathcal{E}^{\dagger}$ (defined in~\eqref{edd}) we have~$\abs{X_i}\leq r^{\dagger}\sqrt{d}$, it holds on~$\mathcal{E}^{\dagger}$ that
\begin{equation*}
\mathcal{R}_2( N(\theta_{\textrm{map}}^{(-0)},\bar{L}^{-1}I_d) || N(\theta_{\textrm{map}}^{(-i)},\bar{L}^{-1}I_d) )
\leq \bar{L}(y^{\dagger}r^{\dagger} + (3e^{\hat{r}}/2)r^{\dagger})^2\cdot d(\tilde{c}(2R)n)^{-2}.
\end{equation*}
Combining this with~\eqref{r3e}, by a chain rule for R\'enyi divergences~\cite[Lemma~2.19]{MR4709755} (since~$fp,f'p'$ therein are joint distributions of adjacent iterates, we also apply data-processing~\cite[Theorem~9]{MR3225930}), under~\eqref{rnc} 
and on~$B^{\dagger}\cap B_*\cap \mathcal{E}^{\dagger}$, inequality~\eqref{nja} holds for all~$k\in\mathbb{N}$.
\end{proof}

The main result of this section is as follows.
\begin{prop}\label{lom}
Let Assumptions~\ref{A1},~\ref{asp} hold and let~$\pi$ be deterministic. Assume~$0\in\argmax\pi$. 
Let~$\delta,\epsilon\in(0,1)$, 
let~$K\in\mathbb{N}\setminus\{0\}$, 
let~$\widetilde{L}$ be given by~\eqref{lzd} and let~$\bar{L}\geq \widetilde{L}$. Let~$W$ be an~$\mathbb{R}^d$-valued r.v. independent of~$(i_k)_k,(u_k)_k,Z^{(n)}$ with~$W\sim N(0,\bar{L}^{-1}I_d)$. 
Let~$r>0$ satisfy
\begin{equation}\label{rco}
r\geq R\cdot(2\lambda_{\max}(\Sigma) \ln(14e^2\cdot (K+1)n^3\epsilon^{-2}\delta^{-1}))^{1/2}.
\end{equation}
There exists an absolute constant~$\bar{c}_4''\geq 1$ and an event~$B$ with~$\mathbb{P}(B)\geq 1-\delta$ such that if
\begin{subequations}\label{nlk}
\begin{align}
\hat{r} &\geq r^*(1/2,\epsilon^2\delta/(14e^2\cdot n^4)),\\
n &\geq \max\big( c^*(\bar{c}_4'',\epsilon\delta/n) \max(\kappa(\Sigma)^2,C_{\pi}(\lambda_{\min}(\Sigma))^{-1})\nonumber\\
&\quad\cdot \big[d\ln\big(\bar{L}\cdot 7e^{4R(\lambda_{\max}(\Sigma))^{1/2}}/(\lambda_{\min}(\Sigma)n)\big) +\ln(Kn\epsilon^{-1}\delta^{-1})\big],\nonumber\\
&\quad 7e^{4R(\lambda_{\max}(\Sigma))^{1/2}}(\lambda_{\min}(\Sigma))^{-1}\max(\bar{L}^{1/2}d^{1/2}z_*,2Kz_*^2) \big)
\end{align}
\end{subequations}
with~$z_* := (y_{\delta/(7n)}^*+3e^{\hat{r}}/2)(2\lambda_{\max}(\Sigma)\ln(14nd/\delta))^{1/2}$ and where~$r^*,y_{\delta/(7n)}^*,c^*$ are given by~\eqref{rsd},~\eqref{ypa0},~\eqref{por} respectively, 
then the following statements hold on~$B$. 
\begin{enumerate}[label=(\roman*)]
\item There exists a unique~$\theta_{\textrm{map}}^{(-0)}\in B_{(\lambda_{\max}(\Sigma))^{-1/2}}(\theta^*)$ with~$\nabla \ln\widetilde{\pi}^{(-0)}(\cdot|Z^{(n)})|_{\theta_{\textrm{map}}^{(-0)}}=0$
\item If~$\tilde{\theta}_0^{(-0)}=\theta_{\textrm{map}}^{(-0)} + W$, then it holds for any~$k\in[0,K]\cap\mathbb{N}$ that
\begin{equation*}
\mathbb{P}( \cup_{i=1}^n\{X_i\cdot \tilde{\theta}_k^{(-0)} \geq r \}|Z^{(n)})
\leq \epsilon.
\end{equation*}
\end{enumerate}
\end{prop}
\begin{proof}
For~$\delta,\epsilon$ as in the statement, let~$\hat{\delta},\hat{\epsilon}\in(0,1)$ be given by
\begin{equation*}
\hat{\delta}=\delta/7,\qquad\hat{\epsilon}=\epsilon/e.
\end{equation*}
Below, we denote~$B_*$ and~$B^{\dagger}$ to be the events from Lemmata~\ref{lop} and~\ref{loo} with~$\delta$ therein replaced by~$\hat{\epsilon}^2\hat{\delta}/n^3$. In particular, we have
\begin{equation}\label{bbp}
\mathbb{P}(B_*\cap B^{\dagger})\geq 1-2\hat{\epsilon}^2\hat{\delta}/n^3.
\end{equation}
Let
\begin{equation*}
y^{\dagger} = y_{\hat{\delta}/n}^*,\qquad r^{\dagger} = \big(2\lambda_{\max}(\Sigma)\ln(2nd/\hat{\delta})\big)^{1/2},
\end{equation*}
where~$y_{\hat{\delta}/n}^*$ is given by~\eqref{ypa0}. By Lemma~\ref{poiy}, these values of~$y^{\dagger}$ and~$r^{\dagger}$ imply
\begin{align}
&\mathbb{P}(\cap_{i=1}^n\{Y_i<y^{\dagger}\} \cap\cap_{i=1}^n\cap_{l=1}^d\{\abs{(X_i)_l}\leq r^{\dagger}\}) \nonumber\\
&\quad\geq 1-\textstyle \sum_{i=1}^n\mathbb{P}(Y_i\geq y^{\dagger}) - \sum_{i=1}^n\sum_{l=1}^d\mathbb{P}(\abs{(X_i)_l}> r^{\dagger})\nonumber\\
&\quad\geq 1- \hat{\delta} - nd\cdot 2\exp(-(r^{\dagger})^2/(2\lambda_{\max}(\Sigma)))\nonumber\\
&\quad\geq 1- 2\hat{\delta}.\label{yxp}
\end{align}
\indent For any~$k\in\mathbb{N}$, the (conditional) probability in the assertion satisfies a.s. that
\begin{equation}\label{fkq}
\Big(\delta_{\tilde{\theta}_0^{(-0)}}P_{(-0)}^k\Big) ( \cup_{i=1}^n\{\theta\in\mathbb{R}^d:X_i\cdot \theta \geq r\}) 
\leq 
\sum_{i=1}^n\Big(\delta_{\tilde{\theta}_0^{(-0)}}P_{(-0)}^k\Big) ( \{\theta\in\mathbb{R}^d:X_i\cdot \theta \geq r\}) 
\end{equation}
and for any~$i\in[1,n]\cap\mathbb{N}$, by H\"older's inequality, we have
\begin{align}
\Big(\delta_{\tilde{\theta}_0^{(-0)}}P_{(-0)}^k\Big) ( \{\theta\in\mathbb{R}^d:X_i\cdot \theta \geq r\}) 
&\leq \Big(\Big(\delta_{\tilde{\theta}_0^{(-i)}}P_{(-i)}^k\Big) ( \{\theta\in\mathbb{R}^d:X_i\cdot \theta \geq r\}) \Big)^{1/2}\nonumber\\
&\quad\cdot \exp\Big((1/2)\mathcal{R}_2\Big(\delta_{\tilde{\theta}_0^{(-0)}}P_{(-0)}^k\Big|\Big|\delta_{\tilde{\theta}_0^{(-i)}}P_{(-i)}^k\Big)\Big),\label{aja}
\end{align}
where the initializations~$\tilde{\theta}_0^{(-i)}$ are as in Lemma~\ref{loo}, which are valid on the event~\eqref{fot}. 
Fix~$i\in[1,n]\cap\mathbb{N}$. 
The last factor in~\eqref{aja} has been estimated in~\eqref{nja}. In particular, Lemma~\ref{loo} implies that if~\eqref{rnc} with~$\delta$ replaced by~$\hat{\epsilon}^2\hat{\delta}/n^3$ and
\begin{equation}\label{njg}
n\geq (\tilde{c}(2R))^{-1}\max\big(\bar{L}^{1/2}d^{1/2}(y^{\dagger}r^{\dagger}+3e^{\hat{r}}r^{\dagger}/2), 2K(y^{\dagger}r^{\dagger} + 3e^{\hat{r}}r^{\dagger}/2)^2\big)
\end{equation}
hold (these are satisfied under~\eqref{nlk} by definition~\eqref{dbri} of~$\tilde{c}(2R)$), 
then on the event~\eqref{fot} it holds for any~$k\in[0,K]\cap\mathbb{N}$ that
\begin{equation}\label{fre}
\exp\Big((1/2)\mathcal{R}_2\Big(\delta_{\tilde{\theta}_0^{(-0)}}P_{(-0)}^k\Big|\Big|\delta_{\tilde{\theta}_0^{(-i)}}P_{(-i)}^k\Big)\Big)\leq e.
\end{equation}
For the first factor on the right-hand side of~\eqref{aja}, we can take the expectation w.r.t. the~$\sigma$-algebra~$\mathcal{F}_{-i}$ generated by~$(Y_j,X_j)_{j\in\mathbb{N}\cap[1,n]\setminus\{i\}}$ to obtain a.s. for any~$k\in\mathbb{N}$ that
\begin{equation}\label{myt}
\mathbb{E}\Big[\Big(\delta_{\tilde{\theta}_0^{(-i)}}P_{(-i)}^k\Big) ( \{\theta\in\mathbb{R}^d:X_i\cdot \theta \geq r\}) \Big|\mathcal{F}_{-i}\Big] = \int_{\mathbb{R}^d}\mathbb{P}( X_i\cdot \bar{\theta} \geq r ) \Big(\delta_{\tilde{\theta}_0^{(-i)}}P_{(-i)}^k\Big)(d\bar{\theta}).
\end{equation}
We split the integral on the right-hand side into two parts over~$B_R$ and~$\mathbb{R}^d\setminus B_R$. By Lemma~\ref{poiy} and~\eqref{rco}, the first part satisfies a.s. that
\begin{equation}\label{hoy}
\int_{B_R}\mathbb{P}( X_i\cdot \bar{\theta} \geq r ) \Big(\delta_{\tilde{\theta}_0^{(-i)}}P_{(-i)}^k\Big)(d\bar{\theta}) \leq e^{-r^2/(2R^2\lambda_{\max}(\Sigma))}\leq \frac{\hat{\delta}\hat{\epsilon}^2}{2n^3(K+1)}.
\end{equation}
For the second part, we apply Lemma~\ref{nfo}. For~$U,U_0,r$ therein, we set~$U=\widetilde{U}^{(-i)}(\cdot+\theta_{\textrm{map}}^{(-i)})$ (as in~\eqref{tpid} and Lemma~\ref{loo}\ref{loo1}),~$U_0=0$ and~$r=(\lambda_{\max}(\Sigma))^{-1/2}$. 
By Lemma~\ref{lop} and Lemma~\ref{loo}\ref{loo1}, under~\eqref{rnc} (with~$\delta$ replaced by~$\hat{\epsilon}^2\hat{\delta}/n^3$, which is satisfied under~\eqref{nlk}) and on~$B_*\cap B^{\dagger}$, these choices for~$U,U_0$ satisfy Conditions~\ref{cond:curvature3} and~\ref{cond:smooth} with~$m_0=\tilde{c}(2R)n$ (as in~\eqref{dbri}) and~$L=\widetilde{L}$. 
In Lemma~\ref{nfo}, the condition~$r^2m_0(r)\geq 9d$ therein is~$n\geq 9\lambda_{\max}(\Sigma)d/\tilde{c}(2R)$, and this is satisfied given~\eqref{nlk}. This application of Lemma~\ref{nfo} yields under~\eqref{nlk} and on~$B_*\cap B^{\dagger}$ that it holds for any~$k\in\mathbb{N}$ that
\begin{align}
\int_{\mathbb{R}^d\setminus B_R}\mathbb{P}(X_i\cdot\bar{\theta}\geq r)\Big(\delta_{\tilde{\theta}_0^{(-i)}}P_{(-i)}^k\Big)(d\bar{\theta}) 
&\leq \Big(\delta_{\tilde{\theta}_0^{(-i)}}P_{(-i)}^k\Big)(\mathbb{R}^d\setminus B_R) \nonumber\\
&\leq \Big(\delta_{\tilde{\theta}_0^{(-i)}}P_{(-i)}^k\Big)(\mathbb{R}^d\setminus B_{(\lambda_{\max}(\Sigma))^{-1/2}}(\theta_{\textrm{map}}^{(-i)})) \nonumber\\
&\leq 5(\bar{L}/(\tilde{c}(2R)n))^{d/2}\cdot e^{2d-3\tilde{c}(2R)n/(8\lambda_{\max}(\Sigma))}\nonumber\\
&\leq \hat{\delta}\hat{\epsilon}^2/(2n^3(K+1)),\label{hoy2}
\end{align}
where we have used
\begin{equation*}
n\geq \frac{8\lambda_{\max}(\Sigma)}{3\tilde{c}(2R)}\bigg(2d+\frac{d}{2}\ln\bigg(\frac{\bar{L}}{\tilde{c}(2R)n}\bigg) + \ln\bigg(\frac{10(K+1)n^3}{\hat{\epsilon}^2\hat{\delta}}\bigg)\bigg)
\end{equation*}
(which is satisfied given~\eqref{nlk}). 
Substituting~\eqref{hoy} and~\eqref{hoy2} into~\eqref{myt}, then using Markov's inequality, yields under~\eqref{nlk} that
\begin{align}
&\mathbb{P}\Big(\max_{k\in[0,K]\cap\mathbb{N}}\Big(\delta_{\tilde{\theta}_0^{(-i)}}P_{(-i)}^k\Big) ( \{\theta\in\mathbb{R}^d:X_i\cdot \theta \geq r\}) \geq \hat{\epsilon}^2/n^2 \Big)\nonumber\\
&\quad\leq n^2\hat{\epsilon}^{-2} \mathbb{E}\Big[\mathbb{E}\Big[\max_{k\in[0,K]\cap\mathbb{N}}\Big(\delta_{\tilde{\theta}_0^{(-i)}}P_{(-i)}^k\Big) ( \{\theta\in\mathbb{R}^d:X_i\cdot \theta \geq r\}) 
\Big|\mathcal{F}_{-i}\Big]\nonumber\\
&\qquad\cdot (\mathds{1}_{\Omega\setminus (B_*\cap B^{\dagger})} + \mathds{1}_{B_*\cap B^{\dagger}})\Big]\nonumber\\
&\quad\leq n^2\hat{\epsilon}^{-2}\cdot 3\hat{\epsilon}^2\hat{\delta}/n^3 \nonumber\\
&\quad= 3\hat{\delta}/n,\label{lpa}
\end{align}
where we have used~\eqref{bbp}. 
Therefore, by substituting this and~\eqref{fre} into~\eqref{aja}, under~\eqref{nlk}, on the intersection between the event~\eqref{fot} and the complement of the event on the left-hand side of~\eqref{lpa}, 
it holds for any~$k\in[0,K]\cap\mathbb{N}$ that
\begin{equation*}
\Big(\delta_{\tilde{\theta}_0^{(-0)}}P_{(-0)}^k\Big) ( \{\theta\in\mathbb{R}^d:X_i\cdot \theta \geq r\}) \leq e\hat{\epsilon}/n = \epsilon/n.
\end{equation*}
By taking a union bound over~$i$ of all of these events, using the estimates~\eqref{bbp},~\eqref{yxp} and~\eqref{lpa} on their probabilities, the proof concludes by~\eqref{fkq}.
\end{proof}

\subsection{Main complexity result}
We proceed with our main result about the Gibbs sampler. With a slight abuse of notation, we denote by~$\pi(\cdot|Z^{(n)})$ and~$\bar{\pi}(\cdot|Z^{(n)})$ the probability measures given by the densities~\eqref{picdef} and~\eqref{piba}. 
For any~$x\in\mathbb{R}^d$ and~$i\in[1,d]\cap\mathbb{N}$, 
let~$F_{x,i}
:\mathbb{R}\rightarrow[0,1]$ be the cumulative distribution function of the conditional distribution of~$y_i$ given~$y_{-i}=x_{-i}$ under~$y\sim \pi(\cdot|Z^{(n)})$. 
Let~$(i_k)_{k\in\mathbb{N}}$ and~$(u_k)_{k\in\mathbb{N}}$ be the i.i.d. sequences as in the beginning of Section~\ref{gibsec}. Let~$(\theta_k)_{k\in\mathbb{N}}$ 
be a sequence 
of~$\mathbb{R}^d$-valued r.v.'s given inductively by
\begin{equation*}
(\theta_{k+1})_i = 
\begin{cases}
F_{\theta_k,i_k}^{-1}(u_k) &\textrm{if } i=i_k,\\
(\theta_k)_i &\textrm{if } i\neq i_k,
\end{cases}
\end{equation*}
with~$\theta_0
$ independent of~$(i_k)_k,(u_k)_k$. 
For any~$k\in\mathbb{N}\setminus\{0\}$, denote~$\delta_{\theta_0},\delta_{\theta_0}P^k$ to be the conditional distributions of~$\theta_0,\theta_k$ given~$Z^{(n)}$. 
\begin{theorem}\label{poigib}
Assume the setting and notations of Theorem~\ref{poiwei}. Suppose~$\pi$ is deterministic. Let~$R>0$ be given by~\eqref{rhdef2}. 
Let~$\bar{L}>0$ 
and let~$W$ be an~$\mathbb{R}^d$-valued r.v. independent of~$(i_k)_k,(u_k)_k,Z^{(n)}$ with~$W\sim N(0,\bar{L}^{-1}I_d)$. 
Let~$k\in\mathbb{N}\setminus\{0\}$. 
Suppose
\begin{equation*}
\hat{r} = \max\bigg( r^*\bigg(\frac{\epsilon}{6k},\frac{\epsilon^2\delta}{3024e^2\cdot k^2n^4}\bigg), R\cdot \bigg(2\lambda_{\max}(\Sigma)\ln\bigg(\frac{3024e^2\cdot k^2(k+1)n^3}{\epsilon^2\delta}\bigg)\bigg)^{1/2}\bigg),
\end{equation*}
where~$r^*$ is given by~\eqref{rsd} in which~$y_{\cdot}^*$ is given by~\eqref{ypa0}. 
Suppose
\begin{equation}\label{lpd}
\bar{L}\geq \widetilde{L}:= 5C_{\pi}d + 4(n/7)\lambda_{\min}(\Sigma)
e^{-4R(\lambda_{\max}(\Sigma))^{1/2}} + e^{\hat{r}}\cdot 9n\lambda_{\max}(\Sigma),
\end{equation}
and suppose
\begin{equation*}
k \geq \frac{14\bar{L}de^{4R(\lambda_{\max}(\Sigma))^{1/2}}}{n\lambda_{\min}(\Sigma)}\bigg[2\ln\bigg(\frac{6k}{\epsilon}\bigg) + \ln\bigg(1+\frac{d}{4}\ln\bigg(\frac{14d\bar{L}e^{4R(\lambda_{\max}(\Sigma))^{1/2}}}{n\lambda_{\min}(\Sigma)}\bigg)\bigg) \bigg].
\end{equation*}
There exists an event~$B_*$ with~$\mathbb{P}(B_*)\geq 1-\delta$ and an absolute constant~$c_5>0$ such that if
\begin{align}
n &\geq \bigg[c^*\bigg(c_5,\frac{\epsilon\delta}{kn}\bigg)\max\bigg( (\kappa(\Sigma))^2,\frac{C_{\pi}}{\lambda_{\min}(\Sigma)}\bigg) \cdot \ln\bigg(\frac{kn}{\epsilon\delta}\bigg)\cdot\bigg[d + \ln\bigg(\frac{kn}{\epsilon\delta}\bigg)\bigg]\bigg]\nonumber\\
&\qquad \vee\big[7e^{4R(\lambda_{\max}(\Sigma))^{1/2}}(\lambda_{\min}(\Sigma))^{-1}\max(\bar{L}^{1/2}d^{1/2}z_*,2kz_*^2)\big],\label{nrc3}
\end{align}
where~$z_*=(y_{\delta/(42n)}^*+3e^{\hat{r}}/2)(2\lambda_{\max}(\Sigma)\ln(84nd/\delta))^{1/2}$ with~$y_{\delta/(42n)}^*$ given by~\eqref{ypa0} and~$c^*
$ is given by~\eqref{por}, 
then on~$B_*$, there exists a unique~$\theta_{\textrm{map}}\in B_{(\lambda_{\max}(\Sigma))^{-1/2}}(\theta^*)$ with~$\nabla \pi(\cdot|Z^{(n)})|_{\theta_{\textrm{map}}}=0$ and for~$\theta_0 = \theta_{\textrm{map}}+W$ it holds that
\begin{equation*}
\textrm{TV}(\delta_{\theta_0}P^k,\pi(\cdot|Z^{(n)})) \leq \epsilon.
\end{equation*}
\end{theorem}
\begin{remark}
Our conditions on~$k$ and~$n$ reduce, after some bookkeeping, to~$k\gtrsim d$ and~$n\gtrsim d$ respectively, where~$\gtrsim$ hides sub-polynomial factors. We conduct this bookkeeping in the proofs of the main results in Section~\ref{poimse}, where further assumptions are made on~$\abs{\theta^*}^2\lambda_{\max}(\Sigma)$. 
\end{remark}
The broad strategy for the proof of Theorem~\ref{poigib} is similar to that of Theorem~\ref{gm}. 
We construct a Gibbs trajectory targeting~$\widetilde{\pi}^{(-0)}(\cdot|Z^{(n)})$ (as given in~\eqref{tpid}), which enjoys the entropy contraction of~\cite{ascolani2024e} by the properties shown in Lemma~\ref{lop}. Thereafter, we (effectively) maximally couple this trajectory with Gibbs trajectories targeting~$\pi(\cdot|Z^{(n)}),\bar{\pi}(\cdot|Z^{(n)})$, and use previous negligibility results to show that they stick together often enough to yield the assertion.
\begin{proof}[Proof of Theorem~\ref{poigib}]
For~$\delta,\epsilon,k$ as in the assertion (note that~$k\geq 1$), we set
\begin{equation*}
\bar{\delta}=\delta/3,\qquad\bar{\epsilon}=\epsilon/(6k).
\end{equation*}
We set
\begin{equation}\label{rsh}
\hat{r} = \max\big(r^*(\bar{\epsilon},\bar{\epsilon}^2 \bar{\delta}/(28e^2n^4)), R\cdot(2\lambda_{\max}(\Sigma) \ln(28e^2\cdot(k+1) n^3\bar{\epsilon}^{-2}\bar{\delta}^{-1}))^{1/2}\big),
\end{equation}
where~$r^*$ is given by~\eqref{rsd} with~$y_{\cdot}^*$ given by~\eqref{ypa0}. 
Note that~$r^*$ (the function as in~\eqref{rsd}) is non-increasing in both arguments; in particular, this satisfies our conditions on~$\hat{r},r$ (in case~$r=\hat{r}$) in Theorems~\ref{poiwei} and Proposition~\ref{lom}, which we apply below.

We begin by gathering properties for posteriors already obtained in three previous results, then defining~$B_*$ by the intersection of the events where they are valid. These events will be denoted~$\bar{B},\hat{B},\bar{B}',\bar{B}''$, with~$B_*=\bar{B}\cap\hat{B}\cap\bar{B}'\cap\bar{B}''$. 
Recall that~$\bar{A}:\mathbb{R}\rightarrow\mathbb{R}$ is given by~\eqref{aba2} (with the above~$\hat{r}$) and~$\bar{\pi}(\cdot|Z^{(n)}):\mathbb{R}^d\rightarrow(0,\infty)$ is given by~\eqref{piba}. 
Let~$Q:\mathbb{R}^d\rightarrow[0,\infty)$ be given for any~$x\in\mathbb{R}^d$ by~\eqref{qua}, 
where~$\tilde{c}(2R)$ is given by~\eqref{dbri}. Let~$\widetilde{U}^{(-0)}:\mathbb{R}^d\rightarrow\mathbb{R}$ be given by~\eqref{tpid}. 
We apply Proposition~\ref{lom} with~$\epsilon,\delta$ therein replaced by~$\bar{\epsilon},\bar{\delta}/2$, 
to obtain an event~$\bar{B}$ with the properties therein. 
In particular, we have~$
\mathbb{P}(\bar{B})\geq 1-\bar{\delta}/2
$, and on~$\bar{B}$, given
\begin{align}
n &\geq \max\big( c^*(\bar{c}_4'',\bar{\epsilon}\bar{\delta}/(2n)) \max(\kappa(\Sigma)^2,C_{\pi}(\lambda_{\min}(\Sigma))^{-1})\nonumber\\
&\quad\cdot \big[d\ln\big(\bar{L}\cdot 7e^{4R(\lambda_{\max}(\Sigma))^{1/2}}/(\lambda_{\min}(\Sigma)n)\big) +\ln(2kn\bar{\epsilon}^{-1}\bar{\delta}^{-1})\big],\nonumber\\
&\quad 7e^{4R(\lambda_{\max}(\Sigma))^{1/2}}(\lambda_{\min}(\Sigma))^{-1}\max(\bar{L}^{1/2}d^{1/2}z_*,2kz_*^2) \big),\label{nty1}
\end{align}
with~$z_*=(y_{\bar{\delta}/(14n)}^*+3e^{\hat{r}}/2)(2\lambda_{\max}(\Sigma)\ln(28nd/\bar{\delta}))^{1/2}$, 
there exists a unique~$
\theta_{\textrm{map}}\in B_{(\lambda_{\max}(\Sigma))^{-1/2}}(\theta^*)\subset B_R$ with~$
\nabla \widetilde{U}^{(-0)}|_{\theta_{\textrm{map}}}=0$. Note that~\eqref{nty1} follows from our conditions on~$k$ and~$n$ in the assertion (given large enough~$c_5$) by definition~\eqref{por} of~$c^*$. Moreover, on~$\bar{B}$ and under~\eqref{nty1} we have for any~$l\in[0,k]\cap\mathbb{N}$ that Gibbs iterate~$\tilde{\theta}_l^{(-0)}$, initialized appropriately with~$W$ (coinciding as in the present assertion), satisfy
\begin{equation}\label{tre}
\mathbb{P}( \cup_{i=1}^n\{X_i\cdot \tilde{\theta}_l^{(-0)} \geq \hat{r} \}|Z^{(n)})
\leq \bar{\epsilon}.
\end{equation}
We use the~$l=0$ case to show that~$\theta_{\textrm{map}}$ satisfies a similar inequality. Suppose for contradiction that on~$\bar{B}$ there exists~$i^*\in[1,n]\cap\mathbb{N}$ with~$X_{i^*}\cdot \theta_{\textrm{map}}\geq \hat{r}$. By definition of~$W$, on~$\bar{B}$ we have
\begin{equation*}
\mathbb{P}( \cup_{i=1}^n\{X_i\cdot \tilde{\theta}_0^{(-0)} \geq \hat{r} \}|Z^{(n)})\geq \mathbb{P}(X_{i^*}\cdot (\theta_{\textrm{map}} + W) \geq \hat{r} |Z^{(n)}) \geq \mathbb{P}(
W\cdot X_{i^*}\geq 0|Z^{(n)}) = \frac{1}{2},
\end{equation*}
which contradicts~\eqref{tre}. Therefore it holds on~$\bar{B}$ and under~\eqref{nty1} that
\begin{equation*}
X_i\cdot\theta_{\textrm{map}} <\hat{r}\qquad \forall i\in[1,n]\cap\mathbb{N}.
\end{equation*}
In particular, by our definitions for~$\pi(\cdot|Z^{(n)}),\bar{\pi}(\cdot|Z^{(n)})$ and by~$\theta_{\textrm{map}}\in B_{(\lambda_{\max}(\Sigma))^{-1/2}}(\theta^*)\subset B_R$, on~$\bar{B}$, the point~$\theta_{\textrm{map}}$ is also a critical point for~$\pi(\cdot|Z^{(n)})$ and~$\bar{\pi}(\cdot|Z^{(n)})$. Namely, it holds on~$\bar{B}$ and under~\eqref{nty1} that~$\nabla \ln\pi(\cdot|Z^{(n)})|_{\theta_{\textrm{map}}} = \nabla \ln\bar{\pi}(\cdot|Z^{(n)})|_{\theta_{\textrm{map}}} = 0$. We may apply Theorem~\ref{poiwei}\ref{pot0} with~$\hat{r}=\infty$ therein, to obtain that the critical point~$\theta_{\textrm{map}}$ is unique in~$B_{(\lambda_{\max}(\Sigma))^{-1/2}}(\theta^*)$ on~$\bar{B}\cap\hat{B}$ for~$\pi(\cdot|Z^{(n)})$ and for some event~$\hat{B}$ with~$\mathbb{P}(\hat{B})\geq 1-\bar{\delta}/2$.\\
\indent Let~$\widetilde{U}:\mathbb{R}^d\rightarrow\mathbb{R}$ be the negative log-density given by
\begin{equation}\label{utdef2}
\widetilde{U}(\theta)= Q(\theta+\theta_{\textrm{map}})-\ln\pi(\theta+\theta_{\textrm{map}})-\textstyle\sum_{i=1}^n  (Y_i(\theta+\theta_{\textrm{map}})^{\top}X_i - \bar{A}((\theta+\theta_{\textrm{map}})^{\top}X_i))
\end{equation}
and let~$\widetilde{\pi}(\cdot|Z^{(n)})$ denote the measure with density~$e^{-\widetilde{U}(\cdot-\theta_{\textrm{map}}))}/\int_{\mathbb{R}^d}e^{-\widetilde{U}}$. We will use the same notation for the density, and do the analogous for~$\pi(\cdot|Z^{(n)})$ and~$\bar{\pi}(\cdot|Z^{(n)})$. 
By Lemma~\ref{lop} with~$\delta$ replaced by~$\bar{\delta}$, given~\eqref{nty1} (so that~\eqref{nln}, with~$\delta$ therein replaced by~$\bar{\delta}$, is satisfied), there exists an event~$\bar{B}'$ with~$\mathbb{P}(\bar{B}')\geq 1-\bar{\delta}$ such that it holds on~$\bar{B}\cap \bar{B}'$ that~$\widetilde{U}$ is~$\tilde{c}(2R)n$-strongly convex on~$\mathbb{R}^d$ and~$\widetilde{L}$-smooth, where here and below~$\widetilde{L}$ is given by
\begin{equation}\label{lxd}
\widetilde{L}= 5 C_{\pi}d + 4\tilde{c}(2R)n + e^{\hat{r}}\cdot 9n\lambda_{\max}(\Sigma)
\end{equation}
as in the assertion.\\
\indent On~$\Omega\setminus \bar{B}$, we define~$
\theta_{\textrm{map}}=0$ (as placeholder). 
Let~$\bar{U},\bar{U}_0$ be given by~\eqref{uuo} (with~$\chi:\mathbb{R}^d\rightarrow[0,1]$ given by~\eqref{chidef}). 
We apply Theorem~\ref{poiwei} with~$\delta,\epsilon$ therein replaced by~$\bar{\delta},\bar{\epsilon}$, to obtain an event~$\bar{B}''$ with the properties therein. 
By inspection of the proof of Theorem~\ref{poiwei}, the conclusions of Lemma~\ref{uuc} also apply on~$\bar{B}''$. In particular, we have~$\mathbb{P}(\bar{B}'')\geq 1-\bar{\delta}$. Moreover, on~$\bar{B}''$ and given
\begin{equation}\label{nty2}
n\geq c^*(c_4', \bar{\delta})\cdot c^{\dagger}\cdot (d+\ln (\bar{\delta}^{-1}) + \ln(\bar{\epsilon}^{-1})),
\end{equation}
where~$c^*$ is given by~\eqref{por} and~$c^{\dagger}$ is given in Corollary~\ref{nvc}, 
the functions~$\bar{U},\bar{U}_0$ satisfy Conditions~\ref{cond:curvature3} and~\ref{cond:smooth} with~\eqref{molm}. On~$\bar{B}''$ and under~\eqref{nty2}, inequalities~\eqref{meco} also hold.\\
\indent Before proceeding with the main coupling arguments, we check that our condition~\eqref{nrc3} on~$n$ (for large enough~$c_5$) implies~\eqref{nty2}. We keep in mind~$k\geq 8$ so that~$\ln(k/\epsilon)\geq 1$. We also denote by~$C\geq 1$ a generic absolute constant that may change from line to line. 
We show that~$c^{\dagger}$ within~\eqref{nty2} (defined just after~\eqref{por} in Corollary~\ref{nvc}) satisfies
\begin{equation}\label{csb}
c^{\dagger}\leq \max\big((\kappa(\Sigma))^2,C_{\pi}/\lambda_{\min}(\Sigma)\big)\cdot c^*(C,\epsilon\delta n^{-1}k^{-1})\ln(nk\epsilon^{-1}\delta^{-1})
\end{equation}
in the present Poisson case~\eqref{pois},~\eqref{por}. Together with the fact that~$c^*$ as in~\eqref{por} is non-decreasing and non-increasing respectively in its first and second argument, inequality~\eqref{csb} gives that~\eqref{nrc3} is stronger than~\eqref{nty2} as required. It suffices to show that the middle argument in the definition of~$c^{\dagger}$ satisfies~\eqref{csb} in place of~$c^{\dagger}$. The~$\ln$ factor therein is bounded as
\begin{align}
&\ln(C_{\pi}/\lambda_{\min}(\Sigma) + e^{\hat{r}}\kappa(\Sigma))\nonumber\\
&\quad\leq \ln(2\max(C_{\pi}/\lambda_{\min}(\Sigma),e^{\hat{r}}\kappa(\Sigma)))\nonumber\\
&\quad\leq \ln2 + 0\vee\ln(C_{\pi}/\lambda_{\min}(\Sigma)) + \hat{r} + \ln\kappa(\Sigma)\nonumber\\
&\quad\leq 1 + (C_{\pi}/\lambda_{\min}(\Sigma))^{1/2} 
+\hat{r} + \kappa(\Sigma)^{1/2}.\label{csp2}
\end{align}
To bound~$\hat{r}$ as in~\eqref{rsh}, we control each argument in the maximum. Firstly, by respective definitions~\eqref{rsd},~\eqref{ypa0},~\eqref{por} of~$r^*$,~$y_{\delta/(10n)}^*$ and~$c^*$, we have
\begin{align*}
r^*(\bar{\epsilon},\bar{\epsilon}^2\bar{\delta}/(28e^2n^4)) &\leq \max\bigg[c^*\bigg(C,\frac{\bar{\epsilon}\bar{\delta}}{n}\bigg), C\bigg( 1+\ln\bigg(\frac{n}{\bar{\epsilon}\bar{\delta}}\bigg)\bigg)\bigg]\\
&\leq c^*(C,\epsilon\delta/(nk)) \cdot \ln(nk/(\epsilon\delta)).
\end{align*}
Secondly, by definition~\eqref{rhdef2} of~$R$ and again that~\eqref{por} of~$c^*$, we also have
\begin{equation*}
R\cdot (2\lambda_{\max}(\Sigma)\ln(28e^2\cdot (k+1)n^3\bar{\epsilon}^{-2}\bar{\delta}^{-1}))^{1/2} \leq c^*(C,\epsilon\delta/(nk)) \cdot \ln(nk/(\epsilon\delta)).
\end{equation*}
Thus inequality~\eqref{csp2} implies 
\begin{align*}
\kappa(\Sigma)\ln(C_{\pi}/\lambda_{\min}(\Sigma) + e^{\hat{r}}\kappa(\Sigma)) &\leq \max\big((\kappa(\Sigma))^2,C_{\pi}/\lambda_{\min}(\Sigma)\big)\\
&\quad\cdot c^*(C,\epsilon\delta/(nk)) \cdot \ln(nk/(\epsilon\delta))
\end{align*}
as required.\\
\indent 
Define~$B_*=\bar{B}\cap \hat{B}\cap\bar{B}'\cap\bar{B}''$. Note that
\begin{equation}\label{bsp}
\mathbb{P}(B_*) \geq 1-3\bar{\delta}=1-\delta
\end{equation}
as required. Note also that~$B_*$ is an event determined by~$Z^{(n)}$, and it is independent of any randomness arising from the Gibbs sampler algorithm. We use this in the sequel without further mention. We also take~\eqref{nty1} and~\eqref{nty2} for granted in the following.

We construct a coupling of Gibbs trajectories similar to the proof of Theorem~\ref{gm}, but targeting~$\pi(\cdot|Z^{(n)})$,~$\bar{\pi}(\cdot|Z^{(n)})$ and~$\widetilde{\pi}(\cdot|Z^{(n)})$. 
We have~$\bar{L}\geq \widetilde{L}$ and~$\theta_0=\theta_{\textrm{map}}+W$ (with~$W$ as in the assertion, in particular~$W\sim N(0,\bar{L}^{-1}I_d)$). 
Let~$\hat{\theta}_0=\bar{\theta}_0=\widetilde{\theta}_0=\theta_0$ and let~$(\hat{\theta}_k)_{k\in\mathbb{N}},(\bar{\theta}_k)_{k\in\mathbb{N}},(\widetilde{\theta}_k)_{k\in\mathbb{N}}$ be given inductively as follows. For any~$k\in\mathbb{N}$, assume~$\hat{\theta}_k,\bar{\theta}_k,\widetilde{\theta}_k$ are given. 
For any~$x\in\mathbb{R}^d$ and any~$i\in[1,d]\cap\mathbb{N}$, 
let~$\bar{F}_{x,i}, \widetilde{F}_{x,i}:\mathbb{R}\rightarrow[0,1]$ be the random cumulative distribution function of the conditional distribution of~$y_i$ given~$y_{-i}=x_{-i}$ under~$y\sim \bar{\pi}(\cdot|Z^{(n)})$ and~$y\sim \widetilde{\pi}(\cdot|Z^{(n)})$ respectively. 
Note that~$F_{x,i}
$ has already been defined analogously for~$\pi(\cdot|Z^{(n)})
$ just before the assertion alongside~$(u_k),(i_k)$. Denote~$\vartheta_k=(\hat{\theta}_k,\bar{\theta}_k,\widetilde{\theta}_k,i_k)$. For any~$\vartheta = (\vartheta_1,\vartheta_2,\vartheta_3,i)\in\mathbb{R}^{d}\times \mathbb{R}^{d}\times\mathbb{R}^{d}\times([1,d]\cap\mathbb{N})$, 
let~$(\xi_k^{\vartheta},\bar{\xi}_k^{\vartheta},\widetilde{\xi}_k^{\vartheta})$ be the coupling of~$(F_{\vartheta_1,i}^{-1}(u_k),\bar{F}_{\vartheta_2,i}^{-1}(u_k),\widetilde{F}_{\vartheta_3,i}^{-1}(u_k))$ defined as follows. 
Let
\begin{align*}
\bar{I}&:= 
\{x\in\mathbb{R}:((\vartheta_2)_1,\dots,(\vartheta_2)_{i-1},x,(\vartheta_2)_{i+1},\dots,(\vartheta_2)_d)\in B_R
\},\nonumber\\
\bar{I}'&:=\cap_{j=1}^n\{x\in\mathbb{R}:((\vartheta_2)_1,\dots,(\vartheta_2)_{i-1},x,(\vartheta_2)_{i+1},\dots,(\vartheta_2)_d)^{\top}X_j \leq \hat{r}\}.
\end{align*} 
These are the respective regions where we will force~$(\bar{\xi}_k^{\vartheta},\widetilde{\xi}_k^{\vartheta})$ and~$(\xi_k^{\vartheta},\bar{\xi}_k^{\vartheta})$ to synchronize. 
For any~$y\in\mathbb{R}^d$, let~$\mu_y,\bar{\mu}_y,\widetilde{\mu}_y$ be probability measures on~$\mathbb{R}$ with densities (also denoted~$\mu_y,\bar{\mu}_y,\widetilde{\mu}_y$) given for any~$x\in\mathbb{R}$ by
\begin{align*}
\mu_y(x) &:=(\pi(\cdot|Z^{(n)}))_i(x,y_{-i})/\textstyle\int_{\mathbb{R}}(\pi(\cdot|Z^{(n)}))_i(\cdot,y_{-i}),\\
\bar{\mu}_y(x) &:=(\bar{\pi}(\cdot|Z^{(n)}))_i(x,y_{-i})/\textstyle\int_{\mathbb{R}}(\bar{\pi}(\cdot|Z^{(n)}))_i(\cdot,y_{-i}),\\
\widetilde{\mu}_y(x) &:=(\widetilde{\pi}(\cdot|Z^{(n)}))_i(x, y_{-i}) /\textstyle\int_{\mathbb{R}}(\widetilde{\pi}(\cdot|Z^{(n)}))_i(\cdot,y_{-i}).
\end{align*}
Note that we have by~$Q\geq 0$ and by~$\bar{A}\leq A$ (from definition~\eqref{aba2} of~$\bar{A}$) that
\begin{align*}
\mu_{\vartheta_2}(x) &\geq \bar{\mu}_{\vartheta_2}(x) \qquad\forall x\in \bar{I}',\\
\widetilde{\mu}_{\vartheta_2}(x) &\geq \bar{\mu}_{\vartheta_2}(x) \qquad\forall x\in \bar{I}.
\end{align*}
Let~$\nu_1,\nu_3$ denote the probability measures on~$\mathbb{R}$ with densities proportional to~$\mathbb{R}\ni x\mapsto\mu_{\vartheta_2}(x)-\bar{\mu}_{\vartheta_2}(x)\mathds{1}_{\bar{I}'}(x)$ and~$\mathbb{R}\ni x\mapsto\widetilde{\mu}_{\vartheta_2}(x)-\bar{\mu}_{\vartheta_2}(x)\mathds{1}_{\bar{I}}(x)$ (whenever they have non-zero mass and say standard normal densities otherwise). These are `residual' measures that contain all the mass outside of the synchronous coupling region. 
Let~$\bar{\xi}_k^{\vartheta}=\bar{F}_{\vartheta_2}^{-1}(u_k)$. Let
\begin{align*}
\xi_k^{\vartheta} &= \begin{cases}
\bar{\xi}_k^{\vartheta} &\textrm{if }\bar{\xi}_k^{\vartheta} \in \bar{I}' \textrm{ and }\vartheta_1=\vartheta_2,\\
\Xi_1 &\textrm{if }\bar{\xi}_k^{\vartheta} \notin \bar{I}' \textrm{ and }\vartheta_1=\vartheta_2,\\
\Xi_1' &\textrm{if }\vartheta_1\neq \vartheta_2,
\end{cases}\\
\widetilde{\xi}_k^{\vartheta} &= \begin{cases}
\bar{\xi}_k^{\vartheta} &\textrm{if }\bar{\xi}_k^{\vartheta} \in \bar{I} \textrm{ and }\vartheta_2=\vartheta_3,\\
\Xi_3 &\textrm{if }\bar{\xi}_k^{\vartheta} \notin \bar{I} \textrm{ and }\vartheta_2=\vartheta_3,\\
\Xi_3' &\textrm{if }\vartheta_2\neq \vartheta_3,
\end{cases}
\end{align*}
where~$\Xi_1\sim \nu_1$,~$\Xi_3\sim \nu_3$,~$\Xi_1'\sim \mu_{\vartheta_1}$,~$\Xi_3'\sim \widetilde{\mu}_{\vartheta_3}$ are independent of each other and all other r.v.'s. 
Let~$\hat{\theta}_{k+1},\bar{\theta}_{k+1},\widetilde{\theta}_{k+1}$ be given for any~$i\in[1,d]\cap\mathbb{N}$ by
\begin{equation*}
((\hat{\theta}_{k+1})_i,(\bar{\theta}_{k+1})_i,(\widetilde{\theta}_{k+1})_i) = 
\begin{cases}
(\xi_k^{\vartheta_k},\bar{\xi}_k^{\vartheta_k},\widetilde{\xi}_k^{\vartheta_k}) &\textrm{if } i=i_k,\\
((\hat{\theta}_k)_i,(\bar{\theta}_k)_i,(\widetilde{\theta}_k)_i) &\textrm{if } i\neq i_k.
\end{cases}
\end{equation*}
Note that measurability of the various densities and~\cite[Lemma~4.22]{MR4226142} imply the measurability of~$(\hat{\theta}_{k+1},\bar{\theta}_{k+1},\widetilde{\theta}_{k+1})$ (w.r.t. the probability space). 
For any~$k\in\mathbb{N}\setminus\{0\}$ we denote by~$\delta_{\bar{\theta}_0}=\delta_{\hat{\theta}_0}=\delta_{\widetilde{\theta}_0}$,~$\delta_{\bar{\theta}_0}\bar{P}^k$,~$\delta_{\hat{\theta}_0}\hat{P}^k$,~$\delta_{\widetilde{\theta}_0}\widetilde{P}^k$ to be 
the conditional distributions of~$\bar{\theta}_0=\hat{\theta}_0=\widetilde{\theta}_0$ and~$\bar{\theta}_k,\hat{\theta}_k,\widetilde{\theta}_k$ given~$Z^{(n)}$ respectively. 
For any~$k\in\mathbb{N}$,~$\hat{\theta}_k$ has the same law as~$\theta_k$. 

By Lemma~2.4 and Theorem~3.2 both in~\cite{ascolani2024e}, on~$B_*$, it holds for any~$k\in\mathbb{N}$ that
\begin{equation}\label{KLcon2}
\textrm{KL}(\delta_{\widetilde{\theta}_0}\widetilde{P}^k| \widetilde{\pi}(\cdot|Z^{(n)}))\leq \bigg(1-\frac{\tilde{c}(2R)n}{\widetilde{L} d}\bigg)^{\!k} \textrm{KL}(\delta_{\widetilde{\theta}_0}| \widetilde{\pi}(\cdot|Z^{(n)})).
\end{equation}
By Lemma~31 in~\cite{pmlr-v178-c}, again monotonicity of R\'enyi divergences and Lemma~30 in~\cite{MR4309974}, 
it holds on~$B_*$ that 
\begin{equation}\label{hoi22}
\textrm{KL}(\delta_{\widetilde{\theta}_0}| \widetilde{\pi}(\cdot|Z^{(n)}))  =\textrm{KL}(N(\theta_{\textrm{map}},\bar{L}^{-1}I_d)|\widetilde{\pi}(\cdot|Z^{(n)}))\leq 2 + (d/2)\ln(2d\bar{L}/(\tilde{c}(2R)n)).
\end{equation}
Substituting~~\eqref{hoi22} into~\eqref{KLcon2} yields by Pinsker's inequality on~$B_*$ for any~$k\in\mathbb{N}$ that
\begin{equation}\label{tvbR2}
\textrm{TV}(\delta_{\widetilde{\theta}_0}\widetilde{P}^k, \widetilde{\pi}(\cdot|Z^{(n)}))\leq R_k,
\end{equation}
where
\begin{equation}\label{rkl}
R_k :=\bigg(1-\frac{\tilde{c}(2R)n}{\widetilde{L} d}\bigg)^{\!k/2} \bigg( 1 + \frac{d}{4}\ln\bigg(\frac{2d\bar{L}}{\tilde{c}(2R)n}\bigg)\bigg)^{\!1/2}.
\end{equation}
Next, we prove by induction it holds for any~$t\in\mathbb{N}$ that
\begin{equation}\label{cei}
\hat{\mathcal{E}}_t\cap B_* \subset \{\bar{\theta}_t=\hat{\theta}_t =\widetilde{\theta}_t\} \cap B_*,
\end{equation}
where
\begin{equation}\label{cek2}
\hat{\mathcal{E}}_t := 
(\cap_{l=1}^t\{\bar{\theta}_l\in B_{(\lambda_{\max}(\Sigma))^{-1/2}}(\theta_{\textrm{map}})\})
\cap (\cap_{l=1}^t\cap_{i=1}^n\{\bar{\theta}_l^{\top} X_i\leq \hat{r}\}).
\end{equation}
Just as in the proof of Theorem~\ref{gm}, this induction~\eqref{cei} says that if~$(\bar{\theta}_k)_k$, the trajectory targeting~$\bar{\pi}(\cdot|Z^{(n)})$, stays inside an appropriately set, then the coupled~$(\hat{\theta}_k)_k,(\widetilde{\theta}_k)_k$ also stay in the same set. Again we do this by considering the conditional densities from which Gibbs draws and showing that those for~$(\bar{\theta}_k)_k$ are always bounded above by the others in that set, and so trajectories remain coupled as long as~$\bar{\theta}_k$ remain in that set. To that end, 
let~$U:\mathbb{R}^d\rightarrow\mathbb{R}$ be the negative log-density given by
\begin{equation*}
U(\theta) = -\ln\pi(\theta+\theta_{\textrm{map}}) - \textstyle\sum_{i=1}^n  (Y_i(\theta+\theta_{\textrm{map}})^{\top}X_i - A((\theta+\theta_{\textrm{map}})^{\top}X_i)),\\
\end{equation*}
and recall~$\bar{U},\widetilde{U}$ have been defined above with~\eqref{uuo1} and~\eqref{utdef2}. 
Fix~$t\in\mathbb{N}$ and assume~\eqref{cei}. 
With the definition~\eqref{rhdef2} of~$R$ and that~$\theta_{\textrm{map}}\in B_{(\lambda_{\max}(\Sigma))^{-1/2}}(\theta^*)$ on~$B$, 
we have on the event~$B_*\cap\hat{\mathcal{E}}_t = B_*\cap\hat{\mathcal{E}}_t\cap\{\bar{\theta}_t=\hat{\theta}_t = \widetilde{\theta}_t\}$ 
that
\begin{align*}
(U(\cdot-\theta_{\textrm{map}}))_{i_t}(x,(\hat{\theta}_t)_{-i_t}) 
= (\bar{U}(\cdot-\theta_{\textrm{map}}))_{i_t}(x,(\bar{\theta}_t)_{-i_t}) 
&= (\widetilde{U}(\cdot-\theta_{\textrm{map}}))_{i_t}(x,(\widetilde{\theta}_t)_{-i_t})  \\
&\qquad\qquad\forall x\in I_t\subset\mathbb{R},
\end{align*}
where~$I_t$ is given by
\begin{align*}
I_t&:= 
\{x\in\mathbb{R}:((\bar{\theta}_t)_1,\dots,(\bar{\theta}_t)_{i_t-1},x,(\bar{\theta}_t)_{i_t+1},\dots,(\bar{\theta}_t)_d)\in B_{(\lambda_{\max}(\Sigma))^{-1/2}}(\theta_{\textrm{map}})
\}\cap\nonumber\\
&\qquad\cap_{i=1}^n\{x\in\mathbb{R}:((\bar{\theta}_t)_1,\dots,(\bar{\theta}_t)_{i_t-1},x,(\bar{\theta}_t)_{i_t+1},\dots,(\bar{\theta}_t)_d)^{\top}X_i \leq \hat{r}\}
\end{align*} 
(the first set in~$I_t$ gives the~$\bar{U}=\widetilde{U}$ part, since~$\widetilde{U}=\bar{U}+Q(\cdot +\theta_{\textrm{map}})$; the rest of the definition of~$I_t$ gives the~$U=\bar{U}$ part). 
We have by~$Q\geq 0$ that
\begin{equation}\label{udo}
\int_{\mathbb{R}}\exp(-\widetilde{U}_{i_t}(\cdot,(\bar{\theta}_t)_{-i_t}))\leq 
\int_{\mathbb{R}} \exp(-\bar{U}_{i_t}(\cdot,(\bar{\theta}_t)_{-i_t}))
\end{equation}
and by~$\bar{A}\leq A$ (from definition~\eqref{aba2} of~$\bar{A}$) that
\begin{equation*}
\int_{\mathbb{R}} \exp(-U_{i_t}(\cdot,(\bar{\theta}_t)_{-i_t}))\leq 
\int_{\mathbb{R}} \exp(-\bar{U}_{i_t}(\cdot,(\bar{\theta}_t)_{-i_t})).
\end{equation*}
These imply respectively
\begin{align*}
\frac{(\widetilde{\pi}(\cdot|Z^{(n)}))_{i_t}(x,(\widetilde{\theta}_t)_{-i_t})}{\int_{\mathbb{R}}(\widetilde{\pi}(\cdot|Z^{(n)}))_{i_t}(\cdot,(\widetilde{\theta}_t)_{-i_t})} &\geq \frac{(\bar{\pi}(\cdot|Z^{(n)}))_{i_t}(x,(\bar{\theta}_t)_{-i_t})}{\int_{\mathbb{R}}(\bar{\pi}(\cdot|Z^{(n)}))_{i_t}(\cdot,(\bar{\theta}_t)_{-i_t})} && \forall x\in I_t,\\
\frac{(\pi(\cdot|Z^{(n)}))_{i_t}(x,(\hat{\theta}_t)_{-i_t})}{\int_{\mathbb{R}}(\pi(\cdot|Z^{(n)}))_{i_t}(\cdot,(\hat{\theta}_t)_{-i_t})}&\geq \frac{(\bar{\pi}(\cdot|Z^{(n)}))_{i_t}(x,(\bar{\theta}_t)_{-i_t})}{\int_{\mathbb{R}}(\bar{\pi}(\cdot|Z^{(n)}))_{i_t}(\cdot,(\bar{\theta}_t)_{-i_t})} && \forall x\in I_t,
\end{align*}
the left-hand sides of which are the conditional density of~$(\widetilde{\theta}_{t+1})_{i_t} = \widetilde{\xi}_t^{\vartheta_t}$ given~$\widetilde{\theta}_t, i_t$ and that of~$(\hat{\theta}_{t+1})_{i_t} = \xi_t^{\vartheta_t}$ given~$\hat{\theta}_t, i_t$ respectively, and the right-hand sides of which are the conditional density of~$(\bar{\theta}_{t+1})_{i_t} = \bar{\xi}_t^{\vartheta_t}$ given~$\bar{\theta}_t, i_t$. 
Therefore, by construction of the coupling, 
it holds on the event~$\{\bar{\xi}_t^{\vartheta_t}\in I_t\}
\cap\hat{\mathcal{E}}_t\cap B_* = \hat{\mathcal{E}}_{t+1}\cap B_*$ that~$\xi_t^{\vartheta_t}=\bar{\xi}_t^{\vartheta_t} = \widetilde{\xi}_t^{\vartheta_t}$, which implies 
\begin{align*}
\hat{\mathcal{E}}_{t+1}\cap B_* &= \hat{\mathcal{E}}_{t+1}\cap B_*\cap\{\bar{\theta}_t=\hat{\theta}_t = \widetilde{\theta}_t\} \\
&\subset \{\xi_t^{\vartheta_t}=\bar{\xi}_t^{\vartheta_t}=\widetilde{\xi}_t^{\vartheta_t}\}\cap B_*\cap\{\bar{\theta}_t=\hat{\theta}_t=\widetilde{\theta}_t\} \\
&\subset \{\bar{\theta}_{t+1}=\hat{\theta}_{t+1}=\widetilde{\theta}_{t+1}\}\cap B_*,
\end{align*}
and this is~\eqref{cei} with~$t$ replaced by~$t+1$. 
For any~$k\in\mathbb{N}$, it follows by~\eqref{tvbR2} that on~$B_*$ we have
\begin{align}
\textrm{TV}(\delta_{\hat{\theta}_0}\hat{P}^k, \widetilde{\pi}(\cdot|Z^{(n)})) &\leq \textrm{TV}(\delta_{\hat{\theta}_0}\hat{P}^k, \delta_{\widetilde{\theta}_0}\widetilde{P}^k) + \textrm{TV}(\delta_{\widetilde{\theta}_0}\widetilde{P}^k, \widetilde{\pi}(\cdot|Z^{(n)}))\nonumber\\
&\leq 1-\mathbb{P}(\widetilde{\theta}_k=\hat{\theta}_k|Z^{(n)}) + R_k\nonumber\\
&\leq 1-\mathbb{P}(\hat{\mathcal{E}}_k|Z^{(n)})
+ R_k.\label{tvb3}
\end{align}
On~$B_*$, we have by definition~\eqref{cek2} of~$\hat{\mathcal{E}}_k$ that 
\begin{align}
\mathbb{P}(\hat{\mathcal{E}}_k|Z^{(n)})&\geq 1 - \textstyle\sum_{l=1}^k\mathbb{P}(\abs{\bar{\theta}_l-\theta_{\textrm{map}}}>(\lambda_{\max}(\Sigma))^{-1/2}|Z^{(n)}
) \nonumber\\
&\quad- \mathbb{P}(
\cup_{l=1}^k\cup_{i=1}^n\{\bar{\theta}_l^{\top}X_i> \hat{r}\}|Z^{(n)}).\label{edr}
\end{align}
\indent We control the two negative expressions on the right-hand side. For the sum, we apply Lemma~\ref{nfo} with\footnote{The left-hand sides refer to the objects in Lemma~\ref{nfo}, the right-hand sides refer to those in this proof.}
\begin{align*}
&U=\bar{U}
,\qquad U_0=\bar{U}_0,\qquad m_0=\tilde{c}(2R)n\mathds{1}_{B_{(\lambda_{\max}(\Sigma))^{-1/2}}},\\
&L=\widetilde{L},
\qquad\textrm{and}\qquad r=(\lambda_{\max}(\Sigma))^{-1/2},
\end{align*}
where recall~$\bar{U}_0$ is given by~\eqref{uuo2}. With these choices, Conditions~\ref{cond:curvature3},~\ref{cond:smooth} in Lemma~\ref{nfo} are satisfied on~$B_*\subset \bar{B}''$ 
(as verified above at the beginning of the proof); 
the~$0\in\argmin_{\mathbb{R}^d}U_0$ condition in Lemma~\ref{nfo} is satisfied by the respective definitions~\eqref{uuo2},~\eqref{chidef},~\eqref{rhdef2} of~$\bar{U}_0,\chi,R$, by~$\theta_{\textrm{map}}\in B_{(\lambda_{\max}(\Sigma))^{-1/2}}(\theta^*)$ and by our assumption~$0\in\argmax \pi$; the~$\bar{L}\geq L$ condition in Lemma~\ref{nfo} is satisfied given our assumption on~$\bar{L}$; the~$r^2m_0(r)\geq 9d$ condition in Lemma~\ref{nfo} is~$n\geq 9\lambda_{\max}(\Sigma)d/\tilde{c}(2R)$, which is satisfied given our assumption~\eqref{nty1} on~$n$ (for large enough~$c_4'$; also note that the first~$\ln$ factor in the expression~\eqref{nty1} is at least~$1$ by~$\bar{L}\geq \widetilde{L}$ and the definition~\eqref{lxd} of~$\widetilde{L}$) together with the definition~\eqref{dbri} of~$\tilde{c}(2R)$. Therefore this application of Lemma~\ref{nfo} yields for any~$l\in\mathbb{N}$ on~$B_*$ that
\begin{equation}\label{nfa2}
\mathbb{P}(\abs{\bar{\theta}_l-\theta_{\textrm{map}}}> (\lambda_{\max}(\Sigma))^{-1/2}|Z^{(n)}) \leq 5(\bar{L}/(\tilde{c}(2R)n))^{d/2}\cdot e^{2d-(3/8)\tilde{c}(2R)n(\lambda_{\max}(\Sigma))^{-1}}.
\end{equation}
By our assumption~\eqref{nty1} (with large enough~$\bar{c}_4''$) 
on~$n$, which by definition~\eqref{dbri} of~$\tilde{c}(2R)$ implies in particular
\begin{equation*}
n\geq (8/3)\lambda_{\max}(\Sigma)(\tilde{c}(2R))^{-1}\big(2d+(d/2)\ln(\bar{L}(\tilde{c}(2R)n)^{-1}) + \ln(5k/\bar{\epsilon})\big),
\end{equation*}
inequality~\eqref{nfa2} implies 
\begin{equation}\label{gob}
\mathbb{P}(\abs{\bar{\theta}_l-\theta_{\textrm{map}}}>(\lambda_{\max}(\Sigma))^{-1/2}|Z^{(n)}) \leq \bar{\epsilon}/k,
\end{equation}
so that the first sum in~\eqref{edr} satisfies on~$B_*$ that
\begin{equation}\label{edr2}
\textstyle\sum_{l=1}^k\mathbb{P}(\abs{\bar{\theta}_l-\theta_{\textrm{map}}}>(\lambda_{\max}(\Sigma))^{-1/2}|Z^{(n)})\leq  
\bar{\epsilon}.
\end{equation}
\indent For the last term in~\eqref{edr}, note first by definition~\eqref{rhdef2} of~$R$ and~$\theta_{\textrm{map}}\in B_{(\lambda_{\max}(\Sigma))^{-1/2}}(\theta^*)$ we have for any~$l\in\mathbb{N}\setminus\{0\}$ on~$B_*$ that
\begin{equation*}
\mathbb{P}(\abs{\bar{\theta}_l} > R|Z^{(n)}) \leq \mathbb{P}(\abs{\bar{\theta}_l- \theta_{\textrm{map}}} > (\lambda_{\max}(\Sigma))^{-1/2}|Z^{(n)}),
\end{equation*}
which implies by~\eqref{gob} 
that
\begin{equation}\label{aog}
\mathbb{P}(\hat{\mathcal{E}}_{l-1}\cap \{\abs{\bar{\theta}_l} > R\}|Z^{(n)})
\leq \mathbb{P}(\{\abs{\bar{\theta}_l} > R\}|Z^{(n)}) \leq 
\bar{\epsilon}.
\end{equation}
By~\eqref{cei}, the coupling construction and~\eqref{udo} together with~$\bar{U}(\cdot-\theta_{\textrm{map}})\cdot\mathds{1}_{B_R} = \widetilde{U}(\cdot-\theta_{\textrm{map}})\cdot\mathds{1}_{B_R}$, it holds for any measurable~$A\subset \mathbb{R}^d$ and~$l\in\mathbb{N}\setminus\{0\}$ on~$B_*$ that 
\begin{align*}
\mathbb{P}(\hat{\mathcal{E}}_{l-1}\cap\{\bar{\theta}_l\in A\}|Z^{(n)}) &\leq \mathbb{P}(\hat{\mathcal{E}}_{l-1}\cap\{\widetilde{\theta}_l\in A\}\cap\{\bar{\theta}_l=\widetilde{\theta}_l\}|Z^{(n)})\\
&\quad+ \mathbb{P}(\hat{\mathcal{E}}_{l-1}\cap\{\bar{\theta}_l\in A\}\cap\{\abs{\bar{\theta}_l}>R\}|Z^{(n)})\\
&\leq \mathbb{P}(\widetilde{\theta}_l\in A|Z^{(n)}) + \mathbb{P}( \hat{\mathcal{E}}_{l-1}\cap\{\abs{\bar{\theta}_l}>R\}|Z^{(n)}).
\end{align*}
Setting~$A=\cup_{i=1}^n\{\theta\in \mathbb{R}^d:\theta\cdot X_i>\hat{r}\}$ and using~\eqref{aog}, we have for any such~$l$ on~$B_*$ that
\begin{equation*}
\mathbb{P}(\hat{\mathcal{E}}_{l-1}\cap\cup_{i=1}^n\{\bar{\theta}_l\cdot X_i>\hat{r}\}|Z^{(n)}) \leq \mathbb{P}(\cup_{i=1}^n \{\widetilde{\theta}_l\cdot X_i>\hat{r}\}|Z^{(n)}) + 
\bar{\epsilon},
\end{equation*}
so that by~\eqref{tre} as well as~$B_*\subset\bar{B}$ and since~$\tilde{\theta}_l^{(-0)}$ has the same distribution as~$\widetilde{\theta}_l$, 
it holds on~$B_*$ 
that
\begin{equation}\label{pin}
\mathbb{P}(\hat{\mathcal{E}}_{l-1}\cap\cup_{i=1}^n\{\bar{\theta}_l\cdot X_i>\hat{r}\}|Z^{(n)}) 
\leq 
2\bar{\epsilon}.
\end{equation}
We prove by induction 
for any~$l\in\mathbb{N}$ on~$B_*$ that
\begin{equation}\label{ild}
\mathbb{P}(\hat{\mathcal{E}}_l |Z^{(n)}) \geq 1-3l\bar{\epsilon}.
\end{equation}
The base case~$l=0$ holds by definition~\eqref{cek2} of~$\hat{\mathcal{E}}_0$. 
Fix~$l\in\mathbb{N}\setminus\{0\}$. Assume~\eqref{ild} holds with~$l$ replaced by~$l-1$. By~\eqref{pin} and the inductive assumption, we have on~$B_*$ that
\begin{align*}
&\mathbb{P}(\cup_{j=1}^l\cup_{i=1}^n\{\bar{\theta}_j\cdot X_i>\hat{r}\}|Z^{(n)})\\
&\quad\leq \mathbb{P}(\hat{\mathcal{E}}_{l-1}\cap\cup_{i=1}^n\{\bar{\theta}_l\cdot X_i>\hat{r}\}|Z^{(n)}) + \mathbb{P}(\Omega\setminus\hat{\mathcal{E}}_{l-1}|Z^{(n)})\\
&\quad\leq 
2\bar{\epsilon} + 3(l-1)\cdot \bar{\epsilon}.
\end{align*}
Therefore by~\eqref{edr} and~\eqref{edr2}, 
we have on~$B_*$ that
\begin{equation*}
\mathbb{P}(\hat{\mathcal{E}}_l|Z^{(n)}) \geq 
1- 3\bar{\epsilon} - 3(l-1)\cdot \bar{\epsilon} = 
1- 3l \bar{\epsilon},
\end{equation*}
which is~\eqref{ild}.\\
\indent Together with~\eqref{tvb3}, 
inequality~\eqref{ild} implies 
on~$B_*$ that
\begin{align}
\textrm{TV}(\delta_{\hat{\theta}_0}\hat{P}^k, \pi(\cdot|Z^{(n)})) &\leq \textrm{TV}(\delta_{\hat{\theta}_0}\hat{P}^k, \widetilde{\pi}(\cdot|Z^{(n)})) +  \textrm{TV}(\widetilde{\pi}(\cdot|Z^{(n)}),\pi(\cdot|Z^{(n)}))\nonumber\\
&\leq 
3k\bar{\epsilon}
+R_k+ \textrm{TV}(\widetilde{\pi}(\cdot|Z^{(n)}),\pi(\cdot|Z^{(n)})).\label{tkm}
\end{align}
For the last total variation term, 
we have
\begin{equation}\label{tk0}
\textrm{TV}(\widetilde{\pi}(\cdot|Z^{(n)}),\pi(\cdot|Z^{(n)})) \leq  \textrm{TV}(\widetilde{\pi}(\cdot|Z^{(n)}),\bar{\pi}(\cdot|Z^{(n)})) +  \textrm{TV}(\bar{\pi}(\cdot|Z^{(n)}),\pi(\cdot|Z^{(n)})).
\end{equation}
First, we bound the first term on the right-hand side of~\eqref{tk0}. Let~$\bar{\psi},\widetilde{\psi}$ be~$\mathbb{R}^d$-valued r.v.'s with conditional distributions~$\bar{\pi}(\cdot|Z^{(n)}),\widetilde{\pi}(\cdot|Z^{(n)})$ given~$Z^{(n)}$ respectively. Let~$(\bar{\psi}',\widetilde{\psi}')$ be the maximal coupling of~$(\bar{\psi},\widetilde{\psi})$ as constructed in the proof of Theorem~7.3 in~\cite[Chapter~3]{MR1741181}. Clearly, a.s.\ the probability density function~$\bar{\pi}(\cdot|Z^{(n)})$ of~$\bar{\psi}$ is less than or equal to that~$\widetilde{\pi}(\cdot|Z^{(n)})$ of~$\widetilde{\psi}$ on~$B_R$, therefore we have
\begin{equation*}
1-\textrm{TV}(\bar{\pi}(\cdot|Z^{(n)}),\widetilde{\pi}(\cdot|Z^{(n)}))=\mathbb{P}(\bar{\psi}'=\widetilde{\psi}'|Z^{(n)})\geq \bar{\pi}(B_R|Z^{(n)}) = 1- \bar{\pi}(\mathbb{R}^d\setminus B_R|Z^{(n)}).
\end{equation*}
To control~$\bar{\pi}(\mathbb{R}^d\setminus B_R|Z^{(n)})$, 
note that we defined~$B_*\subset \bar{B}''$ to be an event from Theorem~\ref{poiwei}, so 
on~$B_*$ we have by the first inequality in~\eqref{meco} (with~$\epsilon$ replaced by~$\bar{\epsilon}$) that~$\bar{\pi}(\mathbb{R}^d\setminus B_R|Z^{(n)})\leq \bar{\epsilon}$. The last term in~\eqref{tk0} can be bounded similarly. Since on~$B_*$, the second inequality in~\eqref{meco} also holds (with~$\epsilon$ therein replaced by~$\bar{\epsilon}$), and we have~$\bar{A}\leq A$ (which implies~$\bar{\pi}(\cdot|Z^{(n)})\leq \pi(\cdot|Z^{(n)})$ on~$\cap_{i=1}^n\{\theta\in\mathbb{R}^d:\theta^{\top}X_i< r\}$), by a similar maximal coupling argument, 
it holds on~$B_*$ that
\begin{equation*}
 \textrm{TV}(\bar{\pi}(\cdot|Z^{(n)}),\pi(\cdot|Z^{(n)})) \leq \bar{\epsilon}.
\end{equation*}
Therefore substituting into~\eqref{tk0}, 
on~$B_*$, we have~$\textrm{TV}(\widetilde{\pi}(\cdot|Z^{(n)}),\pi(\cdot|Z^{(n)}))
\leq 2\bar{\epsilon}$, 
which implies by~\eqref{tkm}
that
\begin{equation*}
\textrm{TV}(\delta_{\hat{\theta}_0}\hat{P}^k, \pi(\cdot|Z^{(n)})) \leq 
5k\bar{\epsilon} +R_k,
\end{equation*}
where we have used~$k\geq 1$ by our condition on~$k$. 
By definition~\eqref{rkl} of~$R_k$, under
\begin{equation*}
k \geq \frac{2\widetilde{L}d}{\tilde{c}(2R)n}\ln\bigg(\frac{1}{\bar{\epsilon}}\bigg(1+\frac{d}{4}\ln\bigg(\frac{2d\bar{L}}{\tilde{c}(2R)n}\bigg)\bigg)^{1/2}\,\bigg)
\end{equation*}
(which is our condition on~$k$ in the assertion by definition~\eqref{dbri} of~$\tilde{c}(2R)$), we have~$R_k\leq \bar{\epsilon}\leq k\bar{\epsilon}$. This implies~$\textrm{TV}(\delta_{\theta_0}P^k, \pi(\cdot|Z^{(n)})) = \textrm{TV}(\delta_{\hat{\theta}_0}\hat{P}^k, \pi(\cdot|Z^{(n)})) \leq 6k\bar{\epsilon}= \epsilon$ as required.
\end{proof}

\begin{appendix}

\section{Algorithms}\label{Ago}
We state the Random Scan Gibbs sampler (from e.g.~\cite{ascolani2024e}) in Algorithm~\ref{alg1} and the HMC algorithm (from~\cite{bourabee2022unadjusted}) in Algorithm~\ref{alg2} both when targeting~$\pi(\cdot|Z^{(n)})$. These are the algorithms giving the announced complexities.  
\begin{algorithm}
\caption{Random Scan Gibbs sampler}\label{alg1}
\begin{algorithmic}
\State{Fix~$N\in\mathbb{N}$ and initialize~$\theta_0\in\mathbb{R}^d$}
\For{$k=0,1,\dots,N-1$}
\State{Draw~$i\sim\textrm{Unif}([1,d]\cap\mathbb{N})$;}
\State{Draw~$\bar{\theta}_i\sim (\pi(\cdot|Z^{(n)}))_i(\cdot,(\theta_k)_{-i})/\int_{\mathbb{R}}(\pi(\cdot|Z^{(n)}))_i(\cdot,(\theta_k)_{-i})$;}
\State{Set~$\theta_{k+1}=((\theta_k)_1,\dots,(\theta_k)_{i-1},\bar{\theta}_i,(\theta_k)_{i+1},\dots,(\theta_k)_d)$;}
\EndFor
\end{algorithmic}
\end{algorithm}

The conditional draw in Algorithm~\ref{alg1} may for example be implemented with Rejection Sampling (RS); we refer the reader to~\cite[Section~4.1]{ascolani2024e} for details. 
Note that part of our analysis is to show under~$n\gtrsim d$ that the tails of the conditional distributions have negligible weight, and that on the complement of the tails, the distribution is strongly log-concave, so RS is expected to be efficient for both the Gaussian and non-Gaussian priors considered in this paper. 
The iterates~$(\theta_k)_{k\in\mathbb{N}}$ in Algorithm~\ref{alg1} are also defined for an arbitrary target negative log-density~$U$ at the beginning of Section~\ref{gibsec}. 
\begin{algorithm}
\caption{Unadjusted Hamiltonian Monte Carlo}\label{alg2}
\begin{algorithmic}
\State{Fix~$N\in\mathbb{N}$,~$h>0$,~$T\in h\mathbb{N}$ and initialize~$\theta_0\in\mathbb{R}^d$}
\For{$k=0,1,\dots,N-1$}
\State{Draw~$v\sim N(0,I_d)$}
\State{Set~$q_0=\theta_k$ and~$p_0=v$;}
\For{$i=0,1,\dots,T/h-1$}
\State{Draw~$U\sim \textrm{Unif}(0,h)$;}
\State{Set~$F=\nabla \ln\pi(q_i + Up_i|Z^{(n)})$;}
\State{Set~$q_{i+1} = q_i + hp_i + h^2F/2$;}
\State{Set~$p_{i+1} = p_i + hF$;}
\EndFor
\State{Set~$\theta_{k+1} = q_{T/h}$;}
\EndFor
\end{algorithmic}
\end{algorithm}

In our analysis of HMC in Section~\ref{HMC}, we accommodate also a Verlet integrator which is not stated explicitly in this section. Note that comparable complexity analysis of more dynamic implementations of HMC~\cite{durmus2024co} (including the No U-Turn Sampler) does not (yet) exist even in the strongly convex case, to the best of our knowledge. Again, the iterates~$(\theta_k)_{k\in\mathbb{N}}$ in Algorithm~\ref{alg2} are also defined for an arbitrary target negative log-density~$U$ at the beginning of Section~\ref{HMC} (it is the case~$u=1$).

\section{Proofs for Section~\ref{OGLM}}\label{ovp}
The results in Section~\ref{OGLM} are corollaries to statements in the other sections. This section provides the necessary bookkeeping.

In the rest of this section, we say a scalar value is a semi-absolute constant if it depends only on~$\sigma$ (polynomially) for linear regression and is an absolute constant for logistic regression. We also call it generic if it changes possibly from line to line. Throughout this section, we will use the assumption that~$C_{\pi}/\lambda_{\max}(\Sigma)$,~$\abs{\theta^*}^2\lambda_{\max}(\Sigma)$ are bounded by the absolute constant~$C^*$ possibly without mention.
\begin{proof}[Proof of Proposition~\ref{kqd}]
The Hessian of $U$ is 
\begin{align}
D^2 U(\theta)=Q+X^{\top}D(\theta)X
\end{align}
for all~$\theta\in\mathbb{R}^d$, 
with~$D(\theta)$ diagonal and~$D_{ii}(\theta)=A''(X_i^{\top}\theta)$. Recall that~$A''(z)\in(0,1/4]$ for all~$z\in\R$,~$A''(0)=1/4$ and~$A''(z)\to 0$ as~$|z|\to\infty$, see e.g.~\cite{anceschi2024optimal} and references therein.
The Hessian of~$\tilde U$ is 
\begin{equation*}
D^2\tilde{U}(\theta)=I_d+Q^{-1/2}X^{\top}D(Q^{-1/2}\theta)XQ^{-1/2}
\end{equation*}
for all~$\theta\in\mathbb{R}^d$. 
Consider first~$\inf_{\theta\in\mathbb{R}^d}\lambda_{\min}(D^2\tilde{U}(\theta))$. 
Since~$\lambda_{\min}(D(\theta))\geq 0$, one has~$\lambda_{\min}(D^2\tilde{U}(\theta))\geq 1$ for all $\theta$.
Since $n<\infty$, and 
by the fact that~$X_i$ are drawn from~$P_X$ with a density w.r.t.\ the Lebesgue measure and so~$X_i\neq0$ a.s., 
there exists~$\theta_0\in\R^d\setminus\{0\}$ such that it holds a.s.\ that~$X_i^{\top}\theta_0\neq 0$ for all $i\in\mathbb{N}\cap[1,n]$. 
Thus 
$$
\inf_{\theta\in\R^d}\lambda_{\min}(D^2\tilde{U}(\theta))\leq
\lim_{r\to\infty}
\lambda_{\min}(D^2\tilde{U}(r\theta_0))
=\lambda_{\min}(I_d)=
1,
$$
where we have used~$\lim_{r\to\infty}
D_{ii}(r\theta_0)=0$ for all $i$, which follows from $\lim_{r\to\infty}
|x_i^{\top}r\theta_0|=\infty$ and $\lim_{|z|\to\infty}A''(z)= 0$. 
Thus $\inf_{\theta\in\R^d}\lambda_{\min}(D^2\tilde{U}(\theta))=1$.

Consider now~$\sup_{\theta\in\mathbb{R}^d}\lambda_{\max}(D^2\tilde{U}(\theta))$. One has~$\lambda_{\max}(D(\theta))\leq 1/4$ for all~$\theta$ by~$A''(z)\in(0,1/4]$ for all~$z$, which implies~$\lambda_{\max}(D^2\tilde{U}(\theta))\leq 1+\lambda_{\max}(Q^{-1/2}X^{\top}XQ^{-1/2})/4$ 
for all $\theta$.
We also have~$\lambda_{\max}(D^2\tilde{U}(0))=1+\lambda_{\max}(Q^{-1/2}X^{\top}XQ^{-1/2})/4$, because $D(0)= I_{n}/4$ by~$A''(0)=1/4$.
\end{proof}

\begin{proof}[Proof of Lemma~\ref{mainmap}]
Let~$R>0$ be given by~\eqref{rhdef2}. 
By 
Lemma~\ref{uuc} with~$\bar{A}=A$,~$\hat{r}=\infty$, 
the conclusion holds 
given~\eqref{nmw}, where~$\hat{c}_3,c^*
$ are defined in the same lemma. 
It remains to show that the conditions on~$n$ in the assertion 
imply~\eqref{nmw}. Below we denote by~$C$ a generic semi-absolute constant. 
We upper bound the prefactor~$c^*
$ in~\eqref{nmw}. Consider first the linear regression case. 
By the definition~\eqref{rhdef2} of~$R$ and the assumption on~$\abs{\theta^*}^2\lambda_{\max}(\Sigma)$, we have
\begin{equation*}
c^*(c_3,\delta) \leq C\ln(\kappa(\Sigma)+n/\delta)\leq C\ln(2\max(\kappa(\Sigma),n/\delta))\leq C(1+\ln(\kappa(\Sigma)) + \ln n + \ln(\delta^{-1})).
\end{equation*}
Thus a sufficient condition for~\eqref{nmw} is
\begin{equation*}
n\geq C(1+\ln\kappa(\Sigma) + \ln n +\ln(\delta^{-1}))(\kappa(\Sigma))^2(d + \ln(\delta^{-1})).
\end{equation*}
Denoting~$\tilde{C} = C(\kappa(\Sigma))^2(d+\ln(\delta^{-1}))$, this is equivalent to
\begin{align*}
n&\geq \tilde{C}(1+\ln\kappa(\Sigma) + \ln(\delta^{-1}) +\ln\tilde{C}) +  \tilde{C}\ln\big(n/\tilde{C}\big),
\end{align*}
for which, by~$\ln z\leq z/2$ for all~$z>0$, a sufficient condition is
\begin{equation*}
n\geq \tilde{C}(1+\ln\kappa(\Sigma) + \ln(\delta^{-1}) +\ln\tilde{C}) +  n/2.
\end{equation*}
This is satisfied if
\begin{equation*}
n\geq C(\kappa(\Sigma))^2(d+\ln(\delta^{-1}))(1+\ln\kappa(\Sigma) + \ln(\delta^{-1}) + \ln d)
\end{equation*}
as required. For the logistic case~\eqref{loge}, we have~$c^*\leq C$ 
by the same calculations as above, using~$\phi=1$ and~$\sup_{\mathbb{R}^d}A''\leq C$.
\end{proof}

\begin{proof}[Proof of Theorem~\ref{gibbs1}]
Let~$R$ be given by~\eqref{rhdef2} and let~$\chi:\mathbb{R}^d\rightarrow[0,1]$ be defined by~\eqref{chidef}. 
By Lemma~\ref{uuc} (with~$\hat{r}=\infty$
therein), there exists an event~$B$ with~$\mathbb{P}(B)\geq 1-\delta-\rho_{\pi}$ such that 
if~\eqref{nmw} holds, 
then, on~$B$, there exists~$\theta_{\textrm{map}}\in B_{(\phi/\lambda_{\max}(\Sigma))^{1/2}}(\theta^*)$ with~$\nabla \ln \pi(\cdot |Z^{(n)})|_{\theta_{\textrm{map}}}=0$ and the functions~$U,U_0$ given by~\eqref{uuo} satisfy Conditions~\ref{cond:curvature3},~\ref{cond:smooth},~\ref{cond:curvature2} with~$m_0,L,m$ given by~\eqref{molm}. We take a slightly enlarged~$L=C_{\pi}d + 9\phi^{-1}n\lambda_{\max}(\Sigma)$ in the sequel for simplicity. Note that this event~$B$ is indeed the same event as in Lemma~\ref{mainmap} by inspection of the proof.

By definition~\eqref{cbb} of~$L_{\textrm{sc}}$, the enlarged~$L$ 
satisfies
\begin{equation}\label{lig}
L = C_{\pi}d + 9\phi^{-1}n\lambda_{\max}(\Sigma)= L_{\textrm{sc}}\lambda_{\max}(\Sigma)n.
\end{equation}
Note that if~$n\geq d$, then
\begin{equation}\label{bph}
9\phi^{-1}\leq L_{\textrm{sc}} \leq 9\phi^{-1} + C_0
\end{equation}
for an absolute constant~$C_0>0$, so that~$L_{\textrm{sc}}$ may be regarded as a semi-absolute constant. 

With these values of~$L,m,m_0$, we estimate~$r_*,\kappa,\kappa_0,c_U$ with~$\lambda=\lambda_{\max}(\Sigma)$ 
in Corollary~\ref{rak}, before applying the corollary to obtain our complexity result. By~\eqref{lig} and from the expression~\eqref{rdt} for~$r_*$, we have
\begin{equation*}
r_* = (L_{\textrm{sc}}\lambda_{\max}(\Sigma))^{-1/2} + d^{1/2}(\pi L_{\textrm{sc}} \lambda_{\max}(\Sigma) n)^{-1/2}.
\end{equation*}
If~$n\geq 2d$, then by~\eqref{bph} we have~$r_*\leq (1/2)(\phi/\lambda_{\max}(\Sigma))^{1/2}$, so that 
\begin{equation*}
\frac{L}{m_0(r_*)} = \frac{L}{m(2r_*)} \leq \frac{L_{\textrm{sc}}\lambda_{\max}(\Sigma)n}{(6\phi)^{-1}\lambda_{\min}(\Sigma)\cdot \inf_{B_{2R(\lambda_{\max}(\Sigma))^{1/2}}}A''\cdot n/2}\leq \bar{C}\kappa(\Sigma)
\end{equation*}
for~$\bar{C}$ satisfying~\eqref{cba}, where we note~$\bar{C}$ is a semi-absolute constant.  
Moreover, by definition~\eqref{uuo2} of~$U_0=\bar{U}_0$ and the assumption~$0\in\argmax\pi$, we have~$U_0=-\ln\pi(0)=\min_{\mathbb{R}^d}U_0$ on~$B_{r_*}\subset B_{(\phi/\lambda_{\max}(\Sigma))^{1/2}}$. This implies~$c_U$ as in~\eqref{cud} satisfies~$c_U=1$. Lastly, we consider the assumption~$r_*^2m_0(r_*)\geq 9d$ in Corollary~\ref{rak}. We have~$r_*^2m_0(r_*)\geq (L_{\textrm{sc}}\lambda_{\max}(\Sigma))^{-1}\hat{c}(R)n/2$, with~$\hat{c}(R)$ given by~\eqref{chdef}. Thus using~\eqref{bph}, the definition~\eqref{rhdef2} of~$R$ and the assumption on~$\abs{\theta^*}^2\lambda_{\max}(\Sigma)$, it suffices to have~$n\geq C \kappa(\Sigma)d$ for large enough semi-absolute constant~$C$, which is easily satisfied given~\eqref{ngib}.

Combining into Corollary~\ref{rak}, setting~$N$ with~\eqref{ngm}, 
we have that if~\eqref{nmw} and 
\begin{equation}\label{ntc2}
n\geq (8/3)\bar{C}\kappa(\Sigma)\big[2d+(d/2)\ln(5\bar{C}\kappa(\Sigma)) + \ln(30N\epsilon^{-1})\big]
\end{equation}
hold, then~\eqref{weg} holds on~$B$. By our choice~\eqref{ngm} of~$N$, it follows that~\eqref{ngib} (with large enough semi-absolute constant~$C$) implies~\eqref{ntc2}.
Moreover, we have checked that in the proof of Lemma~\ref{mainmap} that~\eqref{ngib} (with large enough semi-absolute constant~$C$) 
implies~\eqref{nmw}. 
\end{proof}

\begin{proof}[Proof of Theorem~\ref{mainhmc}]
We first introduce exactly the same objects as in the beginning of the proof of Theorem~\ref{gibbs1}. Let~$R$ be given by~\eqref{rhdef2} and let~$\chi:\mathbb{R}^d\rightarrow[0,1]$ be defined by~\eqref{chidef}. 
By Lemma~\ref{uuc} (with~$\hat{r}=\infty$), 
there exists an event~$B$ with~$\mathbb{P}(B)\geq 1-\delta-\rho_{\pi}$ such that 
if~\eqref{nmw} holds, 
then, on~$B$, there exists~$\theta_{\textrm{map}}\in B_{(\phi/\lambda_{\max}(\Sigma))^{1/2}}(\theta^*)$ with~$\nabla \ln \pi(\cdot |Z^{(n)})|_{\theta_{\textrm{map}}}=0$ and the functions~$U,U_0$ given by~\eqref{uuo} satisfy Conditions~\ref{cond:curvature3},~\ref{cond:smooth},~\ref{cond:curvature2} with~$m_0,L,m$ given by~\eqref{molm}. Note that this event~$B$ is indeed the same event as in Lemma~\ref{mainmap} by inspection of the proof. Again for~$n\geq d$ we have~\eqref{lig} and~\eqref{bph} for some absolute constant~$C_0>0$ and~$L_{\textrm{sc}}$ given by~\eqref{cbb}. 

We check the assumptions of Corollary~\ref{cor1} with~$\lambda=\lambda_{\max}(\Sigma)$, and subsequently apply the corollary to obtain our complexity result. 
By~\eqref{bph}, we have~$L_{\textrm{sc}}\lambda_{\max}(\Sigma)\geq 9\phi^{-1}\lambda_{\max}(\Sigma)$, so that the arguments within~\eqref{kaps} satisfy
\begin{equation*}
\max((8/3)(L_{\textrm{sc}}\lambda)^{-1/2}, (4/3)(L_{\textrm{sc}}\lambda)^{-1/2}, 2(L_{\textrm{sc}}\lambda)^{-1/2})\leq (\phi/\lambda_{\max}(\Sigma))^{1/2}.
\end{equation*}
Together with~\eqref{molm}, this implies in~\eqref{kaps} that
\begin{equation*}
\max(\kappa,\kappa_0,\bar{\kappa})\leq \frac{L_{\textrm{sc}}\lambda_{\max}(\Sigma)n}{(6\phi)^{-1}\lambda_{\min}(\Sigma)\cdot \inf_{B_{2R(\lambda_{\max}(\Sigma))^{1/2}}}A''\cdot n/2}\leq \bar{C}\kappa(\Sigma),
\end{equation*}
where~$\bar{C}$ is given by~\eqref{cba}. Moreover, repeating an argument from the proof of Theorem~\ref{gibbs1}, by definition~\eqref{uuo2} of~$U_0=\bar{U}_0$ and the assumption~$0\in\argmax\pi$, we have~$U_0=-\ln\pi(0)=\min_{\mathbb{R}^d}U_0$ on~$B_{r_*}\subset B_{(\phi/\lambda_{\max}(\Sigma))^{1/2}}$, where~$r_*=(4/3)(L_{\textrm{sc}}\lambda)^{-1/2}$. This implies~$(\sup_{\mathbb{R}^d\setminus B_{r_*}}e^{-U_0})(\inf_{B_{r_*}}e^{-U_0})^{-1}=1$.

Let~$\hat{\epsilon}\in(0,1)$. 
Gathering the above into Corollary~\ref{cor1} with
\begin{equation*}
\lambda=\lambda_{\max}(\Sigma),\qquad \epsilon 
= \hat{\epsilon}/(\lambda n)^{1/2},\qquad \abs{\theta_0}\leq (L_{\textrm{sc}}\lambda_{\max}(\Sigma))^{-1/2}
\end{equation*}
and~\eqref{hmcset} (with~$\epsilon$ replaced by~$\hat{\epsilon}$ therein) 
yields by~\eqref{bph} that if~\eqref{nmw} and
\begin{equation}\label{nj1}
n \geq \hat{C}\kappa(\Sigma)^2(d+\kappa(\Sigma)(\ln n + \ln(1/\hat{\epsilon})))
\end{equation}
hold for a large enough semi-absolute constant~$\hat{C}>0$, then it holds on~$B$ that
\begin{equation*}
(\lambda n)^{1/2} W_2(\nu_N^{\textrm{HMC}},\pi(\cdot|Z^{(n)}))\leq \hat{\epsilon}.
\end{equation*}
Note that since for any~$r>1$,~$n\geq 2r\ln r$ implies~$n\geq r\ln n$, if
\begin{equation}\label{nct2}
n\geq C\max\big(\kappa(\Sigma)^3\ln\kappa(\Sigma),\kappa(\Sigma)^2(d+\kappa(\Sigma) \ln(1/\hat{\epsilon}))\big),
\end{equation}
then~$n\geq 2\hat{C}\kappa(\Sigma)^3\ln n$ for large enough semi-absolute~$C$, and so~\eqref{nct2} implies~\eqref{nj1} (after possibly increasing~$C$). 
Therefore there exists a semi-absolute constant~$C>0$ such that if~\eqref{nct2} holds, then it holds on~$B$ that~$(\lambda n)^{1/2}W_2(\nu_N^{\textrm{HMC}},\pi(\cdot|Z^{(n)}))\leq \hat{\epsilon}$. 
The proof concludes by noting that~\eqref{nhmc} implies both~\eqref{nct2} and the condition on~$n$ in Lemma~\ref{mainmap} (which 
implies~\eqref{nmw}, as shown in the proof of Lemma~\ref{mainmap}).
\end{proof}

\begin{proof}[Proof of Lemma~\ref{mainmap2}]
We apply Theorem~\ref{poiwei}\ref{pot0} with~$r=r^*(1/2,\delta)$,~$\hat{r}=\infty$ and~$\epsilon=1/2$. This application of Theorem~\ref{poiwei} shows the conclusion holds given
\begin{equation}\label{ntz}
n\geq c^*(c_4,\delta)\max(\kappa(\Sigma)^2,C_{\pi}(\lambda_{\min}(\Sigma))^{-1})(d+\ln(\delta^{-1})).
\end{equation}
The rest of the proof is to show that the condition on~$n$ in the assertion is sufficient. By the assumptions on~$\abs{\theta^*}^2\lambda_{\max}(\Sigma)$,~$C_{\pi}/\lambda_{\max}(\Sigma)$ and by definition~\eqref{por} of~$c^*$, condition~\eqref{ntz} 
is satisfied if
\begin{equation}\label{njk}
n\geq \hat{C}\exp\big(\hat{C}((\ln\kappa(\Sigma))^{1/2}+(\ln n)^{1/2}+(\ln(\delta^{-1}))^{1/2})\big) 
\cdot 
\kappa(\Sigma)^2
(d+\ln(\delta^{-1})).
\end{equation}
for a large enough absolute constant~$\hat{C}\geq 1$. 
We resolve the~$\exp(\hat{C}(\ln n)^{1/2})$ factor on the right-hand side. Let~$x>0$ and let~$g:\mathbb{R}\rightarrow\mathbb{R}$ be given by
\begin{equation*}
g(y)=(y/\hat{C})^2-y-x.
\end{equation*}
The function~$g$ is increasing for~$y>\hat{C}^2/2$, and its roots are~$(\hat{C}^2/2)(1\pm(1+4x/\hat{C}^2)^{1/2})$. Thus
\begin{equation*}
y\geq \hat{C}^2(1+x^{1/2}/\hat{C}) \geq (\hat{C}^2/2)(1+(1+4x/\hat{C}^2)^{1/2})\implies g(y)\geq 0.
\end{equation*}
Fix
\begin{equation*}
x=\ln(\hat{C})+\hat{C}(\ln(\kappa(\Sigma)))^{1/2} +\hat{C}(\ln(\delta^{-1}))^{1/2} + \ln( \kappa(\Sigma)^2 (d+\ln(\delta^{-1})))
\end{equation*}
and~$y=\ln n-x$. We have that if
\begin{equation}\label{njh1}
\ln n - x\geq \hat{C}^2(1+x^{1/2}/\hat{C}),
\end{equation}
then~$g(\ln n - x)\geq 0$, which is~\eqref{njk} (after straightforward manipulations). 
Let~$\bar{P}:[1,\infty)^2\rightarrow[1,\infty)$ be given by
\begin{equation}\label{pbm}
\bar{P}(z,\bar{C})=e^{\bar{C}^2}z\exp(\bar{C}(\ln z)^{1/2}) \qquad\forall (z,\bar{C})\in[1,\infty)^2.
\end{equation}
Inequality~\eqref{njh1} is~$n\geq e^{x+\hat{C}^2+\hat{C}x^{1/2}}$, which can be rewritten as
\begin{equation*}
n\geq \bar{P}(e^x,\hat{C})= \bar{P}\big(\hat{C}\exp\big(\hat{C}(\ln\kappa(\Sigma))^{1/2}+\hat{C}(\ln(\delta^{-1}))^{1/2}\big) \cdot \kappa(\Sigma)^2(d+\ln(\delta^{-1})),\hat{C}\big).
\end{equation*}
By the definition of~$\bar{P}$, and the bound
\begin{align*}
\exp(\hat{C}x^{1/2}) &\leq \hat{C}\exp\big(\hat{C}[\ln(\kappa(\Sigma))]^{1/2} + \hat{C}[\ln(\delta^{-1})]^{1/2} + \hat{C}[\ln(\kappa(\Sigma)(d+\ln\delta^{-1}))]^{1/2}\big)\\
&\leq \hat{C}\exp\big(\hat{C}[\ln(\kappa(\Sigma)d\delta^{-1})]^{1/2}\big),
\end{align*}
where we have possibly increased~$\hat{C}$ on the right-hand sides, 
the proof concludes.
\end{proof}

\begin{proof}[Proof of Theorem~\ref{mainpoi}]
We use Theorem~\ref{poigib}. 
Below we set~$\bar{C}>0$ to be an absolute constant that may change from line to line. 
By our assumption that~$\abs{\theta^*}^2\lambda_{\max}(\Sigma)$,~$C_{\pi}/\lambda_{\max}$ are bounded by an absolute constant, if~$n\geq d$, then the quantities within~$\hat{r}$
in Theorem~\ref{poigib} satisfy
\begin{align*}
y_{\epsilon^2\delta/(15120e^2\cdot N^2n^5)}^* &\leq \bar{C}e^{\bar{C}(\ln(nN\delta^{-1}\epsilon^{-1}))^{1/2}},\\
r^*\bigg(\frac{\epsilon}{6N},\frac{\epsilon^2\delta}{3024e^2\cdot N^2n^4}\bigg) &\leq \bar{C} (\ln(nN/(\epsilon\delta)))^{1/2},\\
\textrm{and so}\qquad\hat{r} &\leq \bar{C} (\ln(nN/(\epsilon\delta)))^{1/2}.
\end{align*}
Moreover,~$\tilde{L}$ as in~\eqref{lry2} (with~$C$ replaced by~$\bar{C}$) is lower bounded by~$\widetilde{L}$ as in Theorem~\ref{poigib}. 
Thus Theorem~\ref{poigib} asserts that 
if~\eqref{fea2} holds with~$\bar{W}\sim N(0,\tilde{L}^{-1}I_d)$ independent of all other random elements, where~$\tilde{L}$ is again given by~\eqref{lry2} with~$C$ replaced by~$\bar{C}$, 
and if
\begin{align}
N &\geq \hat{C}e^{\hat{C}(\ln(nN/(\epsilon\delta)))^{1/2}}d\kappa(\Sigma)\nonumber\\
&\quad\cdot\bigg[\ln\bigg(\frac{N}{\epsilon}\bigg) + \ln\bigg(1+d\ln\bigg(de^{\hat{C}(\ln(nN/(\epsilon\delta)))^{1/2}}\kappa(\Sigma)\bigg)\bigg) \bigg],\label{Nj1}\\
n &\geq \hat{C}e^{\hat{C}
(\ln (nN/(\epsilon\delta)))^{1/2}
} 
\cdot\max(\kappa(\Sigma)^2d, \kappa(\Sigma)N)\label{nj2}
\end{align}
hold for a large enough constant~$\hat{C}\geq 1$,
then~\eqref{weg} holds with probability~$1-\delta
$. 

It remains to verify that our assumptions imply~\eqref{Nj1} and~\eqref{nj2}. Note that by our assumption~\eqref{nqw} on~$N$, it suffices to prove that our assumptions imply~\eqref{Nj1} and
\begin{equation}\label{nj2p}
n\geq \hat{C}e^{\hat{C}(\ln(nN/(\epsilon\delta)))^{1/2}}\cdot \kappa(\Sigma)^2d
\end{equation}
after possibly increasing~$\hat{C}$. 
We use arguments similar to those from the proof of Lemma~\ref{mainmap2} to resolve first the fact that~$N$ appears on the right-hand side of~\eqref{Nj1}. Subsequently, we resolve the fact that~$N$ and~$n$ appear on the right-hand side of~\eqref{nj2p}, using the same tricks. 
Let~$P:[1,\infty)^2\rightarrow[1,\infty)$ be given by
\begin{equation}\label{pob}
P(z,\bar{C}')=\bar{C}'z\exp(\bar{C}'(\ln z)^{1/2}) \qquad\forall (z,\bar{C}')\in[1,\infty)^2.
\end{equation}
Let~$x>0$ and let~$g:\mathbb{R}\rightarrow\mathbb{R}$ be given by
\begin{equation*}
g(y)=(y/\hat{C})^2-y-x.
\end{equation*}
The function~$g$ is increasing for~$y>\hat{C}^2/2$, and its roots are~$(\hat{C}^2/2)(1\pm(1+4x/\hat{C}^2)^{1/2})$. Thus if~$y\geq \hat{C}^2(1+x^{1/2}/\hat{C}) \geq (\hat{C}^2/2)(1+(1+4x/\hat{C}^2)^{1/2})$, then~$g(y)\geq 0$. For~\eqref{Nj1}, we estimate first the terms in the square brackets. We have~$\ln(N/\epsilon)\leq e^{\bar{C}(\ln(nN/(\epsilon\delta)))^{1/2}}$. Moreover, if~$n\geq 2$, then
\begin{equation*}
e^{\hat{C}(\ln(nN/(\epsilon \delta)))^{1/2}}\geq 2,
\end{equation*}
so that
\begin{align*}
&\ln\big(1+d\ln\big(de^{\hat{C}(\ln(nN/(\epsilon\delta)))^{1/2}}\kappa(\Sigma)\big)\big) \\
&\quad\leq \bar{C}\ln\big(d\ln\big(de^{\hat{C}(\ln(nN/(\epsilon\delta)))^{1/2}}\kappa(\Sigma)\big)\big)\\
&\quad\leq \bar{C} \ln (d^2 + d\hat{C}(\ln(nN/(\epsilon\delta)))^{1/2} + d\ln\kappa(\Sigma))\\
&\quad\leq \bar{C} \big(\ln d + \ln (1+\hat{C}(\ln(nN/(\epsilon\delta)))^{1/2}+ \ln\kappa(\Sigma) )\big)\\
&\quad\leq \bar{C}\big(\ln d + \ln(1+2\max(\hat{C}(\ln(nN/(\epsilon\delta)))^{1/2},\ln\kappa(\Sigma)))\big)\\
&\quad\leq \bar{C}\big(\ln d + \ln(1+2\hat{C}(\ln(nN/(\epsilon\delta)))^{1/2}) + \ln(1+\ln\kappa(\Sigma))\big)\\
&\quad\leq \bar{C}\big(\ln d + \exp(\bar{C}(\ln(nN/(\epsilon\delta)))^{1/2}) + \ln(1+\ln\kappa(\Sigma))\big)\\
&\quad\leq \bar{C} \exp(\bar{C}(\ln(nN/(\epsilon\delta)))^{1/2}) \big(1+\ln d + \ln(1+\ln\kappa(\Sigma))\big).
\end{align*}
Substituting into~\eqref{Nj1}, we have that a sufficient condition for~\eqref{Nj1} is
\begin{equation*}
N\geq\bar{C}e^{\bar{C}(\ln(nN/(\epsilon\delta)))^{1/2}}d\kappa(\Sigma) \big(1+\ln d + \ln(1+\ln\kappa(\Sigma))\big)
\end{equation*}
and a sufficient condition for this is
\begin{equation}\label{kaw}
N\geq\hat{C}e^{\hat{C}((\ln(n/(\epsilon\delta)))^{1/2} + (\ln N)^{1/2})}d\kappa(\Sigma) \big(1+\ln d + \ln(1+\ln\kappa(\Sigma))\big)
\end{equation}
for a large enough absolute constant~$\hat{C}\geq 1$. Using the function~$g$ above, and setting 
\begin{equation*}
x=\ln\big( \hat{C}e^{\hat{C}(\ln(n/(\epsilon\delta)))^{1/2} }d\kappa(\Sigma) \big(1+\ln d + \ln(1+\ln\kappa(\Sigma))\big)\big)
\end{equation*}
with~$y=\ln N - x$, if
\begin{equation}\label{nou}
\ln N - x\geq \hat{C}^2(1+x^{1/2}/\hat{C}),
\end{equation}
then~$g(\ln N-x)\geq 0$, which is~\eqref{kaw}. Therefore a sufficient condition for~\eqref{kaw} (thus~\eqref{Nj1}) is
\begin{equation}\label{nqw2}
N\geq e^{\hat{C}^2}P(e^x,\hat{C})/\hat{C},
\end{equation}
where~$P$ is given by~\eqref{pob}. Note that for large enough~$C$,~\eqref{nqw} implies~\eqref{nqw2} by the computation
\begin{equation*}
\exp(\hat{C}(\ln(e^x))^{1/2}) \leq \exp(\hat{C}x^{1/2})\leq \exp(\bar{C}[\ln(nd\kappa(\Sigma)\epsilon^{-1}\delta^{-1})]^{1/2}).
\end{equation*}
We fix such a~$C>1$. 
For~\eqref{nj2p}, we use again the function~$g$ above, but set
\begin{equation*}
x=\ln\big(\hat{C}\exp(\hat{C}[\ln(N\epsilon^{-1}\delta^{-1})]^{1/2})\cdot \kappa(\Sigma)^2 \cdot d
\big)
\end{equation*}
and~$y=\ln n-x$. We have that if
\begin{equation}\label{njh2}
\ln n - x\geq \hat{C}^2(1+x^{1/2}/\hat{C}),
\end{equation}
then~$g(\ln n - x)\geq 0$, which implies~\eqref{nj2p}. Inequality~\eqref{njh2} is~$n\geq e^{x+\hat{C}^2+\hat{C}x^{1/2}}$. 
Note that since~$x\leq \bar{C}\ln(\kappa(\Sigma)dN\delta^{-1}\epsilon^{-1})$, the right-hand side satisfies
\begin{align*}
e^{x+\hat{C}^2+\hat{C}x^{1/2}}&\leq\bar{C}\exp\big(
\bar{C}[\ln(N\epsilon^{-1}\delta^{-1})]^{1/2}\big)\cdot \kappa(\Sigma)^2\cdot d
\cdot \exp(\bar{C}[\ln(\kappa(\Sigma)dN\delta^{-1}\epsilon^{-1})]^{1/2})\\
&\leq\bar{C}\exp(\bar{C}[\ln(\kappa(\Sigma)dN\delta^{-1}\epsilon^{-1})]^{1/2})\cdot \kappa(\Sigma)^2\cdot d\\
&\leq \bar{C}\exp(\bar{C}[\ln(\kappa(\Sigma)d\delta^{-1}\epsilon^{-1})]^{1/2})\cdot \exp(\bar{C}[\ln N]^{1/2})\cdot\kappa(\Sigma)^2\cdot d
\end{align*}
Therefore if
\begin{equation}\label{jnk}
n\geq \bar{C}\kappa(\Sigma)^2d  e^{\bar{C}(\ln(\kappa(\Sigma)d\delta^{-1}\epsilon^{-1}))^{1/2}}\cdot e^{\bar{C}(\ln N)^{1/2}},
\end{equation}
then~\eqref{nj2p} holds. 
Given~\eqref{nqw} and~$n\geq \bar{C}$, 
we have
\begin{align*}
\ln N&\leq \bar{C}\big([\ln (d\kappa(\Sigma)n\epsilon^{-1}\delta^{-1})]^{1/2} + \ln (d\kappa(\Sigma))\big)\\
&\leq \bar{C}\big(\ln n + \ln (d\kappa(\Sigma)\epsilon^{-1}\delta^{-1})\big),
\end{align*}
so a sufficient condition for~\eqref{jnk} is
\begin{equation}\label{kpo2}
n\geq \hat{C}\kappa(\Sigma)^2d \cdot e^{\hat{C}(\ln(\kappa(\Sigma)d\delta^{-1}\epsilon^{-1}))^{1/2}}\cdot e^{\hat{C}(\ln n)^{1/2}}
\end{equation}
for a large enough absolute constant~$\hat{C}\geq 1$. 
Therefore using again the function~$g$ above with the updated~$\hat{C}$ and with
\begin{equation*}
x=\ln\big(\hat{C}\kappa(\Sigma)^2d \cdot e^{\hat{C}(\ln(\kappa(\Sigma)d\delta^{-1}\epsilon^{-1}))^{1/2}}\big)
\end{equation*}
and~$y=\ln n-x$, if~\eqref{njh2} holds, then~$g(\ln n-x)\geq 0$, which is~\eqref{kpo2}. Inequality~\eqref{njh2} is~$n\geq e^{x+\hat{C}^2+\hat{C}x^{1/2}}$, for which a sufficient condition is~\eqref{nsd0} 
for large enough~$C$, since
\begin{equation*}
e^{\hat{C}x^{1/2}}\leq e^{\hat{C}[\ln\hat{C} + 2\ln\kappa(\Sigma) + \ln d + \hat{C}[\ln(\kappa(\Sigma)d\delta^{-1}\epsilon^{-1})]^{1/2}]^{1/2}} 
\leq \bar{C}e^{\bar{C}[\ln(\kappa(\Sigma)d\delta^{-1}\epsilon^{-1})]^{1/2}}.\qedhere
\end{equation*} 
\end{proof}

\section{Auxiliary results}\label{Axe}
In this section, we state some results about the Hamiltonian integrators appearing in the HMC algorithm from the beginning of Section~\ref{HMC}. These are variations of some results in the literature. The next Lemma~\ref{veli} shows that the Verlet integrator for Hamiltonian dynamics is Lipschitz (with respect to an appropriately scaled metric). We assume the notation in Section~\ref{HMC}, in particular that~$u=0$ describes the Verlet integrator (as opposed to~$u=1$ for the randomized midpoint version). 
\begin{lemma}\label{veli}
Let~$u=0$. Assume~$U\in C^2$ and there exists~$L>0$ such that~$\abs{D^2 U(\theta)u}\leq L\abs{u}$ for all~$u,\theta\in\R^d$. Assume~$Lh^2\leq 1/8$. It holds for any~$x,v,y,w\in\R^d$ that
\begin{equation*}
\abs{q_h(x,v) - q_h(y,w)}^2 + \abs{p_h(x,v) - p_h(y,w)}^2/L \leq (1+3\sqrt{L}h)(\abs{x-y}^2 + \abs{v-w}^2/L).
\end{equation*}
\end{lemma}
\begin{proof}
It holds that
\begin{align}
\abs{q_h(x,v) - q_h(y,w)} &\leq \abs{x-y + h(v-w) + h^2/2(\nabla U(x) - \nabla U(y))} \nonumber\\
&\leq (1+Lh^2/2)\abs{x-y} + \sqrt{L}h\abs{v-w}/\sqrt{L},\label{qq1}
\end{align}
so, by the Young's inequality~$2(1+Lh^2/2)\abs{x-y}\cdot\sqrt{L}h\abs{v-w}/\sqrt{L}\leq \sqrt{L}h\cdot ((1+Lh^2/2)\abs{x-y})^2 + (\sqrt{L}h)^{-1}(\sqrt{L}h\abs{v-w}/\sqrt{L})^2$, that
\begin{equation*}
\abs{q_h(x,v) - q_h(y,w)}^2\leq (1+Lh^2/2)^{2}(1+\sqrt{L}h)\abs{x-y}^2 + Lh^2(1+1/(\sqrt{L}h))\abs{v-w}^2/L.
\end{equation*}
Moreover,~\eqref{qq1} implies
\begin{align*}
&\abs{p_h(x,v) - p_h(y,w)}\\
&\quad= \abs{v-w + (h/2)(\nabla U(x) - \nabla U(y)) + (h/2) (\nabla U(q_h(x,v) ) - \nabla U( q_h(y,w)))}\\
&\quad\leq \abs{v-w} + (Lh/2)\abs{x-y} + (Lh/2)\abs{q_h(x,v) -  q_h(y,w)}\\
&\quad\leq \sqrt{L}(1+Lh^2/2)\abs{v-w}/\sqrt{L} + (Lh/2)(2+Lh^2/2)\abs{x-y},
\end{align*}
so that, using again Young's inequality,
\begin{align*}
\abs{p_h(x,v) - p_h(y,w)}^2/L &\leq (1+Lh^2/2)^2(1+\sqrt{L}h)\abs{v-w}^2/L\\
&\quad + (Lh^2/4)(2+Lh^2/2)^2(1+1/(\sqrt{L}h))\abs{x-y}^2.
\end{align*}
Gathering the bounds concludes the proof.
\end{proof}
In the rest of this section, we prove some contraction and nonexpansion results for each iterate of the Hamiltonian integrator~\eqref{qpdef}, which will be used in the proof of Lemma~\ref{nelem}. These are essentially Propositions~4.3,~4.10 in~\cite{chak2024r}, but since there are small mistakes (namely the left-hand sides therein do not correspond to a~$h/2$-step of the randomized midpoint integrator, because the factor~$1/4$ should be~$1/8$), the proofs below serve also as corrections for those results (note the step-size is~$h/2$ there, whereas it is~$h$ below, and instead of~$\alpha=1/8$ in Proposition~\ref{lbaR} below,~$\alpha=1/4$ is the analogous constant with the proof below following verbatim to obtain the contraction constant appearing in Proposition~4.3 in~\cite{chak2024r}). 

\begin{prop}\label{lbaR}
Let~$\bar{b}:\mathbb{R}^d\times\Xi\rightarrow\mathbb{R}^d$ be measurable such that~$\bar{b}(\cdot,\xi) \in C^1$ for all~$\xi\in\Xi$ and~$\nabla\!_{\theta}\bar{b}(\theta,\xi)$ is symmetric for all~$\theta\in\mathbb{R}^d$ and~$\xi\in\Xi$. 
Let~$L>0$ 
be such that
\begin{equation}\label{A1leg2}
\sup_{\theta\in\mathbb{R}^d}\lambda_{\max}(\nabla\!_{\theta}\bar{b}(\theta,\xi))\leq L\qquad\forall \xi\in\Xi.
\end{equation}
Let~$\bar{u}\in [0,h]$,~$\eta_0,\eta_1\in[0,1]$ satisfy~$0<\eta_0-\eta_1\leq1$ 
and let~$M\in\mathbb{R}^{2d\times 2d}$ be given 
by
\begin{equation}\label{cMdef}
M = \begin{pmatrix}
1&\frac{h}{\eta_0-\eta_1}\\
\frac{h}{\eta_0-\eta_1}& \frac{2h^2}{(\eta_0-\eta_1)^2}
\end{pmatrix}
\otimes I_d.
\end{equation}
Let~$\|v\|_M^2=v^{\top}Mv$ for any~$v\in\mathbb{R}^{2d}$. 
Let~$\alpha = 1/8$ and 
assume~$
Lh^2\leq \min(\alpha(\eta_0-\eta_1)^2,1/32)$. 
It holds that
\begin{align}
&\bigg\|\begin{pmatrix} 
I_d & 0\\
0 & \eta_1 I_d
\end{pmatrix}
\begin{pmatrix}
\bar{q}-\bar{q}' + h(\bar{p}-\bar{p}') - \frac{h^2}{2}(\bar{b}(\bar{q}+\bar{u}\bar{p},\xi) - \bar{b}(\bar{q}'+\bar{u}\bar{p}',\xi))\\
\bar{p} - \bar{p}' - h(\bar{b}(\bar{q}+\bar{u}\bar{p},\xi) - \bar{b}(\bar{q}'+\bar{u}\bar{p}',\xi))
\end{pmatrix} \bigg\|_M^2 \nonumber\\
&\quad\leq (1-c)\bigg\| \begin{pmatrix}
I_d & 0\\
0 & \eta_0 I_d
\end{pmatrix}
\begin{pmatrix}
\bar{q} - \bar{q}'\\
\bar{p} - \bar{p}'
\end{pmatrix}\bigg\|_M^2 \label{lbaeqR}
\end{align}
for all~$\xi\in\Xi$,~$\bar{q},\bar{q}',\bar{p},\bar{p}'\in\mathbb{R}^d$ satisfying
\begin{align}
\hat{m} &:= 2\inf_{\xi\in \Xi}\lambda_{\min}\bigg(\int_0^1 \nabla\!_{\theta} \bar{b}(\bar{q}'+\bar{u}\bar{p}' + t(\bar{q} +\bar{u}\bar{p} - \bar{q}'-\bar{u}\bar{p}'),\xi)dt\bigg)\geq0,\label{mdef}\\
c&:=\frac{(c_0\eta_0 + c_1\eta_1)\hat{m}h^2}{4(\eta_0-\eta_1)},\qquad c_0 = \frac{3}{2},\qquad c_1 = 1.\label{cdef}
\end{align}
If instead~$\alpha=1/4$, then
the same conclusion holds with~$c_0=3/2$ and~$c_1=0$.
\end{prop}
\begin{proof}
By the mean value theorem, it holds that
\begin{align}
&\begin{pmatrix}
\bar{q} - \bar{q}' + h(\bar{p} - \bar{p}')  - \frac{h^2}{2}(\bar{b}(\bar{q}+\bar{u}\bar{p},\xi) - \bar{b}(\bar{q}'+\bar{u}\bar{p}',\xi)) \\
\bar{p} - \bar{p}' - h(\bar{b}(\bar{q}+\bar{u}\bar{p},\xi) - \bar{b}(\bar{q}'+\bar{u}\bar{p}',\xi))
\end{pmatrix} \nonumber\\
&\quad= \begin{pmatrix}
I_d - \frac{h^2}{2}H & hI_d - \frac{h^2\bar{u}}{2}H\\
-hH & I_d - h\bar{u}H
\end{pmatrix}
\begin{pmatrix}
\bar{q} - \bar{q}'\\
\bar{p} - \bar{p}'
\end{pmatrix},\label{mvtR}
\end{align}
where 
\begin{equation}\label{mvthR}
H = \int_0^1 \nabla\!_{\theta} \bar{b}(\bar{q}'+\bar{u}\bar{p}' + t(\bar{q} +\bar{u}\bar{p} - \bar{q}'-\bar{u}\bar{p}'),\xi)dt.
\end{equation}
By~\eqref{A1leg2} and 
definition~\eqref{mdef} of~$\hat{m}$, the matrix~$H$ is positive semi-definite with eigenvalues between~$\hat{m}/2$ and~$L$ for all~$\xi\in\Xi$,~$\bar{q},\bar{q}',\bar{p},\bar{p}'\in\mathbb{R}^d$. 
To prove~\eqref{lbaeqR}, by Proposition~4.2 in~\cite{MR4748799} and the corresponding statement for \textit{semi}-positive definite matrices in the commuting case (which follows as an easy corollary), it suffices to show that the matrix
\begin{align}
\mathcal{H}&:=(1-c)
\begin{pmatrix}
I_d & 0\\
0 & \eta_0 I_d
\end{pmatrix}M
\begin{pmatrix}
I_d & 0\\
0 & \eta_0 I_d
\end{pmatrix} -
\begin{pmatrix}
I_d - \frac{h^2}{2}H & hI_d- \frac{h^2\bar{u}}{2}H\\
-hH& I_d - h\bar{u}H
\end{pmatrix}^{\!\!\!\top\!\!\!}
\nonumber\\
&\quad
\cdot
\begin{pmatrix}
I_d & 0\\
0 & \eta_1I_d
\end{pmatrix}
M
\begin{pmatrix}
I_d & 0\\
0 & \eta_1I_d
\end{pmatrix}
\begin{pmatrix}
I_d - \frac{h^2}{2}H & hI_d - \frac{h^2\bar{u}}{2}H\\
-hH & I_d - h\bar{u}H
\end{pmatrix}\label{chR}
\end{align}
is semi-positive definite for all~$\xi\in\Xi$,~$\bar{q},\bar{q}',\bar{p},\bar{p}'\in\mathbb{R}^d$. 
The matrix~$\mathcal{H}$ is semi-positive definite if and only if the square matrices~$A,B,C$ given by
\begin{equation}\label{habc}
\mathcal{H}=\begin{pmatrix}A&B\\B &C\end{pmatrix}
\end{equation}
are such that~$A$ and~$AC-B^2$ are semi-positive definite. Denoting
\begin{equation}\label{bhatdef}
\hat{b} = h/(\eta_0-\eta_1),
\end{equation}
it holds by direct calculation that
\begin{subequations}\label{abcdefsR} 
\begin{align}
A&= -cI_d + h\bigg(h + 2\hat{b}\eta_1\bigg)H-h^2\bigg(\frac{h^2}{4} + \hat{b}h\eta_1 + 2\hat{b}^2\eta_1^2\bigg) H^2,\label{AdefR}\\
B &= (1-c)\hat{b}\eta_0I_d - \bigg[\bigg(\bigg(I_d-\frac{h\bar{u}}{2}H\bigg) h + \bigg(I_d - h\bar{u}H\bigg)\eta_1 \hat{b}\bigg)\bigg(I_d - \frac{h^2}{2}H\bigg)\bigg]\nonumber\\
&\quad + \bigg[hH\bigg(\hat{b}\eta_1 h\bigg(I_d - \frac{h\bar{u}}{2}H\bigg) + 2\eta_1^2\hat{b}^2\bigg(I_d - h\bar{u}H\bigg)\bigg)\bigg],\label{BdefR}\\
C &= 2\hat{b}^2(1-c)\eta_0^2I_d 
-h^2\bigg(I_d - \frac{h\bar{u}}{2}H\bigg)^2 - 2h\hat{b}\eta_1\bigg(I_d - h\bar{u}H\bigg)\bigg(I_d - \frac{h\bar{u}}{2}H\bigg)\nonumber\\
&\quad - 2\eta_1^2\hat{b}^2(I_d - h\bar{u}H)^2.\label{CdefR}
\end{align}
\end{subequations}
For any eigenvalue~$\hat{m}/2\leq\lambda\leq L$ of~$H$, the matrices~$A,B,C$ admit the same corresponding eigenspace with eigenvalues~$A_{\lambda},B_{\lambda},C_{\lambda}$ given by~\eqref{abcdefsR} with~$I_d,H$ replaced by~$1,\lambda$. 
In order to show~$AC-B^2$ is positive semi-definite, we lower bound~$A_{\lambda},C_{\lambda}$ and upper bound~$B_{\lambda}^2$. 
For any~$\lambda$ we have by~\eqref{cdef},~\eqref{bhatdef} and the assumptions~$Lh^2\leq \min(\alpha(\eta_0-\eta_1)^2,1/32)$,~$\eta_1\leq 1$ that
\begin{align}
A_{\lambda} &= 2\hat{b}h\bigg(-\frac{(c_0\eta_0 + c_1\eta_1)\hat{m}}{8} + \frac{\lambda}{2}(\eta_0 + \eta_1) - \frac{h^2}{8}\lambda^2(\eta_0-\eta_1) - \frac{h^2}{2}\lambda^2\eta_1 \nonumber\\
&\quad- \hat{b}h\lambda^2\eta_1^2\bigg)\nonumber\\
&\geq 2\hat{b}h\lambda\bigg(-\frac{c_0\eta_0 + c_1\eta_1}{4}+\frac{\eta_0+\eta_1}{2} - \frac{(\eta_0-\eta_1)}{8\cdot32} - \frac{\eta_1}{64} - \frac{\eta_1\sqrt{\alpha}}{\sqrt{32}}\bigg)\nonumber\\
&\geq \hat{b}^2\lambda(\eta_0-\eta_1)\bigg(\bigg(\frac{127}{128}-\frac{c_0}{2}\bigg)\eta_0+\bigg(\frac{125}{128} - \frac{2\sqrt{\alpha}}{\sqrt{32}}-\frac{c_1}{2}\bigg)\eta_1\bigg).\label{Al}
\end{align}
For~$B_{\lambda}^2$, we first write~$B_{\lambda}=(1-c)\hat{b}\eta_0I_d-B_{\lambda}^{(1)}+B_{\lambda}^{(2)}$, where~$B_{\lambda}^{(1)},B_{\lambda}^{(2)}$ correspond to the two respective square brackets in~\eqref{BdefR}. We have by~\eqref{bhatdef} that
\begin{align*}
B_{\lambda}^{(1)} &= \bigg(h + \eta_1\hat{b} - \frac{h^2\bar{u}}{2}\lambda - h\bar{u}\lambda\eta_1\hat{b}\bigg) \bigg(1-\frac{h^2}{2}\lambda\bigg)\nonumber\\
&= \hat{b}\eta_0 - \frac{h^2}{2}\lambda\hat{b}\eta_0 - \bigg(\frac{h^2\bar{u}}{2}\lambda + h\bar{u}\lambda\eta_1\hat{b}\bigg)\bigg(1-\frac{h^2}{2}\lambda\bigg).
\end{align*}
We use~$\bar{u}\leq h$,~$\lambda\in[0,L]$ and~$\sqrt{L}h\leq \min(\sqrt{\alpha}(\eta_0-\eta_1),1/(4\sqrt{2}))$ in the sequel without further mention. We estimate terms in~$B_{\lambda}^{(1)}$ with
\begin{align*}
(h^2/2)\lambda\hat{b}\eta_0 &= \hat{b}^2\eta_0(\eta_0-\eta_1)\cdot\lambda(h/2)\leq \hat{b}^2\eta_0(\eta_0-\eta_1)\sqrt{\lambda}/(8\sqrt{2}),\\
(h^2\bar{u}/2)\lambda &\leq (h^3/2)\lambda= \hat{b}^2(\eta_0-\eta_1)^2\cdot (h/2)\lambda \leq  \hat{b}^2(\eta_0-\eta_1)^2\cdot \sqrt{\lambda}/(8\sqrt{2}),\\
h\bar{u}\lambda\eta_1\hat{b}&\leq \hat{b}^2(\eta_0-\eta_1)\eta_1h\lambda\leq \hat{b}^2(\eta_0-\eta_1)\eta_1\sqrt{\lambda}/(4\sqrt{2}),
\end{align*}
so that~$B_{\lambda}^{(1)}\geq \hat{b}\eta_0 
- \hat{b}^2\eta_0(\eta_0-\eta_1)\sqrt{\lambda}/(8\sqrt{2}) 
- \hat{b}^2(\eta_0-\eta_1)^2\cdot \sqrt{\lambda}/(8\sqrt{2}) 
- \hat{b}^2(\eta_0-\eta_1)\eta_1\sqrt{\lambda}/(4\sqrt{2})$. In the other direction, the second term in~$B_{\lambda}^{(1)}$ satisfies
\begin{equation*}
(h^2/2)\lambda\hat{b}\eta_0 = \hat{b}^3(\eta_0-\eta_1)^2\lambda\eta_0/2,
\end{equation*}
so that~$B_{\lambda}^{(1)}\leq \hat{b}\eta_0-\hat{b}^3(\eta_0-\eta_1)^2\lambda\eta_0/2$. 
Combining, we have
\begin{align}
B_{\lambda}^{(1)} 
&\in \big[\hat{b}\eta_0 - (\hat{b}^2\sqrt{\lambda\,}/(8\sqrt{2}))(\eta_0-\eta_1)\eta_0 - (\hat{b}^2\sqrt{\lambda}/(8\sqrt{2}))(\eta_0-\eta_1)^2 \nonumber\\
&\qquad - (\hat{b}^2\sqrt{\lambda\,}/(4\sqrt{2}))\eta_1(\eta_0-\eta_1),
\hat{b}\eta_0 - (\hat{b}^3\lambda/2)(\eta_0-\eta_1)^2\eta_0\big].
\label{bl1}
\end{align}
Moreover, the second square bracket in~\eqref{BdefR} (which we denoted~$B_{\lambda}^{(2)}$) satisfies
\begin{equation}\label{bl21}
B_{\lambda}^{(2)}\geq  h\lambda \bigg(\frac{63\hat{b}\eta_1h}{64} + \frac{31\eta_1^2\hat{b}^2}{16}\bigg) = \hat{b}^3\lambda(\eta_0-\eta_1)\bigg(\frac{63}{64}\eta_1(\eta_0-\eta_1) + \frac{31}{16}\eta_1^2\bigg)
\end{equation}
and
\begin{align}
B_{\lambda}^{(2)} &= h\lambda\bigg(\hat{b}\eta_1h\bigg(1-\frac{h\bar{u}}{2}\lambda\bigg) + 2\eta_1^2\hat{b}^2(1-h\bar{u}\lambda)\bigg)\nonumber\\
&\leq h\lambda(\hat{b}\eta_1h + 2\eta_1^2\hat{b}^2 )\nonumber\\
&= h\lambda(\hat{b}^2\eta_1(\eta_0-\eta_1) + 2\eta_1^2\hat{b}^2)\nonumber\\
&\leq (\hat{b}^2\sqrt{\lambda\alpha\,})(\eta_0-\eta_1)\eta_1(\eta_0+\eta_1).\label{bl22}
\end{align}
By the definitions~\eqref{cdef},~\eqref{bhatdef} for~$c$,~$\hat{b}$ and the inequality~$\hat{m}/2\leq \lambda$, the upper bound in equation~\eqref{bl1} and inequality~\eqref{bl21} imply
\begin{align}
B_{\lambda}&=(1-c)\hat{b}\eta_0 - B_{\lambda}^{(1)} + B_{\lambda}^{(2)}\nonumber\\
&\geq (\hat{b}^3\lambda/2)(\eta_0-\eta_1)( - (c_0\eta_0 + c_1\eta_1)\eta_0 + \eta_0(\eta_0-\eta_1) + (63/32)\eta_1(\eta_0-\eta_1) \nonumber\\
&\quad+ (31/8)\eta_1^2)\nonumber\\
&\geq (\hat{b}^3\lambda/2)(\eta_0-\eta_1)\eta_0((1 - c_0)\eta_0 + (31/32 - c_1)\eta_1)\nonumber\\
&\geq (\hat{b}^2\sqrt{\lambda}/(4\sqrt{2})) (\eta_0-\eta_1)
((1 - c_0)\eta_0 + (31/32 - c_1)\eta_1),\label{Bl}
\end{align}
where in the last line we have used~$\sqrt{\lambda}h\leq \sqrt{L}h\leq  (1/\sqrt{8})(\eta_0-\eta_1)$,~$\eta_0\leq 1$ and the definitions~\eqref{cdef} of~$c_0,c_1$. 
Moreover, by definition~\eqref{cdef} of~$c$, the lower bound in~\eqref{bl1},~\eqref{bl22}, we have
\begin{align}
B_{\lambda} &= (1-c)\hat{b}\eta_0 - B_{\lambda}^{(1)} + B_{\lambda}^{(2)}\nonumber\\
&\leq \hat{b}^2\sqrt{\lambda\,}(\eta_0 - \eta_1)\bigg(\frac{\eta_0}{8\sqrt{2}} + \frac{\eta_0-\eta_1}{8\sqrt{2}} + \frac{\eta_1}{4\sqrt{2}} + \sqrt{\alpha}\,\eta_1(\eta_0+\eta_1)
\bigg)\nonumber\\
&= \hat{b}^2\sqrt{\lambda\,}(\eta_0-\eta_1)\big((2\eta_0+\eta_1)/(8\sqrt{2}) + \sqrt{\alpha}\,\eta_1(\eta_0+\eta_1)\big),\label{Bl2}
\end{align}
so that together with~\eqref{Bl} (implying by definitions~\eqref{cdef} of~$c_0,c_1$ that it suffices to square the upper bound~\eqref{Bl2} to bound~$B_{\lambda}^2$), we have
\begin{equation}\label{bsb0}
B_{\lambda}^2\leq \hat{b}^4\lambda(\eta_0-\eta_1)^2 \big((2\eta_0+\eta_1)/(8\sqrt{2})+ \sqrt{\alpha}\eta_1(\eta_0+\eta_1)\big)^2.
\end{equation}
By~$\eta_0,\eta_1\leq 1$, the last square term on the right-hand side of~\eqref{bsb0} satisfies
\begin{align*}
&\big((2\eta_0+\eta_1)/(8\sqrt{2})+ \sqrt{\alpha}\eta_1(\eta_0+\eta_1)\big)^2\\
&\quad\leq \big(\eta_0/(4\sqrt{2}) + (1/(8\sqrt{2}) + 2\sqrt{\alpha})\eta_1\big)^2\\
&\quad= \beta_1\eta_0^2
+\beta_2\eta_0\eta_1 + \beta_3\eta_1^2,
\end{align*}
where
\begin{equation*}
\beta_1 = \frac{1}{32},\qquad\beta_2 = 
\frac{1}{32} + \sqrt{\frac{\alpha}{2}}, \qquad \beta_3 = 
\frac{1}{128}+\frac{1}{2}\sqrt{\frac{\alpha}{2}} + 4\alpha,
\end{equation*}
which implies
\begin{equation}\label{bsb}
B_{\lambda}^2 \leq \hat{b}^4\lambda(\eta_0-\eta_1)^2  ( \beta_1\eta_0^2+ \beta_2\eta_0\eta_1 + \beta_3\eta_1^2 ).
\end{equation}
On the other hand, 
by~$\eta_0\leq 1$, 
the corresponding eigenvalue~$C_{\lambda}$ for~$C$ defined in~\eqref{CdefR} satisfies
\begin{align}
C_{\lambda} &\geq 2\hat{b}^2(1-c)\eta_0^2 - (h^2 + 2\hat{b}h\eta_1 + 2\hat{b}^2\eta_1^2)\nonumber\\
&= 2\hat{b}^2\eta_0^2 - (c_0\eta_0+c_1\eta_1)\hat{m}\hat{b}^2\eta_0^2h^2/(2(\eta_0-\eta_1)) - \hat{b}^2((\eta_0-\eta_1)^2\nonumber\\
&\quad + 2\eta_1(\eta_0-\eta_1) + 2\eta_1^2)\nonumber\\
&\geq 2\hat{b}^2 \eta_0^2 - (c_0\eta_0+c_1\eta_1)L\hat{b}^2\eta_0^2h^2/(\eta_0-\eta_1) - \hat{b}^2(\eta_0^2+\eta_1^2)\nonumber\\
&\geq \hat{b}^2(\eta_0^2-\eta_1^2) - (c_0\eta_0+c_1\eta_1)\hat{b}^2\eta_0^2(\eta_0-\eta_1)\alpha\nonumber\\
&\geq \hat{b}^2(\eta_0-\eta_1)((1-c_0\alpha)\eta_0+(1-c_1\alpha)\eta_1).\label{csc}
\end{align}
Together with~\eqref{Al}, this implies
\begin{align}
A_{\lambda}C_{\lambda} &\geq \hat{b}^2\lambda(\eta_0-\eta_1)\bigg(\bigg(\frac{127}{128} - \frac{c_0}{2}\bigg)\eta_0 + \bigg(\frac{125}{128}-\frac{2\sqrt{\alpha}}{\sqrt{32}} - \frac{c_1}{2}\bigg)\eta_1\bigg)\nonumber\\
&\quad \cdot \hat{b}^2(\eta_0-\eta_1)((1-c_0\alpha)\eta_0+(1-c_1\alpha)\eta_1)\nonumber\\
&= \hat{b}^4 \lambda (\eta_0-\eta_1)^2(\bar{\beta}_1\eta_0^2  + \bar{\beta}_2\eta_0\eta_1+\bar{\beta}_3\eta_1^2).\label{ort}
\end{align}
where
\begin{align*}
\bar{\beta}_1 &= (127/128-c_0/2)(1-c_0\alpha),\\
\bar{\beta}_2 &= \bigg(\frac{127}{128}-\frac{c_0}{2} \bigg)(1-c_1\alpha)+ \bigg(\frac{125}{128} - \frac{2\sqrt{\alpha}}{\sqrt{32}} - \frac{c_1}{2}\bigg)(1-c_0\alpha),\\
\bar{\beta}_3 &= (125/128-2\sqrt{\alpha}/\sqrt{32} - c_1/2) (1-c_1\alpha).
\end{align*}
By substituting the values~$c_0 = 3/2$,~$c_1 = 1$,~$\alpha = 1/8$, we have~$\bar{\beta}_i\geq \beta_i$ for all~$i\in\{1,2\}$ and~$\bar{\beta}_1 + \bar{\beta}_2 + \bar{\beta}_3 \geq \beta_1 + \beta_2 + \beta_3$, 
so that~\eqref{ort},~\eqref{bsb} and~$\eta_0^2\geq \eta_1^2$,~$\eta_0\eta_1\geq\eta_1^2$ imply~$A_{\lambda}C_{\lambda}-B_{\lambda}^2\geq 0$. Substituting instead~$c_0=3/2$,~$c_1=0$ in the~$\alpha=1/4$ case yields the same conclusion.
\end{proof}

\begin{prop}\label{RvRth}
Let~$\bar{b}:\mathbb{R}^d\times\Xi\rightarrow\mathbb{R}^d$ be measurable such that~$\bar{b}(\cdot,\xi) \in C^1$ for all~$\xi\in\Xi$ and~$\nabla\!_{\theta}\bar{b}(\theta,\xi)$ is symmetric for all~$\theta\in\mathbb{R}^d$ and~$\xi\in\Xi$. 
Let~$L>0$ be such that~\eqref{A1leg2} holds and~$\inf_{\theta\in\mathbb{R}^d}\lambda_{\min}(\nabla\!_{\theta}\bar{b}(\theta,\xi))\geq -L$ for all~$\xi\in\Xi$. 
Let~$\eta_0,\eta_1\in[0,1]$ satisfy~$\eta_0>\eta_1$ and let~$\bar{u}\in[0,h]$. Let~$\bar{c}\in\mathbb{R}$ be given by
\begin{equation}\label{cbardef}
\bar{c}=10L\eta_0 h^2/(\eta_0-\eta_1)
\end{equation}
and let~$M$ be given by~\eqref{cMdef}. Assume~$Lh^2\leq \min((\eta_0-\eta_1)^2/4,1/32)$. Inequality~\eqref{lbaeqR} holds for all~$\xi\in\Xi$,~$\bar{q},\bar{q}',\bar{p},\bar{p}' \in\mathbb{R}^d$ with~$c$ replaced by~$-\bar{c}$. 
\end{prop}

\begin{proof}
Let~$\mathcal{H}$ be given by~\eqref{chR} with~$c$ replaced by~$-\bar{c}$ and~$H$ given by~\eqref{mvthR}. It suffices to show that~$\mathcal{H}$ is positive definite for any~$\bar{q},\bar{q}',\bar{p},\bar{p}'$. The matrix~$\mathcal{H}$ is positive definite if and only if the square matrices~$A,B,C$ given by~\eqref{habc} 
are such that~$A$ and~$AC-B^2$ are positive definite. Denoting again~$\hat{b}$ by~\eqref{bhatdef},
it holds by direct calculation that~$A,B,C$ are given by~\eqref{abcdefsR} all with~$c$ replaced by~$-\bar{c}$. For any eigenvalue~$-L\leq \lambda\leq L$ of~$H$, the matrices~$A,B,C$ admit the same corresponding eigenspace with eigenvalues~$A_{\lambda},B_{\lambda},C_{\lambda}$. By similar derivations to~\eqref{Al},~$A_{\lambda}$ satisfies
\begin{equation}\label{Akl}
A_{\lambda} \geq 2\hat{b}hL\bigg(5\eta_0 - \frac{\eta_0+\eta_1}{2} - \frac{(\eta_0-\eta_1)}{8\cdot32} - \frac{\eta_0+\eta_1}{64} - \frac{\eta_1}{8\sqrt{2}}\bigg)\geq 7\hat{b}hL\eta_0.
\end{equation}
Moreover, we have~$B_{\lambda} = (1+\bar{c})\hat{b}\eta_0 - B_{\lambda}^{(1)}+B_{\lambda}^{(2)}$, where~$B_{\lambda}^{(1)}$,~$B_{\lambda}^{(2)}$ are the 
two square brackets in~\eqref{BdefR}. Note that
\begin{equation*}
-B_{\lambda}^{(1)}= (h+\eta_1\hat{b})(-1+\lambda h^2/2) + \bar{B}_{\lambda}^{(3)}=-\hat{b}\eta_0+\hat{b}\eta_0\lambda h^2/2+\bar{B}_{\lambda}^{(3)},
\end{equation*}
where~$\bar{B}_{\lambda}^{(3)}$ satisfies
\begin{equation*}
\abs{\bar{B}_{\lambda}^{(3)}} \leq B_{\lambda}^{(3)} = \bigg(\frac{h^2\bar{u}}{2}L + h\bar{u}L\eta_1\hat{b}\bigg)\bigg(1+\frac{h^2}{2}L\bigg).
\end{equation*}
Therefore it holds that
\begin{equation}\label{Bnl}
\abs{B_{\lambda}} \leq \bar{c}\hat{b}\eta_0 + (h^2/2)L\hat{b}\eta_0 +B_{\lambda}^{(3)} + B_{\lambda}^{(4)},
\end{equation}
where
\begin{equation*}
B_{\lambda}^{(4)}= hL\bigg(\hat{b}\eta_1h\bigg(1+\frac{h\bar{u}}{2}L\bigg) + 2\eta_1^2\hat{b}^2\bigg(1+h\bar{u}L\bigg)\bigg).
\end{equation*}
By definitions~\eqref{cbardef} and~\eqref{bhatdef} of~$\bar{c}$ and~$\hat{b}$, the first two terms on the right-hand side of~\eqref{Bnl} satisfy
\begin{equation*}
\bar{c}\hat{b}\eta_0 = 10\hat{b}^2L\eta_0^2h,\qquad
(h^2/2)L\hat{b}\eta_0 = (h/2)L\hat{b}^2\eta_0(\eta_0-\eta_1) 
\end{equation*}
The third term on the right-hand side of~\eqref{Bnl} satisfies
\begin{align*}
B_{\lambda}^{(3)} &\leq (65/64)(h^3L/2 + h^2L\eta_1\hat{b})\\
&= (65/64) \hat{b}^2(\eta_0-\eta_1)L(h(\eta_0-\eta_1)/2 + h\eta_1)\\
&= (65/64) \hat{b}^2Lh(\eta_0-\eta_1)(\eta_0+\eta_1)/2.
\end{align*}
The last term on the right-hand side of~\eqref{Bnl} satisfies
\begin{align*}
B_{\lambda}^{(4)} &\leq hL\big((65/64)\hat{b}\eta_1h + (66/32)\eta_1^2\hat{b}^2\big)\\
&= \hat{b}^2hL\big((65/64)\eta_1(\eta_0-\eta_1) + 66\eta_1^2/32\big)\\
&=  \hat{b}^2Lh\eta_1((65/64)\eta_0+(67/64)\eta_1)\\
&\leq (66/64)\hat{b}^2Lh\eta_1(\eta_0+\eta_1).
\end{align*}
Gathering, we have
\begin{align}
\abs{B_{\lambda}} &\leq 2\hat{b}^2Lh\big(5\eta_0^2 + \eta_0(\eta_0-\eta_1)/4 + (67/64)(\eta_0+\eta_1)^2/4\big)\nonumber\\
&= 2\hat{b}^2Lh\big(5\eta_0^2 + (1+67/64)\eta_0^2/4 + (35/32)\eta_0\eta_1/4 + (67/64)\eta_1^2/4\big)\nonumber\\
&\leq 2(5+67/64)\hat{b}^2Lh\eta_0^2\label{bnf}
\end{align}
Finally, 
from the definition~\eqref{CdefR} of~$C_{\lambda}$ (with~$c,I_d,H$ replaced by~$-\bar{c},1,\lambda$ respectively), 
it holds that
\begin{align}
C_{\lambda} &\geq 2\hat{b}^2(1+\bar{c})\eta_0^2 - h^2(1+Lh^2/2)^2 - 2h\hat{b}\eta_1(1+Lh^2)(1+Lh^2/2) \nonumber\\
&\quad- 2\eta_1^2\hat{b}^2(1+Lh^2)^2\nonumber\\
&\geq 2\hat{b}^2\eta_0^2 + 20L\hat{b}^3\eta_0^3h - [h^2 + 2h\hat{b}\eta_1 + 2\eta_1^2\hat{b}^2](1+Lh^2)^2.\label{fkl}
\end{align}
Recall from the calculation in~\eqref{csc} that the square brackets on the right-hand side of~\eqref{fkl} may be written as~$\hat{b}^2(\eta_0^2+\eta_1^2)$. In addition, we have~$(1+Lh^2)^2=1+2Lh^2 + L^2h^4 \leq 1+65Lh^2/32$ by the assumption~$Lh^2\leq 1/32$. Therefore it holds that
\begin{align}
C_{\lambda} &\geq 2\hat{b}^2\eta_0^2 + 20L\hat{b}^3\eta_0^3h - \hat{b}^2(\eta_0^2+\eta_1^2)(1+65Lh^2/32)\nonumber\\
&= \hat{b}^2(\eta_0^2-\eta_1^2) + 20L\hat{b}^3\eta_0^3h - 65L\hat{b}^3h(\eta_0^2+\eta_1^2)(\eta_0-\eta_1)/32.\label{jfl}
\end{align}
For the last term on the right-hand side of~\eqref{jfl}, since~$(\eta_0^2+\eta_1^2)(\eta_0-\eta_1) = \eta_0^3(1+\eta_1^2/\eta_0^2)(1-\eta_1/\eta_0)$ and the function~$[0,1]\ni x\mapsto (1+x^2)(1-x)$ has codomain in~$[0,1]$, together with the assumption~$(\eta_0-\eta_1)^2\geq 4Lh^2$, inequality~\eqref{jfl} implies
\begin{align*}
C_{\lambda} &\geq \hat{b}^2(\eta_0-\eta_1)(\eta_0+\eta_1) + (17+31/32)L\hat{b}^3\eta_0^3h\\
&\geq \hat{b}^3h^{-1}(\eta_0-\eta_1)^2\eta_0 + (17+31/32)L\hat{b}^3\eta_0^3h\\
&\geq  (21+31/32)\hat{b}^3Lh\eta_0^3 .
\end{align*}
Combining with~\eqref{Akl} and~\eqref{bnf}, we have~$A_{\lambda}C_{\lambda}-B_{\lambda}^2>0$, which concludes the proof.
\end{proof}
\end{appendix}

\paragraph{Acknowledgments}
The authors thank Botond Szabo, Randolf Altmeyer and Shahar Mendelson for interesting discussions on the topic. Support from the European Research Council (ERC),
through StG “PrSc-HDBayLe” grant ID 101076564, is acknowledged.
\bibliography{bibliography}

\end{document}